%% file: Thesis.tex
\documentclass[a4paper,11pt,twoside]{ThesisStyle}

\include{formatAndDefs}

\usepackage{fancyvrb}
\usepackage{pifont}
\usepackage{multicol}
\usepackage{igo}
\usepackage{dsfont}

\begin{document}

\include{TitlePage}


\cleardoublepage
\vspace{1.5cm}
\hspace{-0.3cm}Reviewers

\vspace{0.7cm}

\hspace{-0.3cm} \emph{Prof.} \textsc{Concettina Guerra}\\
College of Computing, Georgia Tech, GA, USA\\
Dipartimento di Ingegneria Informatica, Universit di Padova, Italy

\vspace{1.1cm}

\hspace{-0.3cm} \emph{Prof.} \textsc{Paola Sebastiani}\\
Department of Biostatistics, Boston University School of Public Heath, MA, USA\\

\cleardoublepage
\begin{vcenterpage}
\noindent\rule[2pt]{\textwidth}{0.5pt}
\begin{center}
{\large\textbf{Algorithms for Internal Validation Clustering Measures in the Post Genomic Era\\}}
\end{center}
{\large\textbf{Abstract:}}
Inferring cluster structure in microarray datasets is a fundamental task for the so-called -omic sciences. It is also a fundamental question in Statistics, Data Analysis and Classification, in particular with regard to the prediction of the number of clusters in a dataset, usually established via internal validation measures. Despite the wealth of internal measures available in the literature, new ones have been recently proposed, some of them specifically for microarray data.\\
In this dissertation,  a study of internal validation measures is given, paying particular attention to the stability based ones.  Indeed, this class of  measures is particularly  prominent and promising in order to have a reliable  estimate of the correct number of clusters in a dataset. For this kind of measures,
a new general algorithmic paradigm is proposed here that highlights the richness of measures in this class and accounts for the ones already available in the literature.
Moreover, some of the most representative data-driven validation measures are also considered.
Extensive experiments on twelve benchmark microarray datasets are  performed, using both Hierarchical and K-means clustering algorithms, in order to assess both the intrinsic ability of a measure to predict the correct number of clusters in a dataset and its merit relative to the other measures. Particular attention is given both to precision and speed. The main result is a hierarchy of internal validation measures in terms of precision and speed, highlighting some of their merits and limitations not reported before in the literature. This hierarchy shows that the faster the measure, the less accurate it is. In order to reduce the time performance gap between the fastest and the most precise measures, the technique of designing fast approximation algorithms is systematically applied. The end result is a speed-up of many of the measures studied here that brings the gap between the fastest and the most precise within one order of magnitude in time, with no degradation in their prediction power. Prior to this work, the time gap was at least two orders of magnitude.

Finally, for the first time in the literature, a benchmarking  of Non-negative Matrix Factorization as a clustering algorithm on microarrays is provided. Such a benchmarking is novel and sheds further light on the use of Non-negative Matrix Factorization as a data mining tool in bioinformatics. Given the increasing popularity  of Non-negative Matrix Factorization for data mining in biological data, the results reported here seem to contribute to the proper use of the technique, being well aware of its limitations, in particular the extensive use of computational resources it needs.

{\large\textbf{Keywords:}}
Algorithms and Data Structures, Experimantal Analysis of Algorithms,
General Statistics, Analysis of Massive Datasets, Machine Learning,
Computational Biology, Bioinformatics.
\\
\noindent\rule[2pt]{\textwidth}{0.5pt}
\end{vcenterpage}

\dominitoc

\pagenumbering{roman}

 \cleardoublepage

\section*{Acknowledgments}

I owe a great deal of thanks to many people for making this thesis possible. First, I would like to express my gratitude for my advisor Prof. Raffaele Giancarlo, who has been leading and supporting me and my research to be fruitful in his patience.

\medskip

\noindent A huge thanks goes to Davide Scaturro for his collaboration in the first stage of my work. Without his skillful support my projects would not have
been possible.

\medskip

\noindent I truly appreciate to Dr. Giusi Castiglione that read drafts of the thesis and offered important
suggestions for improvement and more important her friendship.

\medskip

\noindent Thanks to my fellow PhD friends, in particular Fabio Bellavia, Marco Cipolla, Filippo Millonzi and Luca Pinello
for our broad-ranging discussions and for sharing the joys and worries of academic research.

\medskip

\noindent Furthermore I am deeply indebted to my colleagues at Department of Mathematics and Computer Science that have provided the environment for sharing their experiences about the problem issues involved as well as participated in stimulating team exercises developing solutions to the identified problems. I would specially like to thank Dr. Giosu\'{e} Lo Bosco and Prof. Marinella Sciortino for their extremely valuable experiences, support, insights and more important their friendship.

\medskip

Finally, I wish to express my gratitude to my family and friends who provided continuous understanding, patience, love and energy. In particular, I would like to express a heartfelt thanks to my parents and my girlfriend Marcella for their infinite support in my research endeavors.

\medskip

\medskip

Thanks to all of you.

\cleardoublepage

\section*{Originality Declaration}
This work contains no material which has been accepted for the award of any other degree or diploma in any university or other tertiary institution and, to the best of my knowledge and belief, contains no material previously published or written by another person, except where due reference has been made in the test. I give consent to this copy of my thesis, when deposited in the University Library, begin available for loan and photocopying.

\vspace{2cm}

Signed \ldots\ldots\ldots\ldots\ldots\ldots\ldots\ldots \hspace{150pt} February 2011

\tableofcontents
\listoffigures
\listoftables


\mainmatter
\include{Introduction}

\include{Chapter1}

\include{Chapter2}

\include{Chapter3}

\include{Chapter4}

\include{Chapter5}

\include{Chapter6}

\include{Chapter7}

\include{Conclusions}



\bibliographystyle{plain}
\bibliography{Thesis}

\end{document}

%% file: formatAndDefs.tex
\usepackage{color}

\newcommand{\ignore}[1]{}

\newcommand{\vir}[1]{``#1''}

\definecolor{green}{rgb}{0,0.70,0}

\newcommand{\codice}[1]{\mbox{\textsc{#1}}}

\usepackage{fancybox,Modificatopseudocode}
\usepackage{amsmath,amssymb,dsfont,amsthm}             
\usepackage[english]{babel}
\usepackage[latin1]{inputenc}
\usepackage[T1]{fontenc}
\usepackage[left=1.5in,right=1.3in,top=1.1in,bottom=1.1in,includefoot,includehead,headheight=13.6pt]{geometry}

\usepackage[nottoc, notlof, notlot]{tocbibind}
\usepackage{minitoc}
\usepackage{aecompl}

\usepackage[intoc]{nomencl}

\makenomenclature

\usepackage{ifpdf}

\ifpdf
  \usepackage[pdftex]{graphicx}
  \DeclareGraphicsExtensions{.jpg}
  \usepackage[a4paper,pagebackref,hyperindex=true]{hyperref}
\else
  \usepackage{graphicx}
  \DeclareGraphicsExtensions{.ps,.eps}
  \usepackage[a4paper,dvipdfm,pagebackref,hyperindex=true]{hyperref}
\fi

\graphicspath{{.}{images/}}

\usepackage{color}
\definecolor{linkcol}{rgb}{0,0,0.4}
\definecolor{citecol}{rgb}{0.5,0,0}

\hypersetup
{
bookmarksopen=true,
pdftitle="Algorithms for Internal Validation Clustering Measures in the Post
Genomic Era",
pdfauthor="Utro Filippo",
pdfsubject="Benchmarking and speedups of Internal Validation Clustering Measures", 
pdfmenubar=true, 
pdfhighlight=/O, 
colorlinks=true, 
pdfpagemode=None, 
pdfpagelayout=SinglePage, 
pdffitwindow=true, 
linkcolor=linkcol, 
citecolor=citecol, 
urlcolor=linkcol 
}

\def\argmax{\operatornamewithlimits{arg\,max}}

\usepackage{rotating}                    
\usepackage{fancyhdr}                    

\let\headruleORIG\headrule
\renewcommand{\headrule}{\color{black} \headruleORIG}

\usepackage{colortbl}
\usepackage{epsfig}
\arrayrulecolor{black}

\fancypagestyle{plain}{
  \fancyhead{}
  \fancyfoot{}
  
}

\usepackage{algorithm}
\usepackage[noend]{algorithmic}

\makeatletter

\def\cleardoublepage{\clearpage\if@twoside \ifodd\c@page\else%
  \hbox{}%
  \thispagestyle{empty}
  \newpage%
  \if@twocolumn\hbox{}\newpage\fi\fi\fi}

\makeatother

{%

\hrulefill
\vspace*{0.5cm}%
\end{minipage}
}

\let\minitocORIG\minitoc
\renewcommand{\minitoc}{\minitocORIG \vspace{1.5em}}

\usepackage{multirow}
\usepackage{slashbox}

{ \begin{list}%
	{$\bullet$}%
	{\setlength{\labelwidth}{25pt}%
	 \setlength{\leftmargin}{30pt}%
	 \setlength{\itemsep}{\parsep}}}%
{ \end{list} }

\newtheorem{theorem}{Theorem}
\newtheorem{lemma}[theorem]{Lemma}

\renewcommand{\epsilon}{\varepsilon}

\newenvironment{vcenterpage}
{\newpage\vspace*{\fill}\thispagestyle{empty}}
{\vspace*{\fill}}

%% file: TitlePage.tex
\begin{titlepage}
\begin{center}
\begin{figure}[htbp]
    \begin{center}
     \epsfig{file=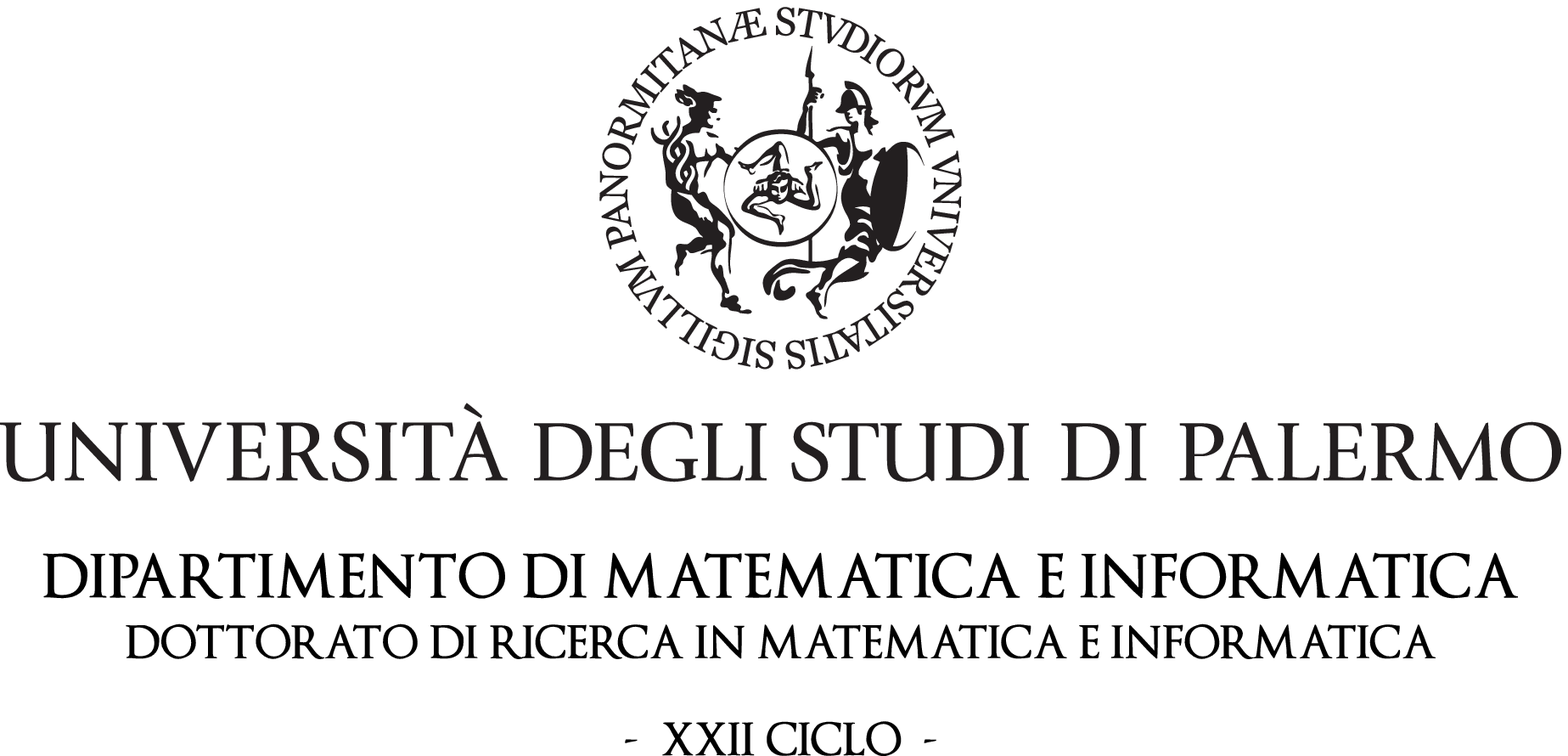,scale=0.7}
    \end{center}
  \end{figure}
\vspace*{0.2cm}
\vspace*{0.2cm}
\vspace*{0.3cm}
\vspace*{0.3cm}
\vspace*{0.3cm}
\vspace*{0.4cm}
\vspace*{0.8cm}
\noindent {\Huge \textbf{\textsc{Algorithms for Internal Validation Clustering Measures in the Post
Genomic Era}}} \\
\vspace*{6.0cm}

\begin{minipage}[t]{6cm}
\center Author\\ \vspace{-0.2cm}
\center \textsc{Filippo Utro}\\
\end{minipage}
\hfill
\begin{minipage}[t]{6cm}
\center Coordinator\\ \vspace{-0.2cm}
\center \emph{Prof. }\textsc{Camillo Trapani}\\ 
\vskip 0.3cm
\end{minipage}
\vspace{1cm}
\begin{center}
\center Thesis Advisor\\ \vspace{-0.2cm}
\center \emph{Prof. }\textsc{Raffaele Giancarlo}\\ \null\vspace{-0.18cm}\hrulefill\\
Settore Scientifico Disciplinare INF/01
\end{center}

\end{center}
\end{titlepage}
\sloppy

\titlepage

%% file: Introduction.tex
\chapter*{Introduction}
\label{chap:intro}\markboth{Introduction}{Contributions and Thesis Outline}
\addcontentsline{toc}{chapter}{Introduction} \pagenumbering{arabic}

In the past 15 years, a paradigm shift in Life Sciences research has taken place, thanks to the availability of genomic and proteomic data on an unprecedented scale. Such a revolutionary change has posed new challenges to Mathematics, Statistics and Computer Science, since the conceptual tools proper of those three disciplines are fundamental for the study of biological questions via computational tools, on a genomic scale. In the following, by way of example, the converging views of Hood and Galas, two authorities in the Life Sciences, and Knuth, an authority in Computer Science, are summarized.

In the Fifty Years Commemorative Issue of Nature on the discovery of DNA, Hood and Galas~\cite{HoodGals} conclude their contribution by clearly stating that a mathematical notion of \vir{biological information} readily usable for the development of computational tools for -omic investigations is lacking and that such a notion would be a fundamental contribution of the Exact Sciences to the Life Sciences. Moreover, Knuth~\cite{knuth} assert that the contributions given by Mathematics, Statistics and Computer Science to Molecular Biology would have been unpredictable, both in depth and breadth, only at the end of the 90s. Yet, those contributions are only a small fraction of the challenges faced by the Computer and Information Sciences in strategic junctions of this domain. Therefore, the development of a new core area of Mathematics and Computer Science is taking place. Although the topics discussed in this dissertation are classical, they will be developed having well in mind such an important new direction. In particular, this dissertation focuses on various aspects of clustering when used in conjunction with microarray data.

\emph{Microarrays } are a useful and, by now, well established technology in genomic investigation. Indeed,  experiments based on that technology are increasingly being carried out in biological and medical research to address a wide range of problems, including the classification of tumors~\cite{Aliz00,AlonU1999,CLEST,golub99,Perou1999,Pollack99,Ross00}, where a reliable and precise classification is essential for successful diagnosis and treatment. By allowing the monitoring of gene expression levels on a genomic scale, microarray experiments may lead to a more complete understanding of the molecular variations among tumors and hence to a finer and more reliable classification. An important statistical problem associated with tumor classification is the identification of new tumor classes using gene expression profiles, which has revived interest in cluster analysis. However, the novelty, noisiness and high dimensionality of microarray data provide new challenges even to a classic and well studied area such as clustering. More in general, new methodological and computational challenges are proposed daily~\cite{Shmulevichbook,CLEST,Kikuchi2003,McAdams-1995}.  As a results, there has been a \vir{malthusian growth} of new statistical and computational methods for genomic analysis. Unfortunately, many papers for -omic research describe development or application of statistical methods for microarray data that are questionable~\cite{Mehta2004}. In view of this latter peril, in this dissertation an effort has been made to use  a methodology that is sound and coherent for the  experimental validation of the computational methods proposed here.

In microarray data analysis there are two essential aspects of clustering: finding a \vir{good} partition of the datasets and estimating the number of clusters, if any, in a dataset.
The former problem is usually solved by  the use of a clustering algorithm. In the Literature, a large number of clustering algorithms has been proposed
and many of these have been applied to genomic data~\cite{Nature_Clustering}, the most famous are: K-means~\cite{JainDubes},
fuzzy c-means, self-organizing maps~\cite{Perrin2003,Sachs02,Shmulevich2002}, hierarchical clustering~\cite{JainDubes},
and model-based clustering~\cite{Silvescu01temporalboolean,Smolen2000}. Some of those
studies concentrate both on the ability of an algorithm to obtain a
high quality partition of the data and on its performance in terms
of computational resources, mainly CPU time (see~\cite{Clusterin_high,Kmeanscoresets,Kraus10,Optimal_Hier} and references
therein).

However, the most fundamental issue is the latter problem, i.e., the determination of the number of clusters.
Despite the vast amount of knowledge available in the
general data mining literature, e.g.,~\cite{Breckenridge89,
Dudoit2003,Gordon1, Gordon,Tibshrbook,JainDubes,KR90,Kerr00bootstrappingcluster}, gene
expression data provide unique challenges, in particular with
respect to internal validation indices. Indeed, they  must predict
how many clusters are really present in a dataset, an already
difficult task, made even worse by the fact that the estimation must
be sensible enough to capture the inherent biological structure of
functionally related genes. Despite their potentially important
role, both the use of classic internal validation  indices  and the
design of new ones, specific for microarray data, do not seem to
have great prominence in bioinformatics, where attention is mostly
given to clustering algorithms. The excellent survey by Handl et al.~\cite{Handl05} is a big step forward in making the study of those
techniques a central part of both research and practice in
bioinformatics, since it provides both a technical presentation as
well as valuable general guidelines about their use for post-genomic
data analysis. Although much remains to be done, it is,
nevertheless, an initial step.

For instance, in the general data mining literature, there are
several studies, e.g.,~\cite{MC85}, aimed at establishing  the
intrinsic, as well as the relative, merit of an index. To this end,
the two relevant questions are:
\begin{itemize}
\item[(i)] What is the precision of an index, i.e., its ability to predict the correct number of clusters in a dataset? That is usually established by comparing the number of clusters predicted by  the index against the number of clusters in the true solution of several datasets, the true solution being a partition of the dataset in classes that can be trusted to be correct, i.e.,  distinct groups of functionally related genes.

\item[(ii)] Among a collection of indices, which is more accurate, less algorithm dependent, etc.,?.
\end{itemize}

\noindent Precision versus the  use of computational resources, primarily execution time, would be an important discriminating factor.

\section*{Contributions and Thesis Outline}
\addcontentsline{toc}{section}{Contributions and Thesis Outline}

From the previous brief description of the state of the art it is evident that, for the special case of microarray data, the
experimental assessment of the \vir{fitness} of a measure has been
rather ad hoc and studies in that area provide only partial and
implicit comparison among measures. Moreover, contrary to research
in the clustering literature, the performance of validation methods
in terms of computational resources, again mainly CPU time, is
hardly assessed both in absolute and relative terms. This dissertation is an attempt to tackle the stated limitations of the state of the art in a homogeneous way. It is organized as follows:

\begin{itemize}
\item[$\bullet$] Chapter~\ref{chap:1} provides the background information relevant to this thesis. Indeed, a formal definition of the clustering problem and the notation used in this dissertation is given. Moreover, the two main classes of clustering algorithms proposed in the Literature, some methods for the assessment and evaluation of cluster quality, and data generation/perturbation methods, are also detailed.

\item[$\bullet$] Chapter~\ref{chap:ValidationMeasuers} provides a presentation of some basic validation techniques. In detail, external and internal indices are outlined. In particular, three external indices that assess the agreement between two partitions are presented. Moreover, four internal measures useful to estimate the correct number of clusters present, if any, in a dataset are detailed. They are based on: compactness, hypothesis testing in statistics and jackknife techniques. The described measures have been selected since they are particularly prominent in the clustering literature. Moreover, they are all characterized by the fact that, for their prediction, they make use of nothing more than the dataset available, i.e., they are all data-driven. It is worth pointing out that another class of measure is based on Bayesian Model and it is an aproach both to cluster analysis and the estimation of the \vir{best} number of clusters in a dataset. In order to keep this thesis focus it will not be discussed in this dissertation. The interested reader is referred to~\cite{BayesianCluster} and references therein for an in depth treatment of the relevant topics regarding Bayesian Model Based Clustering.

\item[$\bullet$] Chapter~\ref{chap:Stability} provides one of the main  topics of this thesis. Indeed, for the first time in the Literature, a general algorithmic paradigm of stability internal validation measure is introduced.  It can be seen as a generalization of earlier works by Breckenridge and Valentini. Moreover, it is shown that each of the known stability based measures is an instance of such a novel paradigm. Surprisingly, also the Gap Statistics falls within the new paradigm. Moreover, from this general algorithmic paradigm it is simple to design new stability internal measure combining the building blocks of the measures detailed. As will be evident in this dissertation,  this particular category of internal validation measure obtains excellent results in terms of estimation of number of clusters in a dataset. In fact, prior to this study, they were perceived as a most promising avenue of research in the development of internal validation measures. Therefore, the identification of an algorithmic paradigm describing the entire class seems to be a substantial methodological contribution to that area.

\item[$\bullet$] Chapter~\ref{chap:NMF} provides a formal description of one of the methodologies that has gained prominence in the data analysis literature: Non-negative Matrix Factorization. In particular, of relevance for this thesis, is the use  of Non-negative Matrix Factorization as a clustering algorithm.

\item[$\bullet$] Chapter~\ref{chap:5} describes the experimental methodology used in this thesis. The experimental setup include, to the best of our knowledge, the most complete representative collection of datasets used in the Literature. In particular, this collection is composed of nine microarray datasets that seem to be a \emph{de facto} standard in the specialistic literature and three artificial dataset generated in order to evaluate specific aspects of the clustering methodology, when used on microarray data.
    \noindent Moreover, this chapter provides an exhaustive study of the three external indices detailed in this dissertation. Furthermore, a benchmarking of Non-negative Matrix Factorization as  a clustering algorithm on microarrays data. Such a benchmarking is novel and sheds further light on the use of Non-negative Matrix Factorization as a data mining tool in bioinformatics.

\item[$\bullet$] Chapter~\ref{chap:6} provides a benchmarking of some of the  most prominent  internal validation measures in order to  establish the intrinsic, as well as the relative, merit of a measure taking into account both its predictive power and its computational demand. This benchmarking shows that there is a natural hierarchy, in terms of the trade-off time/precision, for the measures taken into account. That is, the faster the measure, the less accurate it is. Although this study has been published only recently,  it is already being referenced, even \vir{back to back} to fundamental studies on clustering such as  the papers by D'haeseleer~\cite{Nature_Clustering} and Handl et al.~\cite{Handl05}.

\item[$\bullet$] Chapter~\ref{chap:7}, based on the benchmarking in Chapter~\ref{chap:6}, investigates systematically the application of the idea of algorithmic approximation to internal validation measures in order to obtain speedups. That is, the development of new methods that closely track the behavior of existing methods, but that are substantially faster in time. Such a systematic study seems to be quite novel. In this chapter, several approximation algorithms  and two general approximation schemes are proposed. In particular, an approximation of a \vir{star} of the area as Gap Statistics is proposed and it is shown that it grants clearly superior results. Indeed, depending on the dataset, it is from two to three orders of magnitude faster than the Gap Statistics,  with a better prediction of the correct number of clusters in a datasets. Finally, an approximation of Consensus Clustering it is also proposed. In terms of the trade-off time/precision, it turns out to be the best among all measures studied in this dissertation. Even more remarkably, it reduces the time performance gap between the fastest measures and the most precise to one order of magnitude. Prior to this work, the gap was at least two orders of magnitude.

\item[$\bullet$] Chapter~\ref{chap:8} offers some conclusions as well as some future lines of research for further development of the ideas presented in this dissertation.

\end{itemize}

%% file: Chapter1.tex
\chapter{Background on Cluster Analysis}
\label{chap:1}



In this chapter, some fundamental aspects of cluster analysis are presented. In particular, the two main classes of clustering algorithms proposed in the literature are described. Moreover, methods for the assessment and evaluation of cluster quality are discussed as well as data generation/perturbation methods which can be applied to the former.\\

\section{Basic Mathematical Problem Formulations}\label{sec:basic}
Consider a  set of $n$ items $\Sigma=\{\sigma_1,\ldots, \sigma_n\}$, where  $\sigma_i$, with  $1 \leq i \leq n$, is defined by $m$ numeric values, referred to as features or
conditions. That is, each $\sigma_i$ is an element in a
$m$-dimensional space. Let $\mathcal{C}_k=\{c_{1},c_{2}, \ldots, c_{k}\}$ be a \emph{partition} of $\Sigma$, i.e.,
a set of subsets of $\Sigma$ such that $\bigcup_{i=1}^{k}c_{i}=\Sigma$ and $c_{i}\cap c_{j}=\emptyset$ for $1\leq i\neq j \leq k$. Each subset $c_{i}$, where $1\leq i\leq k$, is referred to as a \emph{cluster}, and $\mathcal{C}_k$ is referred to as a \emph{clustering solution}. The aim of cluster analysis is to determine a partition of $\Sigma$ according to a similarity/distance $S$, which is referred to as \emph{similarity/distance metric}. It is defined on the elements in $\Sigma$. In particular, one wants that items in the same cluster have \vir{maximal similarity}, while items in different clusters are \vir{dissimilar}.
For instance, an example comes from molecular data analysis~\cite{Speed03}, in which a set of genes are the items and the features are the expression level measurements in $m$ different experimental conditions or in $m$ different time periods. Clustering would highlight groups of genes that are, for instance, functionally correlated or that have the same response to medical treatments.

Usually, the set $\Sigma$ containing the items to be clustered is represented in one of two different ways: (1) a data matrix $D$, of size $n\times m$, in which the rows represent the items and the columns represent the condition values; (2) a similarity/dissimilarity matrix $S$, of size $n\times n$, in which each entry $S_{i,j}$, with $1\leq i\neq j \leq n$, is the value of similarity/dissimilarity of the pair $(i,j)$ of items. Specifically, the value of $S_{i,j}$ can be computed using rows $i$ and $j$ of  $D$. Hence, $S$ can be derived from $D$, but not viceversa.  The specification and formalization of a similarity metric, via mathematical functions, depends heavily on the application domain and it is one of the key steps in clustering, in particular in the case of microarray data. The state of the art,  as well as some relevant progress in the identification of good distance functions for microarrays, is presented in~\cite{Priness07}.


\section{Clustering Algorithms}
Usually, the partition of the items in $\Sigma$ is
accomplished by means of a clustering algorithm $A$. In this
chapter, only the class of clustering algorithms
that takes as input $D$ and an integer $k$ and return a partition
$\mathcal{C}_k$ of $\Sigma$ into $k$ subsets is taken in account. There is a rich literature about clustering algorithms, and there are many different classifications of them~\cite{JainDubes,KR90}. A survey of classic as well as more innovative clustering algorithms, specifically designed for microarray data, is given in~\cite{Roded4}. A classical classification is \emph{hierarchical} versus \emph{partitional} algorithms.
The hierarchical clustering algorithms produce a partition by a nested sequence of partitions and they are outlined in Section~\ref{sec:Hierarchical}, where three of them are detailed~\cite{JainDubes}, referred to as Average Link (Hier-A for short), Complete Link (Hier-C for short), and Single Link (Hier-S for short).
The partitional clustering algorithms directly decompose the dataset into a partition $\mathcal{C}_k$. One of the most prominent in that class, i.e., K-means~\cite{JainDubes}, is detailed in Section \ref{sec:Partional}.

\ignore{
 Moreover, we describe one  \ignore{We  both in the version that starts the clustering from a random partition of the data and in the version where it takes as part of the input an initial partition produced by one of the chosen hierarchical methods. The acronyms of those versions are K-means-R, K-means-A, K-means-C and K-means-S, respectively.}
}
\subsection{Hierarchical Algorithms}\label{sec:Hierarchical}

In hierarchical clustering, items are related by a tree\footnote{One assumes that the reader is familiar with that general concepts of graph theory such as trees and planar graphs. The reader is referred to standard references for an appropriate background~\cite{Bondy}.}  structure, referred to as \emph{dendogram} such that similar items are at the leaves of the same subtree. Each internal node represents a cluster and the leaves correspond to the items.
These algorithms can be either agglomerative (\vir{bottom-up}), in which one starts at the leaves and successively merges clusters together; or divisive (\vir{top-down}) in which one starts at the root and recursively splits the clusters. For instance, in Fig.~\ref{fig:dendogram}, a dendogram obtained by a run of Hier-A on a dataset of 24 elements according to the Euclidean distance is given. The partition $\mathcal{C}_4$ is marked.  In what follows, only the agglomerative approach is outlined. In particular, the Hier-A, Hier-C and Hier-S clustering algorithms are considered in this thesis. Each of them is an instance of the \codice{Hierarchical} paradigm described in Fig.~\ref{alg:Hierarchical} as a procedure. The interested reader is referred to~\cite{Ever,Hartigan,JainDubes} for an in depth treatment of the hierarchical clustering algorithms.

\begin{figure}[ht]
\centering
\epsfig{file=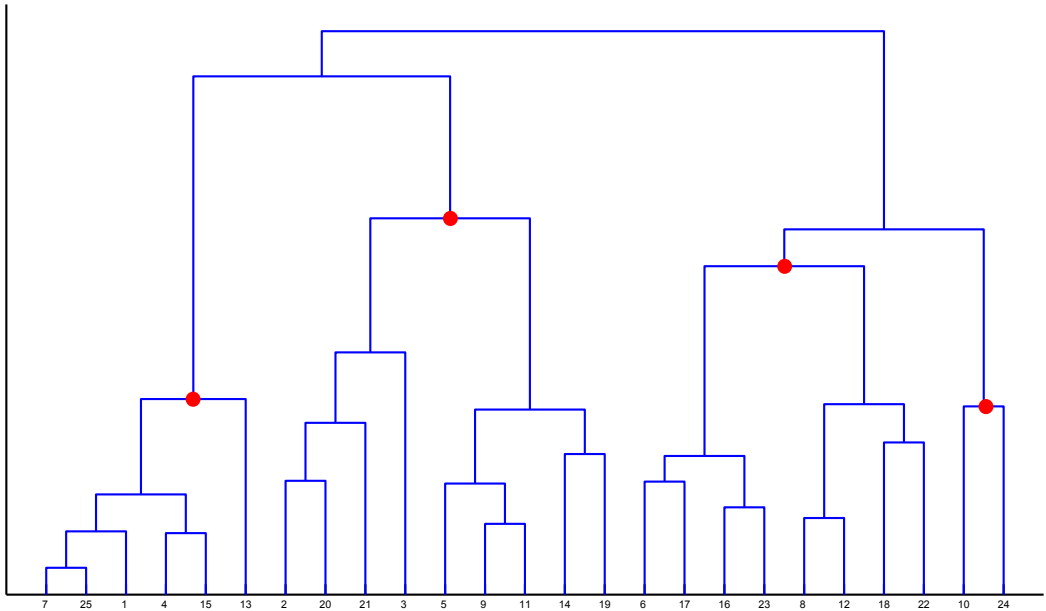,scale=0.4}
\caption{Example of a dendogram. The nodes in red indicate the partition $\mathcal{C}_4$ where the leafs (e.g. items) in the same subtree are in the same cluster.}\label{fig:dendogram}
\end{figure}

\codice{Hierarchical} takes as input a data matrix $D$ and the desired number $k$ of clusters.
\begin{figure}
\newlength{\mylength}
\[
\setlength{\fboxsep}{11pt}
\setlength{\mylength}{\linewidth}
\addtolength{\mylength}{-2\fboxsep}
\addtolength{\mylength}{-2\fboxrule}
\ovalbox{
\parbox{\mylength}{
\setlength{\abovedisplayskip}{0pt}
\setlength{\belowdisplayskip}{0pt}

\begin{pseudocode}{Hierarchical}{D, k}
1. \mbox{ Compute a similarity matrix } S\\
2. \mbox{ Initialize each item as a cluster}\\

\FOR i\GETS 1 \TO n-k \DO \\
\BEGIN

3. \mbox{ }\mbox{ } \mbox{ Select the two most \vir{similar} clusters}\\
4. \mbox{ }\mbox{ } \mbox{ Merge them}.\\

\END\\
\RETURN{clustering\_solution}
\end{pseudocode}

}
}
\]
\caption{The \codice{Hierarchical} procedure.}\label{alg:Hierarchical}
\end{figure}
In step 1, a similarity matrix $S$ is computed from the data matrix $D$. In step 2, each item is considered as a cluster, i.e., this step corresponds to the leaf level of the dendogram. One single iteration of the {\bf for} loop is discussed, which is repeated until $n-k$ steps are performed, i.e., until $k$ clusters are obtained. In step 3, the  most \vir{similar} clusters $c_{i}$ and $c_{j}$ are selected.
The \vir{similarity} between two clusters is measured by a distance function $Dist$. Therefore, if $c_{i}$ and $c_{j}$ are the most similar $Dist(c_{i},c_{j})$ is minimum.
In step 4, $c_{i}$ and $c_{j}$ are merged. Finally, the clustering solution obtained is given as output.

The three hierarchical clustering algorithms differ one from the other only for the distance function used to select the two clusters in step 3.

\noindent Hier-A computes the distance between two clusters $c_{i}$ and $c_{j}$ as the average of the values of the similarity metric between the items of $c_{i}$ and $c_{j}$, respectively. Formally:
$$
Dist(c_{i},c_{j})=\frac{1}{|c_{i}||c_{j}|}\sum\limits_{x \in c_{i}}\sum\limits_{y\in c_{j}}S_{x,\,y}
$$
\noindent where $|c_{i}|$ and $|c_{j}|$ are the sizes of the clusters $c_{i}$ and $c_{j}$, respectively.

\noindent Hier-C computes the distance between two clusters as the maximal item-to-item similarity matric value. Formally:
$$
Dist(c_{i},c_{j})=\underset{x\in c_{i},\, y\in c_{j}}{\mathop{\max }}\,S_{x,\,y}.
$$

\noindent Finally,  in Hier-S, the distance between two clusters is computed as the minimal item-to-item similarity matric value Formally:
$$
Dist(c_{i},c_{j})=\underset{x\in c_{i},\, y\in c_{j}}{\mathop{\min }}\,S_{x,\,y}.
$$

\subsection{Partitional Algorithms}\label{sec:Partional}

The goal of partitional clustering algorithms is to decompose directly the dataset into a set of disjoint clusters, obtaining a partition which should optimize a given objective function. Intuitively, the criteria one follows are to minimize the dissimilarity between items in the same cluster and to maximize the dissimilarity between items of different clusters.
Therefore, clustering can be seen as an optimization problem, where one tries to minimize/maximize an objective function. Unfortunately, it can be shown to be NP-Hard~\cite{GareyJohnson}. For completeness, the number of possible partitions can be computed via the Stirling numbers of the second kind~\cite{JainDubes}:
$$
\frac{1}{k!}\sum\limits_{i=0}^{k}{{{(-1)}^{k-i}}\left( \begin{matrix}
   k  \\
   i  \\
\end{matrix} \right){{i}^{n}}}.
$$
Even for small $k$ and $n$, it is such a substantially large number to discourage exhaustive search. Therefore, existing algorithms provide different heuristics to solve the various versions of clustering as an optimization problem~\cite{Hansen}.
Here, K-means~\cite{MacQueen} is detailed. The interested reader is referred to~\cite{JainDubes,KR90} for an in depth treatment of the partitional clustering algorithm.

In K-means, a cluster is represented by its \vir{center}, which is referred to as \emph{centroid}.
The aim of the algorithm  is to minimize a squared error function, i.e., an indicator of the distance of the $n$ data points from their respective cluster centers. Formally:
$$
\sum\limits_{j=1}^{k}\sum\limits_{x\in c_{j}}\|x-\overline{c_j}\|^{2},
$$

\noindent where $\overline{c_j}$ is the centroid of cluster $c_j$.

The procedure summarizing K-means is reported in Fig.~\ref{alg:Kmeans}. In addition to the input parameters of the \codice{Hierarchical} procedure, \codice{K-means} takes also the  maximum number of iterations allowed to the algorithm, which is referred to as $Niter$.
\begin{figure}
\[
\setlength{\fboxsep}{12pt}
\setlength{\mylength}{\linewidth}
\addtolength{\mylength}{-2\fboxsep}
\addtolength{\mylength}{-2\fboxrule}
\ovalbox{
\parbox{\mylength}{
\setlength{\abovedisplayskip}{0pt}
\setlength{\belowdisplayskip}{0pt}

\begin{pseudocode}{K-means}{D, k, Niter}
1. \mbox{ Extract, at random, } k \mbox{ items from } D \mbox{ and use them as centroids}\\
2. \mbox{ Assign each item of } D \mbox{ to the centroid with minimum distance}\\
3. \mbox{ }step \GETS 0\\
\WHILE (H \OR step>Niter) \DO \\
\BEGIN
4. \mbox{ }\mbox{ }step\GETS step + 1\\
    \mbox{ }\mbox{ } \mbox{ }\mbox{ }\FOREACH i \in D \DO\\
    \mbox{ }\mbox{ }    \mbox{ }\mbox{ }    \BEGIN
5.  \mbox{ }\mbox{ }    \mbox{ }\mbox{ }\mbox{ Compute the distance between } i \mbox{ and each centroid} \\
6.  \mbox{ }\mbox{ }    \mbox{ }\mbox{ }\mbox{ Assign } i \mbox{ to the centroid with minimum distance}\\
7.  \mbox{ }\mbox{ }\mbox{ }\mbox{ }\mbox{ Compute the new } k \mbox{ centroids}\\
    \mbox{ }\mbox{ }    \mbox{ }\mbox{ }\END\\
\END\\
\RETURN{clustering\_solution}
\end{pseudocode}

}
}
\]
\caption{The \codice{K-means} procedure.}\label{alg:Kmeans}
\end{figure}

In steps 1 and 2, the dataset $D$ is partitioned into $k$ clusters, by a random selection of the $k$ centroids.
One single iteration of the {\bf while} loop is discussed, which is repeated until at least one of the following two conditions is satisfied: (i) the clustering solution has not changed or (ii) the maximum number of iterations has been reached. The former condition is indicates in the procedure with $H$, while the latter condition prevents endless oscillations~\cite{JainDubes}.
The main part of the algorithm consists of steps 5-7,  where each item in $D$ is assigned to the centroids with minimum distance. Indeed, for each $i\in D$, with $1\leq i\leq n$,  the distance between $i$ and each centroid is computed in step 5. In step 6, item $i$ is assigned to the cluster whose distance from $i$ is minimal. Finally, the clustering solution is given as output. Figs.~\ref{fig:KmeansExample}(b)-(d) report an example of successive iterations of K-means, where the points in red are the two centroids. The algorithm is applied to the  dataset, with two well-separated clusters, reported in Fig.~\ref{fig:KmeansExample}(a).\\
\noindent Moreover, it is worth pointing out that a clustering solution obtained by another clustering algorithm (e.g. hierarchical) can be used as an initial clustering solution, instead of the random partition generated in steps 1 and 2, by the \codice{K-means} procedure. Accordingly, when \codice{K-means} starts the clustering from a random partition it is referred to as K-means-R, while when it starts from an initial partition produced by one of the chosen hierarchical methods it is referred to as K-means-A, K-means-C and K-means-S, respectively.

\begin{figure}[ht]
\begin{center}
\epsfig{file=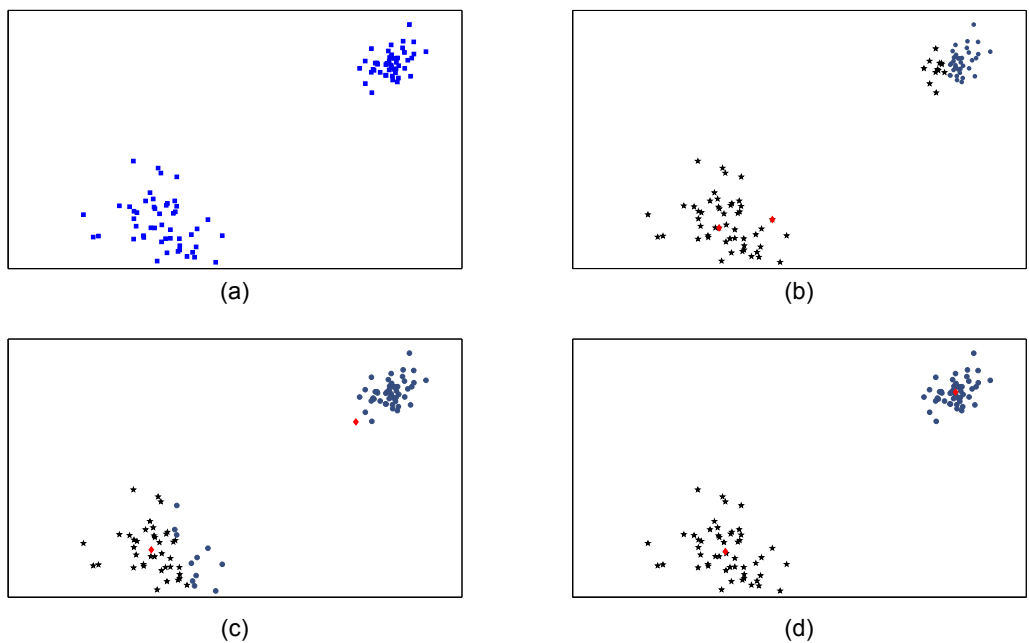,scale=0.6}
\end{center}
\caption{(a) Dataset. (b) Cluster membership after the first step. (c) Cluster membership after the second step. (d)  Cluster membership after the last step.}\label{fig:KmeansExample}
\end{figure}

\section{Assessment of Cluster Quality: Main Problems Statement}\label{sec:asseClusterQuality}

In bioinformatics, a sensible biological
question would be, for instance, to find out how many functional
groups of genes are present in a dataset. Since the presence of \vir{statistically
significant patterns} in the data is usually an indication of their
biological relevance~\cite{zscore}, it makes sense to ask whether a
division of the items into groups is statistically significant. In what follows, the three problem statements in which that question can be cast~\cite{JainDubes} are detailed.

Let $C_j$ be a reference classification for $\Sigma$ consisting of $j$
classes. That is, $C_j$ may either be a partition of $\Sigma$ into $j$
groups, usually referred to as the \emph{gold standard}, or a
division of the universe generating $\Sigma$ into $j$ categories, usually
referred to as \emph{class labels}. An \emph{external index} $E$
is a function that takes as input a reference classification $C_j$
for $\Sigma$ and a partition $P_k$ of $\Sigma$ and returns a value assessing
how close the partition is to the reference classification. It is
external because the quality assessment of the partition is
established via criteria external to the data, i.e., the reference
classification. Notice that it is not required that $j=k$. An
\emph{internal index} $I$ is a function defined on the set of all
possible partitions of $\Sigma$ and with values in $\mathds{R}$. It should
measure the quality of a partition according to some suitable
criteria. It is internal because the quality of the partition is
measured according to information contained in the dataset without
resorting to external knowledge. The first two problems are:

\begin{itemize}

\item [(\bf Q.1)] Given $C_j$, $P_k$ and $E$,  measure how far is
$P_k$ from $C_j$, according to $E$.\\

\item [(\bf Q.2)] Given $P_k$ and $I$, establish whether the value of $I$ computed on $P_k$ is unusual and therefore surprising.  That is, significantly small or significantly large.
\end{itemize}

Notice that the two questions above try to assess the quality of a clustering solution $P_k$ consisting of $k$ groups, but they give no indication on what the \vir{right number} of clusters is. In order to get such an indication, one is interested in the following:

\begin{itemize}

\item[({\bf Q.3})] Given: (Q.3.a) A sequence of clustering solutions $P_1, \ldots, P_s$, obtained for instance via repeated application of a clustering algorithm $A$; (Q.3.b) a function $R$, usually referred to as a \emph{relative index},  that estimates the relative merits of a set of clustering solutions. One is interested in identifying  the partition $P_{k^*}$ among the ones given in (Q.3.a) providing the best value of $R$. In what follows, the optimal number of clusters according to $R$ is referred to as $k^*$.

\end{itemize}

The clustering literature is extremely rich in mathematical
functions suited for the three problems outlined above
~\cite{Handl05}. The crux of the matter is to establish
quantitatively the threshold values allowing one to say that the
value of an index is significant enough. That naturally leads to briefly mention hypothesis testing in statistics, from which one
can derive procedures to assess the statistical significance of an
index. As will be evident in the following sections, those
procedures are rarely applied in microarray data analysis,
being preferred to less resource-demanding heuristics that are
validated experimentally.


\section{Cluster Significance for a Given
Statistic: a General Procedure}\label{sec:statistics}

A \emph{statistic} $T$ is a function of the data capturing useful information about it, i.e., it can be one of the indices mentioned earlier. In mathematical terms, it is  a random variable and its distribution describes the relative frequency with which values of $T$ occur, according to some assumptions. In turn, since $T$ is  a random variable, one implicitly assumes the existence of a background or reference probability distribution for its values. That implies the existence of a sample space. A \emph{hypothesis} is a
statement about the frequency of events in the sample space. It is tested by observing a value of $T$ and by deciding how unusual it is, according to the probability distribution one is assuming for the sample space. In what follows, one assumes that the higher the value of $T$, the more unusual it is, the symmetric case being dealt with similarly.

The most common hypothesis tested for in clustering is the
\emph{null hypothesis} $H_0$: there is no structure in the data (i.e. $k = 1$).
Testing for $H_0$ with a statistic $T$ in a dataset $D$ means to compute $T$ on $D$ and then decide whether to reject  or not to reject $H_0$. In order to decide,  one needs to establish how significant is the value found with respect to a background probability distribution  of the statistic $T$ under $H_0$. That means one has to formalize  the concept of \vir{no structure} or \vir{randomness} in the data. Among the many possible ways, generally referred to as \emph{null models}, the
most relevant proposed in the clustering literature~\cite{Bock1985,Gordon,JainDubes,Sarle1983} are introduced, together with an identification of which one is well suited for microarray data analysis~\cite{CLEST,Tibshirani}:

\paragraph{Unimodality Hypothesis.} A new dataset $D^{\prime}$ is generated as follows: the variables describing the items are randomly selected from a unimodal distribution (e.g. normal). This null model typically is not applied to microarray data, since it gives a high probability of rejection of the null hypothesis. For instance, that happens when the data are sampled from a distribution with a lower kurtosis than the normal distribution, such as the uniform distribution~\cite{Sarle1983}. Fig.~\ref{fig:Null}(a) reports an example of a dataset generated via the unimodality hypothesis.

\paragraph{Random Graph Hypothesis.}
The entries of the dissimilarity/distance matrix $S$ are random.
That is, one assumes that, in terms of a linear order relation capturing proximity, all the entries of
the lower triangular part of $S$ are equally likely, i.e., $S_{i,j}=\frac{1}{[n(n-1)/2]!}$ for $1\leq i\leq n$ and $1\leq j\leq i$. This null model is not applied to microarray data, since it does not preserve the distances that may be present among items.

\paragraph{Random Label Hypothesis.} All permutations of the items are equally likely with respect to some characteristic, such as a \emph{priori} class membership. In order to use this model, one needs to specify the a priori classification of the data. Each permutation has a probability $\frac{1}{n!}$.  In particular, for microarray data, it coincides with the so called \emph{Permutational Model}, detailed in what follows.\\
\begin{itemize}

\item \emph{Permutational Model} ({\tt Pr} for short), it generates a random data matrix by randomly permuting the elements within the rows and/or the columns of $D$. Some variants of this model have been studied for binary pattern matrices~\cite{Harper,Strauss,Vassiliou1989}. In order to properly implement this model, care must be taken in specifying a proper permutation for the data, since some similarity and distance functions are insensitive to permutations of coordinates. That is, although $D'$ is a random permutation of $D$, it may happen that the distance or similarity among the points in $D^{\prime}$ is the same as in $D$,  resulting in indistinguishable datasets for clustering algorithms. This latter model may not be suitable for microarray data with very small sample sizes (conditions), since one will not obtain enough observations (data points) to estimate the null model, even if one generates all possible permutations.

\end{itemize}

\paragraph{Random Position Hypothesis.}  The items can be represented by points that are randomly drawn from a region $\mathcal{R}$ in $m$-dimensional space. In order to use this model, one needs to specify the region within which the points have to be uniformly distributed. Two instances applied to microarray data~\cite{CLEST,giancarlo08,Tibshirani} are distinguished:
\begin{itemize}

\item \emph{Poisson Model} ({\tt Ps} for short), where the region $\mathcal{R}$ is specified from the data. The simplest regions that have been considered are the $m$-dimensional hypercube and hypersphere enclosing the points specified by the matrix $D$~\cite{Gordon}. Another possibility, in order to make the model more data-dependent, is to choose the convex hull enclosing the points specified by $D$. Fig.~\ref{fig:Null}(b) reports an example of a dataset $D'$ generated by {\tt Ps}, where the region $\mathcal{R}$ (the box in red) is obtained from the dataset $D$ reported in Fig.\ref{fig:KmeansExample}(a).

\item \emph{Poisson Model Aligned with Principal Components of the Data} ({\tt Pc} for short), where Tibshirani et al.~\cite{Tibshirani}, following Sarle~\cite{Sarle1983}, propose to align the region $\mathcal{R}$ with the principal components of the data matrix $D$. In detail, assuming that the columns of $D$ have mean zero, let $D=UXV^T$ be its singular value decomposition . Let $\widehat{D}=DV$. One uses $\widehat{D}$ as in {\tt Ps} to obtain a dataset $\widehat{D}'$. Then one back transforms via $D'=\widehat{D}'V^T$ to obtain the new dataset. Fig.~\ref{fig:Null}(c) reports an example of a dataset $D'$ generated by {\tt Pc}, where the region $\mathcal{R}$ (the box in red) is obtained from the dataset $D$ reported in Fig.\ref{fig:KmeansExample}(a).

\end{itemize}

It is worth pointing out that in a multivariate situation, one is not able to choose a generally applicable and useful reference distribution: the geometry of the particular null distribution matters~\cite{Sarle1983,Tibshirani}. Therefore, in two or more dimensions, and depending on the test statistic, the results can be very sensitive to the region of support of the reference distribution~\cite{CLEST,Sarle1983}.

\begin{figure}[ht]
\centering
\epsfig{file=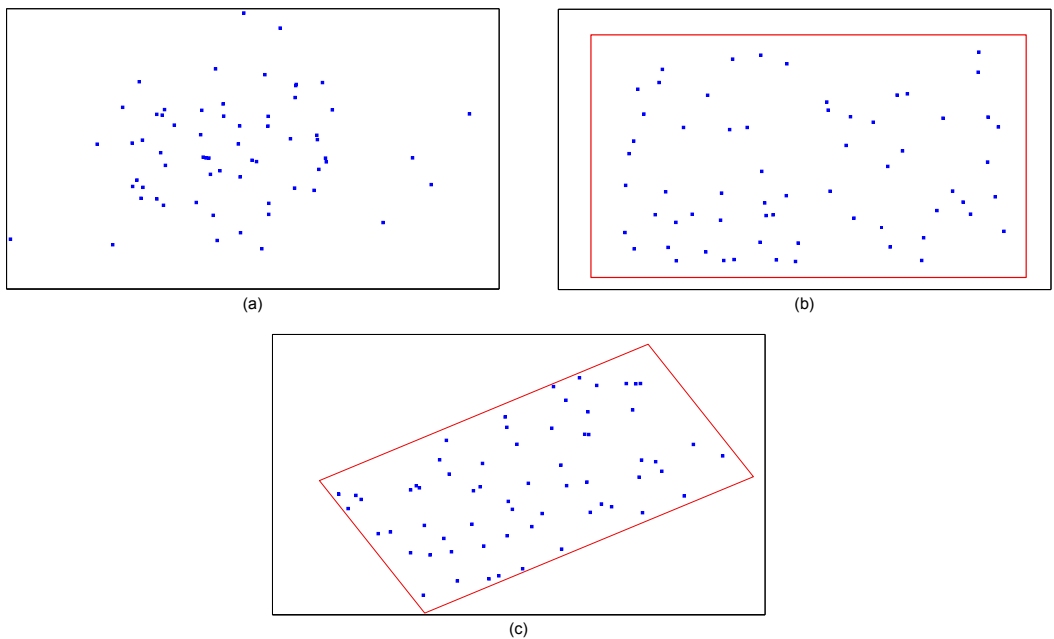, scale=0.19}
\caption{Dataset generated via (a) the Unimodality Hypothesis; (b) the Poisson Null Model; (c) the Poisson Null Model Aligned with Principal Components of the Data.}\label{fig:Null}
\end{figure}

\medskip

\medskip

Once a null model has been agreed upon, one would like to
obtain formulas giving the value of $T$ under
the null model and for a specific set of parameters. Unfortunately,
not too many such formulae are available. In fact, in most cases,
one needs to resort to a Monte Carlo simulation applied to the context of assessing the significance
of a partition of the data into $k$ clusters. In technical terms, it
is a p-value test assessing whether in the dataset there exist $k$
clusters, based on  $T$ and the null model for $H_0$.  It is referred  to as \codice{MECCA}, an abbreviation for Monte Carlo Confidence Analysis, and it is described in Fig.~\ref{alg:Mecca}.
The procedure is also a somewhat more general version of a
significance test proposed and studied by Gordon
~\cite{Gordon1,Gordon} for the same problem. It takes as input an
integer $\ell$ (the number of iterations in the Monte Carlo
simulation),  a clustering algorithm $A$, a dataset $D$, the
function $T$, a partition $P_k$ of $D$ obtained via algorithm $A$
and a parameter $\alpha \in [0,1]$ indicating the level of
\vir{significance} for the rejection of $H_0$. It returns a value
$p \in [0, 1]$. If $ p <\alpha$, the null hypothesis of no cluster
structure in the data is to be rejected at significance level
$\alpha$. Else, it cannot be rejected at that significance level.
\medskip

\begin{figure}
\[
\setlength{\fboxsep}{12pt}
\setlength{\mylength}{\linewidth}
\addtolength{\mylength}{-2\fboxsep}
\addtolength{\mylength}{-2\fboxrule}
\ovalbox{
\parbox{\mylength}{
\setlength{\abovedisplayskip}{0pt}
\setlength{\belowdisplayskip}{0pt}

\begin{pseudocode}{MECCA}{\ell,  A, D, T, P_k, \alpha}
\FOR i\GETS 1 \TO \ell\DO\\
\BEGIN
1.\mbox{ }\mbox{ }\mbox{ Compute a new data matrix } D_i\mbox{, by using the chosen null model.}\\
2.\mbox{ }\mbox{ }\mbox{ Partition } D_i \mbox{ into a set of } k\mbox{ clusters } P_{i,k} \mbox{ by using the algorithm } A.\\
\END\\
\FOR i\GETS 1 \TO \ell\DO\\
\BEGIN
3.\mbox{ }\mbox{ }\mbox{ Compute  } T \mbox{ on } P_{i,k}.\\
4.\mbox{ }\mbox{ }\mbox{ Let }SL\mbox{ be the non-decreasing sorted array of $T$ values.}\\
\END\\
5.\mbox{ Let } V \mbox{ denote the value of } T \mbox{ computed on } P_k.\\
6.\mbox{ Let } p \mbox{ be the proportion of the values in } SL \mbox{ larger
than } V.\\
\RETURN{p}
\end{pseudocode}

}
}
\]
\caption{The \codice{MECCA} procedure.}\label{alg:Mecca}
\end{figure}

A few remarks are in order. As pointed out by Gordon, significance
tests aiming at assessing how reliable is a clustering solution are
usually not carried out in data analysis. Microarrays are no
exception, although sophisticated statistical techniques specific for those data have been designed (e.g.~\cite{Dudoit2003}). One of the reasons is certainly their high computational demand.
Another, more subtle, reason is that researchers expect that
\vir{some structure} is present in their data.  Nevertheless, a
general procedure, like \codice{MECCA}, is quite useful as a paradigm
illustrating how one tries to assess cluster quality via a null
model and a statistic $T$. Indeed, one computes the observed value
of $T$ (on the real data). Then, one computes, via a Monte Carlo
simulation, enough values of $T$, as expected from the formalization
of $H_0$ via the null model. Finally, one checks how \vir{unusual}
is the value of the observed statistic with respect to its expected
value, as estimated by a Monte Carlo simulation. In Chapters \ref{chap:ValidationMeasuers} and \ref{chap:Stability}, two
methods that one describes in this thesis resort to the same
principles and guidelines of \codice{MECCA}, although they are more
specific about the statistic that is relevant in order to
identify the number of clusters in a dataset.

\section{Data Generation/Perturbation Techniques}\label{sec:Perturbation}

The null models described in Section~\ref{sec:statistics} can be seen as an instance of a very general procedure, in what follows referred to as \codice{DGP}, that generates new datasets from a given one. Such a procedure takes as input a dataset $D$, of size $n\times m$, together with other parameters and returns a new dataset $D^{\prime}$ of size $n^{\prime} \times m^{\prime}$, with $n^{\prime}\leq n$ and $m^{\prime}\leq m$. In this section, additional instances of \codice{DGP} are detailed. They are used in microarray data analysis to generate data points in order to compute a cluster quality measure.
Each of them can be thought of as a paradigm in itself and therefore in this section  only an outline is provided.

\subsection{Subsampling}\label{subsec:Subsampling}

The simplest way to generate a new dataset $D^{\prime}$ from $D$ is to take random samples from it.  Although simple, this approach critically depends on whether the sampling is performed without or with replacement. The first type of method is referred to as \emph{subsampling}. It is widely used in clustering and briefly discussed here. The second method is referred to as {\em bootstrapping} and, although fundamental in statistics~\cite{ET93}, it is hardly used in cluster validation as pointed out and discussed in~\cite{JainDubes,Monti03}.

Formally, a subsampling procedure takes as input a dataset $D$ and a parameter $\beta$, with $0<\beta<1$, and gives as output a percentage $\beta$ of $D$, i.e., the dataset $D^{\prime}$ has size $n^{\prime}\times m$, with $n^{\prime}=\lceil\beta n\rceil$.
$D^{\prime}$ is obtained via the extraction of $n^{\prime}$ items (i.e. rows) from $D$, which are usually selected uniformly and at random.

The aim of subsampling procedures is to generate a reduced dataset $D^{\prime}$ that captures the structure (i.e. the right number of clusters) in the original data. Intuitively, both the chances to achieve that goal and the time required by procedures using $D^{\prime}$ increase with $\beta$. In order to have a good trade-off between the representativeness of $D^{\prime}$ and the speed of the methods using it, a value of $\beta \in [0.6,0.9]$ is used in the literature (e.g.~\cite{BenHur02,CLEST,FC,Monti03}).

The subsamping technique does not guarantee that each cluster of $D$ is represented in $D^{\prime}$, i.e., the random extraction could not select any elements of a given cluster. Therefore, Hansen et al.~\cite{Hansen93} propose a heuristic referred to as \emph{proportionate stratified sampling} as an alternative that may take care of the mentioned problem. In that case, $D^{\prime}$ is generated first by clustering $D$ and then by selecting a given percentage $\beta$ of the elements in each cluster. Proportionate stratified sampling gives no formal guarantee that the entire cluster structure of $D$ is present in $D^{\prime}$.

\subsection{Noise Injection}\label{subsec:NoiseInjection}

{\em Noise injection} is a widely applied perturbation methodology in computer science (see ~\cite{noise1,bittner,Kerr00bootstrappingcluster,MCShane02,Raviv96bootstrappingwith,wolfinger01} and reference therein). However, it is not widely applied in clustering. The main idea is to generate $D^{\prime}$ by adding a random value, i.e., a \vir{perturbation}, to each of the elements of $D$. Perturbations are generated  via some random process, i.e., a probability distribution whose parameters can be directly estimated from $D$. In a study about melanoma patients, Bittner et al.~\cite{bittner} propose to perturb the original dataset by adding Gaussian noise to its elements  in order to assess cluster stability.  Following up, Wolfinger et al.~\cite{wolfinger01} report that  perturbing the data via a Gaussian distribution provides good stability results for several microarray datasets. As for parameter estimation, McShane et al.~\cite{MCShane02} propose to compute the variance of experiments in each row of $D$ and then to use the median of the observed distribution as the variance in the Gaussian distribution.

\subsection{Dimensionality Reduction Methods}\label{subsec:SubspaceMethods}

The methodology described in this section is, in most cases,  the
dual of subsampling, since the main idea is to obtain $D^{\prime}$
by reducing the number of columns of $D$ while trying to preserve
its cluster structure. Since each element,  i.e. row,  of $D$ is a
point in  $m$-dimensional space, one has a dimensionality reduction
in the data.

A well established dimensionality reduction method in data analysis is the \emph{Principal Component Analysis} (PCA for short)~\cite{JainDubes}. Although it is a standard method for dimensionality reduction from a statistical point of view, it is not used in conjunction with stability-based internal validation measures (detailed in Chapter~\ref{chap:Stability}) because of its determinism in generating  $D^{\prime}$. Note, however, that the main idea of principal components is used in conjunction with  null models (see for example the (M.2) model described in Section~\ref{sec:statistics}).

The following three techniques of dimensionality reduction seem to be of use in this area.
The first one is rather trivial since it consists of randomly
selecting the columns of $D$ (cf.~\cite{SmolkinGhosh}). However, the
cluster structure of $D$ is unlikely to be preserved and this
approach may introduce large distortions into gene expression data,
which then result in the introduction of biases into stability
indices (detailed in Chapter~\ref{chap:Stability}),  as reported in~\cite{BertoniV06}. More sensible approaches
for dimensionality reduction are Non-negative Matrix Factorization (NMF for short) and randomized dimensionality reduction techniques
reported in the following. This latter is discussed next, while discussion of the former is given in Chapter~\ref{chap:NMF}, while an exhaustive benchmarking of NMF as a clustering algorithm is provided in Chapter~\ref{chap:5}.

\subsubsection{Randomized Dimensionality Reduction}

The technique consists of the use of a family of transformations that try to preserve the distance
with an \vir{$\varepsilon$ distortion level} between the elements of $D$.
Intuitively, if two elements are \vir{close} in $D$, according to
some distance function $d$, they should be \vir{close} in $D^{\prime}$. Let $f$ be a transformation from $D$ in $D'$,  $f(\sigma_i)$ and
$f(\sigma_j)$ be the projections of two elements $\sigma_i$ and
$\sigma_j$ of $D$ into $D^{\prime}$. Let
$$
d_{f}(\sigma_i,\sigma_j)=\frac{d(\sigma_i,\sigma_j)}{d(f(\sigma_i),f(\sigma_j))}
.$$

If $d_{f}=1$ the distance of the two elements is preserved. When $1-\varepsilon \leq d_{f}(\sigma_i,\sigma_j) \leq 1 + \varepsilon$, one says that the function $f$ preserves the distance with an \vir{$\varepsilon$ distortion level}. The Johnson-Lindenstrauss Lemma and \emph{random projections} are the keys to all the randomized dimensionality reduction based techniques. Intuitively, for a fixed distortion level $\varepsilon$, the Johnson-Lindenstrauss Lemma gives nearly optimal bounds to the value of $m^{\prime}$ (cf.~\cite{Alon08}), formally:

\begin{lemma}[Johnson-Lindenstrauss~\cite{JL}]
For any $0<\varepsilon<1$ and any integer $n$, let $m^{\prime}$ be a positive integer such that
$$
m^{\prime}> 4 (\varepsilon^2/2 - \varepsilon^2/3)^{-1}\log n.
$$

For any set $V$ of $n$ points in $\mathds{R}^{m}$ there is a map function $f:\mathds{R}^{m}\rightarrow\mathds{R}^{m^{\prime}}$ such that for all $u,v \in V$

$$
(1-\varepsilon) \|u-v\|^{2} \leq \|f(u)-f(v) \|^{2} \leq (1 + \varepsilon) \|u-v\|^{2}
$$

\end{lemma}

The interested reader will find two independent simplified versions of the proof of the above Lemma in~\cite{DasguptaG03,IndykMotwani} and extensions to other spaces and distances in~\cite{AilonChazelle,BhattacharyaKP09,CormodeDIM03,JohnsonN09}.
It is possible to determine a function $f$ that satisfies the Lemma with high probability, e.g., at least 2/3, in randomized polynomial time~\cite{DasguptaG03,IndykMotwani}.

Since the projection into the new smaller space is a time consuming task, several heuristics have been proposed in the literature. Some of them are based on sparse projection matrices~\cite{Achlioptas,BinghamMannila}, while a more innovative and  recent approach has been proposed by Ailon and Chazelle~\cite{AilonChazelle} with the addition of  the  Fast Fourier Transform to the Johnson-Lindenstrauss Lemma. In conclusion, it is also worthy of mention that the dimensionality reduction techniques described here are tightly connected to problems as approximate nearest neighbor~\cite{AndoniIndyk08,IndykMotwani} of relevance also for clustering.

\ignore{
\subsubsection{Nonnegative Matrix Factorization}

Nonnegative matrix factorization (NMF for short) is introduced by Lee and Seung~\cite{NMF,lee00algorithms} and generalize the positive matrix factorization from Paatero and Tapper~\cite{PaateroTapper}. Formally, given an $n\times m$ nonnegative matrix $V$ where $m$ is the number of conditions and $n$ is the number of examples in the dataset (i.e. $V$ is the transpose of $D$). This matrix is then approximately factorized into an $n \times r$ matrix $W$ and an $r\times m$ matrix H, i.e., $V\equiv WH$ (see Fig.~\ref{fig:NMFexample}).
Usually $r$ is chosen to be smaller than $n$ or $m$, so that $W$ and $H$ are smaller than the original matrix $V$. This results in a compressed version of the original data matrix.
By convention, the columns of $W$ are considered to be the components of $V$, while the matrix $H$ is said to contain the coefficients or \vir{weights} that indicate the linear combination of components which approximate each column of $V$. That is, $W$ contains the basis of the new space of size $r$ and $H$ contains the coefficient.
NMF has been primarily applied in an unsupervised setting in image and natural language processing~\cite{DingNMF,NMF,PatrikNMF,Wild}. More recently, it has been successfully utilized in a variety of applications in computational biology~\cite{BrunetNMF,Devarajan,KimPark,WangKO06}. In terms of the subject of this manuscript, it is used as a clustering algorithm rather than a dimensionality reduction techniques~\cite{BrunetNMF} (see also Section~\ref{subsec:Consensus}).

\ignore{
\begin{figure}
\begin{center}
\includegraphics[scale=0.75]{img/nmf.jpg}
\caption{An example of NMF factorization.}\label{fig:NMFexample}
\end{center}
\end{figure}
}
}

%% file: Chapter2.tex
\chapter{Fundamental Validation Indices}\label{chap:ValidationMeasuers}

In this chapter, some basic validation techniques are presented. In detail, external and internal indices are outlined. In the scholarly literature, the terms index is also referred to as measure. Following the literature, in this dissertation both nomenclatures are used. The external and internal indices differ fundamentally in their aims, and find application in distinct experimental settings. In particular, three external indices that assess the agreement between two partitions are presented. Moreover, four internal measures useful to estimate the correct number of clusters present in a dataset, based on: compactness, hypothesis testing in statistics and jackknife techniques are also discussed.
One of the topics of this thesis is the study of a relevant paradigm of internal validation measure based on the notion of cluster stability. For this reason this paradigm and the relative instances are thoroughly discussed in Chapter~\ref{chap:Stability}.

\section{External Indices}\label{sec:ext}
In  this section three external indices, namely formulae, are defined. Such measures establish the level of agreement between two partitions. Usually, for a given dataset, one of the partitions is a reference classification of the data while the other one is provided as output by a clustering algorithm.

Let $C=\{c_1,
\ldots,c_r\}$ be a partition of the items in $\Sigma$ into $r$ classes
and $P=\{p_1, \ldots,p_t\}$ be another partition of $\Sigma$ into $t$
clusters. With the notation of Section~\ref{sec:asseClusterQuality}, $C$ is an
external partition of the items, derived from the reference
classification, while $P$ is a partition obtained by some clustering
method. Let $n_{i,j}$ be the number of items in both $c_i$ and
$p_j$, $1 \leq i \leq r$ and $1 \leq j \leq t$. Moreover, let
$|c_i|=n_{i.}$ and $|p_j|=n_{.j}$. Those values can be conveniently
arranged in a \emph{contingency table} (see Table
~\ref{ContingencyTable}).

\subsection{Adjusted Rand Index}\label{sec:AdjRand}

Let $a$ be the number of pairs of items that are placed in the same class in  $C$ and in the same cluster in $P$; let $b$ be the number of pairs of items  placed in the same class in $C$ but not in the same class in $P$; let $c$ be the number of pairs of items in the same cluster in $P$ but not in the same cluster in $C$; let $d$ be the number of pairs of items in different classes and different clusters in both partitions.  The information needed to compute $a$, $b$, $c$ and $d$ can be derived from  Table~\ref{ContingencyTable}. One has:

\begin{table}
\begin{center}
\begin{tabular}{|c|c c c c|c|}
\hline $Class\setminus Cluster$ & $p_{1}$ & $p_{2}$ & \dots & $p_{t}$ & Sums \\
\hline
 $c_{1}$ & $n_{1,1}$ & $n_{1,2}$ & \dots & $n_{1,t}$ & $n_{1.}$ \\
 $c_{2}$ & $n_{2,1}$ & $n_{2,2}$ & \dots & $n_{2,t}$ & $n_{2.}$ \\
\vdots & \vdots & \vdots & & \vdots &  \vdots \\
 $c_{r}$ & $n_{r,1}$ & $n_{r,2}$ & \dots & $n_{r,t}$ & $n_{r.}$ \\
\hline Sums & $n_{.1}$ & $n_{.2}$ & \dots & $n_{.t}$ & $n_{..}=n$\\
 \hline
\end{tabular}
\medskip
\medskip
\caption{Contingency table for comparing two
partitions}\label{ContingencyTable}
\end{center}
\end{table}

\begin{equation}
a=\sum_{i,j}{n_{i,j} \choose 2},
\end{equation}

\begin{equation}
b=\sum_{i}{n_{i.} \choose 2}-a,
\end{equation}

\begin{equation}
c=\sum_{j}{n_{.j} \choose 2}-a.
\end{equation}

\noindent Moreover, since $a+b+c+d= {n \choose 2}$, one has:
\begin{equation}
d={n \choose 2}-(a+b+c).
\end{equation}

Based on those quantities,  the {\tt Rand} index $R$ is defined as
~\cite{rand_ref}:

\begin{equation}\label{eq:rand}
R=\frac{a+d}{a+b+c+d}
\end{equation}

Notice that,  since $a+d$ is the number of pairs of items in which there is agreement between the two partitions,  $R$ is an index of agreement of the two partitions with value in  $[0,1]$. The main problem with $R$ is that its value on two partitions picked at random does not take a constant value, say zero. So, it is difficult to establish, given two partitions, how significant (distant from randomness) is the concordance  between the two partitions, as measured by the value of $R$. In general, given an index, it would be appropriate to take an
adjusted version ensuring that its expected value is zero when the partitions are selected at random and one when they are identical.
That can be done according to the following general scheme:

\begin{equation}
\frac{index- expected\mbox{ }index}{maximum\mbox{
}index-expected\mbox{ }index}\label{eq:genext}
\end{equation}

\noindent where $maximum\mbox{
}index$ is the maximum value of the $index$ and  $expected\mbox{ }index$ is its
expected value derived under a suitably chosen model of random
agreement between two partitions, i.e. the null hypothesis. The {\tt Adjusted Rand} Index $R_A$ is derived from \eqref{eq:rand} and \eqref{eq:genext} using the generalized hypergeometric distribution as the null hypothesis. That is, it is assumed that the row and column sums in Table~\ref{ContingencyTable} are fixed, but the two partitions are picked at random. One has~\cite{AdjRand0}:

$$R_A=\frac{\sum\limits_{i,j}{n_{i,j} \choose 2} - \frac{\left[\sum\limits_{i}{n_{i.} \choose 2}\sum\limits_{j}{n_{.j} \choose 2}\right]}{{n \choose 2}}}{ \frac {1}{2}\left[\sum\limits_{i}{n_{i.} \choose 2}+\sum\limits_{j}{n_{.j} \choose 2}\right]-\frac{\left[\sum\limits_{i}{n_{i.} \choose 2}\sum\limits_{j}{n_{.j} \choose 2}\right]}{{n \choose 2}}}$$

$R_A$ has a maximum value of one, when there is a perfect agreement between the two partitions, while its expected value of zero indicates a level of agreement due to chance. Moreover, $R_A$ can take on a larger range of values with respect to $R$ and, in particular,  may be negative~\cite{KaYeeDiss}. Therefore, the two partitions are in significant agreement if $R_A$ assumes a non-negative value, substantially away from zero. Notice that $R_A$ is a statistic on the level of agreement of two partitions of a dataset (see Section~\ref{sec:statistics}) while $R$ is a simple indication of percentage agreement.
To illustrate this point, consider two partitions of a set of 29 items giving rise to Table~\ref{ContingencyTableExample}. Then $R=0.677$, indicating a good percentage agreement while $R_A=-0.014$ and, being close to its expected value under the null model, it indicates a level of significance in the agreement  close to the random case.
In fact, the entries in the table have been picked at random. $R_A$  is a statistic recommended in the classification literature~\cite{AdjRand2} to compare the level of agreement of two partitions.

\begin{table}
\begin{center}
\begin{tabular}{|c|c c c c c|c|}
\hline $Class\setminus Cluster$ & $p_{1}$ & $p_{2}$ & $p_{3}$ & $p_{4}$ & $v_{5}$ & Sums \\
\hline
 $c_{1}$ & $1$ & $4$ & $2$ & $1$ & $2$ & $10$ \\
 $c_{2}$ & $0$ & $1$ & $1$ & $0$ & $1$ & $3$ \\
 $c_{3}$ & $1$ & $2$ & $0$ & $2$ & $0$ & $5$ \\
 $c_{4}$ & $2$ & $1$ & $0$ & $1$ & $2$ & $6$ \\
 $c_{5}$ & $1$ & $0$ & $1$ & $0$ & $3$ & $5$ \\
\hline Sums & $5$ & $8$ & $4$ & $4$ & $8$ & $n=29$\\
 \hline
\end{tabular}
\medskip
\medskip
\caption{Contingency table example}
\label{ContingencyTableExample}
\end{center}
\end{table}


\ignore{ We now illustrate the methodology that uses $R_A$ to
experimentally validate clustering algorithms, i.e., how to
experimentally assess how reliable is the output of a clustering
algorithm on microarray data. To this end, we use the {\bf RYCC}
dataset and a few classic clustering algorithms. The interested
reader will find extensive sets of experiments, following this
methodology, in~\cite{Genclust,KaYeeDiss}. }

\subsection{Fowlkes and Mallows Index}\label{subsection:FM}

The {\tt FM-index}
~\cite{FMMeasure} is also derived from the contingency Table
~\ref{ContingencyTable}, as follows:

\begin{equation}\label{FMIndex}
FM_{k}=\frac{T_{k}}{\sqrt{U_{k} \cdot V_{k}}}
\end{equation}

\noindent where:

\begin{equation}
T_{k}=\sum_{i=1}^k \sum_{j=1}^k n_{ij}^2 - n
\end{equation}

\begin{equation}
U_{k}=\sum_{i=1}^k n_{i.}^2 - n
\end{equation}

\begin{equation}
V_{k}=\sum_{j=1}^k n_{.j}^2 - n
\end{equation}

\noindent The index has values in the range $[0,1]$, with an
interpretation of the values in that interval analogous to that provided  for  the values of $R$. An example can be obtained as in Section~\ref{sec:AdjRand}. Indeed, for Table
~\ref{ContingencyTableExample}, one has $FM_{5}= 0.186$ indicating a low level of agreement between the two partitions.

\subsection{The F-Index}

The {\tt F-index}~\cite{CVRijs} combines notions from information retrieval, such as precision and recall, in order to evaluate the agreement of a clustering solution $P$ with respect to a reference partition $C$. Again, its definition can be derived from the contingency  Table~\ref{ContingencyTable}.

Given $c_i$ and $p_j$, their relative precision is defined as the ratio of the number of elements of the class $c_{i}$ within cluster $p_{j}$, and by the size of the cluster $p_{j}$. That is:

\begin{equation}\label{eq:Prec}
Prec(c_i, p_j)=\frac{n_{i,j}}{n_{.j}}
\end{equation}

\noindent Moreover, their relative recall is defined as the ratio of the number of elements of the class $c_{i}$ within cluster $p_{j}$, divided by the size of the class $c_{i}$. That is:

\begin{equation}\label{eq:Rec}
Rec(c_i, p_j)=\frac{n_{i,j}}{n_{i.}}.
\end{equation}

\noindent The {\tt F-index}  is then defined as an harmonic mean that uses the  precision and recall values, with weight $b$:

\begin{equation}\label{eq:FtCk}
F(c_{i}, p_{j})=\frac{(b^2+1) \cdot Prec(c_{i}, p_{j}) \cdot
Rec(c_{i}, p_{j})}{b^2 \cdot Prec(c_{i}, p_{j}) + Rec(c_{i}, p_{j})}.
\end{equation}
\\
\noindent Equal weighting for precision and recall is obtained by setting $b = 1$. Finally,  the overall {\tt F-index} is:
\\
\begin{equation}\label{eq:FOverall}
F=\sum_{c_i \in C} \frac{n_{i.}}{n} \cdot \underset{p_k \in P}{\mathop{\max }}\, F(c_i, p_k).
\end{equation}

\medskip

$F$ is an index with value in the range  $[0,1]$, with an
interpretation of the values in that interval analogous to that provided for the values of $R$. An example can be obtained as in Section~\ref{sec:AdjRand}. Indeed, for Table~\ref{ContingencyTableExample}, one has $F = 0.414$, indicating low level of  agreement between the two partitions.

\section{Internal and Relative Indices}\label{sec:int}

Internal indices should assess the merits of a partition, without
any use of external information. Then, a Monte Carlo simulation can
establish if the value of such an index on the given partition is
unusual enough for the user to gain confidence that the partition is
good. Unfortunately, this methodology is rarely used in data
analysis for microarrays, as stated in~\cite{Handl05}. Internal indices are also a
fundamental building block in order to obtain relative indices that
help to select, among a given set of partitions, the \vir{best}
one.

In this section, four relative indices are presented, starting
with the ones based on compactness. Then
methods that are based on hypothesis testing and the jackknife approach are presented.

\subsection{Methods Based on Compactness}

The measures presented here assess cluster compactness. The most popular compactness measures are based on the sum-of-squares. In what follows, two of the prominent measures in that class are detailed.

\subsubsection{Within Cluster Sum-of-Squares}
\label{sec:WCSS-Theory}

An internal measure that gives an assessment of the level of
compactness of each cluster in a clustering solution is the Within Cluster Sum of Squares ({\tt WCSS} for short).
Let $\mathcal{C}=\{c_1,\ldots,c_k\}$ be a clustering solution, with $k$ clusters. Formally, let

\begin{equation}
D_r=\sum_{j \in c_r}||\sigma_j- \overline{\sigma_r}||^2
\end{equation}

\noindent where $\overline{\sigma_r}$ is the centroid of cluster $c_r$.
Then, one has:

\begin{equation}
{\tt WCSS}(k)=\sum_{r=1}^k D_{r}.
\end{equation}

By analyzing the behavior of {\tt WCSS} in $[1,k_{max}]$, as a function of $k$, one can estimate the correct number of cluster $k^*$ in the dataset.
Intuitively, for values $k<k^*$, the compactness of each cluster should substantially increase, causing a substantial decrease in {\tt WCSS}. In other words, one should observe in the {\tt WCSS} curve a decreasing marginal improvement in terms of cluster compactness after the value $k^*$. The following heuristic approach comes out~\cite{Tibshrbook}: Plot the values of {\tt WCSS}, computed on the given clustering solutions, in the range $[1,k_{max}]$; choose as $k^*$ the abscissa closest to the \vir{knee} in the {\tt WCSS} curve. Fig.~\ref{fig:WCSS-example} provides an example of the {\tt WCSS} curve computed on the dataset of Fig.~\ref{fig:KmeansExample}(a) with K-means-R (see Section~\ref{sec:Partional}) for $k\in[1,10]$. Indeed, the dataset has two natural clusters and the plot of the {\tt WCSS} curve in Fig.~\ref{fig:WCSS-example} indicates $k^*=2$. As it will be clear in Section~\ref{sec:wcss-exp}, the prediction of $k^*$ with {\tt WCSS} is not so easy on real datasets, since the behavior of {\tt WCSS} is not so regular as one expects.

\medskip

\medskip

\begin{figure}[ht]
\centering
\epsfig{file=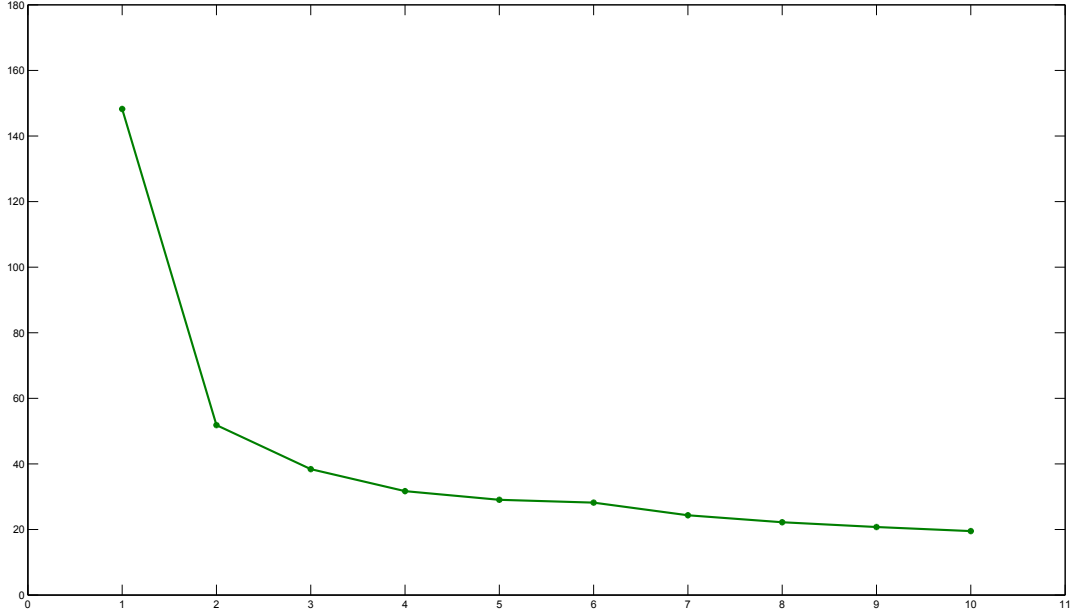, scale=0.4}
\caption{Plot of the values of {\tt WCSS}.}\label{fig:WCSS-example}
\end{figure}

\subsubsection{Krzanowski and Lai Index}\label{subsec:KL}
By elaborating on an earlier proposal by Marriot~\cite{Marriot71},
\ignore{\cite{Marriot71}} Krzanowski and Lai~\cite{KLMeasure}
proposed an internal measure, which is referred to as {\tt KL}. It is based on {\tt WCSS}, but it is automatic, i.e., a numeric value for $k^*$ is returned. Let

\begin{equation}\label{eqn:Diff}
DIFF(k)=(k-1)^{2/m}{\tt WCSS}(k-1)-k^{2/m}{\tt WCSS}(k).
\end{equation}

\noindent with $2 \leq k \leq k_{max}$.

Recall from Section~\ref{sec:WCSS-Theory} the behavior of {\tt
WCSS},  with respect to $k^*$. Based of those considerations, one expects the following behavior for
$DIFF(k)$:
\begin{itemize}
\item[(i)] for $k<k^{*}$, both $DIFF(k)$ and $DIFF(k+1)$ should be large positive values.
\item[(ii)] for $k>k^{*}$,  both $DIFF(k)$ and $DIFF(k+1)$ should be small values, and one or both might be negative.
\item[(iii)] for $k=k^{*}$, $DIFF(k)$ should be large positive, but $DIFF(k+1)$ should be relatively small (might be negative).
\end{itemize}

Based on these considerations, Krzanowski and Lai propose to
choose the estimate on the number of clusters as the $k$ maximizing:

\begin{equation}\label{eqn:KL}
KL(k)=\left|\frac{DIFF(k)}{DIFF(k+1)}\right|.
\end{equation}

That is,
\begin{equation}\label{eqn:KL-rule}
k^{*}=\argmax\limits_{2 \leq k \leq k_{max}} KL(k).
\end{equation}

Notice that  $KL(k)$ is not defined for the important special case of $k=1$, i.e., there is no cluster structure in the data.
Figure~\ref{fig:KL-Example}(a) reports an example of $DIFF(k)$ computation with K-means-R (see Section~\ref{sec:Partional}) on the dataset of Fig.~\ref{fig:KmeansExample}(a); the corresponding $KL$ values are reported in Fig.\ref{fig:KL-Example}(b). Notice that the $KL$ curve has a local maximum on $k=3$ (value close to the correct number of classes) but based on \eqref{eqn:KL-rule} the prediction is $k^{*}=5$.

\begin{figure}[ht]
\centering
\epsfig{file=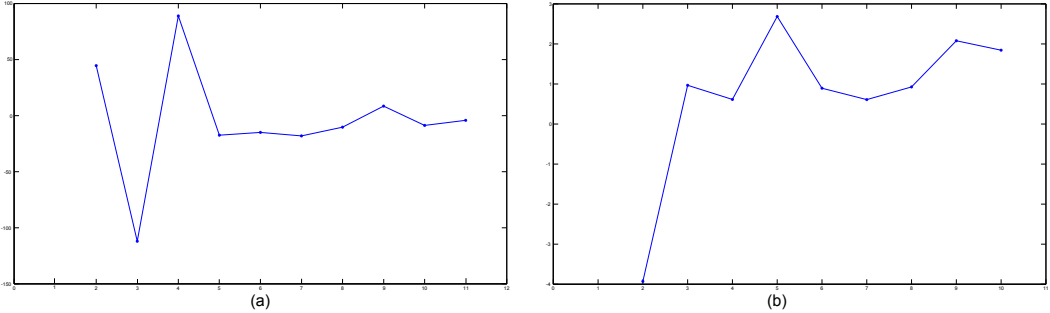,
scale=0.19}
\caption{(a) Plot of the values of $DIFF$. (b) Plot of the values of $KL$.}\label{fig:KL-Example}
\end{figure}

\subsection{Methods Based on Hypothesis Testing in Statistics}\label{sec:Gap}

The measures presented so far are either useless or not defined for the important special case $k=1$.
In this thesis, two methods based on hypothesis testing proposed by Dudoit and Fridlyand~\cite{CLEST} and Tibshirani et al.~\cite{Tibshirani} are considered. The former is a clever combination of the \codice{MECCA} hypothesis testing paradigm (see Section~\ref{sec:statistics}) and stability techniques, and for this reason is detailed in Chapter~\ref{chap:Stability}.

Tibshirani et al.~\cite{Tibshirani} brilliantly combine the ideas of Section~\ref{sec:statistics} with the {\tt WCSS} heuristic, to obtain an index that can deal also with the case $k=1$. It is referred to as the Gap Statistics and, for brevity, it is denoted as {\tt Gap}.\\
\noindent The intuition behind the method is brilliantly elegant. Recall, from the previous subsection that the \vir{knee} in the {\tt WCSS} curve can be used to predict the real number of cluster in the dataset. Unfortunately, the localization of such a value may be subjective. Consider the curves in Fig.~\ref{fig:GAP-exp}. The curve in green at the bottom of the figure is the {\tt WCSS}  given in Fig.~\ref{fig:WCSS-example}. The curve in red at  the top of the figure is the {\em average} {\tt WCSS}, computed on ten datasets generated from the original data via the {\tt Ps} null model. As it is evident from the figure, the curve on the top has a nearly constant slope: an expected behavior on datasets with no cluster structure in them. The vertical lines indicate the gap between the null model curves and the curve computed by K-means-R, which supposedly captures \vir{cluster structure} in the dataset. Since {\tt WCSS} is expected to decrease sharply up to $k^*$, on the real dataset, and it has a nearly constant slope on the null model datasets, the length of the vertical segments is expected to increase up to $k^*$ and then to decrease. In fact, in the figure, if one takes as the  prediction for $k^*$ the first local maximum of the gap values (data not shown), one has $k^*=2$, the correct number of classes in the
dataset. Normalizing the {\tt  WCSS} curves via logs and accounting also for the simulation error, such an intuition can be given under the form of a procedure in Fig.~\ref{algo:Gap}, which is strikingly similar to \codice{MECCA}, as discussed shortly (see Section~\ref{sec:statistics}). The first three parameters are as in that procedure, while the last one states that the search for $k^*$ must be done in the interval $[1, k_{max}]$.

\begin{figure}[ht] \centering \epsfig{file=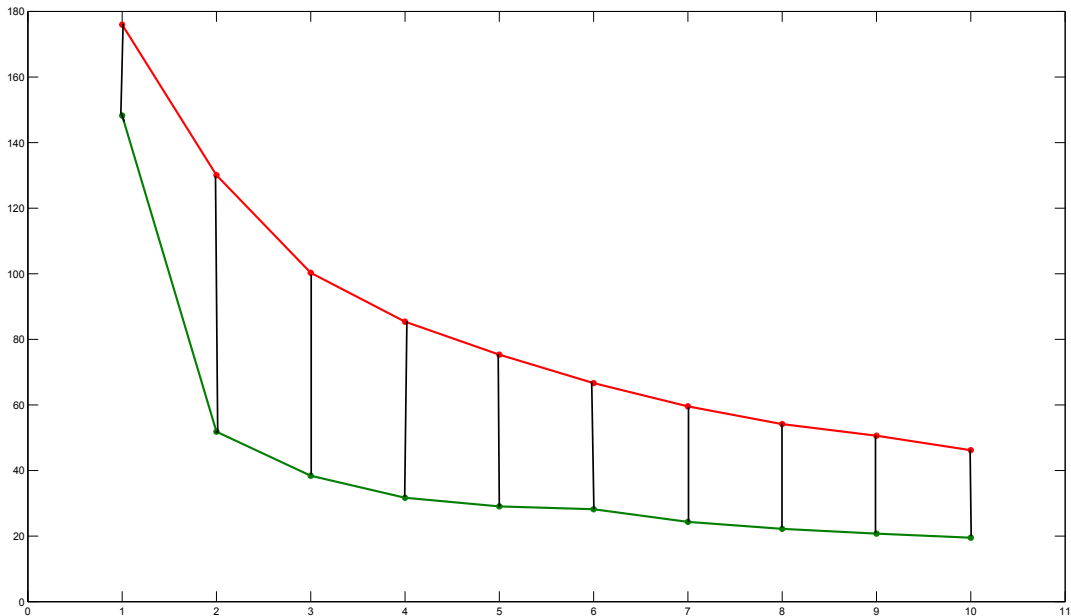,
scale=0.4} \caption{A geometric interpretation of the Gap
Statistics.  }\label{fig:GAP-exp}
\end{figure}

\begin{figure}
\[
\setlength{\fboxsep}{12pt}
\setlength{\mylength}{\linewidth}
\addtolength{\mylength}{-2\fboxsep}
\addtolength{\mylength}{-2\fboxrule}
\ovalbox{
\parbox{\mylength}{
\setlength{\abovedisplayskip}{0pt}
\setlength{\belowdisplayskip}{0pt}

\begin{pseudocode}{GP}{\ell,  A, D, k_{max}}

\FOR i \GETS 1 \TO \ell \DO\\

1.\mbox{ Compute a new data matrix
} D_i\mbox{, using the chosen null model.}\\\mbox{\,\,\,\,\,\,\,\,\,\,\,\,Let } D_0\mbox{ denote the original
data matrix.}\\

\FOR i \GETS 1 \TO \ell \DO\\
\BEGIN
\mbox{ }\mbox{ }\mbox{ }\mbox{ }\FOR k \GETS 1 \TO k_{max} \DO\\
2. \mbox{ }\mbox{ }\mbox{ }\mbox{ Compute a clustering solution } P_{i,k}\mbox{ on } D_i\mbox{ using algorithm }
A.\\
\END\\
\FOR i \GETS 1 \TO \ell \DO\\
\BEGIN
\mbox{ }\mbox{ }\mbox{ }\mbox{ }\FOR k \GETS 1 \TO k_{max} \DO\\
3. \mbox{ }\mbox{ }\mbox{ }\mbox{ Compute  } \log ({\tt WCSS}(k))\mbox{ on } P_{i,k}
\mbox{ and store the result in matrix } SL[i,k].\\
\END\\
\FOR k \GETS 1 \TO k_{max} \DO\\
\BEGIN
4. \mbox{ }\mbox{ }\mbox{ }\mbox{ }{\tt Gap}(k) \GETS \frac{1}{\ell}\sum\limits_{i=1}^\ell SL[i,k] -SL[0,k].\\

5.\mbox{ }\mbox{ }\mbox{ }\mbox{ Compute the standard deviation }sd(k)\mbox{ of the set of numbers }\\
\mbox{\,\,\,\,\,\,\,\,}\mbox{ }\mbox{ }\mbox{ }\{SL[1,k], \ldots, SL[\ell,k]\}\\
6.\mbox{ }\mbox{ }\mbox{ }\mbox{ }s(k)\GETS \left(\sqrt{1 +\frac{1}{\ell}}\right)sd(k).\\
\END\\

7.\mbox{ }k^*\mbox{ is the first value of }k\mbox{ such that }{\tt Gap}(k)\geq {\tt Gap}(k+1)-s(k+1).\\

\RETURN{k^{*}}
\end{pseudocode}

}
}
\]
\caption{The Gap Statistics procedure.}\label{algo:Gap}
\end{figure}

Now, $\log ({\tt WCSS}(k))$ is the statistic $T$ used to assess how reliable is a clustering solution with $k$ clusters. The value of that statistic is computed on both the observed data and on data generated by the chosen null model. Then, rather than returning a p-value, the procedure returns the first $k$ for which \vir{the gap} between the observed and the expected statistic is at a local maximum. With reference to step 7 of procedure \codice{GP} (see Fig.~\ref{algo:Gap}), it is worth pointing out that the adjustment due to the $s(k+1)$ term  is a heuristic
meant to account for the Monte Carlo simulation error in the
estimation of the expected value of $\log ({\tt WCSS}(k))$ (cf.~\cite{Tibshirani}).

 As discussed in Chapter~\ref{chap:6}, the
prediction of $k^*$ is based on running a certain number of times
the procedure \codice{GP}. Then one takes the most frequent outcome as the
prediction. It is worth pointing out that further improvements and
generalizations of {\tt Gap} have been proposed in~\cite{Yan07}.

\subsection{Methods Based on Jackknife Techniques: FOM}

Figure of Merit ({\tt FOM} for short) is a family of internal
validation measures introduced by Yeung et al.
~\cite{KaYeeFOM}, specifically for microarray data. Such a family is based on the jackknife approach and it  has been designed for use as a relative index assessing the predictive power of a clustering algorithm, i.e., its ability to predict the correct number of clusters in a dataset.  It has also been extended  in several directions by Datta and Datta~\cite{Ba03}. Experiments by Yeung et al. show that the {\tt FOM} family of measures satisfies the following properties, with a good degree of accuracy. For a given clustering algorithm, it has a low value in correspondence with the number of clusters that are really present in the data. Moreover, when comparing clustering algorithms for a given number of clusters $k$, the lower the value of {\tt FOM} for a given algorithm, the better its predictive power. In what follow, a review of this work is given, using the {\tt
$2$-norm FOM}, which is the most used instance in the {\tt FOM} family.

Assume that a clustering algorithm is given the data matrix $D$ with column $e$ excluded. Assume also that, with that reduced dataset, the algorithm produces $k$ clusters $c_1, \ldots, c_{k}$. Let $D(\sigma,e)$ be the expression level of gene $\sigma$ and $m_i(e)$ be the
average expression level of condition $e$ for genes in cluster
$c_i$. The {\tt $2$-norm FOM } with respect to $k$ clusters and
condition $e$ is defined as:

\begin{equation}\label{math:fome}
{\tt FOM}(e,k)= \sqrt{\frac{1}{n}\sum_{i=1}^{k}\sum_{x\in
c_i}(D(x,e)-m_i(e))^2}.
\end{equation}

\noindent Notice that {\tt FOM}$(e,k)$ is  essentially a root mean
square deviation. The {\tt aggregate $2$-norm FOM } for $k$ clusters
is then:

\begin{equation}\label{math:foma}
{\tt FOM}(k)=\sum_{e=1}^m {\tt FOM}(e,k).
\end{equation}

Both formulae \eqref{math:fome} and \eqref{math:foma} can be used to
measure the predictive power of an algorithm. The first one offers
more flexibility, since one can pick any condition, while the second one
offers a total estimate over all conditions. So far,
\eqref{math:foma} is the formula used the most in the literature.
Moreover, since the experimental studies conducted by Yeung
et al. show that {\tt FOM}$(k)$ behaves as a decreasing function of
$k$, an adjustment factor has been introduced to properly compare
clustering solutions with different numbers of clusters. A
theoretical analysis by Yeung et al. provides the following
adjustment factor:

\begin{equation}\label{math:adj}
\sqrt{\frac{n-k}{n}}.
\end{equation}

When \eqref{math:adj}  divides \eqref{math:fome},
\eqref{math:fome} and \eqref{math:foma} are referred to as \textit{adjusted} {\tt
FOM}s. The adjusted aggregate {\tt FOM} is used for the experiments in this thesis and, for brevity, it is referred to as  {\tt FOM}.

The use of {\tt FOM} in order to establish how many clusters are present in the data follows the same heuristic methodology outlined for {\tt WCSS}, i.e., one tries to identify the \vir{knee} in the {\tt FOM} plot as a function of the number of clusters. Fig.~\ref{FOM-Example} provides an example, where the {\tt FOM} curve is computed on the dataset of Fig.~\ref{fig:KmeansExample}(a) with K-means-R. In this case, it is easy to see that the predicted value is $k^*=6$, that is a value very far to the correct number of clusters in the dataset.

\medskip
\medskip

\begin{figure}[ht]
\centering
\epsfig{file=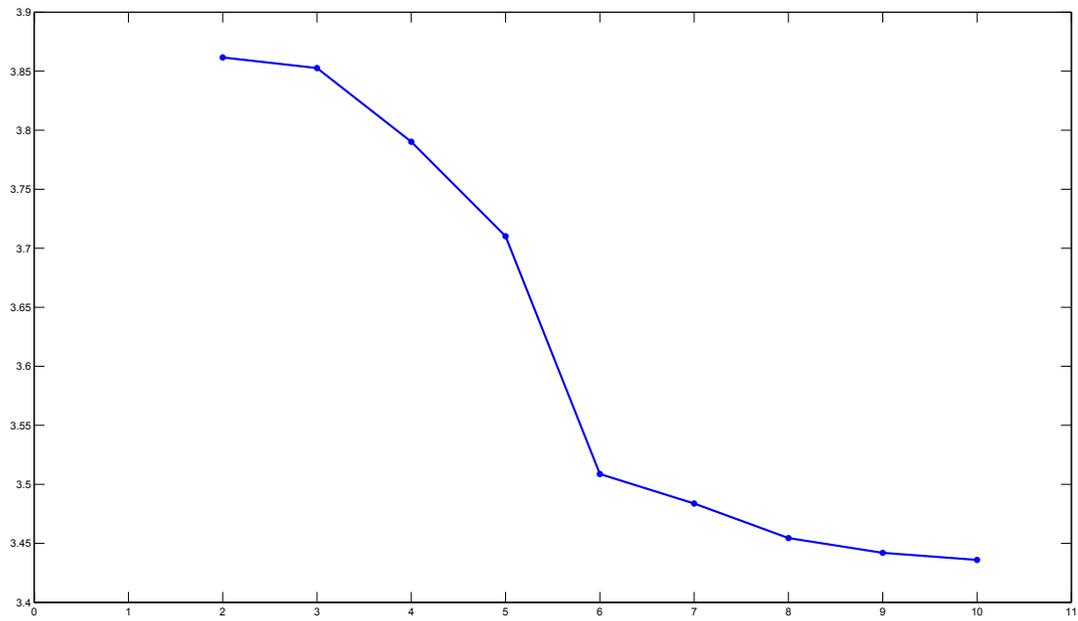,
scale=0.4}
\caption{Plot of the values of {\tt FOM}.}\label{FOM-Example}
\end{figure}

\ignore{
\subsection{Stability Measures}

This category of internal measure, assessing the \vir{stability} of a partitioning, by repeatedly resampling or perturbing the original dataset, and by repeatedly clustering the resulting data. The consistency of the corresponding results provides an estimate of the significance of the clusters obtained from the original dataset.
In the next chapter of this dissertation, we detail in deep this category and we show how the most relevant stability measures proposed in literature are an instance of the same paradigm.
} 

%% file: Chapter3.tex
\chapter{The Stability Measure Paradigm and Its Instances}
\label{chap:Stability}


In this chapter, internal validation measures based on the notion of stability are presented. First, a general algorithmic paradigm is introduced, which can be seen as a generalization of earlier work by Breckenridge and Valentini. Then, it is shown that each of the known stability based measures is an instance of such a novel paradigm. Surprisingly, also {\tt Gap} falls within the new paradigm.

\section{An Intuitive Description}\label{sec:ClusterStability}

All the methods described in this chapter for model selection in clustering are related to the concept of \emph{stability} which is now discussed in intuitive terms.
A \vir{good} algorithm should produce clustering solutions that do not vary much from one sample to another, when data points are repeatedly sampled and clustered. That is, the algorithm must be stable with respect to input randomization. Therefore, the main idea to validate a clustering solution is to use a measure the self-consistency of the data instead of using the classical concepts of isolation and compactness~\cite{Handl05,JainDubes}.

The stability framework can be applied to problems ({\bf Q.2}) and ({\bf Q.3}), detailed in Section~\ref{sec:asseClusterQuality}, and for convenience of the reader reported here again:
\begin{itemize}
\item ({\bf Q.2}) compute a partition of $D$ and assess the confidence of cluster assignments for individual samples; \\

\item ({\bf Q.3}) estimate the number of clusters, if any, in a dataset.
\end{itemize}
These two problems are strongly related, since it is possible to use the former problem to solve the latter.

In order to obtain a stability internal validation method, one needs to specify the following \vir{ingredients}:

\begin{enumerate}
    \item a data generation/pertubation procedure;

    \item a similarity measure between partitions;

    \item a statistics on clustering stability;

    \item rules on how to select the most reliable clustering(s).
\end{enumerate}

\noindent Points 1 and 2 have been addressed in Sections~\ref{sec:statistics}-\ref{sec:Perturbation} and~\ref{sec:ext}, respectively. Points 3 and 4 are the main subject of study of this chapter.
For problem ({\bf Q.2}) and ({\bf Q.3}) a first effort to formalize the steps and the ingredients of a solution based on \vir{stability} are due to Breckenridge~\cite{Breckenridge89} and Valentini~\cite{mosclust}, respectively. In particular, this latter formalization is done via a software library for the statistical computing environment R, which is referred to as {\tt mosclust}. Indeed, this tool provides several macro operations in order to implement a stability internal measure. However, {\tt mosclust} focuses its attention on similarity measures between only two partitions (see point 2) obtained from the same clustering algorithm and it does not provide any macro that takes into account a more general similarity measure.

Extensive experimental results (see~\cite{Dudoit2003,giancarlo08} and Chapter~\ref{chap:6}) show that this class of measures has an excellent predictive power. However there are some drawbacks and open problems associated with their use:

\begin{itemize}

\item[(a)] As shown in~\cite{Handl05}, a given clustering may converge to a suboptimal solution owing to the shape of the data manifold and not to the real structure of the data. Thus, some bias in the stability indices are introduced.

\item[(b)] Ben-David et al.~\cite{Ben-David2006} show that stability methods based on resampling techniques, when cost-based clustering algorithms are used, may fail to estimate $k^{*}$, if the data is not symmetric.

\item[(c)] Stability methods have various parameters that a user needs to specify~\cite{giancarlo08}. Those choices may affect both their time performance and their estimation of $k^{*}$.

\end{itemize}

Whit respect to the problem (b), it is unclear if these results may be extended to other stability based methods or to other more general classes of clustering algorithms. In this dissertation, one focuses on problem (c). Indeed, as it will be shown in Chapter~\ref{chap:6}, repeatedly generating and clustering data, i.e., the main cycle of the stability based methods, has a drastic influence on time performance. Therefore, design of fast approximation algorithms is needed in order to use these measures for large datasets.
It is worth pointing out that in Chapter~\ref{chap:7} an approximation scheme of the stability internal validation measures is proposed.

\ignore{
Although the excellent experimental results obtained with stability-based methods~\cite{Dudoit2003,giancarlo08} there are some drawbacks and open problems associated with these techniques~\cite{giancarlo08}. Indeed, as shown in~\cite{Handl05}, a given clustering may converge to a suboptimal solution owing to the shape of the data manifold and not to the real structure of the data, thus some bias in the stability indices are introduced. Moreover, Ben-David et al.~\cite{Ben-David2006} show that stability methods based on resampling techniques, when cost-based clustering algorithms are used, may fail to estimate $k^{*}$ if the data are not symmetric. However, it is unclear if these results may be extended to other stability based methods or to other more general classes of clustering algorithms. Moreover, it is worth pointing out that the stability methods have various parameters that a user needs to specify~\cite{giancarlo08}. Those choices may affect both time performance and the estimation of $k^{*}$. The performance of the stability based methods are the main practical drawback. Indeed, the repeatedly data generation and clustering data have a drastic influence on the time performance. For this reason design fast approximation is needed.}

In this chapter a generalization of the efforts of Breckenridge and Valentini is proposed via two novel paradigms in order to solve the problems ({\bf Q.2}) and ({\bf Q.3}). The former is described in Section~\ref{sec:Statistic_Stability} and it is referred to as Stability Statistic. The latter is described in Section~\ref{sec:Stability_Measure} is referred to as Stability Measure.
Finally, in Section~\ref{sec:Application} several stability measure are presented as instances of the novel paradigm.

\ignore{
There is an intuitively connection between the two problems. Indeed, the partitions obtained from clustering analysis (problem ({\bf Q.2})) are often used later on for prediction purpose (problem ({\bf Q.3})).

The viceversa is to estimate the number of $k$ clusters into a dataset, and its cluster label for each observation to get a measure of confidence for this cluster assignment if possible. In what follows, this connection will be detailed showing as methodology for the problem ({\bf Q.3}) are a fundamental component of the stability internal validation measure paradigm (problem ({\bf Q.2})).
Indeed, the internal validation methods assessing the stability of a partition via a resample/perturbation of the data (see Section~\ref{sec:Perturbation}) in order to compute a level of significance of the cluster obtained from the original dataset.
}

\section{The Stability Statistic and the Stability Measure Paradigms}\label{sec:paradigm}

Recall from~\cite{JainDubes} that a \emph{statistic} is a function
of the data capturing useful information about it. A
statistic assessing cluster stability is, intuitively, a measure of consistency of
a clustering solution. In turn, information obtained from the statistic is used by the
Stability Measure in order to estimate $k^*$. Since
Stability Statistic is a  \vir{subroutine} of the
Stability Measure paradigm, it is presented first. In what follows,
a statistic is represented by a set $S$ of records. For instance, in
its simplest form, a statistic consists of a single real number,
while in other cases of interest, it is a one or two-dimensional
array of real numbers.

\ignore{ In this section we present two paradigms:
Stability\_Statistic and Stability\_Measure. The relation between
them is as follow: the Stability\_Measure paradigm uses the
Stability\_Statistic paradigm as one of its main \vir{subroutine},
therefore this latter paradigm is described first. The
Stability\_Statistic paradigm collects a statistic for a fixed
number of cluster.}

\subsection{The Stability Statistic Paradigm}\label{sec:Statistic_Stability}

The paradigm for the collection of a statistic on cluster
stability is best presented as a procedure, reported in Fig.~\ref{fig:StabilityStatistic}. Its input parameters and macro
operations are described in abstract form in Figs.~\ref{fig:InputCollect} and~\ref{fig:MacrosCollect}, respectively,
while its basic steps are described below.

\begin{figure}
\[
\setlength{\fboxsep}{11pt}
\setlength{\mylength}{\linewidth}
\addtolength{\mylength}{-2\fboxsep}
\addtolength{\mylength}{-2\fboxrule}
\ovalbox{
\parbox{\mylength}{
\setlength{\abovedisplayskip}{0pt}
\setlength{\belowdisplayskip}{0pt}

\begin{pseudocode}{Stability\_Statistic}{D_0, H, \alpha, \beta, <C_{1},C_2,\ldots,C_t>, k}
 S^{k}=\emptyset\\
 \WHILE H \DO\\
\BEGIN

1. <D_1, D_2,\ldots, D_l> \GETS <\CALL{DGP}{D_0, \beta}, \CALL{DGP}{D_0, \beta}, \ldots, \CALL{DGP}{D_0, \beta}>\\

2. <D_{T,0},D_{T,1}, \ldots, D_{T,l}, D_{L,0},D_{L,1}, \ldots, D_{L,l}> \GETS \CALL{Split}{<D_0,D_1, \ldots, D_l>, \alpha}\\

3. <G> \GETS \CALL{Assign}{<D_{T,0}, D_{T,1},\ldots, D_{T,l}>,<C_{1}, C_{2},\ldots, C_{t}>}\\

4. <C_{i_{1}}, C_{i_{2}},\ldots, C_{i_{q}}> \GETS \CALL{Train}{<G>}\\

5. <\hat{G}> \GETS \CALL{Assign}{<D_{L,0}, D_{L,1},\ldots, D_{L,l}>,<C_{1}, C_{2},\ldots, C_{t}>}\\

6. <P_{1}, P_{2}, \ldots, P_{z}> \GETS \CALL{Cluster}{\hat{G},k}\\

7.~u \GETS \CALL {Collect\_Statistic}{<P_{1}, P_{2}, \ldots, P_{z}>}\\

8.~S^{k} \GETS S^{k} \bigcup \{u\}\\

\END\\
\RETURN{S^{k}}
\end{pseudocode}

}
}
\]
\caption{The \codice{Stability\_Statistic} procedure.}\label{fig:StabilityStatistic}
\end{figure}

\begin{figure}
\[
\setlength{\fboxsep}{8pt}
\setlength{\mylength}{\linewidth}
\addtolength{\mylength}{-2\fboxsep}
\addtolength{\mylength}{-2\fboxrule}
\fbox{
\parbox{\mylength}{
\setlength{\abovedisplayskip}{0pt}
\setlength{\belowdisplayskip}{0pt}

\underline{\codice{Input}}
\begin{itemize}

\item[-] $D_0$: it is the input dataset.

\item[-] $H$: it is a test on the \vir{adequacy} of a statistic $S$, i.e., it evaluates whether $S$ contains enough information. Note that $H$ could simply be a check of as to whether a given number $c$ of iterations has been reached. In what follows, this simple test is denoted as $\hat{H}_c$.

\item[-] $\alpha$: it is a number in the range $[0,1]$.

\item[-] $\beta$: it is a sampling percentage, used by the \codice{DGP} procedure (described in Section~\ref{sec:Perturbation}).

\item[-] $<C_{1},C_2,\ldots,C_t>$: it is a set of procedures, each of which is either a classifier or a clustering algorithm.

\item[-] $k$: it is the number of clusters in which a dataset has to be partitioned.

\end{itemize}
}
}
\]
\caption{List of the input parameters used in the \codice{Stability\_Statistic} procedure.}\label{fig:InputCollect}
\end{figure}


\begin{figure}
\[
\setlength{\fboxsep}{8pt}
\setlength{\mylength}{\linewidth}
\addtolength{\mylength}{-2\fboxsep}
\addtolength{\mylength}{-2\fboxrule}
\fbox{
\parbox{\mylength}{
\setlength{\abovedisplayskip}{0pt}
\setlength{\belowdisplayskip}{0pt}

\underline{\codice{Macro Operations}}
\begin{itemize}

\item[-] \codice{Split}: it takes as input a family of datasets $F_1,F_2,\ldots,F_w$ and a real number $\alpha$ in the range [0,1]. The procedure splits each $F_i$, $1\leq i \leq w$, into two parts according to $\alpha$, referred to as  \emph{learning} and \emph{training} dataset and denoted with $F_{L,i}$ and $F_{T,i}$, respectively. That is, from each $F_i$, $\lceil\alpha n_{i}\rceil$ and $\lfloor(1-\alpha)n_{i}\rfloor$ rows are selected in order to obtain the corresponding $F_{T,i}$ and $F_{L,i}$, respectively, where $n_i$ is the number of rows of $F_i$. Each $F_{T,i}$ and $F_{L,i}$ is given as output.

\item[-] \codice{Assign}: it takes as input a family of datasets and a set of procedures, each of which is either a classifier or a clustering algorithm. It returns a finite set of pairs in which the first element is a dataset and the second one is either a classifier or a clustering algorithm. Such an association is encoded via a bipartite graph $G$, where the datasets are represented by nodes in one partition and procedures in the other partition. Notice that the graph is not a matching, i.e., the same dataset can be assigned to different procedures and viceversa.

\item[-] \codice{Train}: it takes as input a set of pairs $<$dataset, classifier$>$, encoded as a bipartite graph, analogous to the one just discussed. For each pair, it gives as output the classifier trained with the corresponding dataset. Notice that the number $q$ of trained classifiers returned as output is equal to the number of edges in the input graph.

\item[-] \codice{Cluster}: it takes as input a set of pairs $<$dataset, classifier/clustering algorithm$>$ and a positive integer $k$. Again, the set is encoded as a bipartite graph. For each pair, it gives as output a partition in $k$ clusters obtained by the classifier/clustering algorithm on the corresponding input dataset.
    Notice that the number $z$ of partitions returned as output is equal to the number of edges in the input graph.

\item[-] \codice{Collect\_Statistic}: it takes as input a set of partitions. It returns as output the statistic computed on the input set.

\end{itemize}
}
}
\]
\caption{List of the macro operations used in the \codice{Stability\_Statistic} procedure.}\label{fig:MacrosCollect}
\end{figure}

\medskip

A single iteration of the \textbf{while} loop is discussed. The loop is repeated until the condition $H$ is satisfied, i.e., until enough information about the given statistic has been collected. In step 1, a set of perturbed datasets is generated from $D_0$ by a \codice{DGP} procedure (see Section~\ref{sec:Perturbation}). In step 2, $D_0$ and all the datasets generated in the previous step are split in a learning and training dataset, according to the input parameter $\alpha$.
The next two steps train a subset of the classifiers on a subset of the training sets. In step 5, the bipartite graph $\hat{G}$ encodes the association between learning datasets and clustering procedures. In step 6, based on the association encoded by $\hat{G}$, the learning datasets are partitioned. Finally, in step 7, a statistic $S^{k}$ is computed from those partitions and is given as output.\\

In the next subsection,  some instances of this paradigm are discussed.

\subsubsection{Instances}\label{subsec:ExampleStatisticCollect}
Here  three incarnations of the Stability Statistic paradigm are provided. The first is \emph{replicating analysis}, a ground-breaking method due to Breckenridge~\cite{Breckenridge89}. The other two are \emph{BagClust1} and \emph{BagClust2}, due to Dudoit and Fridlyand~\cite{Dudoit2003}. In all three cases, the procedures were proposed to improve a clustering solution for a fixed
value of $k$ (see problem ({\bf Q.2})), rather than to estimate the \vir{true} number of clusters in $D$. However, as one will see in Section~\ref{sec:Application}, \emph{replicating analysis} and \emph{BagClust2} play a key role in many internal stability methods.\\
The presentation of methods in this section is organized as follows: for each example, the input parameters setup is first described (see Fig.~\ref{fig:InputCollect}), then the \codice{Stability\_Statistic} is detailed.

\begin{itemize}

\item[$\bullet$] \emph{Replicating analysis}.
\begin{itemize}

 \item[-] \emph{The input parameters setup}:  $\beta$ is not relevant  and the simple test $\hat{H}_1$ is used to allow only one iteration of the \textbf{while} loop. Moreover, the set of procedures $<C_{1},C_{2},\ldots,C_{t}>$ has  size  two,  i.e., it contains one classifier and one clustering algorithm, referred to as $C_1$ and $C_2$, respectively.\\

 \item[-] \emph{The \codice{Statistics\_Stability} procedure}: step 1 is not performed. In step 2, the \codice{Split} procedure is applied  to $D_0$ only  and it gives as output the training and learning dataset $D_{T,0}$ and $D_{L,0}$, respectively. Then, in steps 3-5, $D_{T,0}$  is used to train the classifier $C_1$. In steps 5 and 6, two partitions $P_1$ and $P_2$ of $D_{L,0}$ are produced,  by $C_1$ and $C_2$, respectively. Finally, in step 7, the \codice{Collect\_Statistic} procedure measures the agreement between the two partitions $P_1$ and $P_2$ via an external index (see Section~\ref{sec:ext}) in order to assess the stability structure of the dataset. For convenience of the reader the \emph{replicating analysis} procedure is given in Fig.~\ref{algo:Breckenridge}.

\end{itemize}

\begin{figure}
\[
\setlength{\fboxsep}{12pt}
\setlength{\mylength}{\linewidth}
\addtolength{\mylength}{-2\fboxsep}
\addtolength{\mylength}{-2\fboxrule}
\ovalbox{
\parbox{\mylength}{
\setlength{\abovedisplayskip}{0pt}
\setlength{\belowdisplayskip}{0pt}

\begin{pseudocode}{Replicating\_Analysis}{D_0, \alpha, <C_{1}, C_{2}>, k}
1. \mbox{ Split the input dataset in $D_{L}$ and $D_{T}$, the  learning and training sets, }\\\mbox{ }\mbox{ }\mbox{ }\mbox{ respectively}.\\
2. \mbox{ Train the classifier }C_{1} \mbox{on }D_{T}.\\
3. \mbox{ Let $P_{1}$ and $P_{2}$ be the partitions of $D_L$ }\\\mbox{ }\mbox{ }\mbox{ }\mbox{ }\mbox{into $k$ cluster with the use of $C_1$ and $C_2$, respectively.}\\
4. \mbox{ Let $e$ be the agreement measure between $P_1$ and $P_2$ obtained }\\\mbox{ }\mbox{ }\mbox{ }\mbox{ via an external index.}\\
\RETURN{e}
\end{pseudocode}
}
}
\]
\caption{The \emph{replicating analysis} procedure.}\label{algo:Breckenridge}
\end{figure}

\item[$\bullet$] \emph{BagClust1}.

\begin{itemize}

 \item[-] \emph{The input parameters setup}: $\hat{H}_c$ is used as test, for a given number of iterations $c$.  The set of procedures $<C_{1},C_{2},\ldots,C_{t}>$ consists only of one clustering algorithm, $\alpha=0$ and $\beta=1$. Moreover, each  \codice{DGP} is an instance of the same bootstrapping subsampling method (see Section~\ref{subsec:Subsampling}). Since $\alpha=0$, the \codice{Split} procedure gives as output only the learning datasets, which are copies of the corresponding input dataset.

 \item[-] \emph{The \codice{Statistics\_Stability} procedure}: in step 1, a single  \codice{DGP} procedure is executed to generate $D_1$. Then, the \codice{Split} procedure takes as input $D_0$ and $D_1$ and it gives as output $D_{L,0}=D_0$ and $D_{L,1}=D_1$. In steps 5 and 6, the clustering procedure is applied to both $D_0$ and $D_1$ in order to obtain the partitions $P_1$ and $P_2$, respectively. The \codice{Collect\_Statistic} procedure permutes the elements assigned to the partition $P_2$ so that there is the maximum overlap with $P_1$. For each iteration of the \textbf{while} loop, the number of overlapping elements are counted and given as output of the method. From that statistic, a new partition is obtained by assigning each element of $D_0$ to a cluster via a majority vote system. That is, each element is assigned to the cluster for which it expressed the maximum of number of preferences. For convenience of the reader the\emph{BagClust1} procedure is given in Fig.~\ref{algo:BagClust1}.
\end{itemize}

\begin{figure}
\[
\setlength{\fboxsep}{12pt}
\setlength{\mylength}{\linewidth}
\addtolength{\mylength}{-2\fboxsep}
\addtolength{\mylength}{-2\fboxrule}
\ovalbox{
\parbox{\mylength}{
\setlength{\abovedisplayskip}{0pt}
\setlength{\belowdisplayskip}{0pt}

\begin{pseudocode}{BagClust1}{D_0, H_c, \beta, <C_{1}>, k}

\FOR i\GETS 0 \TO H_c \DO \\
\BEGIN
1. \mbox{ Generate (via a bootstrap) a data matrix $D_1$}\\
2. \mbox{ }\mbox{Let $P_{1}$ and $P_{2}$ be the corresponding partitions of $D_0$ and $D_1$  }\\\mbox{ }\mbox{ }\mbox{ }\mbox{ }\mbox{into $k$ cluster with the use of $C_1$.}\\
3. \mbox{ Permute the elements assigned to the partition $P_2$ so that there is }\\\mbox{ }\mbox{ }\mbox{ }\mbox{ maximum overlap with $P_1$}\\
4. \mbox{ Let $O_c$ be the number of overlapping elements}\\
\END\\
\RETURN{O_c}

\end{pseudocode}
}
}
\]
\caption{The \emph{BagClust1} procedure.}\label{algo:BagClust1}
\end{figure}

\item[$\bullet$] \emph{BagClust2}.

\begin{itemize}

 \item[-] \emph{The input parameters setup}: as in \emph{BagClust1}.

 \item[-] \emph{The \codice{Statistics\_Stability} procedure}: in step 1, a single  \codice{DGP} procedure is executed to generate $D_1$. Then, the \codice{Split} procedure takes as input $D_0$ and $D_1$ and it gives as output $D_{L,0}=D_0$ and $D_{L,1}=D_1$. In step 5, the bipartite graph $\hat{G}$ consists of only one node per partition, encoding the dataset $D_1$ and the clustering procedure, respectively. In step 6, a clustering partition is obtained from it. Finally, in steps 7 and 8, a dissimilarity matrix $\mathcal{M}$ is computed. Each entry $\mathcal{M}_{i,j}$ of $\mathcal{M}$ is defined as follows:

\begin{equation}\label{eqn:dissBag2}
\mathcal{M}_{i,j} = 1 - \frac{M_{i,j}}{I_{i,j}},
\end{equation}

\noindent where $M_{i,j}$ is the number of times in which items $i$ and
$j$ are in the same cluster and $I_{i,j}$ is the number of times in
which items $i$ and $j$ are in the same learning dataset. The
dissimilarity matrix $\mathcal{M}$ is then used as input to a
clustering procedure in order to obtain a partition. For convenience of the reader the \emph{BagClust2} procedure is given in Fig.~\ref{algo:BagClust2}.

\end{itemize}

\begin{figure}
\[
\setlength{\fboxsep}{12pt}
\setlength{\mylength}{\linewidth}
\addtolength{\mylength}{-2\fboxsep}
\addtolength{\mylength}{-2\fboxrule}
\ovalbox{
\parbox{\mylength}{
\setlength{\abovedisplayskip}{0pt}
\setlength{\belowdisplayskip}{0pt}

\begin{pseudocode}{BagClust2}{D_0, H_c, \beta, <C_{1}>, k}

\FOR i\GETS 0 \TO H_c \DO \\
\BEGIN
1.  \mbox{ Generate (via a bootstrap)  a data matrix $D_1$}\\
2. \mbox{ }\mbox{Let $P_{1}$ be the partition of $D_{1}$ into $k$ clusters with the use of }C_{1}\\
3. \mbox{ Compute $\mathcal{M}$ as in \eqref{eqn:dissBag2}}\\
\END\\
\RETURN{\mathcal{M}}

\end{pseudocode}
}
}
\]
\caption{The \emph{BagClust2} procedure.}\label{algo:BagClust2}
\end{figure}

\end{itemize}

\ignore{
In the following sections, we explain how the statistic collected by the
\codice{Stability\_Statistics} paradigm is used to compute a level
of significance in order to estimate $k^*$.}

\ignore{
{\bf old}

Finally, via the \codice{Collect\_Statistic} the level of stability significance computation is performed. The two most used approach to compute the stability score either (A) via similarity measure, described in the previous section, or (B) it is computed by a directly comparison of the several clustering solutions. In the former, the typically approach is generate, via some data generation/perturbation procedure, two datasets and computed from them two clustering solutions. Finally, consider one of this solution as \vir{gold solution} (or generate it) and compute the similarity measure. For instance, Dudoit and Fridlyand~\cite{CLEST} built a DLDA classifier for the data, via a learning set, then to be used to derive \vir{gold solution} for the training set. That is, the classifier is assumed to be a reliable model for the data.
It is then used to assess the quality of the partitions of the training set obtained by a given clustering algorithm.
In the latter, the clustering solution are directly compared in order to compute the stability score. The most used approach is either via a statistical definition or via some \vir{appropriate} mathematical function. For instance, the most simple statistical way is compute a matrix where either each entry contains the number of times in which the two items are included in the same cluster or the proportion in which a cluster appearers.
Whereas, the Loevinger's measures~\cite{Loevinger} is an example of mathematical function that gives a quality measure of the stability of an individual cluster. Where the stability of an individual cluster is interpreted as a weighted mean of the inherent stabilities in the isolation and cohesion, respectively, of the examined cluster.
}

\subsection{The Stability Measure Paradigm}\label{sec:Stability_Measure}
In this section, the main paradigm of internal
stability methods is described. It is best presented as a procedure, reported in Fig.~\ref{fig:StabilityMeasure}. Its macro operations are described in
abstract form in Fig.~\ref{fig:MacrosMain}, while its basic steps are described below.\\
\begin{figure}
\[
\setlength{\fboxsep}{12pt}
\setlength{\mylength}{\linewidth}
\addtolength{\mylength}{-2\fboxsep}
\addtolength{\mylength}{-2\fboxrule}
\ovalbox{
\parbox{\mylength}{
\setlength{\abovedisplayskip}{0pt}
\setlength{\belowdisplayskip}{0pt}

\begin{pseudocode}{Stability\_Measure}{k_{min}, k_{max}, D, H, \alpha, \beta, <C_{1},C_{2},\ldots,C_{t}>}
\mbox{ }\FOR k\GETS k_{min} \TO k_{max} \DO \\
\mbox{ }\BEGIN
1.\mbox{ }\mbox{ }\mbox{ }\mbox{ }S^{k} \GETS \CALL{Stability\_Statistics}{D, H,  \alpha, \beta, <C_{1},C_{2},\ldots,C_{t}>, k}\\

2.\mbox{ }\mbox{ }\mbox{ }\mbox{ }R^{k} \GETS \CALL{Synopsis}{S^{k}}\\

\mbox{ }\END\\

3. \mbox{ }k^{*} \GETS \CALL{Significance\_Analysis}{R^{k_{min}},\ldots,R^{k_{max}}}\\

\RETURN{k^{*}}
\end{pseudocode}
}
}
\]
\caption{The \codice{Stability\_Measure} procedure.}\label{fig:StabilityMeasure}
\end{figure}
\begin{figure}
\[
\setlength{\fboxsep}{8pt}
\setlength{\mylength}{\linewidth}
\addtolength{\mylength}{-2\fboxsep}
\addtolength{\mylength}{-2\fboxrule}
\fbox{
\parbox{\mylength}{
\setlength{\abovedisplayskip}{0pt}
\setlength{\belowdisplayskip}{0pt}

\underline{\codice{Macro Operations}}
\begin{itemize}

\item[-] \codice{Synopsis}: it takes as input a statistic and
returns as output a concise description of it.

\item[-] \codice{Significance\_Analysis}: it takes as input all the
statistics/information collected as returned by the \codice{Synopsis} procedure. It computes the significance level
of each statistic. It returns as output, explicitly or
implicitly, a prediction about $k^*$. For instance,  an implicit prediction of the value of $k^*$ can be the plot of a histogram or of a curve, as in many methods described in the next section.

\end{itemize}
}
}
\]
\caption{List of the macro operations used in \codice{Stability\_Measure} procedure.}\label{fig:MacrosMain}
\end{figure}
For each $k$ in the range $[k_{min},k_{max}]$ the paradigm collects the statistics $S^{k}$ computed by the \codice{Stability\_Statistic} procedure, then a concise description $R^{k}$ of the statistic $S^{k}$ is computed via the \codice{Synopsis} procedure. Finally, an explicit or implicit prediction of the value of $k^*$ is computed by \codice{Significance\_Analysis} and it is given as output.

In the remaining part of this section, only two examples of the Stability Measure paradigm are detailed. The other incarnations proposed in the literature are discussed in the
remaining part of the chapter.

\subsubsection{Instances}\label{subsec:ExampleStatisticStatistic}
The presentation of methods in this section is organized as follows: for each method, the input parameters setup is first described (see Fig.~\ref{fig:InputCollect}), then the \codice{Stability\_Statistic} and the \codice{Stability\_Measure} procedures are detailed.

\begin{itemize}

\item[$\bullet$] \emph{Model Explorer} by Ben-Hur et al.~\cite{BenHur02} ({\tt ME} for
short) is the simplest incarnation of the Stability Measure paradigm and it can be derived in the following way.

\begin{itemize}

\item \emph{Input parameters setup}: $\hat{H}_c$ is used as test, for a given number of iteration $c$, $\alpha = 0$, $\beta \in [0.6,0.9]$ and the set of procedures $<C_1,C_2,\ldots,C_t>$ consists only of one clustering algorithm $C_1$.

\item \emph{The \codice{Statistics\_Stability} procedure}: in step 1, $D_1$ and $D_2$ are generated by two \codice{DGP} procedures, where each  procedure is an instance of subsampling (see Section~\ref{subsec:Subsampling}). Since $\alpha=0$, the \codice{Split} procedure copies those datasets into the corresponding learning datasets, while  steps 3 and 4 are not performed. In step 5, the graph $\hat{G}$, obtained as output of the \codice{Assign} procedure, encodes two relations: $<D_{L,1},C_1>$ and $<D_{L,2},C_1>$. In step 6, two clustering solutions $P_1$ and $P_2$ are obtained from $<D_{L,1},C_1>$ and $<D_{L,2},C_1>$, respectively. The \codice{Collect\_Statistic} procedure computes the level of agreement  between the two partitions via an external index (see Section~\ref{sec:ext}), but restricted to the common elements of $D_1$ and $D_2$. In step 8, this level of agreement is stored into a one dimensional array $S^{k}$. That is, for each iteration of the \textbf{while} loop, the value returned by the external index is stored in the corresponding entry of $S^{k}$.

\item \emph{The \codice{Stability\_Measure} procedure}: for each $k$, it computes the array $S^{k}$ via the \codice{Statistic\_Stability} procedure while the \codice{Synopsis} procedure performs  a copy of the collected statistic. Finally, the \codice{Significance\_Analysis} procedure provides an implicit estimation of $k^*$: each $k\in [k_{min},k_{max}]$, $R^{k}$ and its values are histogrammed separately. Then, the optimal number of clusters $k^*$ is predicted to be the lowest value of $k$ such that the $R^{k}$ value distribution is close to one and $R^{k+1}$ value distribution is in a wider range of values. An example of the number of clusters prediction is given in Fig.~\ref{fig:MEexample}, where {\tt ME} is computed on the dataset of Fig.~\ref{fig:KmeansExample}(a) with K-means-R (see Section~\ref{sec:Partional}) for $k\in[2,5]$.
\end{itemize}

For convenience of the reader the {\tt ME} procedure is given in Fig.~\ref{algo:ME}.

\begin{figure}
\[
\setlength{\fboxsep}{12pt}
\setlength{\mylength}{\linewidth}
\addtolength{\mylength}{-2\fboxsep}
\addtolength{\mylength}{-2\fboxrule}
\ovalbox{
\parbox{\mylength}{
\setlength{\abovedisplayskip}{0pt}
\setlength{\belowdisplayskip}{0pt}

\begin{pseudocode}{ME}{H_{c}, <C_{1}>, D, k_{max}}
\FOR k\GETS 2 \TO k_{max} \DO \\
\BEGIN
\mbox{ }\mbox{ }\mbox{ }\mbox{ }\FOR i\GETS 1 \TO H_{c} \DO \\
\mbox{ }\mbox{ }\mbox{ }\mbox{ }\BEGIN
1.\mbox{ }\mbox{ }\mbox{ }\mbox{ }\mbox{ Generate (via subsampling)  two data matrices $D_1$ and $D_2$}\\
2.\mbox{ }\mbox{ }\mbox{ }\mbox{ }\mbox{ Let $P_{1}$ and $P_{2}$ be the corresponding partitions of $D_1$ and $D_2$  }\\\mbox{ }\mbox{ }\mbox{ }\mbox{ }\mbox{ }\mbox{ }\mbox{ into $k$ cluster with the use of $C_1$, respectively}\\
3.\mbox{ }\mbox{ }\mbox{ }\mbox{ }\mbox{ Let $S^{k}(i)$ be the level of agreement  between $P_1$ and $P_2$ via an}\\\mbox{ }\mbox{ }\mbox{ }\mbox{ }\mbox{ }\mbox{ }\mbox{  external index but restricted to the common elements of $D_1$ and $D_2$} \\
\mbox{ }\mbox{ }\mbox{ }\mbox{ }\END\\
\END\\
4.\mbox{  Plot separately the histogram of $S^{k}$ values and return a prediction for }k^{*}
\end{pseudocode}
}
}
\]
\caption{The {\tt ME} procedure}\label{algo:ME}
\end{figure}

\begin{figure}[ht]
\begin{center}
\epsfig{file=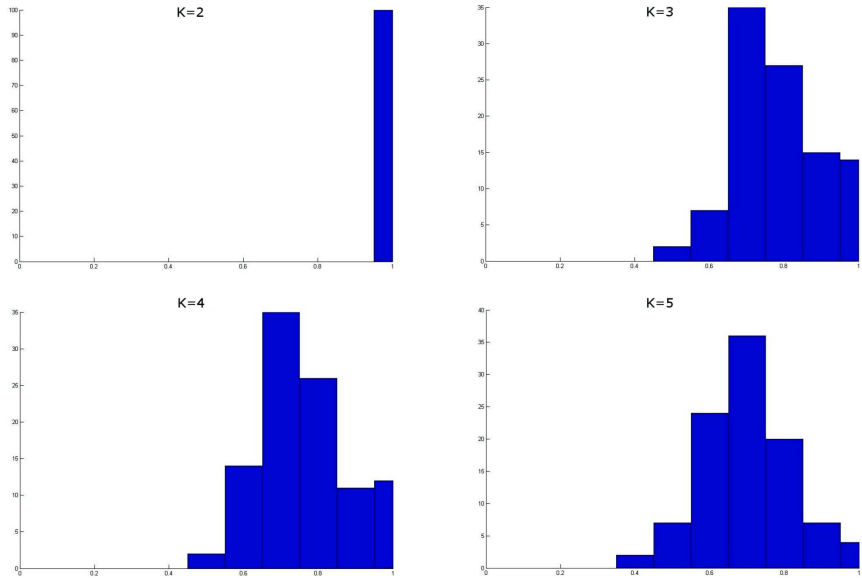,scale=0.2}
\end{center}
\caption{The histograms plotting the $R^{k}$ values distribution for increasing values of $k$. The prediction of $k^*$ correspond to correct number of cluster, i.e., $k^{*} = 2$.}\label{fig:MEexample}
\end{figure}

\item[$\bullet$] \emph{MOSRAM} by  Bertoni and Valentini~\cite{BertoniV07} is strongly related to {\tt ME}, where the most significant change is in the \codice{Significance\_Analysis} procedure.
Indeed, it estimates automatically the \vir{true} number of clusters, and in addition, it detects significant and possibly multi-level structures simultaneously present in $D$ (e.g. hierarchical structures - see Fig.~\ref{fig:datasethier}). It can be derived from the Stability Measure paradigm as follows.

    \begin{itemize}

    \item \emph{The input parameters setup}: as in {\tt ME}.

    \item \emph{The \codice{Statistics\_Stability} procedure}: it is the same proposed in {\tt ME}, except that the two \codice{DGP} procedures performed in step 1 are both an instance of randomized mapping (see Section~\ref{subsec:SubspaceMethods}).

    \item \emph{The \codice{Stability\_Measure} procedure}: in step 2, each $R^{k}$ given as output by the \codice{Synopsis} procedure is an average of the statistics $S^{k}$ computed in step 1. Intuitively, if the value of $R^{k}$ is close to 1, then the clustering solution is stable. Moreover, in order to detect significant and possibly multi-level structures that are simultaneously present in $D$, a statistical hypothesis test is applied. The \codice{Significance\_Analysis} procedure performs a $\chi^2$-based test in order to estimate $k^{*}$ as follows. Let $R=\{R^{k_{min}},\ldots,R^{k_{max}}\}$ and let $\tau$ be a significance level. The null hypothesis $H_0$ considers the set of $k$-clusterings as equally reliable, while the alternative hypothesis $H_1$ considers the set of $k$-clusterings as not equally reliable. When $H_{0}$ is rejected at $\tau$ significance level, it means that at least one $k$-clustering significantly differs from the others. The procedure sorts the values in $R$, and a $\chi^2$-based test is repeated until no significant difference is detected or the only remaining clustering is the top-ranked in $R$. At each iteration, if a significant difference is detected, the bottom-ranked value is removed from the set $R$. Therefore, the \codice{Significance\_Analysis} gives as output the set of the remaining (top sorted) $k$-clusterings that corresponds to the set of the estimate \vir{true} number of clusters (at $\tau$ significance level). For convenience of the reader the \emph{MOSRAM} procedure is given in Fig.~\ref{algo:Mosram}.

\begin{figure}
\[
\setlength{\fboxsep}{12pt}
\setlength{\mylength}{\linewidth}
\addtolength{\mylength}{-2\fboxsep}
\addtolength{\mylength}{-2\fboxrule}
\ovalbox{
\parbox{\mylength}{
\setlength{\abovedisplayskip}{0pt}
\setlength{\belowdisplayskip}{0pt}

\begin{pseudocode}{MOSRAM}{H_{c}, <C_{1}>, D, k_{max}}
\FOR k\GETS 2 \TO k_{max} \DO \\
\BEGIN
\mbox{ }\mbox{ }\mbox{ }\mbox{ }\FOR i\GETS 1 \TO H_{c} \DO \\
\mbox{ }\mbox{ }\mbox{ }\mbox{ }\BEGIN
1.\mbox{ }\mbox{ }\mbox{ }\mbox{ }\mbox{ Generate (via randomized mapping) two data matrices $D_1$ and $D_2$}\\
2.\mbox{ }\mbox{ }\mbox{ }\mbox{ }\mbox{ Let $P_{1}$ and $P_{2}$ be the corresponding partitions of $D_1$ and $D_2$ }\\\mbox{ }\mbox{ }\mbox{ }\mbox{ }\mbox{ }\mbox{ }\mbox{ into $k$ cluster with the use of $C_1$, respectively}\\
3.\mbox{ }\mbox{ }\mbox{ }\mbox{ }\mbox{ Let $S^{k}(i)$ be the level of agreement  between $P_1$ and $P_2$ via an}\\\mbox{ }\mbox{ }\mbox{ }\mbox{ }\mbox{ }\mbox{ }\mbox{ external index but restricted to the common elements of $D_1$ and $D_2$} \\
\mbox{ }\mbox{ }\mbox{ }\mbox{ }\END\\
\END\\
4.\mbox{  Perform a $\chi^2$-based test in order to estimate }k^{*}
\end{pseudocode}
}
}
\]
\caption{The \emph{MOSRAM} procedure}\label{algo:Mosram}
\end{figure}

\begin{figure}
\begin{center}
\epsfig{file=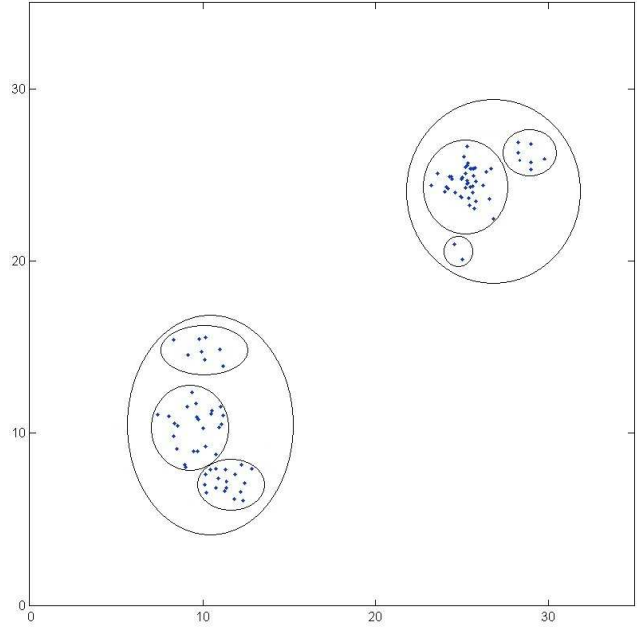,scale=0.4}
\end{center}
\caption{An example of hierarchical structures in a dataset.}\label{fig:datasethier}
\end{figure}

    \end{itemize}

\end{itemize}

\section{Further Instances of the Stability Measure Paradigm}\label{sec:Application}
In this section several incarnations of the Stability Measure paradigm are detailed, describing for each method the input parameters and macro operations listed in Figs.~\ref{fig:InputCollect},~\ref{fig:MacrosCollect} and~\ref{fig:MacrosMain}.
The section follows the same organization of Section~\ref{subsec:ExampleStatisticStatistic}.

\ignore{
For completeness, we repot that the implementation of many methods
proposed here can be engineered in order to have a
speedup in terms of time.

Indeed, our aim is to
show how all the methods perform two main steps: (1) the statistics are collected and (2) those statistics are used in order to provide either
explicit or implicit estimations of $k^{*}$.
The section is organized as follows: for each method, we first describe the input parameters setup, then the \codice{Stability\_Statistic} and the \codice{Stability\_Measure} procedures are detailed.
}

\subsection{Consensus Clustering}\label{subsec:Consensus}

Consensus Clustering by Monti et al.~\cite{Monti03} ({\tt Consensus} for short) is a reference method in internal validation measures, with a prediction power far better than other established methods~\cite{giancarlo08, Monti03}.
\begin{itemize}

    \item[-] \emph{The input parameters setup}: it uses the same input parameters setup of {\tt ME} for the \codice{Statistic\_Stability} procedure (see Fig.~\ref{fig:InputCollect}). Moreover, as in {\tt ME}, the \codice{DGP} procedure is an instance of subsampling (see Section~\ref{subsec:Subsampling}).

    \item[-] \emph{The \codice{Statistics\_Stability} procedure}: it is strongly related to {\em BagClust2}, where the most significant change is that the matrix $\mathcal{M}$ is a similarity instead of a dissimilarity matrix.

   \item[-] \emph{The \codice{Stability\_Measure} procedure}: as in {\tt ME}, the \codice{Synopsis} procedure performs a copy of the collected statistic and the \codice{Significance\_Analysis} procedure provides an implicit estimation of $k^*$, as detailed next. Based on the collected statistics, for each $k$, Monti et al. define a value $A(k)$ measuring the level of stability in cluster assignments, as reflected by the matrix $\mathcal{M}$. Formally,
\end{itemize}

$$
A(k) = \sum_{i=2}^{n} [x_{i} - x_{i-1}]CDF(x_{i})
$$

where $CDF$ is the empirical cumulative distribution defined over the range $[0,1]$, as follows:

$$
CDF(c)=\frac{\sum\limits_{i<j}l\{{\cal M}_{i,j}\leq c\}}{n(n-1)/2}
$$

with $l$ equal to 1 if the condition is true and 0 otherwise.  Finally, based on $A(k)$, one can define:

\begin{displaymath}
\Delta(k)=\left\{\begin{array}{lr}
A(k)   & k=2,  \\
\frac{A(k+1)-A(k)}{A(k)} & k>2. \\
\end{array}\right.
\end{displaymath}

Moreover, Monti et al. suggest the use of the function $\Delta'$ for non-hierarchical algorithms. It is defined as $\Delta$ but one uses $A'(k) = \underset{k'\in [2,k]}{\mathop{\max }}\,A(k')$. The reason is the following: $A(k)$ is a value that is
expected to behaves like a non-decreasing function of $k$, for
hierarchical algorithms. Therefore $\Delta(k)$ would be expected to
be positive or, when negative, not too far from zero. Such a
monotonicity of $A(k)$ is not expected for non-hierarchical
algorithms. Therefore, another definition of $\Delta$ is needed to
ensure a behavior of this function analogous to the hierarchical
algorithms.

Assuming that one has computed the $\Delta$ curve for a given dataset,  the value of $k^{*}$ can be obtained by using the following intuitive idea,  also based on experimental observations.

\begin{itemize}
\item[(i)] For each $k \leq k^{*}$,  the area $A(k)$ markedly increases. This results in an analogous pronounced decrease of the $\Delta$ curve.

\item[(ii)]For $k>k^{*}$, the area $A(k)$ has no meaningful increases. This results in a stable plot of the $\Delta$  curve.
\end{itemize}

From this behavior,  the \vir{rule of thumb} to identify $k^{*}$ is: take as $k^{*}$ the abscissa corresponding to the smallest non-negative value where the curve starts to stabilize; that is, no big variation in the curve takes place from that point on. However, the behavior of the $\Delta$ curve could give an ambiguous estimation of $k^*$. Therefore,	it is advisable to combine the information given by the $\Delta$ curve with an estimation provided by the plot of the $CDF$ curves. In this latter estimation,  $k^*$ is the value of $k$ where the area under the $CDF$ curves does not change more, i.e., the gap between the curves stay almost constant. An example is given in Fig.~\ref{fig:Consensusexample}.

The same considerations and rule applies to the prediction of $k^*$ via $\Delta'$. However, as the  experiments performed in Chapter~\ref{chap:6} bring to light, for the partitional algorithms  used in this dissertation, $\Delta'$  displays nearly the same monotonicity properties of $\Delta$, when used on hierarchical algorithms. The end result is that $\Delta$ can be used for both types of algorithms. It is worth pointing out that, to the best of our knowledge, Monti et al. defined the function $\Delta'$, but they did not experiment with it, since their experimentation was limited to hierarchical algorithms. For convenience of the reader the {\tt Consensus} procedure is given in Fig.~\ref{algo:Consensus}.

For completeness, it may be of interest to the reader to report that  Brunet et al.~\cite{BrunetNMF} propose a different approach to estimate $k^*$ based on the dispersion of the matrix $\mathcal{M}$. Indeed, Brunet et al. compute the cophenetic correlation coefficient $\rho$~\cite{JainDubes} for each $k$. Based on the observation on how the value of $\rho$ changes as $k$ increases, the \vir{rule of thumb} to identify $k^{*}$ is: take as $k^{*}$ the abscissa corresponding to the value where the curve starts to fall.

Finally, it can be useful to observe that the matrix  $\mathcal{M}$  can be naturally thought of as a similarity measure. Accordingly, via standard techniques, it can be transformed into a (pseudo) distance matrix that can be used by clustering algorithms as in {\em BagClust2}  in order to solve problem ({\bf Q.2}).

\ignore{
\begin{figure}
\[
\setlength{\fboxsep}{12pt}
\setlength{\mylength}{\linewidth}
\addtolength{\mylength}{-2\fboxsep}
\addtolength{\mylength}{-2\fboxrule}
\ovalbox{
\parbox{\mylength}{
\setlength{\abovedisplayskip}{0pt}
\setlength{\belowdisplayskip}{0pt}

\begin{pseudocode}{Consensus}{H_{c}, <C_{1}>, D, k_{max}}
\FOR k\GETS 2 \TO k_{max} \DO \\
\BEGIN
\mbox{ }\mbox{ }\mbox{ }\mbox{ }\FOR i\GETS 1 \TO H_{c} \DO \\
\mbox{ }\mbox{ }\mbox{ }\mbox{ }\BEGIN
1.\mbox{ }\mbox{ }\mbox{ }\mbox{ }\mbox{ }D_{1} \GETS \CALL{DGP}{D}\\
2.\mbox{ }\mbox{ }\mbox{ }\mbox{ }\mbox{ }<\hat{G}> \GETS \CALL{Assign}{<D_{1}>,<C_{1}>}\\
3.\mbox{ }\mbox{ }\mbox{ }\mbox{ }\mbox{ }<P_{i}> \GETS \CALL{Cluster}{\hat{G},k}\\
4.\mbox{ }\mbox{ }\mbox{ }\mbox{ }\mbox{ Based on }P_{i}\mbox{ compute the connectivity matrix }M^{k}_{i}\\
\mbox{ }\mbox{ }\mbox{ }\mbox{ }\END\\
5.\mbox{ }\mbox{ }\mbox{ Compute the consensus matrix }\mathcal{M}^{k}\\
\END\\
6.\mbox{ Based on the }k_{max}-1 \mbox{ consensus
matrices, return a prediction for }k^{*}
\end{pseudocode}
}
}
\]
\caption{The \codice{Consensus} procedure}\label{algo:Consensus}
\end{figure}
}

\begin{figure}
\[
\setlength{\fboxsep}{12pt}
\setlength{\mylength}{\linewidth}
\addtolength{\mylength}{-2\fboxsep}
\addtolength{\mylength}{-2\fboxrule}
\ovalbox{
\parbox{\mylength}{
\setlength{\abovedisplayskip}{0pt}
\setlength{\belowdisplayskip}{0pt}

\begin{pseudocode}{Consensus}{H_{c}, <C_{1}>, D, k_{max}}
\FOR k\GETS 2 \TO k_{max} \DO \\
\BEGIN
\mbox{ }\mbox{ }\mbox{ }\mbox{ }\FOR i\GETS 1 \TO H_{c} \DO \\
\mbox{ }\mbox{ }\mbox{ }\mbox{ }\BEGIN
1.\mbox{ }\mbox{ }\mbox{ }\mbox{ }\mbox{ }\mbox{ Generate (via a subsampling)  a data matrix $D_1$}\\
2.\mbox{ }\mbox{ }\mbox{ }\mbox{ }\mbox{ }\mbox{ }\mbox{Let $P_{1}$ be the partition of $D_{1}$ into $k$ clusters with the use of }C_{1}\\
3.\mbox{ }\mbox{ }\mbox{ }\mbox{ }\mbox{ }\mbox{ Based on }P_{1}\mbox{ compute the connectivity matrix }M^{k}_{i}\\
\mbox{ }\mbox{ }\mbox{ }\mbox{ }\END\\
4.\mbox{ }\mbox{ }\mbox{ Compute the consensus matrix }\mathcal{M}^{k}\\
\END\\
5.\mbox{ Based on the }k_{max}-1 \mbox{ consensus
matrices, return a prediction for }k^{*}
\end{pseudocode}
}
}
\]
\caption{The \codice{Consensus} procedure}\label{algo:Consensus}
\end{figure}

\begin{figure}[ht]
\centering
\epsfig{file=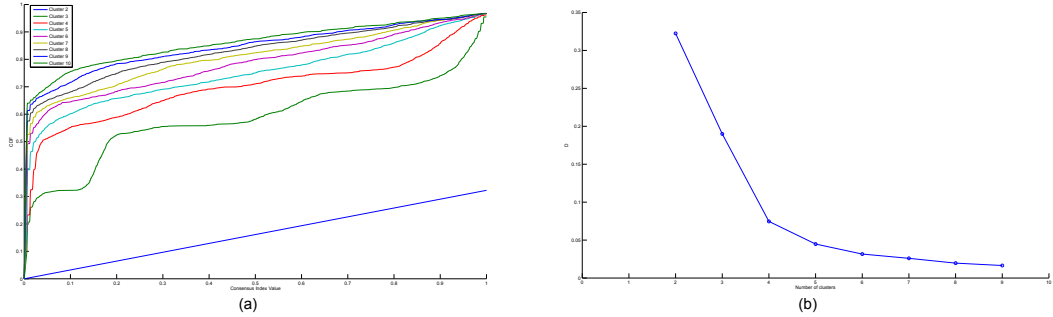,scale=0.18}
\caption{An example of number of cluster prediction with the use of {\tt Consensus}. The experiment is derived from the dataset of Fig.~\ref{fig:KmeansExample}(a) with $k^*=2$, with use of the K-means-R clustering algorithm. The plots of the CDF curves is shown in (a), yielding a monotonically increasing value of $A$, as a function of $k$. The plot of the $\Delta$ curve is shown in (b), where the flattening effect corresponding to $k=4$, while the gap of the area under the CDF curves is evident for $k \geq k^{*} = 2$.}\label{fig:Consensusexample}
\end{figure}

\subsection{Levine and Domany}
This method is due to Levine and Domany~\cite{LevineDomany}.

 \begin{itemize}

    \item[-] \emph{The input parameters setup}: as in {\tt ME}.

    \item[-] \emph{The \codice{Statistics\_Stability} procedure}: it is strongly related to {\em BagClust2}. Indeed, for each iteration, the method computes as statistic a connectivity matrix in which each entry is 1, if the two elements are in the same cluster and 0 otherwise. Moreover, the collected  statistic $S^{k}$ is a set of matrices, $S^{k}=\{S^{k}_{0},\ldots,S^{k}_{c}\}$. Matrix $S^{k}_0$ corresponds to the connectivity matrix for $D_0$, and matrix $S_{i}^{k}$, for $1\leq i\leq c$, corresponds to the connectivity matrix for the dataset $D_1$ generated by the \codice{DGP} procedures at the corresponding iteration. \ignore{ the datasets generated by the \codice{DGP} procedures.}

 \item[-] \emph{The \codice{Stability\_Measure} procedure}: the \codice{Synopsis} procedure compares the collected statistic, via the following formula:

\begin{equation}\label{eqn:LD01}
R^{k} = \ll\delta_{S^{k}_0,S^{k}_{i}}\gg
\end{equation}

\noindent where $\ll\cdot\gg$ is a twofold averaging. That is, for each $S^{k}_{i}$, an average is computed over all pairs which are in the same cluster in the original dataset and have been both selected in the same resample. Then, an average for all $S^{k}_{i}$ is computed. The \codice{Significance\_Analysis} procedure gives as output $k^*$ as the value of $k$ with the local maximum of $R^{k}$, for $k\in[k_{min},k_{max}]$.
\end{itemize}
For convenience of the reader the procedure proposed by Levine and Domany is given in Fig.~\ref{algo:LD01}.

\begin{figure}
\[
\setlength{\fboxsep}{12pt}
\setlength{\mylength}{\linewidth}
\addtolength{\mylength}{-2\fboxsep}
\addtolength{\mylength}{-2\fboxrule}
\ovalbox{
\parbox{\mylength}{
\setlength{\abovedisplayskip}{0pt}
\setlength{\belowdisplayskip}{0pt}

\begin{pseudocode}{LD01}{H, <C_{1}>, D, k_{max}}
\FOR k\GETS 2 \TO k_{max} \DO \\
\BEGIN
1.\mbox{ }\mbox{ }\mbox{ }\mbox{Let $P_{1}$ be the partition of $D$ into $k$ clusters with the use of }C_{1}\\
2.\mbox{ }\mbox{ }\mbox{ Based on }P_{1}\mbox{ compute the connectivity matrix }S^{k}_{0}\\
\mbox{ }\mbox{ }\mbox{ }\mbox{ }\mbox{ }\mbox{ }\FOR i\GETS 1 \TO H \DO \\
\mbox{ }\mbox{ }\mbox{ }\mbox{ }\mbox{ }\mbox{ }\mbox{ }\BEGIN
3.\mbox{ }\mbox{ }\mbox{ }\mbox{ }\mbox{ }\mbox{ }\mbox{ }\mbox{ Generate (via a subsampling)  a data matrix $D_1$}\\
4.\mbox{ }\mbox{ }\mbox{ }\mbox{ }\mbox{ }\mbox{ }\mbox{ }\mbox{ Let $P_{1}$ be the partition of $D_{1}$ into $k$ clusters with the use of }C_{1}\\
5.\mbox{ }\mbox{ }\mbox{ }\mbox{ }\mbox{ }\mbox{ }\mbox{ }\mbox{ Based on }P_{1}\mbox{ compute the connectivity matrix }S^{k}_{i}\\
\mbox{ }\mbox{ }\mbox{ }\mbox{ }\mbox{ }\mbox{ }\mbox{ }\END\\
6.\mbox{ }\mbox{ }\mbox{ Compute } R^{k}\mbox{, as defined in }\eqref{eqn:LD01}\\
\END
\end{pseudocode}
}
}
\]
\caption{The Levine and Domany procedure}\label{algo:LD01}
\end{figure}

\subsection{Clest}
{\tt Clest}, proposed by Dudoit and Fridlyand~\cite{CLEST},  generalizes in many aspects {\em replicating analysis} by Breckenridge (see Section~\ref{sec:Statistic_Stability}). It can be regarded as a clever combination of hypothesis testing and resampling techniques. It estimates  $k^{*}$ by iterating the
following: randomly partition the original dataset in a \emph{learning} set and \emph{training} set. The learning set is used to build  a classifier $\mathcal{C}$ for the data, then to be
used to derive \vir{gold standard} partitions of the training set. That is, the classifier is assumed to be a reliable model for the data. It is then used to assess the quality of  the partitions of the training set obtained by  a given clustering algorithm.

 \begin{itemize}

    \item[-] \emph{The input parameters setup}: it uses the same input parameters of \emph{replicating analysis}, except for the test condition $H$, where in this case $c$ iterations of the \textbf{while} loop are allowed, for a given integer $c>1$.

     \item[-] \emph{The \codice{Statistics\_Stability} procedure}: it corresponds to the {\em replicating analysis}. Therefore, the set $S^{k}$ of records is a one dimensional array, in which each entry stores the value of the external index for the corresponding iteration.

      \item[-] \emph{The \codice{Stability\_Measure} procedure}: the \codice{Synopsis} procedure computes $R^{k}$ as the median of the values stored in $S^{k}$. The \codice{Significance\_Analysis} procedure proposed in {\tt Clest} is best presented as a procedure which is given in Fig.~\ref{fig:SignificanceClest} and it is outlined next. The first step of the procedure generates a new dataset via the \codice{DGP} procedure that in this case it is an instance of null models (see Section~\ref{sec:statistics}), the $p_{max}$ is a \vir{significance level} threshold and $d_{min}$ is a minimum allowed difference between \vir{computed and expected} values. It is worth pointing out that the \codice{Significance\_Analysis} procedure provides an explicit prediction of $k^{*}$.
\end{itemize}
For convenience of the reader the {\tt Clest} procedure is given in Fig.~\ref{algo:ClestOriginal}. It is worth pointing out that steps 1-7 correspond to the \codice{Stability\_Statistics} detailed above. Step 8 is the \codice{Synopsis} call in the \codice{Stability\_Measure} procedure, and the next steps corresponds to the \codice{Significance\_Analysis} procedure described above (see Fig.~\ref{fig:SignificanceClest}).

\begin{figure}
\[
\setlength{\fboxsep}{12pt}
\setlength{\mylength}{\linewidth}
\addtolength{\mylength}{-2\fboxsep}
\addtolength{\mylength}{-2\fboxrule}
\ovalbox{
\parbox{\mylength}{
\setlength{\abovedisplayskip}{0pt}
\setlength{\belowdisplayskip}{0pt}

\begin{pseudocode}{Significance\_Analysis}{R^{k_{min}},\ldots,R^{k_{max}}}
\FOR i\GETS 0 \TO B_{0} \DO \\

\BEGIN
1.\mbox{ }\mbox{ }\mbox{ }D^{i} \GETS \CALL{DGP}{D}\\

2.\mbox{ }\mbox{ }\mbox{ }\FOR k\GETS k_{min} \TO k_{max} \DO \\

\mbox{ }\mbox{ }\mbox{ }\mbox{ }\mbox{ }\BEGIN
3.\mbox{ }\mbox{ }\mbox{ }\mbox{ }\mbox{ }\mbox{ }\mbox{ }\mbox{ }S^{k} \GETS \CALL{Stability\_Statistics}{D^{i}, H, \beta, <C_{1},C_2,\ldots,C_t>, k}\\

4.\mbox{ }\mbox{ }\mbox{ }\mbox{ }\mbox{ }\mbox{ }\mbox{ }\mbox{ }R_{i}^{k} \GETS \CALL{Synopsis}{S^{k}}\\
\mbox{ }\mbox{ }\mbox{ }\mbox{ }\mbox{ }\END\\

\END\\

5.\mbox{ }\FOR k\GETS k_{min} \TO k_{max} \DO \\
\mbox{ }\mbox{ }\mbox{ }\BEGIN
6.\mbox{ }\mbox{ }\mbox{ }\mbox{ }\mbox{ }\mbox{ }t_k^{0} \GETS \mbox{Compute the average of the }R^{k}_{i}\mbox{ values}\\
7.\mbox{ }\mbox{ }\mbox{ }\mbox{ }\mbox{ }\mbox{ }p_k \GETS \mbox{Compute the p-value of }R^{k}\\
8.\mbox{ }\mbox{ }\mbox{ }\mbox{ }\mbox{ }\mbox{ }d_k \GETS R^{k} - t_{k}{0}\\

\mbox{ }\mbox{ }\mbox{ }\END\\

9.\mbox{ }\mbox{Define a set } K=\{k_{min}\leq k \leq k_{max}: p_{k}\leq p_{max} \mbox{ and } d_{k}\geq d_{min}\}\\

10.\mbox{ }\IF K = \emptyset \THEN k^{*}\GETS 1
\ELSE  k^{*} \GETS \argmax\limits_{k\in K} d_{k}\\

\RETURN{k^{*}}
\end{pseudocode}
}
}
\]
\caption{Implementation of the \codice{Significance\_Analysis} procedure proposed by Dudoit and Fridlyand for {\tt Clest}.}\label{fig:SignificanceClest}
\end{figure}

\begin{figure}
\[
\setlength{\fboxsep}{12pt}
\setlength{\mylength}{\linewidth}
\addtolength{\mylength}{-2\fboxsep}
\addtolength{\mylength}{-2\fboxrule}
\ovalbox{
\parbox{\mylength}{
\setlength{\abovedisplayskip}{0pt}
\setlength{\belowdisplayskip}{0pt}

\begin{pseudocode}{Clest}{B_{0}, H_{c},  <C_{1},C_{2}>, E, D,  k_{max}, p_{max}, d_{min}}

\FOR k\GETS 2 \TO k_{max} \DO\\
\BEGIN
1. \mbox{ }\mbox{ }\mbox{ }\FOR h\GETS 1 \TO H_{c} \DO\\
\mbox{ }\mbox{ }\mbox{ }\mbox{ }\mbox{ }\BEGIN
2.\mbox{ }\mbox{ }\mbox{ }\mbox{ }\mbox{ }\mbox{ }\mbox{ }\mbox{ }\mbox{Split the input dataset in $D_{L}$ and $D_{T}$, the  learning and training sets, }\\\mbox{ }\mbox{ }\mbox{ }\mbox{ }\mbox{ }\mbox{ }\mbox{ }\mbox{ }\mbox{ }\mbox{ }\mbox{ }\mbox{respectively}\\
3.\mbox{ }\mbox{ }\mbox{ }\mbox{ }\mbox{ }\mbox{ }\mbox{ }\mbox{ Train the classifier }C_{1} \mbox{on }D_{T}\\

6.\mbox{ }\mbox{ }\mbox{ }\mbox{ }\mbox{ }\mbox{ }\mbox{ }\mbox{ Partition $D_{L}$ into $k$ clusters by $C_{1}$ and $C_{2}$ in order to obtain }\\\mbox{ }\mbox{ }\mbox{ }\mbox{ }\mbox{ }\mbox{ }\mbox{ }\mbox{ }\mbox{ }\mbox{ }\mbox{ }\mbox{$P_{1}$ and $P_{2}$, respectively.}\\
7.\mbox{ }\mbox{ }\mbox{ }\mbox{ }\mbox{ }\mbox{ }\mbox{ }\mbox{ }m_{k,h}=E(P_{1},P_{2})\\
\mbox{ }\mbox{ }\mbox{ }\mbox{ }\mbox{ }\END\\
8. \mbox{ }\mbox{ }\mbox{ }\mbox{ }t_{k} \GETS median(m_{k,1},\ldots,m_{k,H})\\
\mbox{ }\mbox{ }\mbox{ }\mbox{ }\mbox{ }\mbox{ }\FOR b\GETS 1 \TO B_{0} \DO\\
\mbox{ }\mbox{ }\mbox{ }\mbox{ }\mbox{ }\mbox{ }\BEGIN
9.\mbox{ }\mbox{ }\mbox{ }\mbox{ }\mbox{ }\mbox{ }\mbox{ }\mbox{ Generate (via a null model),  a data matrix $D^{b}$}\\
10.\mbox{ }\mbox{ }\mbox{ }\mbox{ }\mbox{ }\mbox{ }\mbox{ }\mbox{Repeat steps 1-8 on }D^{b}\\
\mbox{ }\mbox{ }\mbox{ }\mbox{ }\mbox{ }\mbox{ }\mbox{ }\END\\

11. \mbox{ }\mbox{ }\mbox{ }\mbox{ Compute the average of these H statistics,
and denote it with }t_{k}^{0}\\
12. \mbox{ }\mbox{ }\mbox{ }\mbox{ Compute the p-value }p_{k}\mbox{ of }t_{k} \\
13.\mbox{ } \mbox{ }\mbox{ }\mbox{ } d_{k}\GETS t_{k}-t^{0}_{k}\\
\END\\
14.\mbox{ } K\GETS\{ 2 \leq k \leq k_{max}: p_{k} \leq p_{max}\mbox{ and }d_{k}
\geq d_{min}\}\\

\IF K = \emptyset \THEN k^{*}\GETS 1
 \ELSE  k^{*} \GETS \argmax\limits_{k\in K} d_{k}\\

\RETURN{k^{*}}

\end{pseudocode}
}
}
\]
\caption{The \codice{Clest} procedure.}\label{algo:ClestOriginal}
\end{figure}

\subsection{Roth et al.}

In analogy with {\tt Clest}, this method, by Roth et al.~\cite{Roth02aresampling}, also generalizes \emph{replicating analysis}.
\begin{itemize}
    \item[-] \emph{The input parameters setup}: as in {\tt Clest}.

   \item[-] \emph{The \codice{Statistics\_Stability} procedure}: steps 1-6  that are the same as in \emph{replicating analysis}. Recall from that latter procedure, that $P_1$ and $P_2$ are two partitions obtained by a classifier and a clustering algorithm, respectively. The \codice{Collect\_Statistic} procedure takes as input $P_1$ and $P_2$ and generates a new partition by computing a minimum weighted perfect bipartite matching~\cite{PapadimitriouS82}. Therefore, assuming $P_1$ as a correct solution, \codice{Collect\_Statistic} gives as output the number of misclassified elements, normalized with respect to the case in which the prediction is random.

   \item[-] \emph{The \codice{Stability\_Measure} procedure}: the \codice{Synopsis} procedure computes the average over the assignment cost and it computes the \vir{expected (in)-stability} value defined as the expectation with respect to the two different datasets. Finally, the \codice{Significance\_Analysis} procedure gives as output the value of $k$ with the minimum \vir{expected (in)-stability} value as $k^*$.

\end{itemize}
For convenience of the reader, the procedure proposed by Roth et al. is given in Fig.~\ref{algo:RLBB02}.

\begin{figure}
\[
\setlength{\fboxsep}{12pt}
\setlength{\mylength}{\linewidth}
\addtolength{\mylength}{-2\fboxsep}
\addtolength{\mylength}{-2\fboxrule}
\ovalbox{
\parbox{\mylength}{
\setlength{\abovedisplayskip}{0pt}
\setlength{\belowdisplayskip}{0pt}

\begin{pseudocode}{RLBB02}{H_{c}, <C_{1}, C_{2}>, D, k_{max}}
\FOR k\GETS k_{min} \TO k_{max} \DO  \\
\BEGIN
 \mbox{ }\mbox{ }\mbox{ } \FOR i\GETS 0 \TO H_{c} \DO\\
\mbox{ }\mbox{ }\mbox{ }\mbox{ }\mbox{ }\BEGIN
1.\mbox{ }\mbox{ }\mbox{ }\mbox{ }\mbox{ }\mbox{ }\mbox{ }\mbox{Split the input dataset in $D_{L}$ and $D_{T}$, the  learning and training sets, }\\\mbox{ }\mbox{ }\mbox{ }\mbox{ }\mbox{ }\mbox{ }\mbox{ }\mbox{ }\mbox{ }\mbox{ }\mbox{respectively}.\\
3.\mbox{ }\mbox{ }\mbox{ }\mbox{ }\mbox{ }\mbox{ }\mbox{ Train the classifier }C_1 \mbox{on }D_{T}\\
5.\mbox{ }\mbox{ }\mbox{ }\mbox{ }\mbox{ }\mbox{ }\mbox{ Partition $D_{L}$ into $k$ clusters by $C_{1}$ and $C_{2}$ in order to obtain }\\\mbox{ }\mbox{ }\mbox{ }\mbox{ }\mbox{ }\mbox{ }\mbox{ }\mbox{ }\mbox{ }\mbox{$P_{1}$ and $P_{2}$, respectively.}\\
6.\mbox{ }\mbox{ }\mbox{ }\mbox{ }\mbox{ }\mbox{ }\mbox{ Find the correct permutation by the minimum}\\\mbox{ }\mbox{ }\mbox{ }\mbox{ }\mbox{ }\mbox{ }\mbox{ }\mbox{ }\mbox{ }\mbox{weighted perfect bipartite matching}\\
7.\mbox{ }\mbox{ }\mbox{ }\mbox{ }\mbox{ }\mbox{ }\mbox{ Normalize w.r.t. the random}\\
8.\mbox{ }\mbox{ }\mbox{ }\mbox{ }\mbox{ }\mbox{ }\mbox{ Compute the expected (in)-stability value}\\
\mbox{ }\mbox{ }\mbox{ }\mbox{ }\mbox{ }\END\\
\END
\end{pseudocode}
}
}
\]
\caption{The Roth et al. procedure.}\label{algo:RLBB02}
\end{figure}

\ignore{
\subsection{BB06}

This method, referred to as {\tt BB06}, is due to Bertrand and Bel Mufti~\cite{BertrandM06} is an internal validation measure proposed for partitional algorithms. It uses the Loevinger's measure in order to assess the cluster stability.

{\tt BB06} uses the following experimental setup (see Figs.~\ref{fig:InputCollect} and~\ref{fig:MacrosCollect}): $H_c$ is used as test, for a given number of iterations $c$ (e.g. $c=100$), $\alpha=0$, $\beta\in[0.7,0.9]$ and the \codice{DGP} is an instance of the \emph{proportionate stratified sampling} method (see Section~\ref{subsec:Subsampling}).  It can be derived by the Stability Measure paradigm as follows.}

\subsection{A Special Case: the Gap Statistics}

Although {\tt Gap} (see Section~\ref{sec:Gap}) is not an internal stability measure, the \codice{GP} procedure (see Fig.~\ref{algo:Gap}) can be derived from the Stability Measure paradigm as follows.

\begin{itemize}
    \item[-] \emph{The input parameters setup}: $\hat{H}_1$,  $\alpha=0$, $\beta$ is not relevant and  the set of clustering procedures $<C_{1},C_2,\ldots,C_t>$ consists of only one clustering algorithm $C_1$.

    \item[-] \emph{The \codice{Statistics\_Stability} procedure}: only steps 2 and 5-8 are performed. The \codice{Collect\_Statistic} procedure takes as input a partition $P_1$ relative to the dataset $D$ and it gives as output the {\tt WCSS} value (see Section~\ref{sec:WCSS-Theory}).

     \item[-] \emph{The \codice{Stability\_Measure} procedure}: the \codice{Synopsis} procedures return a copy of the {\tt WCSS} values. Finally, the \codice{Significance\_Analysis} procedure is best presented as a procedure which is given in Fig.~\ref{fig:SignificanceGap}, where in the first step the procedure generates a new dataset via the \codice{DGP} procedure that in this case it is an instance of null models (see Section~\ref{sec:statistics}). Intuitively, in the remaining steps, it compares the value of {\tt WCSS} obtained for $D$ and for the null model datasets. The prediction of $k^*$ is based on running a certain number of times the procedure \codice{Stability\_Measure} taking then the most frequent outcome as the prediction.

\end{itemize}
\begin{figure}
\[
\setlength{\fboxsep}{12pt}
\setlength{\mylength}{\linewidth}
\addtolength{\mylength}{-2\fboxsep}
\addtolength{\mylength}{-2\fboxrule}
\ovalbox{
\parbox{\mylength}{
\setlength{\abovedisplayskip}{0pt}
\setlength{\belowdisplayskip}{0pt}
\begin{pseudocode}{Significance\_Analysis}{R^{k_{min}},\ldots,R^{k_{max}}}

\FOR k\GETS k_{min} \TO k_{max} \DO  \\
\BEGIN
 \mbox{ }\mbox{ }\mbox{ }\mbox{ }\mbox{ }\FOR i\GETS 0 \TO B_{0} \DO\\

\mbox{ }\mbox{ }\mbox{ }\mbox{ }\mbox{ }\mbox{ }\mbox{ }\BEGIN

\mbox{ }\mbox{ }\mbox{ }\mbox{ }\mbox{ }\mbox{ }\mbox{ }\mbox{ }\mbox{ }\mbox{ }\mbox{ }D^{i} \GETS \CALL{DGP}{D}\\

\mbox{ }\mbox{ }\mbox{ }\mbox{ }\mbox{ }\mbox{ }\mbox{ }\mbox{ }\mbox{ }\mbox{ }\mbox{ }R_{i}^{k} \GETS \CALL{Stability\_Statistics}{D^{i}, H, \beta, <C_{1}>, k}\\
\mbox{ }\mbox{ }\mbox{ }\mbox{ }\mbox{ }\mbox{ }\mbox{ }\END\\

\mbox{ }\mbox{ }\mbox{ }\mbox{ }\mbox{ }\mbox{ }\mbox{Compute }Gap(k)= \frac{1}{B_{0}}\sum_{i=1}^{B_{0}} R_{i}^{k} - R^{k}.\\

\mbox{ }\mbox{ }\mbox{ }\mbox{ }\mbox{ }\mbox{ }\mbox{Compute the standard deviation }sd(k) \mbox{of the set of numbers }\{R_{1}^{k}, \ldots,R_{B_{0}}^{k} \}
\\\mbox{ }\mbox{ }\mbox{ }\mbox{ }\mbox{ }\mbox{ }s(k)=\left(\sqrt{1 +\frac{1}{B_{0}}}\right)sd(k).\\

\END\\

k^*\mbox{ is the first value of }k\mbox{ such that } Gap(k)\geq Gap(k+1)-s(k+1).\\

\RETURN{k^{*}}
\end{pseudocode}

}
}
\]
\caption{Implementation of the \codice{Significance\_Analysis} procedure for {\tt Gap}.}\label{fig:SignificanceGap}
\end{figure}

%% file: Chapter4.tex
\chapter{Non-negative Matrix Factorization}\label{chap:NMF}

This chapter describes one of the methodologies that has gained prominence in the data analysis literature: Non-negative Matrix Factorization (NMF for short).
In particular, the mathematical formulation of NMF and some algorithms that have been developed to compute it are presented. Moreover, some applications of NMF are also discussed. It is worth pointing out that, for this dissertation, NMF is of interest as a clustering algorithm. In fact,  in the next chapter, its first benchmarking on microarray data is presented.

\section{Overview}
One common ground on which many
data analysis methods rest is to replace the original data by a lower dimensional representation obtained via
subspace approximations~\cite{lonardibook,Tibshrbook,Rice,Speed03,Witten00}. The goal is to explain the observed data using a limited number of basis components, which when combined together, approximate the original data as accurately as possible. This concise description  of the data allows to find possible structure in them. Indeed, a meaningful dimensionality reduction can be achieved only if the data has common underlying regularities and patterns.
Matrix factorization and principal component analysis are two of the many classical methods used to accomplish both the goals of reducing the number of variables and detecting structure underlying the data. While some of those techniques have been reviewed in Section~\ref{subsec:SubspaceMethods},  NMF is singled-out in this chapter since it is a substantial contribution to this area: the pioneering paper by Lin and Seung has been  followed-up by extensions and uses of NMF in a broad range of domains.  In particular,  molecular biology applications, as described in the survey of Devarajan~\cite{Devarajan}, image processing~\cite{GuillametICPR,NMF} and text mining~\cite{NMF,NielsenBalslevHansen}.

This chapter is organized as follows. Section~\ref{sec:MF} provides a formalization of matrix factorization. Sections~\ref{sec:NMFScheme} and~\ref{sec:NMFAlg} provide the general scheme and different algorithms for NMF, respectively. Finally, in Section~\ref{sec:NMFAlg},  several applications of NMF to image processing, text mining and molecular biology are briefly described.

\ignore{
One of the matrix factorization methodologies that has gained much attention over the past decade is the Non-negative Matrix Factorization due to Lee and Seung~\cite{Lee97unsupervisedlearning,NMF,lee00algorithms}.
}
\section{Matrix Factorization: A Basic Mathematical Formulation}\label{sec:MF}

In this section, a general statement of Matrix Factorization (MF for short) is given, then two restricted versions of it are presented: Positive Matrix Factorization (PMF for short)  and the already mentioned NMF.

\paragraph{MF.} Let  $V$ be  matrix  of size $m\times n$. Usually, when $V$ is a data matrix, $m$ denotes  the number of features and $n$ denotes the number of items in $\Sigma$.
Given $V$ and an integer $r<\mathop{\min }\,\{n, m\}$,  one wants  to find two matrix factors $W$ and $H$ of size $m\times r$ and $r \times n$, respectively, such that:

\begin{equation}\label{eqn:VapproxWH}
V\approx W \times H.
\end{equation}

One has that $v\approx Wh$, where $v$ and $h$ are homologous columns in  $V$ and $H$, respectively.
That is, each column $v$ is approximated by a linear combination of the columns of $W$, weighted by the components of $h$. Therefore, $W$ can be regarded as containing a basis.

\paragraph{PMF.} It is  a variant of MF, by  Paatero and Tapper~\cite{PaateroTapper},   since $V$  is constrained  to be a positive matrix. One possible solution can be obtained by computing a   positive low-rank approximation of $W\times H$, via an optimization of the function:

$$
\underset{W,H\ge 0}{\mathop{\min }}\,{{\left\| A \circ (V-W\times H) \right\|}_{F}},
$$

\noindent where $A$ is a weighted matrix whose elements are associated to the elements of $V$, $\circ$ is the Hadamard product and ${\left\|  \right\|}_{F}$ denotes the Frobenious norm~\cite{MatrixComputation}. Paatero and Tapper also proposed an alternative
least squares algorithm in which one of the two matrices is fixed and the optimization is solved with respect to the other one and viceversa. Later, Paatero developed a series of algorithms~\cite{Paatero1997,PaateroPARAFAC,Paatero99} using a longer product of matrices to replace the approximate $W\times H$.

\paragraph{NMF.} A variant of PMF,  NMF allows  $V$ to be non-negative, this latter being a   constraint much more suitable for  data analysis tasks. It also offers several advantages~\cite{Gaujoux2010,NMF}:

\begin{enumerate}
\item The constraint that the matrix factors are  non-negative  allows for their intuitive interpretation as real underlying components within the context defined by the original data. The basis components can be directly interpreted as parts or basis samples, present in different proportions in each observed sample.\\

\item NMF generally produces sparse results, implying that the basis and/or the mixture coefficients have only a few non-zero entries. \\

\item Unlike other decomposition methods such as singular value decomposition~\cite{TrefethenBau} or independent component analysis~\cite{BellSejnowski,CichockiAmari}, the aim of NMF is not to find components that are orthogonal or independent, therefore allowing overlaps among  the components.\\
\end{enumerate}

\section{The General NMF Scheme}\label{sec:NMFScheme}

In order to compute an NMF, all the methods proposed in the literature use the same simple scheme, reported in Fig.~\ref{algo:NMFmain}, where the following \vir{ingredients} have to be specified:

\begin{figure}
\[
\setlength{\fboxsep}{12pt}
\setlength{\mylength}{\linewidth}
\addtolength{\mylength}{-2\fboxsep}
\addtolength{\mylength}{-2\fboxrule}
\ovalbox{
\parbox{\mylength}{
\setlength{\abovedisplayskip}{0pt}
\setlength{\belowdisplayskip}{0pt}

\begin{pseudocode}{NMF}{V}
 \mbox{ Initialize $W$ and $H$}\\
\WHILE H \DO \\
\BEGIN
\mbox{ }\mbox{ }\mbox{ }\mbox{ }\mbox{ }\mbox{ update $W$ and $H$}\\
\END
\end{pseudocode}
}
}
\]
\caption{The basic \codice{NMF} procedure.}\label{algo:NMFmain}
\end{figure}

\begin{itemize}
    \item[(i)] an initialization for matrices $W$ and $H$;

    \item[(ii)] an update rule;

    \item[(iii)] a stopping criterion.
\end{itemize}

The scheme  is iterative: it starts from the two initial matrices $W$ and $H$, which are repeatedly  updated via a fixed  rule until the stopping criterion is satisfies.

For point (i), the matrices $W$ and $H$  are initialized at random. In that case, different runs of NMF, with the same input, are likely to  produce different results. However, it is worth pointing out that sophisticated  deterministic initialization methods have been proposed to choose appropriate initial values referred to as \vir{seed  NMF algorithms}~\cite{Langville06initializations,BoutsidisGallopoulos,Wild2003}. When one uses the same initial seed, the procedure is deterministic, i.e., it always produces  the same output on a given input.

With respect to points (ii) and (iii), quite many numerical algorithms have been developed for NMF~\cite{LawsonHanson,lee00algorithms,Lin07,Lin07Neural,Piper04}.

For (ii), the most popular follow at least one of the following  principles and techniques: alternating direction iterations,  projected Newton,  reduced quadratic approximation, and  descent search. Correspondingly, specific implementations can be categorized into alternating least squares algorithms (ALS for short)~\cite{PaateroTapper}, multiplicative update algorithms~\cite{NMF,lee00algorithms} combined with  gradient descent search, and hybrid algorithms~\cite{Berry06algorithmsand,Piper04}. For the interested reader, it is worth pointing out that a general assessments of these methods can be found in~\cite{LangvilleMeyerAlbright,Tropp}. In the next section, a brief description of the gradient descent and ALS methods is provided.

As for (iii), the most \vir{popular}  stopping  criteria are: a fixed number of iterations, of a suitably defined matrix, referred to as consensus (see~\cite{BrunetNMF} for formal definitions)  and stationarity of the objective function value.

\section{The NMF Procedure and Its Variants}\label{sec:NMFAlg}
 As anticipated, two of the most popular strategies for the NMF computation are detailed here. Moreover, other approaches specific to particular  applications are briefly described. Finally, a list of the software available for the computation of NMF is reported.

\subsection{Multiplicative Update Rules and Gradient Descent}
Gradient descent is an optimization strategy widely used in the literature~\cite{Snyman,SwagatamSambarta,Yamagishi}.
Let $F$ be a function to be minimized. Unfortunately, it is possible that the computation of the absolute minimum of $F$ may  be practically unfeasible. In that case,  a local minimum is an \vir{acceptable solution}. With reference to Fig.~\ref{fig:NMFGFD}, the gradient descent approach tries to compute such a local minimum as follows: the algorithm starts from a random point and moves to a successive point by minimizing along the local downhill gradient between the two points. This process is iterated until a minimum of the objective function is reached.

Lee and Seung propose the following optimization function :

\begin{equation}\label{eqn:LeeSeungFunction}
f(W,H)=\frac{1}{2}{\left\| V-W\times H\right\|}^{2}
\end{equation}
as a convergence criterion.

In order to optimize that function,  Lee and Seung propose \vir{multiplicative update rules} in conjunction with gradient descent.
As discussed in~\cite{lee00algorithms}, the use of a multiplicative update strategy as opposed to an additive one is that  the latter does not guarantee a systematic decrease in the cost function.

The algorithm starts with a random initialization of the two matrices $W$ and $H$ and then, at each iteration, updates them as follows:

$$
H_{a,\mu}\leftarrow H_{a,\mu} \times \frac{(W^{t}\times V)_{a,\mu}}{(W^{t}\times W \times H)_{a,\mu}}
$$
\noindent and
$$
W_{i,a}\leftarrow W_{i,a} \times \frac{(V\times H^{t})_{i,a}}{(W\times H \times H^{t})_{{i,a}}}
$$

\noindent where $1\leq a\leq n $, $1 \leq i,\mu \leq r$ and $A^{t}$ denotes the transpose of a matrix $A$.

Letting  $\nabla_{H}f(W,H) = W^{t}\times (W\times H - V )$ and  $\nabla_{W}f(W,H) = (W\times H - V )\times H^{t}$, the above update rules can be rewritten as follows:

\begin{equation}
H_{a,\mu}\leftarrow H_{a,\mu} - \frac{H_{a,\mu}}{(W^{t}\times W \times H)_{a,\mu}}\times \nabla_{H}f(W,H)\label{eq:update1}
\end{equation}

and

\begin{equation}
W_{i,a}\leftarrow W_{i,a} - \frac{W_{i,a}}{(W\times H \times H^{t})_{{i,a}}}\times \nabla_{W}f(W,H)\label{eq:update2}
\end{equation}

It is worth pointing out that the optimization function~\eqref{eqn:LeeSeungFunction} in the general NMF model can be modified in several ways, depending on the application at hand. For instance, some penalty terms can be added in order to gain more localization or enforce sparsity, and more constraints such as sparseness can be imposed.

\begin{figure}[ht]
\begin{center}
\epsfig{file=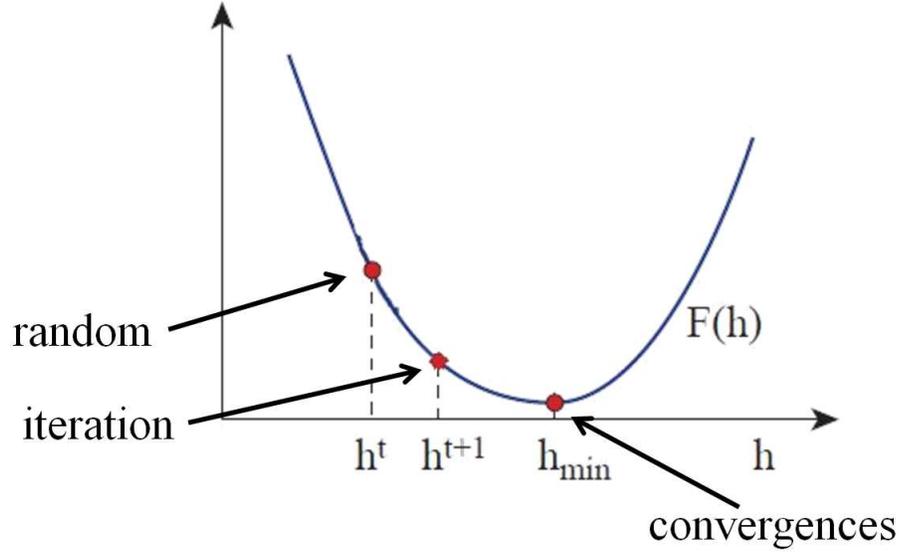,scale=0.3}
\end{center}
\caption{An example of gradient descend. Part of this figure is taken from~\cite{lee00algorithms}.}\label{fig:NMFGFD}
\end{figure}

Two drawbacks of the multiplicative algorithms are that: (1) the denominator of the step size may be zero and (2) once an element in $W$ or $H$ becomes zero, it remains zero and this fact should have some effects on the correct convergence of the algorithm.

Although Lee and Seung~\cite{lee00algorithms} claimed the convergence of the above algorithm, Gonzalez and Zhang~\cite{GonzalesZhang2005a} and Lin~\cite{Lin07,Lin07Neural} independently presented
numerical counter-examples, where the Lee and Seung algorithm fails to approach a stationary point, i.e., to converge to a local minimum.

Moreover, Lin~\cite{Lin07Neural} proposed to modify the  Lee and Seung method,  in order to solve the convergence problem, by keeping the same objective function, while update rules (\ref{eq:update1}) and (\ref{eq:update2}) become:

$$
H_{a,\mu}\leftarrow H_{a,\mu} - \frac{\overline{H}_{a,\mu}}{(W^{t}\times W \times \overline{H})_{a,\mu}+\delta}\times \nabla_{H}f(W,H)
$$
\noindent and
$$
W_{i,a}\leftarrow W_{i,a} - \frac{\overline{W}_{i,a}}{(\overline{W}\times H \times H^{t})_{i,a}+\delta}\,\times \nabla_{W}f(W,H)
$$

\noindent Both  $\epsilon$ and $\delta$ are pre-defined small positive constants,

\begin{displaymath}
\overline{H}_{a,\mu}=\left\{\begin{array}{ll}
H_{a,\mu}  & \mathop{if }\nabla_{H} f(W,H)_{a,\mu}\geq 0,  \\
\mathop{\max }(H_{a,\mu},\epsilon) & otherwise,\\
\end{array}\right.
\end{displaymath}

\noindent and $\overline{W}_{i,a}$ is defined in analogy with  $\overline{H}_{a,\mu}$.

Moreover, Lin also proposed to normalize $W$ so that the sum of its columns is one. With respect to the Lee and Seung method, this modified version requires same extra operations (e.g. to compute  $\overline{H}$ and $\overline{W}$), but that has no substantial effect  on the  complexity of the procedure.

Lin~\cite{Lin07Neural} proves the convergence of this modified algorithm, via the following theorem.

\begin{theorem}
The update sequence of $W$ and $H$ has at least one limit point.
\end{theorem}
\begin{proof}
Let $W^{k}$ and $H^{k}$ be the matrices at iteration $k$. It suffices to prove that ($W^{k}$, $H^{k}$), $k=1,\ldots,\infty$, are in a closed and bounded set. Since $W^{k}$ is normalized, one needs to show that ${H^{k}}$ is bounded. Otherwise, there is a component $H_{a,\mu}$ and an infinite index set $\kappa$ such that
\begin{equation}\label{eqn:lin40}
\begin{array}{lcr}
\underset{k\in \kappa}{\mathop{\lim }}\,H^{k}_{a,\mu} \rightarrow \infty,& H^{k}_{a,\mu}<H^{k+1}_{a,\mu}, &\forall k \in \kappa,\\
\end{array}
\end{equation}
\noindent and exist $\forall i$
$$
\underset{k\in \kappa, k\rightarrow \infty}{\mathop{\lim }}\,W^{k}_{i,a}.
$$

\noindent  One must have that  $W_{i,a}=0$, $\forall i$. Otherwise, there is an index $i$ such that
$$
\underset{k\in \kappa, k\rightarrow \infty}{\mathop{\lim }}\,(W^{k}\times H^{k})_{i,j}\geq \underset{k\in \kappa, k\rightarrow \infty}{\mathop{\lim }}\,W^{k}_{i,a}H^{k}_{a,j}=\infty.
$$

\noindent  Then,

$$
\underset{k\rightarrow \infty}{\mathop{\lim }}\,f(W^{k},H^{k}) \geq \underset{k\in \kappa, k\rightarrow \infty}{\mathop{\lim }}\,{\left\| V_{i,j}-(W^{k}\times H^{k})_{i,j}\right\|}^{2}=\infty
$$

\noindent  contradicting that $f(W^{k},H^{k})$ is strictly decreasing. Since the sum of the columns of $W^{k}$ is either one or zero, $W_{i,a}=0$, $\forall i$ implies
\begin{equation}\label{eqn:lin42}
\begin{array}{lcr}
W_{i,a}=0, & \forall i & \forall k\in \kappa \mbox{ large enough}.\\
\end{array}
\end{equation}

\noindent Then,

\begin{equation}\label{eqn:lin43}
\begin{array}{lcr}
 \nabla_H f(W^{k},H^{k})_{a,\mu}=0, & \forall k\in \kappa, \mbox{ so }H^{k,n}_{a,\mu}=H^{k}_{a,\mu}, &  \forall k\in \kappa .\\
\end{array}
\end{equation}

\noindent  Moreover, by \eqref{eqn:lin42}, one can prove that
$$
\underset{k\in \kappa, k\rightarrow \infty}{\mathop{\lim }}\,W^{k,n}_{i,a} = \underset{k\in \kappa, k\rightarrow \infty}{\mathop{\lim }}\,W^{k}_{i,a} =0, \forall i
$$

Thus, in normalizing $(W^{k,n},H^{k,n})$, $H^{k,n}$'s $a$-th row is either unchanged or decreased. With \eqref{eqn:lin43},
$$
H^{k+1}_{a,\mu} \leq H^{k}_{a,\mu}, \forall k \in \kappa
$$
an inequality contradicting \eqref{eqn:lin40}. Thus, ${H^{k}}$ is bounded. Therefore, $\{W^{k}, H^{k}\}$ is in a compact set, so there is at least one convergent sub-sequence.
\end{proof}

\subsection{Alternating Least Squares Algorithms}

In this class of algorithms, a least squares step is followed by another least squares step in an alternating way. That is, the algorithm updates in alternation  $W$ and $H$ by solving a  distinct matrix equations for each, as detailed in procedure ALS reported in Fig.~\ref{algo:ALS}. In that procedure, in order to keep non-negativity, a simple projection step is performed to set all negative elements resulting from the least squares computation to zero.  Moreover, some additional flexibility,  not available in other algorithms-especially those of the multiplicative update class- is also available. For instance,
as detailed above, in the multiplicative algorithms, if an element in $W$ or $H$ becomes 0, it must remain 0. This is restrictive, since once the algorithm starts heading down a path towards a fixed point, even if it is a poor fixed point, it must continue in that vein. The ALS algorithms is more flexible, allowing the iterative process to escape from a poor path.
It is worth pointing out that some improvements to the basic ALS algorithm scheme appear in~\cite{Berry06algorithmsand,Paatero99} and that, depending on the implementation, ALS algorithms can be very fast.

\begin{figure}
\[
\setlength{\fboxsep}{12pt}
\setlength{\mylength}{\linewidth}
\addtolength{\mylength}{-2\fboxsep}
\addtolength{\mylength}{-2\fboxrule}
\ovalbox{
\parbox{\mylength}{
\setlength{\abovedisplayskip}{0pt}
\setlength{\belowdisplayskip}{0pt}

\begin{pseudocode}{ALS}{V}
\WHILE \mbox{convergence is not reached} \DO \\
\BEGIN
1.\mbox{ Initialize $W$ at random}\\
2.\mbox{ Solve for $H$ in matrix equation $W^{t}\times W\times H=W^{t}\times V$}\\
\mbox{ }\mbox{ }\mbox{ and set all negative elements in $H$ to zero}\\
3.\mbox{ Solve for $W$ in matrix equation $H\times H^{t}\times W^{t}=H\times V^{t}$}\\
\mbox{ }\mbox{ }\mbox{ and set all negative elements in $W$ to zero}\\
\END
\end{pseudocode}
}
}
\]
\caption{The basic \codice{ALS} procedure}\label{algo:ALS}
\end{figure}

\subsection{NMF Algorithms with Application-Dependent Auxiliary Constraints}

A great deal of work has been devoted to the analysis, extension, and application of NMF algorithms in science, engineering and medicine. The NMF has been cast into alternate formulations by various authors. In this section a brief survey on the state of the art is provided.
One of the main improvements is to develop the algorithms by using different objective functions. Lee and Seung~\cite{lee00algorithms} provided an information theoretic formulation based on the Kullback-Leibler divergence~\cite{KullbackLeibler1951} of $V$ from $W\times H$. Dhillon and Sra~\cite{Dhillon05} generalized the NMF methods with Bregman divergence. Cichocki et al.~\cite{Cichocki06} have proposed cost functions based on Csisz\'{a}r's $\varphi$-divergence. Wang et al.~\cite{WangJiaHuTurk} propose a formulation that enforces constraints based on Fisher linear discriminant analysis for improved determination of spatially localized features. Guillamet et al.~\cite{GuillametBressan} have suggested the use of a diagonal weight matrix $Q$ in a new factorization model, $V\times Q \approx W\times H\times Q$, in an attempt to compensate for feature redundancy in the columns of $W$.
Other approaches propose alternative cost function formulations. Smoothness constraints have been used to regularize the computation of spectral features in remote sensing data~\cite{Pauca200629,Piper04}. Chen and Cichocki~\cite{Chen05} used temporal smoothness and spatial correlation constraints to improve the analysis of EEG data for early detection of Alzheimer's disease. Hoyer~\cite{Hoyer,Hoyer04} employed sparsity constraints on either $W$ or $H$ to improve local rather than global representation of data. The extension of NMF to include such auxiliary constraints is problem dependent and often reflects the need to compensate for the presence of noise or other data degradations in $V$.

\subsection{Software Available}

Several algorithms performing NMF have been implemented and  published. The interested reader will find a compendium of them in~\cite{Berry06algorithmsand}.  In this dissertation, some of them are mentioned next.

 Hoyer~\cite{Hoyer} provided a package that implements five different algorithms. Cichocki and Zdunek~\cite{NMFLab} produced an  appealing NMFLAB package that implements a wide range of NMF algorithms, which can be combined, tested and compared via a graphical interface. However, availability only in MATLAB, a proprietary software, limits access to these packages within the scientific community. Some C/C++ implementations are also available~\cite{WangKO06}, including a parallel implementation using the MPI. Recently, Gaujoux and Seoighe~\cite{Gaujoux2010} propose a completely open-source package for the R/BioConductor platform~\cite{Bioconductor}, which is a well established and extensively used standard in statistical and bioinformatics research.

\section{Some Applications}\label{sec:NMFApps}
In this section, the spectrum of possible applications of NMF is highlighted. In particular, the ones in  image analysis  and text mining, originally used to Lee and Seung to show the validity of NMF, are detailed. Moreover, a short summary of the more recent applications of NMF in bioinformatics is also provided.

\subsection{Image Analysis}

In this domain, the use of NMF is a natural choice, since an image can be represent as a non-negative matrix. Moreover, when  it is desirable to process datasets of images represented by column vectors, as composite objects or as separated parts, it is suggested that an NMF would enable the identification and classification of intrinsic \vir{parts} that make up the object being imaged by multiple observations~\cite{Donoho03whendoes,NMF}. An example is provided  in  Fig.~\ref{fig:NMFimage}, which is taken from~\cite{NMF}. Other work on face and image processing applications of NMF includes~\cite{GuillametBressan,GuillametICPR,GuillametSchiele,Guillametface,LawrenceMatusik,LawrenceRusinkiewicz,WangZL05}.

\begin{figure}[ht]
\begin{center}
\epsfig{file=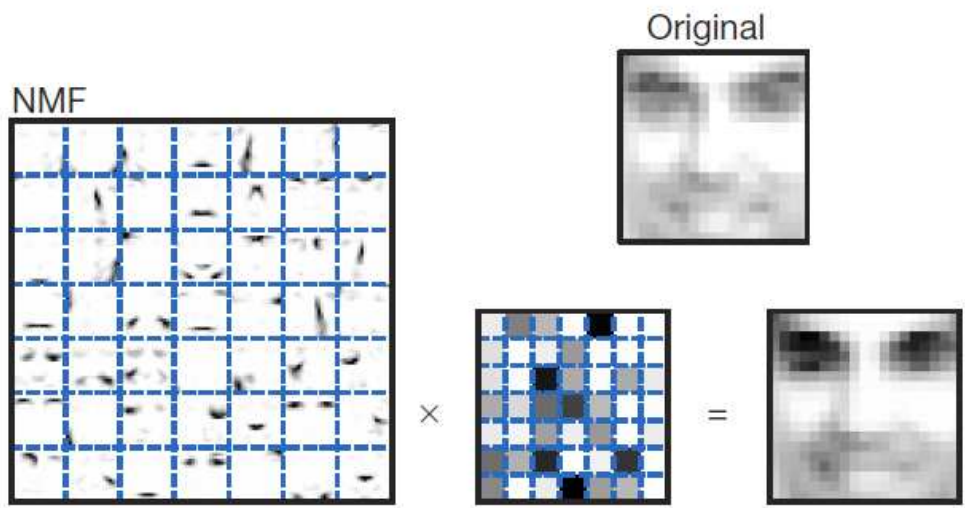,scale=0.6}
\end{center}
\caption{The NMF is applied to a  database of 2.429 facial images, each consisting of $19\times 19$ pixels, and constituting a matrix $V$. The NMF learns to represent a face as a linear combination of basis images.
The NMF basis and encodings contain a large fraction of vanishing coefficients, so both the basis images and image encodings are sparse. The basis images are sparse because they are non-global and contain several versions of mouths, noses and other facial parts, where the various versions are in different locations or forms.}\label{fig:NMFimage}
\end{figure}

\subsection{Text Mining}

The use of NMF on text documents has  highlighted its ability  to tackle semantic issues such as synonymy or even to cluster data. An example is given in Fig.~\ref{fig:NMFtext}, where NMF is used to discover semantic features of 30.991 articles from the Grolier encyclopedia. For each word in a vocabulary of size 15.276, the number of occurrences is counted in each article and it is used to form the matrix $V$. Each column of $V$ contains the word counts for a particular article, whereas each row of $V$ contains the counts of a particular word in different articles. Upper left, four of the $r=200$ semantic features (columns of $W$). As they are very high-dimensional vectors, each semantic feature is represented by a list of the eight words with highest frequency in that feature. The darkness of the text indicates the relative frequency of each word within a feature.
With reference to Fig.~\ref{fig:NMFtext} the bottom of the figure exhibits the two semantic features containing \vir{lead} with high frequencies. Judging from the other words in the features, two different meanings of \vir{lead} are differentiated by NMF.
Right, the eight most frequent words and their counts in the encyclopedia entry on the \vir{Constitution of the United States}. This word count vector was approximated by a superposition that gave high weight to the upper two semantic features, and none to the lower two, as shown by the four shaded squares in the middle indicating the activities of H.

Other work on text mining applications of NMF includes~\cite{Badea05clusteringand,BerryBrowne,Dhillon,PaucaSBP04,ShahnazBPP06,XuLiu}.

\begin{figure}[ht]
\begin{center}
\epsfig{file=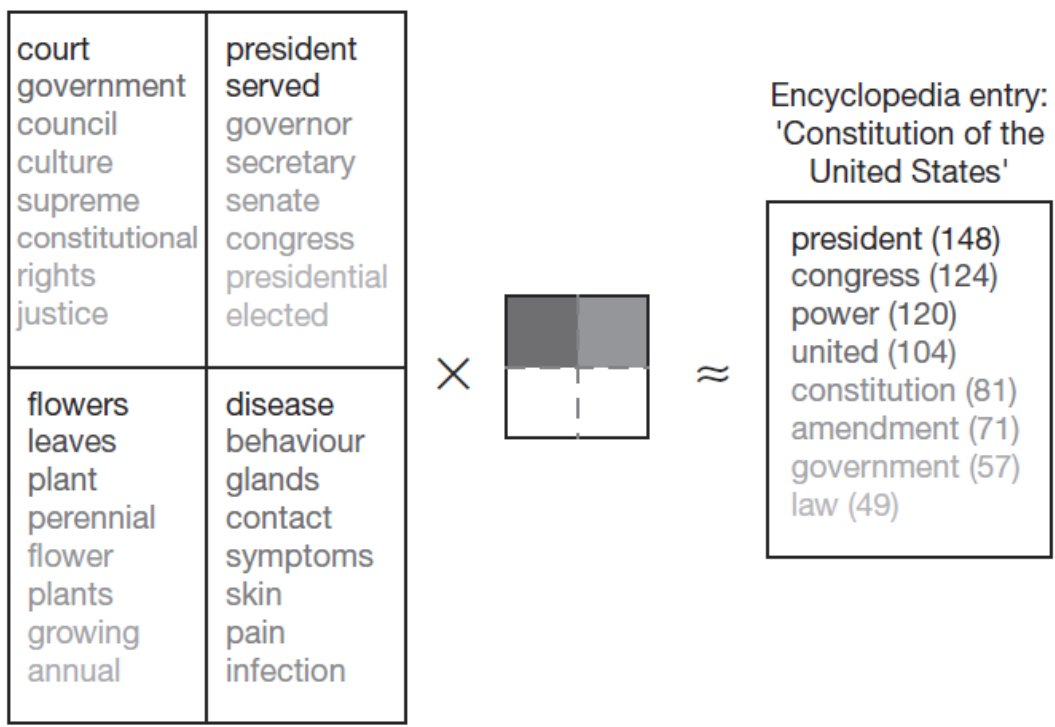,scale=0.6}
\end{center}
\caption{Example of NMF on text. This figure is takes from~\cite{NMF}.}\label{fig:NMFtext}
\end{figure}

\subsection{Bioinformatics}

NMF is a very versatile pattern discovery technique that has received quite a bit of attention in the computational biology literature, as discussed in the review by Devarajan \cite{Devarajan}.
In this section,  some of them are detailed.

\paragraph{Class Comparison and Prediction.} Recently, the NMF is used in this domain. For instance, Fogel et al.~\cite{Fogel2007} apply NMF to identify ordered sets of genes and utilize them in sequential analysis of variance procedures in order to identify differentially expressed genes.
Okun and Priisalu~\cite{Okun2006} apply NMF, successfully,  as a dimension reduction tool in conjunction with several classification methods for protein fold recognition. They report superior performance, in terms of misclassification error rate,  of three classifiers based on nearest neighbor methods when applied to NMF reduced data relative to the original data. Similar approaches have been proposed in~\cite{JungLee} and~\cite{Kelm2007} for fold recognition  and magnetic resonance spectroscopic imaging, respectively.

\paragraph{Cross-Platform and Cross-Species Characterization.}
The rapid advances in high-throughput technologies have resulted in the generation of independent large-scale biological datasets using different technologies in different laboratories.
In this scenario, it is important to assess and interpret potential differences and similarities in these datasets in order to enable cross-platform and cross-species analysis and the possible characterization of the data.
Tamayo et al.~\cite{Tamayo2007} describe an approach referred to as metagene projection in order to reduce noise and technological variation, while capturing invariant biological features in the data. Furthermore, this approach allows the use of prior knowledge based on existing datasets in analyzing and characterizing new data~\cite{Isakoff06122005}. In metagene projection, the dimensionality of a given dataset is reduced using NMF, based on a pre-specified rank $k$ factorization.

\paragraph{Molecular Pattern Discovery.}
One of the most common applications of NMF in bionformatics,  and of  great interest for this thesis,  is in the area of molecular pattern discovery. In particular, for gene and protein expression microarray studies,  where there is  lack of a priori knowledge of the expected expression patterns,  for a given set of genes or phenotypes.  In this area,  NMF is successfully applied in order to discover biologically meaningful classes, i.e., clusters.

For instance, Kim and Tidor~\cite{Kim01072003} apply NMF as a tool to cluster genes and predict functional cellular relationships in Yeast, while Heger and Holm~\cite{HegerHolm} use it for the recognition of sequence patterns among related proteins. Brunet et al.~\cite{BrunetNMF} apply it to cancer microarray data for the elucidation of tumor subtypes. Moreover, they developed a model selection for NMF based on the consensus matrix (see Section~\ref{subsec:Consensus}) that enables the choice of the appropriate number of clusters in a dataset.
Following the same notation as in Brunet et al.~\cite{BrunetNMF} and Devarajan~\cite{Devarajan}, let $V$ represent the outcome of a microarray experiment, where there are $m$ samples,  each  composed of measurements of $n$ genes. In this case, $W$ and $H$ assume two very intuitive roles. $W$ is a matrix whose columns are \vir{metagenes} and $H$ is a matrix whose rows  are \vir{meta expression patterns}. If one is interested in clustering the samples in $r$ groups, as we do here, then one can place sample $i$ in cluster $j$ if  the expression level of sample $i$ is maximum in metagene $j$.  That is,  $H_{i,j}$ is maximum in the  $i$-th column of $H$. An example is given in Figure~\ref{fig:NMFCluster}.

Other approaches of molecular pattern discovery  that use NMF~\cite{CarmonaPascual,Costanzo01012001,YuanGao11012005,KimPark,PascualMontano2006,bioNMF,TiradoMontano}, or its sparse version, have been proposed in the literature.

\begin{figure}[ht]
\begin{center}
\epsfig{file=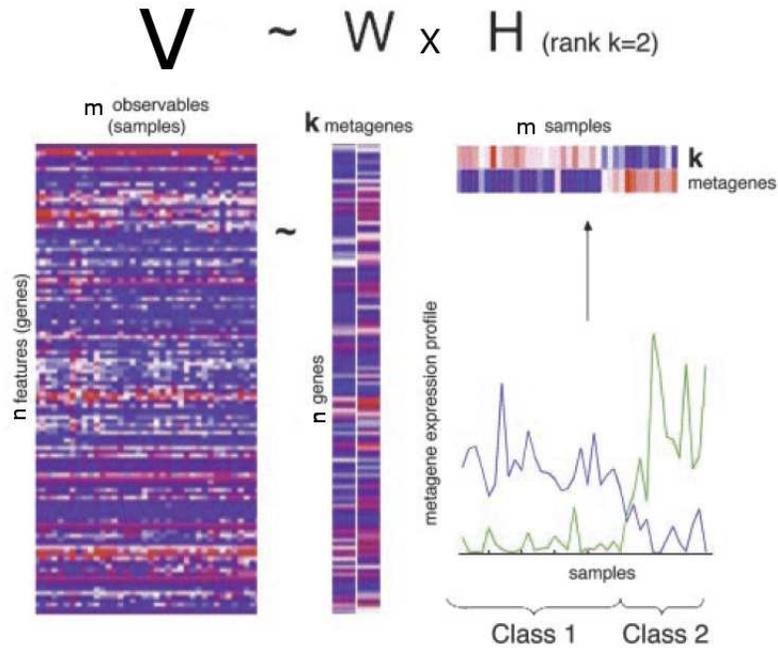,scale=0.4}
\end{center}
\caption{A rank-2 reduction of a DNA microarray of $n$ genes and $m$ samples is
obtained by NMF, $V\approx W\times H$. For better visibility, $H$ and $W$ are shown with
exaggerated width compared with original data in $V$, and a white line
separates the two columns of $W$. Metagene expression levels (rows of $H$) are
color coded by using a heat color map, from dark blue (minimum) to dark red
(maximum). The same data are shown as continuous profiles below. The
relative amplitudes of the two metagenes determine two classes of samples,
class 1 and class 2. Here, samples have been ordered to better expose the class
distinction. This figure is taken from~\cite{BrunetNMF}.}\label{fig:NMFCluster}
\end{figure}

%% file: Chapter5.tex
\chapter{Experimental Setup and Benchmarking of NMF as a Clustering Algorithm}
\label{chap:5}

In this chapter, the experimental framework used in this dissertation is  detailed, i.e.,  datasets, algorithms and hardware.  Moreover, a  benchmarking of NMF as a clustering algorithms on microarray data is proposed. To the best of our knowledge, it is the first one that takes into account both its ability to identify cluster structure in a dataset and the computational resources it needs for that task. A comparative analysis with some classic algorithms is also provided.

\section{Datasets}\label{sec:dataset}

It is useful to recall the definition of \vir{gold solution} that naturally yields a partition of the datasets in two main categories.
Technically speaking, a gold solution for a dataset is a partition of the data in a number of classes known {\it a priori}. Membership in a class is established by assigning the appropriate  class label to each element. In less formal terms, the partition  of the dataset in classes is based on external knowledge that leaves no ambiguity on the  actual number of classes and on the membership of elements to classes. Although there exist real microarray datasets for which such an {\it a priori} division is known, in a few previous studies of relevance here,  a more relaxed criterion has been adopted to
allow also datasets with high quality partitions that have been inferred by analyzing the data, i.e., by the use of internal
knowledge via data analysis tools such as clustering algorithms. In strict technical terms, there is a difference between the two types of \vir{gold solutions}. For their datasets, Dudoit and Fridlyand~\cite{CLEST} elegantly make clear that difference and a closely related approach is used here.

\medskip

\medskip

Each dataset is a matrix, in which each row corresponds to an
element to be clustered and each column to an experimental
condition.
In this dissertation, both microarray and simulated datasets are used. In what follows, a brief description of them is given.

\subsection{Gene-Expression Microarray Data}

The nine datasets from gene-expression microarray, together with the acronyms used in
this dissertation, are reported next. For conciseness, only some relevant facts about them are mentioned.  The interested reader can find additional information in Handl et al.~\cite{Handl05},  for the Leukemia dataset, in Dudoit
and Fridlyand~\cite{CLEST} for the Lymphoma and NCI60 datasets, in Monti et al.~\cite{Monti03} for the Normal, Novartis and St.Jude datasets, finally in Di Ges\'u
et al.~\cite{Genclust},  for the  remaining ones.
In all of the referenced papers, the datasets were used for validation studies. Moreover, in
those papers, the interested reader can find additional pointers to
validation studies using the same datasets. For completeness, it is worth reporting that they have also been used for benchmarking in the context of clustering analysis of
microarray data~\cite{lonardibook,giancarlo08,Monti03,Priness07}.

Particularly relevant is their use in two studies strictly related to this thesis: Giancarlo et al.~\cite{giancarlo08} and Monti et al.~\cite{Monti03}. Indeed, Giancarlo et al. use the following six datasets: CNS Rat, Leukemia, Lymphoma, NCI60, Yeast and PBM, while the remain three datasets are used by Monti et al..
The choice to use all nine is motivated as follows: the datasets of Giancarlo et al. allow to make an accurate and uniform benchmarking of several internal validation measures, taking into account both time and precision as well as to compare our results with extant ones in the Literature. This study is reported in Chapter~\ref{chap:6} and it is quite unique in literature. The datasets of Monti et al. are used to complete the study of {\tt Consensus} and to compare it with a speedup proposed in Chapter~\ref{chap:7}.

Although the Giancarlo et al.'s datasets have relatively few items to classify and relatively few dimensions, it is worth
mentioning that Lymphoma, NCI60 and Leukemia have been obtained by Dudoit and Fridlyand and Handl et al., respectively, via an accurate statistical screening of the three relevant microarray experiments that involved thousands of conditions (columns). That screening process eliminated most of the conditions since there was no statistically significant variation across items (rows). Indeed, one would hardly attempt the clustering of microarray experiments without a preliminary statistical screening aimed at identifying the \vir{relevant parts} of the experiment~\cite{Eisen98}.  It is also worth pointing out that the three mentioned datasets are quite representative of microarray cancer studies. The CNS Rat and Yeast datasets  come from gene functionality studies. The sixth one, PBM, is a dataset that corresponds to a cDNA with a large number of items to classify and it is used to show the current limitations of existing validation methods (see Giancarlo et al. for additional details). Indeed, they have been established  with PBM as input. In particular, when given to {\tt Consensus} as input, the computational demand is such that all experiments were stopped after four days, or they would have taken weeks to complete.  It is also worth pointing out that the remain three datasets are very high dimension ($\simeq 1000$ features). Therefore, they naturally complement the first six since they all have relatively few features (less than 200). In summary, the nine datasets used for the experimentation in this dissertation seem to be a reliable sample of microarray studies, where clustering is used as an exploratory data analysis technique.

\medskip

\medskip

\noindent CNS Rat: It is a $112\times17$ data matrix, obtained from  the expression levels of 112
genes during a rat's central nervous system development. The dataset is  studied by Wen et al.~\cite{CNSRat}, where they
suggest a  partition  of the genes into six
classes, four of which are composed of biologically,
functionally-related genes. This partition is taken as the gold solution,
which is the same one used for the validation of {\tt FOM}.

\medskip

\medskip

\noindent Leukemia: It is a $38\times100$ data matrix, where each
row corresponds to a patient with acute leukemia and each column to
a gene. The original microarray experiment consists of a $72\times6817$ matrix, due to Golub et al.~\cite{golub99}. In order to obtain the current dataset, Handl et al.~\cite{Handl05} extracted from it a $38\times6817$ matrix, corresponding to the \vir{learning set}  in the study of Golub et al. and, via preprocessing steps, they reduced it to the current dimension by excluding genes that exhibited no significant variation across samples. The interested reader can find details of the extraction process in Handl et al.. For this dataset, there is a partition into
three classes and it is taken as gold solution. It is also worthy of mention that Leukemia has become a benchmark standard in the cancer classification community~\cite{BrunetNMF}.

\medskip

\medskip

\noindent Lymphoma: It is a $80\times100$ data matrix, where each row corresponds to a tissue
sample and each column to a gene. The dataset comes from the study of Alizadeh et
al.~\cite{Aliz00} on the three most common adult lymphoma tumors. There is a
partition into three classes and it is taken as the  gold solution.
The dataset has been obtained from the original microarray
experiments, consisting of an $80\times4682$ data matrix, following the same preprocessing steps detailed in Dudoit and Fridlyand~\cite{CLEST}.

\medskip

\medskip

\noindent NCI60: It is a $57\times200$ data matrix, where each row corresponds to a cell line
and each column to a gene. This dataset originates from a microarray study in
gene expression variation among the sixty cell lines of the National
Cancer Institute anti-cancer drug screen~\cite{NCI60}, which consists of a $61\times5244$ data matrix. There is a partition of
the dataset into eight classes, for a total of $57$ cell lines, and it is taken as the  gold
solution. The dataset has been obtained from the original microarray
experiments as described by Dudoit and Fridlyand~\cite{CLEST}.

\medskip

\medskip

\noindent Normal: It is a $90\times1277$ data matrix, where each row corresponds to a tissue
sample and each column to a gene. The dataset comes from the study of Su et
al.~\cite{Su2002} on four distinct cancer types. There is a
partition into four classes and it is taken as the gold solution.

\medskip

\medskip

\noindent Novartis: It is a $103\times1000$  data matrix, where each row corresponds to a tissue
sample and each column to a gene. The dataset comes from the study of Ramaswamy et al.~\cite{Ramaswamy2001} on 13 distinct tissue types. There is a
partition into 13 classes and it is taken as the gold solution.

\medskip

\medskip

\noindent PBM: It is a $2329\times139$ data matrix, where each row corresponds to a cDNA with a fingerprint of
139 oligos. According to
Hartuv et al.~\cite{Hartuv00}, the cDNAs in the dataset originated
from 18 distinct genes, i.e., the classes are known.
The partition  of the dataset into 18 groups was obtained by lab
experiments at Novartis in Vienna. Following that study,  this partition is taken as the
gold solution.

\medskip

\medskip

\noindent St.Jude: It is a $248\times985$  data matrix, where each row corresponds to a tissue
sample and each column to a gene. The dataset comes from the study of Yeoh et al.~\cite{Yeoh2002} on diagnostic bone marrow samples from pediatric acute leukemia patients corresponding to 6 prognostically important leukemia sub-types. There is a
partition into 6 classes and it is taken as the gold solution.

\medskip

\medskip

\noindent Yeast: It is a $698\times72$ data matrix, studied by Spellman et
al.~\cite{Spell98} whose analysis suggests a  partition  of the genes into five
functionally-related classes which is taken as the gold solution and which has been used by Shamir and Sharan for a case study on the performance of clustering
algorithms~\cite{Roded4}.

\medskip

\medskip

In this dissertation, it is referred to as  {\tt Benchmark 1} the following group of datasets: CNS Rat, Leukemia, Lymphoma, NCI60, PBM and Yeast. While the remaining three datasets are referred to as {\tt Benchmark 2}.

\subsection{Simulated Data}

In order to compare some of the algorithms of this thesis, in particular in the speedup of {\tt Consensus} proposed in Chapter~\ref{chap:7}, with the work of Monti et al.~\cite{Monti03} on {\tt Consensus}, some artificial datasets from that study are used.
These datasets have known characteristics, typical of  microarray data, but since they have no \vir{noise}, they are \vir{easy} to classify. Therefore, the experimental results involving them  are considered as complementary to those on real microarray datasets.
The three datasets, together with the acronyms used in
this dissertation, are reported next.

\medskip

\medskip

\noindent Gaussian3: It is a $60 \times 600$ data matrix. It is generated by having 200 distinct features out of the 600 assigned to each cluster. There is a
partition into three classes and that is taken as the  gold solution. The data simulates a pattern whereby a distinct set of 200 genes is up-regulated in one of the three clusters, and down-regulated in the remaining two.
\medskip

\medskip

\noindent Gaussian5: It is a $500 \times 2$ data matrix, and it is show in Fig.~\ref{fig:Gaussian5}.
It represents the union of observations from 5 bivariate Gaussians, 4
of which are centered at the corners of the square of side length $\lambda$, with the 5th Gaussian centered at ($\lambda/2$, $\lambda/2$). A total of 250 samples, 50 per class, were generated, where two values of $\lambda$ are used, namely, $\lambda= 2$ and $\lambda= 3$, to investigate different levels of overlapping between clusters. There is a
partition into five classes and that is taken as the  gold solution.

\begin{figure}[ht]
\begin{center}
\epsfig{file=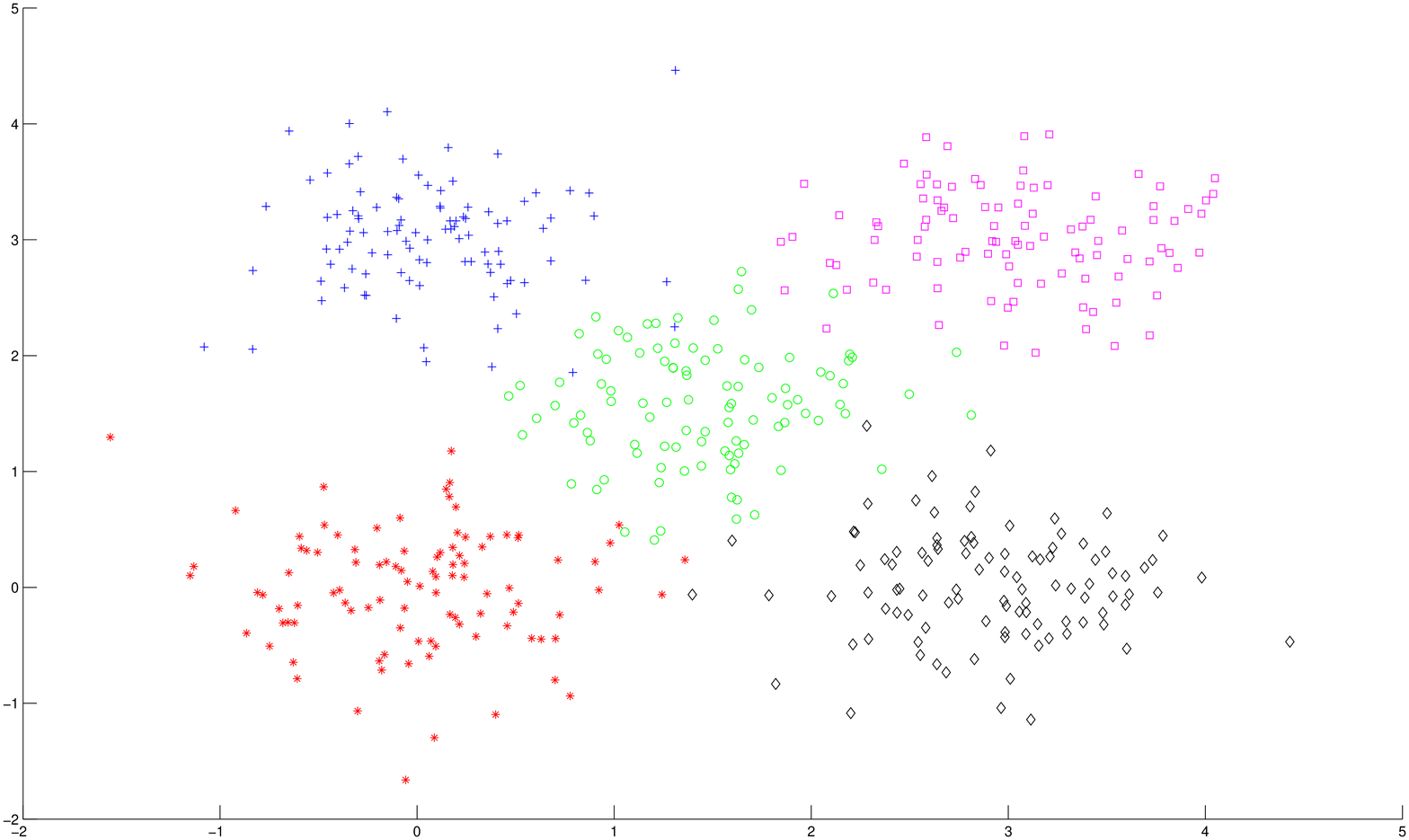,scale=0.3}
\end{center}
\caption{The Gaussian5 dataset.}\label{fig:Gaussian5}
\end{figure}

\medskip

\medskip

\noindent Simulated6: It is a $60 \times 600$ data matrix. It consists of a 600-gene by 60-sample dataset. It can be partitioned into 6 classes with 8, 12, 10, 15, 5, and 10 samples respectively, each marked by 50 distinct genes uniquely up-regulated for that class. Additionally, 300 noise genes (i.e., genes having the same distribution within all clusters) are included. The genes for the different clusters are of varying \vir{sharpness}. That is, the 50 genes marking the first class are the sharpest- with highest differential expression and lowest variation-followed by the 50 genes for the second cluster, etc. Fig.~\ref{fig:Simulated} depicts the expression profile of the 600 genes within each cluster. This partition into 6 classes is taken as the  gold solution.

\begin{figure}[ht]
\begin{center}
\epsfig{file=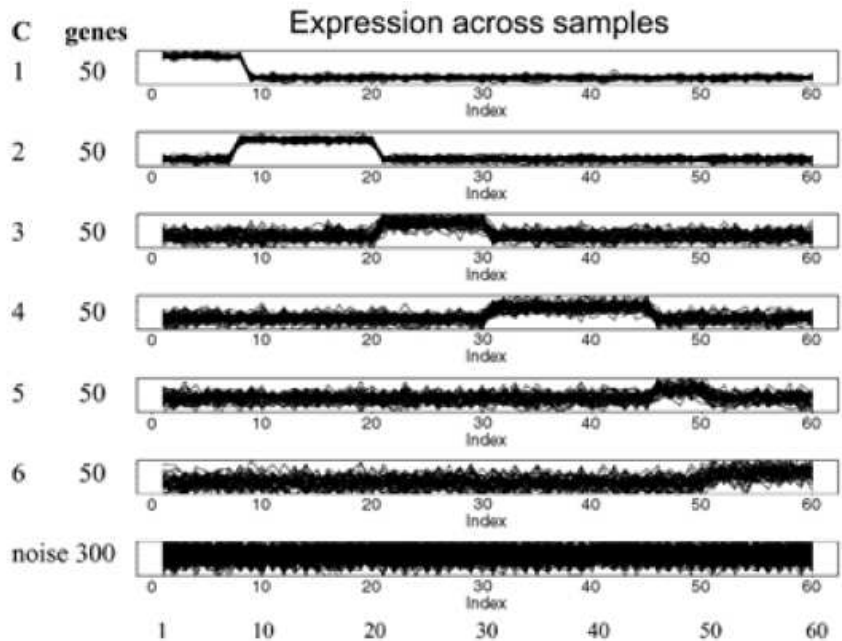,scale=0.6}
\end{center}
\caption{Expression profiles for each
gene within each cluster on the Simulated6 dataset. This figure is takes from~\cite{Monti03}}\label{fig:Simulated}
\end{figure}

\medskip

\medskip

This three datasets are also included in the {\tt Benchmark 2}.

\section{Clustering Algorithms and Their Stability}

In this dissertation, a suite of clustering algorithms is used. Among the hierarchical
methods~\cite{JainDubes}  Hier-A (Average Link), Hier-C (Complete
Link), and Hier-S (Single Link) (see Section~\ref{sec:Hierarchical}).

Moreover, both   K-means~\cite{JainDubes} and {\tt NMF} are used (see Section~\ref{sec:Partional} and Chapter~\ref{chap:NMF}) , both in the version that starts the clustering
from a random partition of the data and in the version where each takes, as part of its input, an initial partition produced by one of the chosen hierarchical methods. For K-means,  the acronyms of those versions are K-means-R, K-means-A, K-means-C and K-means-S, respectively.  An analogous notation is followed for {\tt NMF}.

It is worth pointing out that K-means-R is a randomized algorithm that may provide different answers on the same input dataset. That might make the values of many of the measures studied in this dissertation to  depend critically on the particular execution of the algorithm. Such a dependance is
important for {\tt WCSS},  {\tt KL} and {\tt FOM}. For those
measures and their approximations, the
computation of the relevant curves, on all datasets, with K-means-R is repeated five times. Only negligible differences from run to run is observed. Therefore, in what follows, all reported results refer to a single  run of the algorithms, except for the cases in which an explicit Monte Carlo simulation is required.

For completeness, it is also reported that in this thesis a C/C++ implementation of the {\tt NMF} is used, which is based on the Matlab script available at the Broad institute~\cite{BroadInstitute}. Indeed, it was  converted to a C/C++ version that  was then validated by ensuring it produced the same results as for the Matlab version, in a number of simulations. Notice that this implementation also allows for {\tt NMF} to start from two matrices $W$ and $H$ that actually correspond to a partition of the data into clusters, rather than choosing $W$ and $H$ at random. Such an option is analogous to the well known one offered by K-means, which can start from a given solution rather than randomly. As with K-means, {\tt NMF} also has a faster convergence to a solution when not initialized at random, although the improvement seems not to be significant.

\section{Similarity/Distance Functions}
All of the algorithms use Euclidean distance in order to assess similarity of single elements to be clustered. Such a choice is natural and conservative, as now explained. It places all
algorithms in the same position without introducing biases due to distance function performance, rather than to the algorithm.
Moreover, time course data have been properly standardized (mean equal to  zero  and variance equal to one), so that Euclidean distance would not be penalized on those data. This is standard procedure, e.g.,~\cite{KaYeeFOM},  for those data. The results obtained are conservative since, assuming that one has a provably much better similarity/distance function, one can only hope to get better estimates than ours (else the used distance function is not better than Euclidean distance after all). As it is clear from the upcoming results presented in the next chapters, such better estimates will cause no
dramatic change in the general picture of our findings. The choice is  also natural, in view of the debate regarding the identification of a proper similarity/distance function for clustering gene expression data and the number of such measures available. The state of the art as well some relevant progress in the identification of such measure  is well presented in~\cite{Giancarlo2010,giancarlocibb,Priness07}.

\section{Hardware}
All experiments for the assessment of the precision of each measure were performed in part on several state-of-the-art PCs and in part on  a 64-bit AMD Athlon 2.2 GHz bi-processor with 1 GB of main memory running Windows Server 2003. All the timing experiments reported were performed on the bi-processor,  using one processor per run. The usage of different machines for the experimentation was deemed necessary in order to complete the full set of experiments in a reasonable amount of time. Indeed, as detailed later, some measures require weeks to complete execution on large datasets. It is worth pointing out that all the Operating Systems supervising the computations have a 32 bits precision.

\section{NMF Benchmarking}

In this section a benchmarking of NMF as a clustering algorithms is described.
In order to perform it, the performance of NMF is measured via the three external indices described in Section~\ref{sec:ext}:
{\tt Adjusted Rand} Index, {\tt FM-Index} and {\tt F-Index}.

External indices  can be very useful in evaluating the performance of algorithms and internal/relative indices, with the use of datasets that have a gold standard solution. A brief illustration is given of the methodology for the external validation of a clustering  algorithm, via an external index that needs to be maximized. The same methodology applies to internal/relative indices, as discussed in~\cite{KaYeeFOM}. For a given dataset, one plots the values of the index computed by the algorithm as a function of $k$, the number of clusters. Then, one expects the curve to grow to reach its maximum close or at the number of classes in the reference classification of the dataset. After that number, the curve should fall.

In what follows, the results of the experiments are presented, with the use of the indices. The experiments summarized here refer  to the {\tt Benchmark 1} datasets and the simulated datasets in {\tt Benchmark 2}. For each dataset and each clustering algorithm, each index is computed for a number of cluster values in the range $[2,30]$. Moreover, the time performance of NMF  on the microarray  datasets  is compared with that of the classical clustering algorithms.
To this end the execution time in millisecond of  each algorithm on each  datasets is reported in Table~\ref{table:time}. A dash indicates that the experiment was stopped because of its high computational demand. Indeed, given the dimension of the PBM datasets, NMF is stopped after four days, for this reason the results on this dataset are not reported here.
From the results in Table~\ref{table:time}, it is possible to see that NMF is very slow, at least four order of magnitude of difference with the other clustering algorithms. Moreover, NMF is not able to complete the experiment on PBM dataset.

\begin{table}
\begin{tabular}{|c|c|c|c|c|c|c|}
  \hline
  & CNS Rat & Leukemia & NCI60 & Lymphoma & Yeast  & PBM  \\
  Hier-A & 875 & 219 & 500 & 921 & 594 & $4.4 \times 10^{5}$ \\
  Hier-C & 865 & 250 & 469 & 750 & 625 & $4.6 \times 10^{5}$ \\
  Hier-S & 860 & 296 & 516 & 641 & 609 & $4.3 \times 10^{5}$ \\
  K-means-R & $3.2 \times 10^{3}$ & $2.1 \times 10^{3}$ & $3.2 \times 10^{3}$ & $7.2 \times 10^{3}$ & $1.1 \times 10^{5}$ & $1.1 \times 10^{6}$ \\
  K-means-A & $3.2 \times 10^{3}$ & $1.1 \times 10^{3}$ & $4.2 \times 10^{3}$ & $4.4 \times 10^{3}$ & $1.1 \times 10^{5}$ & $1.7 \times 10^{6}$ \\
  K-means-C & $2.9 \times 10^{3}$ & $1.1 \times 10^{3}$ & $4.0 \times 10^{3}$ & $4.2 \times 10^{3}$ & $1.0 \times 10^{5}$ & $1.3 \times 10^{6}$ \\
  K-means-S & $3.3 \times 10^{3}$ & $1.3 \times 10^{3}$ & $5.2 \times 10^{3}$ & $5.4 \times 10^{3}$ & $1.2 \times 10^{5}$ & $1.4 \times 10^{6}$ \\
  NMF-R & $9.0 \times 10^{6}$ & $8.6 \times 10^{4}$ & $3.9 \times 10^{5}$ & $5.2 \times 10^{5}$ & $2.9 \times 10^{8}$ & - \\
  NMF-A & $3.0 \times 10^{6}$ & $2.4 \times 10^{4}$ & $7.9 \times 10^{4}$ & $1.1 \times 10^{5}$ & $5.5 \times 10^{7}$ & - \\
  NMF-C & $2.4 \times 10^{6}$ & $2.5 \times 10^{4}$ & $7.4 \times 10^{4}$ & $1.1 \times 10^{5}$ & $5.9 \times 10^{7}$ & - \\
  NMF-S & $5.7 \times 10^{6}$ & $2.3 \times 10^{4}$ & $6.9 \times 10^{4}$ & $1.1 \times 10^{5}$ & $4.5 \times 10^{7}$ & - \\
  \hline
\end{tabular}
\caption{Time results in millisecond for all the algorithms on {\tt Benchmark 1} datasets. For
PBM, the experiments were terminated  due to their high
computational demand (weeks to complete).}
\label{table:time}
\end{table}

\medskip

\paragraph{Adjusted Rand Index}
For those experiments, the relevant plots are in Fig.~\ref{Fig:Adjusted_Rand}-\ref{Fig:Adjusted_RandSim} for {\tt Benchmark1} and simulated datasets, respectively.
Based on the results on {\tt Benchmark1} datasets (see  Fig.~\ref{Fig:Adjusted_Rand}), it is possible to see that all the algorithms perform very well on Leukemia and NCI60  datasets. For the remain three datasets, their  performance is somewhat mixed, and sometimes they are not  precise. In particular,  the Hier-S, NMF-R, and NMF-S algorithms.
From the results on the simulated datasets (see Fig.~\ref{Fig:Adjusted_RandSim}), it is possible to see that, except for Hier-S and K-means-S, all the algorithms  perform very well on Gaussian3 and Gaussian5. Whereas, on Simulated6 all the algorithms give an useless indication.

\begin{figure}[ht]
\begin{center}
\epsfig{file=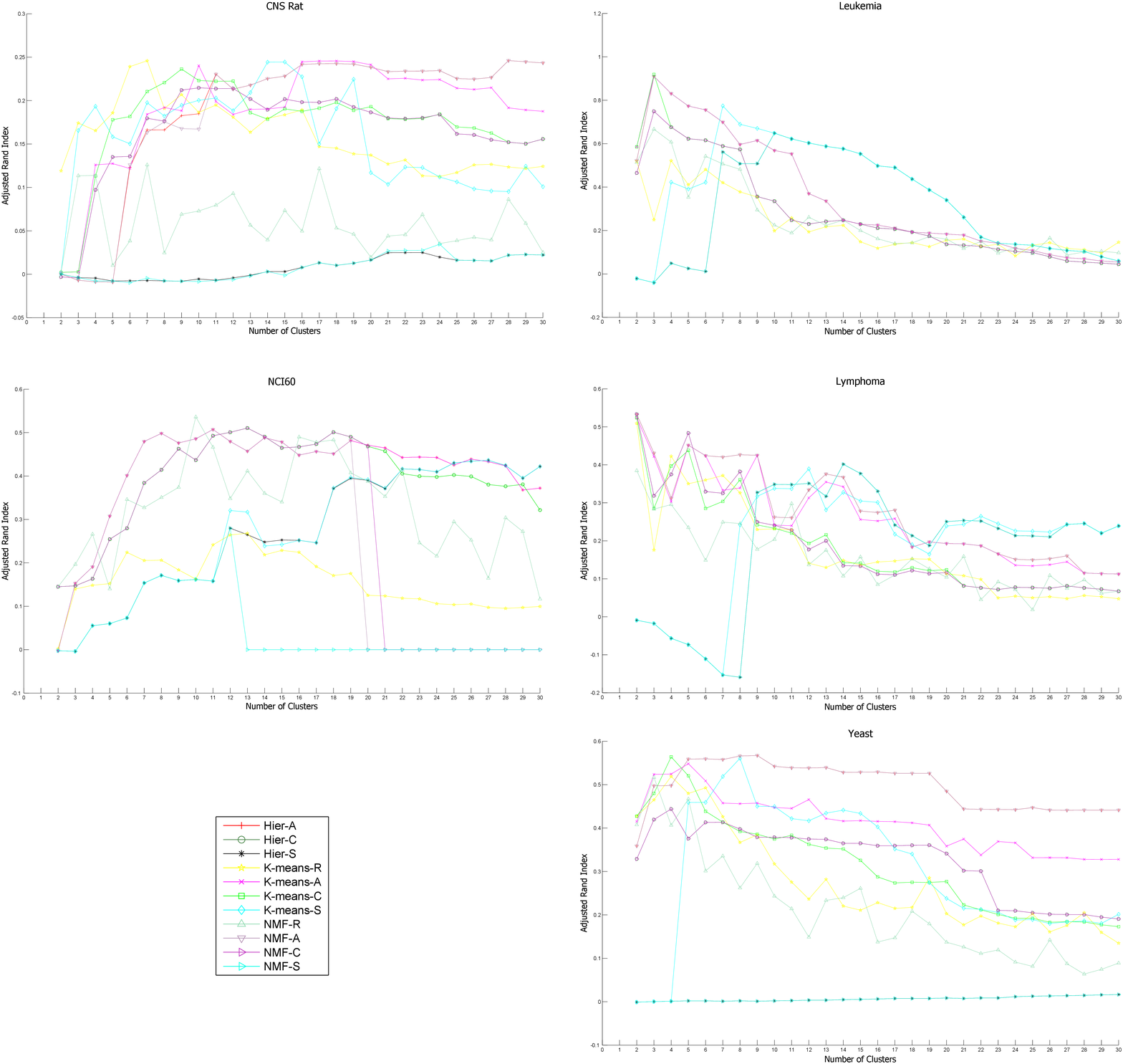,scale=0.2}
\end{center}
\caption{The {\tt Adjusted Rand} Index curves, for each of the {\tt Benchmark1} datasets. In each
figure, the plot of the index, as a function of the number of
clusters, is plotted  differently for each algorithm. For
PBM, the experiments on NMF were terminated  due to their high
computational demand and the corresponding plots has been removed
from the figure.}\label{Fig:Adjusted_Rand}
\end{figure}

\begin{figure}[ht]
\begin{center}
\epsfig{file=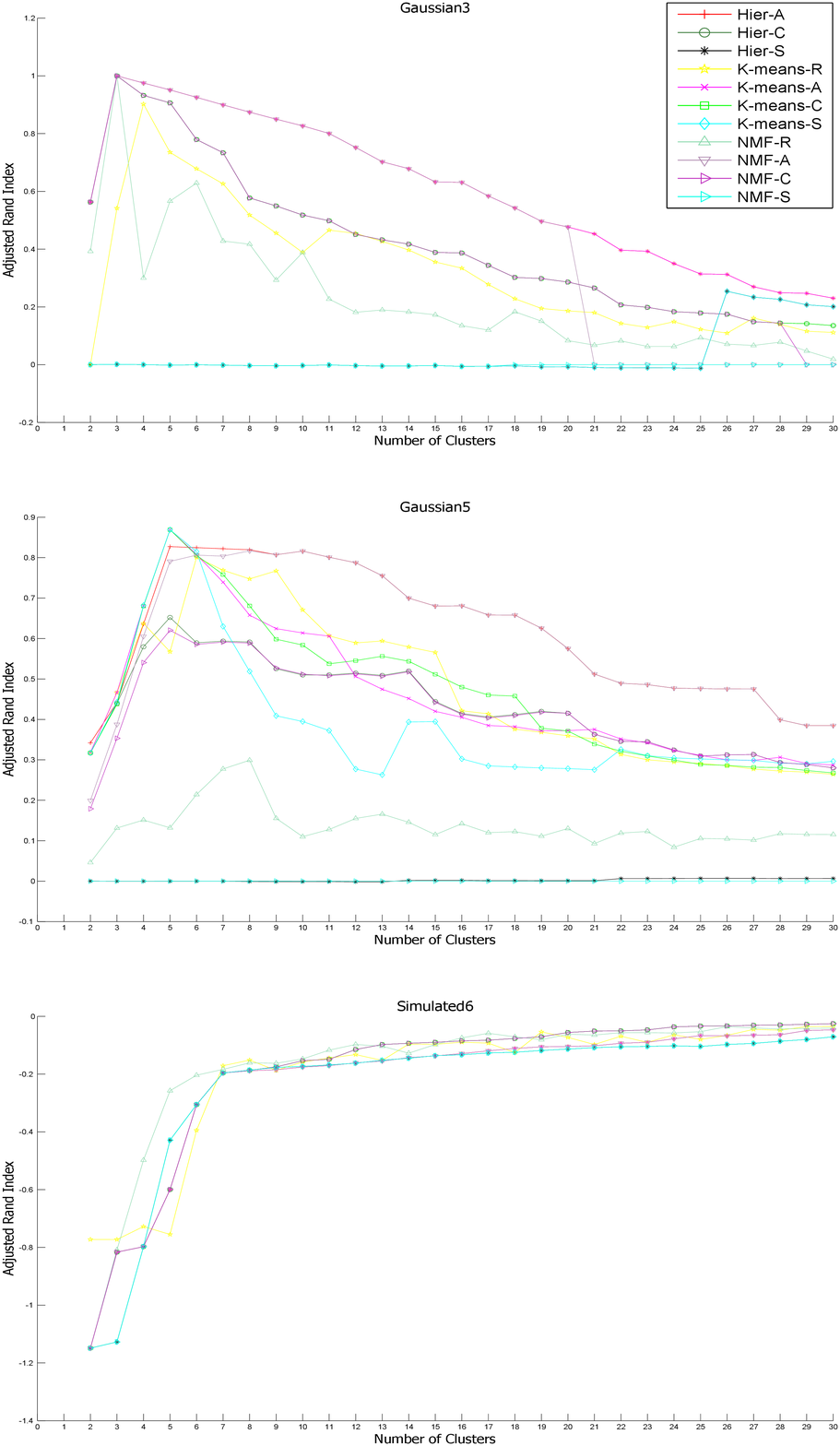,scale=0.4}
\end{center}
\caption{The {\tt Adjusted Rand} Index curves, for each of the simulated dataset in {\tt Benchmark 2}. In each
figure, the plot of the index, as a function of the number of
clusters, is plotted  differently for each algorithm.}\label{Fig:Adjusted_RandSim}
\end{figure}

\paragraph{FM-Index} For those experiments, the
relevant plots are in Fig.~\ref{Fig:FM}-\ref{Fig:FMSim}. Based on them, Hier-S, NMF-S are still the worst among the algorithms,
however now there is no   consistent indication given by  the other
algorithms. Moreover, on Gaussian5 and Simulated6 NMF, both with random and hierarchical initialization, does not perform very well.

\begin{figure}[ht]
\begin{center}
\epsfig{file=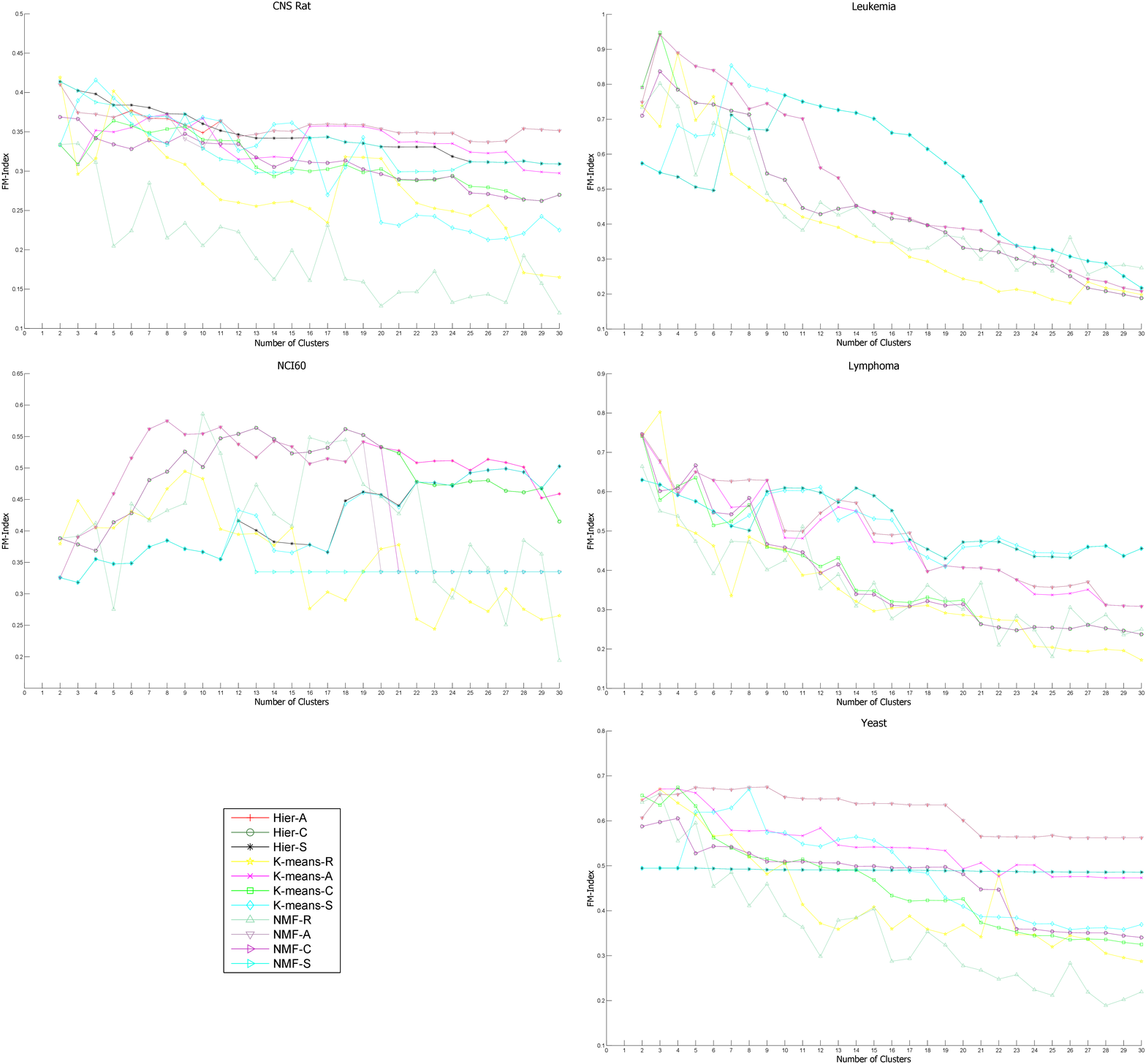,scale=0.2}
\end{center}
\caption{The {\tt FM-Index} curves, for each of the {\tt Benchmark1} datasets. In each
figure, the plot of the index, as a function of the number of
clusters, is plotted  differently for each algorithm. For
PBM, the experiments on NMF were terminated  due to their high
computational demand and the corresponding plots has been removed
from the figure. }\label{Fig:FM}
\end{figure}

\begin{figure}[ht]
\begin{center}
\epsfig{file=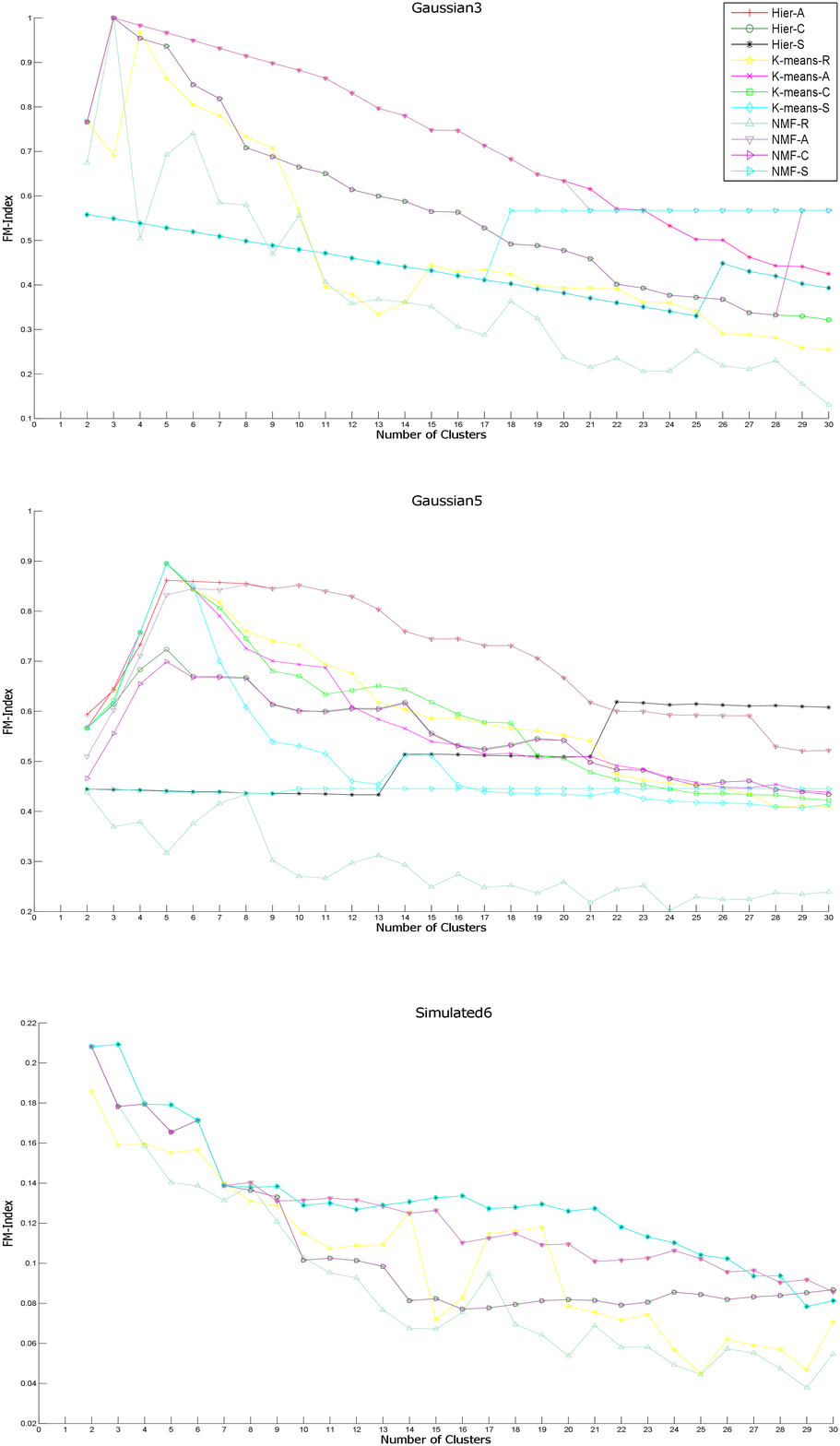,scale=0.4}
\end{center}
\caption{The {\tt FM-index} curves, for each of the simulated dataset in {\tt Benchmark 2}. In each
figure, the plot of the index, as a function of the number of
clusters, is plotted  differently for each algorithm. }\label{Fig:FMSim}
\end{figure}

\paragraph{F-Index} For those experiments, the
relevant plots are in Figs.~\ref{Fig:F}-\ref{Fig:FSim}. On {\tt Benchmark1} datasets, Hier-S and NMF-S are still the worst among the algorithms
and the indications in this case about the other algorithms
are essentially the same as in the case of {\tt Adjusted Rand} Index. Whereas, on simulated dataset all the algorithms have a disappointing performance.\\

\begin{figure}[ht]
\begin{center}
\epsfig{file=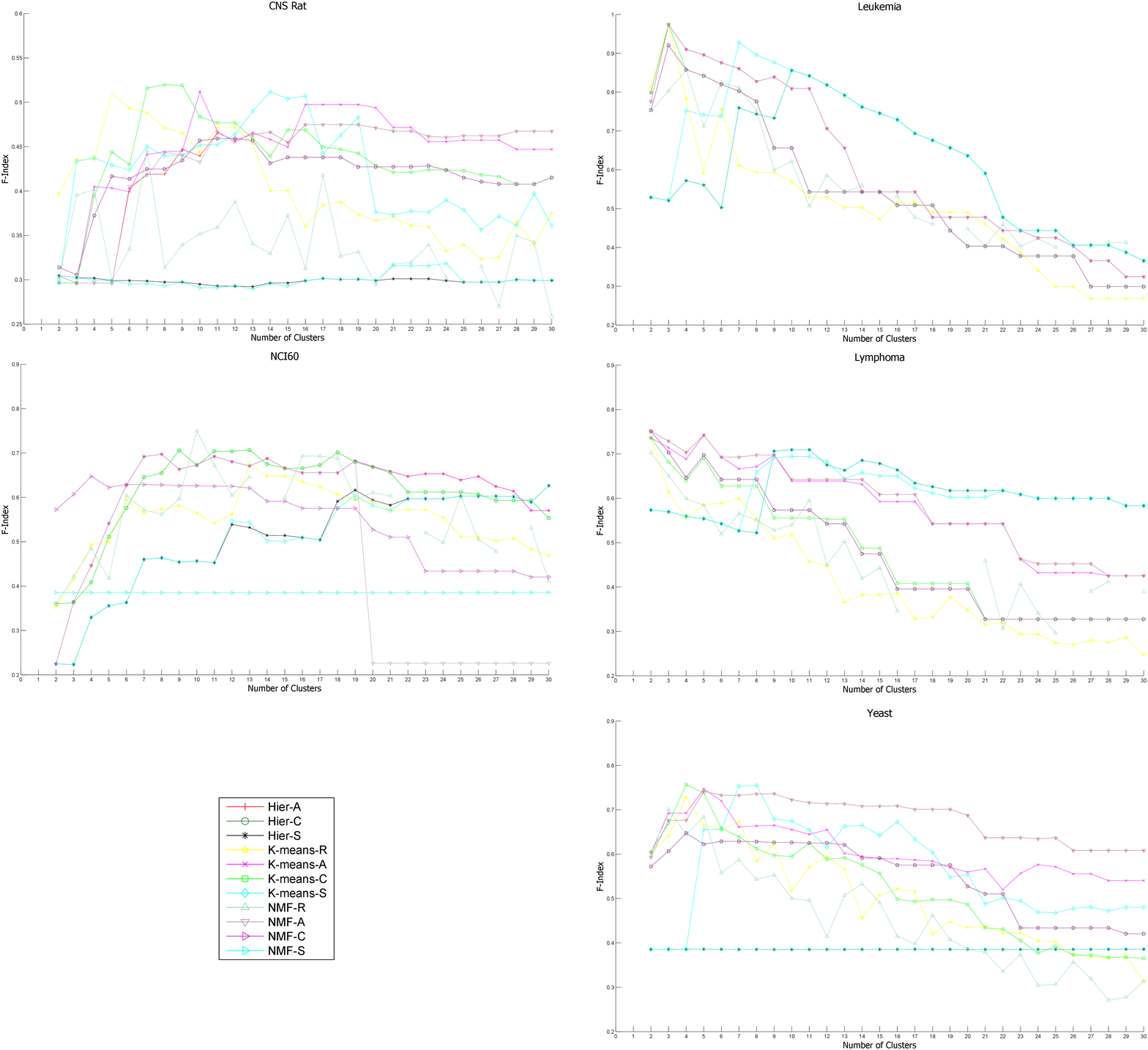,scale=0.2}
\end{center}
\caption{The {\tt F-Index} curves, for each of the {\tt Benchmark1} datasets. In each
figure, the plot of the index, as a function of the number of
clusters, is plotted  differently for each algorithm. For
PBM, the experiments on NMF were terminated  due to their high
computational demand and the corresponding plots has been removed
from the figure.}\label{Fig:F}
\end{figure}

\begin{figure}[ht]
\begin{center}
\epsfig{file=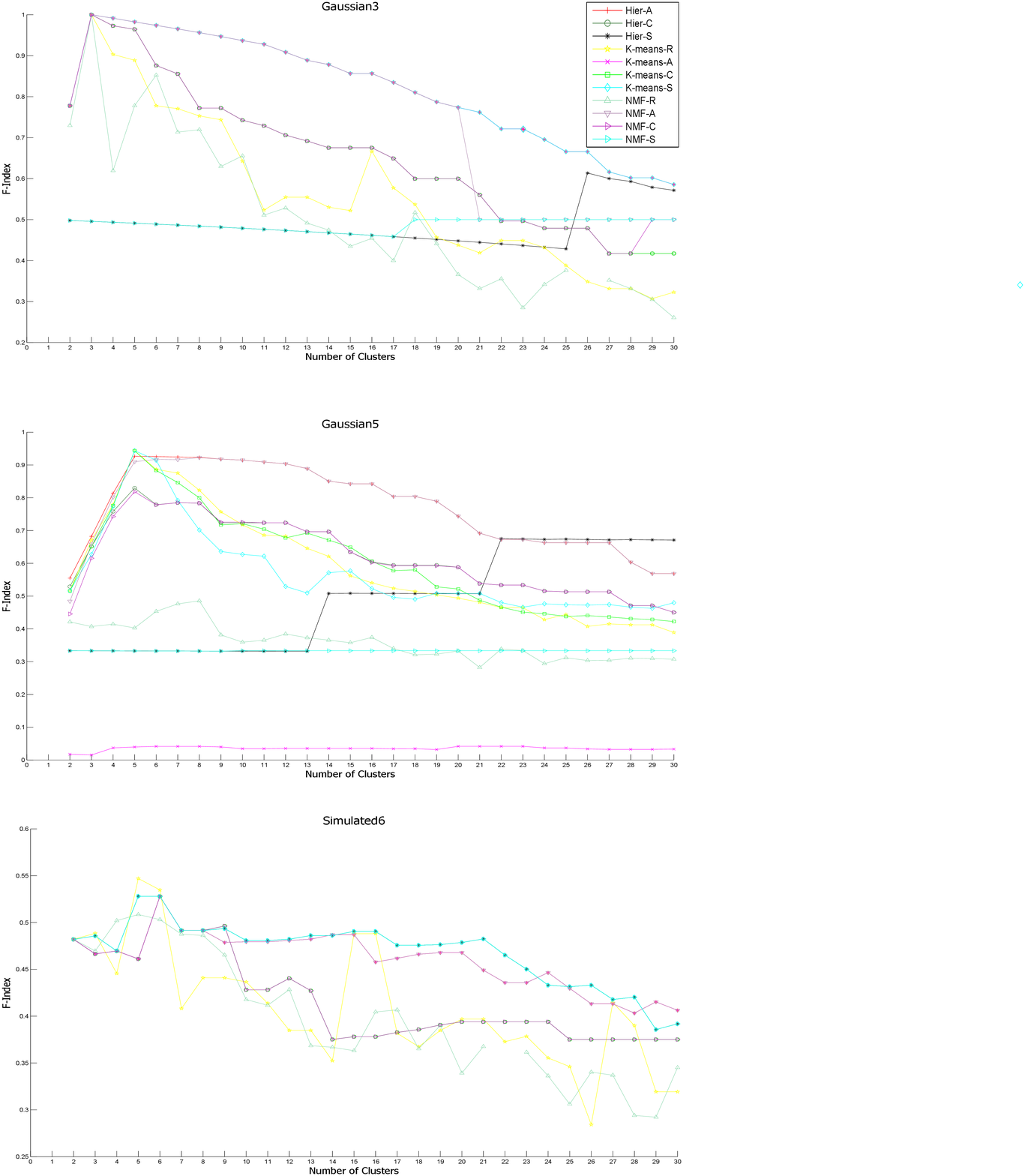,scale=0.4}
\end{center}
\caption{The {\tt F-Index} curves, for each of the simulated dataset in {\tt Benchmark 2}. In each
figure, the plot of the index, as a function of the number of
clusters, is plotted  differently for each algorithm.} \label{Fig:FSim}
\end{figure}

\medskip

\medskip

In Conclusion, from all the results proposed in this section, it should be mentioned that although
all of the three indices have solid statistical justifications, the {\tt Adjusted Rand} Index seems to be the best performer while the {\tt FM-index} is
somewhat  disappointing. Moreover, the use of NMF as  a clustering algorithm is not suggested, in particular for large datasets. Indeed, given the steep computational price (see Table~\ref{table:time}) one has to afford,  its use does not seem to be justified since Hier-A is at least four orders of magnitude faster and with a better precision. In fact, the main power of NMF rests on its pattern discovery ability, and its use as a clustering algorithm seems to be very limiting for this technique.

%% file: Chapter6.tex
\chapter{An Experimental Assessment of Internal Validation Measures}
\label{chap:6}

In this chapter, a benchmarking of some of the internal validation measures described in the previous chapters is presented. In particular, all measures presented in Chapter~\ref{chap:ValidationMeasuers} are considered, while only {\tt ME}, {\tt Consensus} and {\tt Clest} are studied here, since they seem to be the most  representative of the stability based ones. This study tries to  establish the intrinsic, as well as the relative, merit of a measure taking into account both its predictive power and its computational demand. To the best of our knowledge, this is the first study of this kind available in the Literature. It is worthy to anticipate that, based on the results reported here, a speedup of some of the measure presented here becomes a very natural and well motivated problem, that is addressed in the next chapter.

\section{Problems Statement and State of the Art}

An established and rich research area  in bioinformatics
is  the design of new internal  validation measures  that should assess the
biological relevance of the clustering solutions found. Despite the
vast amount of knowledge available in this area in the general
data mining literature~\cite{Ever,Hansen,Hartigan,Tibshrbook,jain,KR90,
Mirkin,Rice}, gene expression data provide unique
challenges. Indeed, the internal validation measure  must predict how many clusters are really
present in a dataset, an already difficult task, made even worse by
the fact that the estimation must be sensible enough to capture the
inherent biological structure of functionally related genes. The excellent survey by Handl et al.~\cite{Handl05} is a
big step forward in making the study of internal  validation measures a central
part of both research and practice in bioinformatics, since it
provides both a technical presentation as well as valuable general
guidelines about their use for post-genomic data analysis. Although
much remains to be done, it is, nevertheless, an initial step.

In order to establish  the intrinsic and relative merit of a measure,
the two relevant questions are:

\begin{itemize}

\item [(A)]  What is the precision of a measure, i.e., its ability to
predict the correct number of clusters in a dataset? That is
usually established by comparing the number of clusters predicted by
the measure against the number of clusters in the gold solution of
several datasets. It is worth recalling from  Chapter~\ref{chap:5} that the gold solution is a partition of the dataset
in classes that can be trusted to be correct, i.e., distinct groups
of functionally related genes.

\item [(B)] Among a collection of measures, which is more accurate, less algorithm-dependent,
etc.,?. Precision versus the  use of computational resources,
primarily execution time, would be an important discriminating
factor.

\end{itemize}

Although the classic studies in the general data mining Literature, mentioned earlier, are also of great relevance for
bioinformatics, there is an acute need for analogous studies
conducted on internal measures  introduced recently and specifically
designed for analysis of microarray data.
In this chapter both  of the stated
questions are addressed for several measures. They are all
characterized by the fact that, for their prediction,  they make use
of nothing  more than the dataset available (see Chapters~\ref{chap:ValidationMeasuers} and~\ref{chap:Stability}): {\tt WCSS}, {\tt KL}, {\tt Clest}, {\tt Consensus} , {\tt FOM},
 {\tt Gap}  and {\tt ME}.  In order to perform this study, as anticipated in Chapter~\ref{chap:5}, only the {\tt Benchmark1} datasets  is used, i.e.: CNS Rat, Leukemia, Lymphoma, NCI60, Yeast and PBM.

Initial studies of the mentioned measures, in connection with  both
Questions (A) and (B),  have been done, primarily, in the papers in
which they were originally proposed.
This study carries further
those studies by providing more focused information about using
those measures for the analysis of gene expression data. For
Question (A), that analysis provides further insights into the
properties of the mentioned measures, with particular attention to
time. For Question (B), a first comparative analysis
involving all of those measures that accounts for both precision and
time is provided. This is particularly relevant in regard to the
\vir{stability-based} methods, i.e., {\tt Clest}, {\tt Consensus}
and {\tt ME}. In fact,
\begin{itemize}
\item[(1)] those three measures are excellent representatives  of methods in the class (see  Handl et al. and~\cite{Tib06,LK04});\\
\item[(2)] Dudoit and Fridlyand mention that it would be desirable to relate {\tt Clest} and {\tt ME} but no comparison seems to be available in the literature;\\
\item[(3)] although it is quite common to include {\tt Clest} and {\tt Gap} in comparative analysis for novel measures, {\tt Consensus} is hardly considered. However, the   experiments presented here show that it should definitely be included.
\end{itemize}

Finally, it is worth pointing out that the results and conclusion of this chapter are also available to Giancarlo et al.~\cite{giancarlo08}.

\section{Intrinsic Precision of the Internal Measures}

In this section,  the experiments with the aim to shed some
light on Question (A) are presented. As discussed in Chapters~\ref{chap:ValidationMeasuers} and~\ref{chap:Stability}, for most
measures, the prediction of the \vir{optimal} number $k^*$ of
clusters is based on the visual inspection of curves and histograms.
For conciseness, all the relevant material is provided in the following
supplementary material web site~\cite{Bench} (Figures section). Here
only summary tables are given, based on the corresponding analysis
of the relevant curves and experiments. In this section, two separate tables
for each measure are  reported, one for the precision and the other for timing results.

It is worthy to anticipate that the next section addresses
the relative merits of each measure and two global summary tables are
reported, but only for the best performers. That is,  for each
measure, the experimental parameters are reported (e.g., clustering
algorithm) only if in that setting the prediction of $k^*$ has been
reasonably close to the gold solution (at most an absolute value
difference of one between the predicted number and the real number)
in at least four of the six datasets used in this chapter.

Moreover, in what follows, for each cell in a table displaying
precision results, a number in a circle with a black background
indicates a prediction in agreement with the number of classes
in the dataset, while a number in a circle with a white background
indicates a prediction that differs, in absolute value, by 1
from the number of classes in the dataset; when the
prediction is one cluster, i.e. Gap statistics, this symbol
rule is not applied because the prediction means no
cluster structure in the data; a number not in a circle indicates
the remaining predictions. As detailed in each
table displaying
timing or precision results,  cells with a dash indicate that either the experiment was
stopped,  because of its high computational demand, or that the
measure gives no  useful indication. The timing results are reported
only on the four smallest datasets. Indeed, for Yeast and PBM, the
computational demand is such on some measures  that either they had
to be stopped or they took weeks to complete. For those two
datasets, the experiments reported here were done using more than one
machine.

\subsection{WCSS}\label{sec:wcss-exp}

For each algorithm,
and each dataset, {\tt WCSS} is computed for a number of
cluster values in the range $[2, 30]$. The relevant plots are in the Figures section at the following  supplementary material web site~\cite{Bench}: Fig. S1 for the K-means algorithms and Fig. S2
for the hierarchical algorithms.

As outlined in the Section~\ref{sec:WCSS-Theory}, given the relevant {\tt WCSS}
curve, $k^*$ is predicted as the abscissa closest to the \vir{knee} in
that curve. The values resulting from the application of this
methodology to the relevant plots are reported in Table~\ref{table:WCSS}, while the timing results for the relevant datasets are reported in Table~\ref{table:WCSS-time}.

One has that  {\tt WCSS}  performs well with
K-means-C and  K-means-A (see Table~\ref{table:WCSS}), on the first five  datasets, while it
gives no reasonably correct indication on PBM. It is a poor
performer with the other clustering algorithms. Those facts give
strong indication that {\tt WCSS} is algorithm-dependent. Finally,
the failure of {\tt WCSS}, with all algorithms,  to give a good
prediction for PBM indicates that {\tt WCSS} may not be of any use
on large datasets having  a large number of clusters.

Overall, the  best performer is  K-means-C.
The relative results  are reported in Tables~\ref{table:summary-table} and~\ref{table:summary-table-time}, for comparison with
the performance of the other measures.
\begin{table}
\begin{footnotesize}
\begin{center}
\begin{tabular}{|l|cccccc|}\hline
& \multicolumn{6}{c|}{Precision}
\\\cline{2-7}
  & CNS Rat & Leukemia & NCI60 & Lymphoma & Yeast & PBM \\
   Hier-A & 10 & \ding{184}& 3 & 6 &
\ding{186}&  - \\
 Hier-C & 10 & \ding{184}& \ding{178}& 8 &
9 &  - \\
 Hier-S & 8 & 10  & \ding{178}& 9
& - &- \\
   K-means-R & 4 &   \ding{184}& 3 & 8 &\ding{175}   &  - \\
 K-means-A & 4 &   \ding{184} &
\ding{178}& 6 & \ding{186}&  -
\\
 K-means-C & \ding{176}  &   \ding{184} &
\ding{189}& 8 &\ding{175} &  -
\\
K-means-S &3 &  \ding{175}   &  \ding{178}
& 8 & 24&  -
\\
\hline
{\bf Gold solution}  & {\bf 6} & {\bf 3} & {\bf 8} & {\bf 3} & {\bf 5} & {\bf 18} \\
\hline
\end{tabular}
\end{center}
\end{footnotesize}
\caption{A summary of the precision results for {\tt WCSS} on all algorithms and {\tt Benchmark 1}
datasets. Cells with a dash indicate that  {\tt
WCSS}  did not give any useful indication.}\label{table:WCSS}
\end{table}

\begin{table}
\begin{footnotesize}
\begin{center}
\begin{tabular}{|l|cccc|}\hline
& \multicolumn{4}{c|}{Timing}
\\\cline{2-5}
  &  CNS Rat & Leukemia & NCI60 & Limphoma\\
   Hier-A &
$1.1  \times 10^{3}$ & $4.0  \times 10^{2}$&$2.1 \times 10^{3}$ & $1.9 \times 10^{3}$ \\
 Hier-C &  $7.0 \times
10^{2}$ &$ 4.0 \times 10^{2}$ & $1.7 \times 10^{3}$ & $1.4 \times 10^{3}$\\
 Hier-S & $2.6  \times
10^{3}$ & $6.0 \times 10^{2}$ & $3.2 \times 10^{3}$& $3.8 \times 10^{3}$\\
   K-means-R  &  $2.4 \times 10^{3}$ & $2.0 \times 10^{3}$ & $8.4 \times 10^{3}$ &  $8.4 \times 10^{3}$ \\
 K-means-A  &$2.3 \times
10^{3}$ & $1.3 \times 10^{3}$& $5.4 \times 10^{3}$  & $5.8 \times
10^{3}$
\\
 K-means-C & $1.7 \times
10^{3}$ & $1.3 \times 10^{3}$&$5.0 \times 10^{3}$ & $4.0 \times
10^{3}$
\\
K-means-S & $2.6 \times 10^{3}$ & $1.6 \times 10^{3}$ & $7.3
\times 10^{3}$ & $7.4 \times 10^{3}$
\\

\hline

\end{tabular}
\end{center}
\end{footnotesize}
\caption{A summary of the timing results for {\tt WCSS}.}\label{table:WCSS-time}
\end{table}

\subsection{KL}
Following the same experimental setup of {\tt WCSS}, the {\tt KL} measure is computed,  for each dataset and each algorithm.
The results, summarized in Tables~\ref{table:KL-Index} and~\ref{table:KL-Index-time}, are rather disappointing: the
measure  provides some reliable indication, across algorithms, only
on the Leukemia and the Lymphoma datasets. Due to such a poor
performance, no results are reported in Tables~\ref{table:summary-table} and~\ref{table:summary-table-time},  for comparison with
the performance of the other measures.

\begin{table}
\begin{footnotesize}
\begin{center}
\begin{tabular}{|l|cccccc|}\hline
& \multicolumn{6}{c|}{Precision}
\\\cline{2-7}
 & CNS Rat & Leukemia & NCI60 & Lymphoma & Yeast & PBM \\
Hier-A & \ding{178}& \ding{184} &3 &  \ding{173}& 17 & 12  \\
Hier-C & 10 & \ding{184}&  2&  \ding{173} & 16 & 15  \\
Hier-S & 21 & 7 & \ding{178}&  9 & 15 & 25  \\
K-means-R & 4 & 27 &  3 &  22 & 29 & 24   \\
K-means-A & 25 & \ding{184} & 3 & \ding{173}& 7 & 16   \\
K-means-C & 2 & \ding{184} &\ding{178} & \ding{173}  & 26 & 24  \\
K-means-S & 4 &\ding{175} & 12 & 8 & 13 & 16  \\
\hline
{\bf Gold solution}  & {\bf 6} & {\bf 3} & {\bf 8} & {\bf 3} & {\bf 5} & {\bf 18} \\
\hline
\end{tabular}
\end{center}
\end{footnotesize}
\caption{
A summary of the precision results for {\tt KL} on all algorithms and on {\tt Benchmark 1}
datasets.
}
\label{table:KL-Index}
\end{table}

\begin{table}
\begin{footnotesize}
\begin{center}
\begin{tabular}{|l|cccc|}\hline
 & \multicolumn{4}{c|}{Timing}
\\\cline{2-5}
 &  CNS Rat & Leukemia & NCI60 & Limphoma\\
Hier-A &  $1.9 \times 10^{3}$ & $6.0 \times 10^{2}$& $2.1 \times 10^{3}$ &  $2.5 \times 10^{3}$ \\
Hier-C &  $1.6 \times 10^{3}$ & $1.1 \times 10^{3}$& $2.5 \times 10^{3}$ & $2.1 \times 10^{3}$  \\
Hier-S &  $3.4\times 10^{3}$ & $1.3 \times 10^{3}$& $3.7 \times 10^{3}$ & $4.9\times 10^{3}$  \\
K-means-R &  $2.7 \times 10^{3}$ & $3.4 \times 10^{3}$& $9.3 \times 10^{3}$ &  $9.0 \times 10^{3}$  \\
K-means-A &$2.3 \times 10^{3}$ & $2.4 \times 10^{3}$& $5.7 \times 10^{3}$ &  $6.2 \times 10^{3}$  \\
K-means-C & $3.0 \times 10^{3}$ & $2.6 \times 10^{3}$ & $5.0 \times 10^{3}$ &  $5.8 \times 10^{3}$ \\
K-means-S & $4.0 \times 10^{3}$ & $2.9 \times 10^{3}$ & $8.0 \times 10^{3}$ &  $8.5 \times 10^{3}$  \\
\hline

\end{tabular}
\end{center}
\end{footnotesize}
\caption{
A summary of the timing results {\tt KL} on all algorithms.
}
\label{table:KL-Index-time}
\end{table}

\subsection{Gap}

For each dataset and each clustering algorithm, three
versions of {\tt Gap} are computed, namely {\tt Gap-Ps}, {\tt Gap-Pc} and {\tt
Gap-Pr}, for a number of cluster values in the range $[1,30]$. {\tt
Gap-Ps} uses  the  Poisson null model, {\tt Gap-Pc} the Poisson null
model aligned with the principal components of the data while {\tt
Gap-Pr}  uses the permutational null model (see Section~\ref{sec:statistics}).
For each of them, a Monte Carlo simulation is performed, 20 steps, in
which the measure returns  an estimated number of clusters for each
step. Each simulation step is based on the generation of 10 data
matrices from the null model used by the measure. At the end of each
Monte Carlo simulation,  the number with the majority of estimates
is taken  as the predicted number of clusters. Occasionally, there
are ties and both numbers are reported. The relevant histograms are
displayed  at the  following  supplementary material web site~\cite{Bench} (Figures section):
Figs. S3-S8 for {\tt Gap-Ps}, Figs. S9-S13 for {\tt Gap-Pc}  and
Figs. S14-S19 for {\tt Gap-Pr}. The results are summarized in Tables~\ref{table:Gap} and~\ref{table:Gap-time}. For PBM
and {\tt Gap-Pc},  each  experiment was terminated  after a week,
since no substantial progress was being made towards its completion.

\begin{table}
\begin{footnotesize}
\begin{center}
\begin{tabular}{|l|cccccc|}\hline
& \multicolumn{6}{c|}{Precision}
\\\cline{2-7}
 & CNS Rat & Leukemia & NCI60 & Lymphoma & Yeast & PBM \\
 {\tt Gap-Ps}-Hier-A & 1 & \ding{175} &  1 & 6 & 3 & 1 \\
 {\tt Gap-Ps}-Hier-C & 1 or 2 & \ding{175} &  2 & 1 or 25& 7 & 15 \\
 {\tt Gap-Ps}-Hier-S & 1 & 1&  1 & 1 & 1 & 1 \\
 {\tt Gap-Ps}-K-means-R & \ding{187}  or \ding{178}  & \ding{175}  or 5 &  3& 8 & 9 &7 \\
 {\tt Gap-Ps}-K-means-A & \ding{178}  & \ding{184} &  1 & 8 & 7 & 9 \\
 {\tt Gap-Ps}-K-means-C & \ding{178} & \ding{175} &  2& 1 or 25  & 12 & 6 \\
 {\tt Gap-Ps}-K-means-S & 9 & \ding{184} &  1 & 1 & 7 & 8 \\

 {\tt Gap-Pc}-Hier-A & 1 & \ding{184} or \ding{175} &  1 & 1 & 1 or 2 or 3  & - \\
 {\tt Gap-Pc}-Hier-C & 1 &  \ding{175} &  1 & 1 &  3  & - \\
 {\tt Gap-Pc}-Hier-S & 1 & 1 &  1& 1  &  1  & - \\
 {\tt Gap-Pc}-K-means-R & 2 & 1 &  1& 1  &   \ding{175}  & -\\
 {\tt Gap-Pc}-K-means-A & 2 &  \ding{175}  &  1 & 1& 3 & - \\
 {\tt Gap-Pc}-K-means-C & 2 & 1 &  1 & 1  & \ding{175} & -\\
 {\tt Gap-Pc}-K-means-S & 3 & 1 &  1 & 1  & 1 & -\\

 {\tt Gap-Pr}-Hier-A & 3 & \ding{175} &  1  &  6 & 3 & 1 \\
 {\tt Gap-Pr}-Hier-C  & \ding{178}  & \ding{175} &  1 & 1 or 25 & 16 & 1 \\
 {\tt Gap-Pr}-Hier-S & 1 or \ding{187}  & 1&  2 & 1  & 1 &2 \\
 {\tt Gap-Pr}-K-means-R& \ding{187}  & \ding{175}  &  5 &  8 & 8 & 8 \\
 {\tt Gap-Pr}-K-means-A & 8 & \ding{175}&  1  & 8 & 13 &4 \\
 {\tt Gap-Pr}-K-means-C & \ding{176} & 6 &  1 & 1 or 25 & 8 & 1 \\
 {\tt Gap-Pr}-K-means-S & \ding{178} & \ding{184} &  2 & 1 & 11 & 1 \\
 \hline
 {\bf Gold solution}  & \textbf{6} & {\bf 3} & {\bf 8} & {\bf 3} & {\bf 5} & {\bf 18}  \\
 \hline
\end{tabular}
\caption{A summary of the precision results for {\tt Gap} on all algorithms
and  {\tt Benchmark 1} datasets, with use of three null models.  For
{\tt Gap-Pc}, on PBM, the experiments were stopped due to their high
computational demand.}\label{table:Gap}
\end{center}
\end{footnotesize}
\end{table}

\begin{table}
\begin{footnotesize}
\begin{center}
\begin{tabular}{|l|cccc|}\hline
 & \multicolumn{4}{c|}{Timing}
\\\cline{2-5}
 & CNS Rat & Leukemia & NCI60 & Limphoma\\
 {\tt Gap-Ps}-Hier-A & $2.7 \times 10^{5}$ & $1.4 \times 10^{5}$ &  $6.1 \times 10^{5}$ & $6.4 \times 10^{5}$\\
 {\tt Gap-Ps}-Hier-C  & $2.3 \times 10^{5}$ & $1.1 \times 10^{4}$&  $3.4 \times 10^{5}$  & $3.2 \times 10^{5}$\\
 {\tt Gap-Ps}-Hier-S  & $6.1 \times 10^{5}$ & $1.9 \times 10^{5}$&  $1.1 \times 10^{6}$ &$1.4 \times 10^{6}$ \\
 {\tt Gap-Ps}-K-means-R  & $8.4  \times 10^{5}$ & $5.0 \times 10^{5}$ &  $1.1 \times 10^{6}$ &$1.0 \times 10^{6}$\\
 {\tt Gap-Ps}-K-means-A & $6.1  \times 10^{5}$ & $4.7 \times 10^{5}$&   $1.1 \times 10^{6}$ & $1.0 \times 10^{6}$\\
 {\tt Gap-Ps}-K-means-C & $6.0  \times 10^{5}$ &$6.1  \times 10^{5}$&   $8.8 \times 10^{5}$& $7.6 \times 10^{5}$\\
 {\tt Gap-Ps}-K-means-S & $9.1  \times 10^{5}$ &$6.5  \times 10^{5}$&   $2.1 \times 10^{6}$ & $1.8 \times 10^{6}$\\

 {\tt Gap-Pc}-Hier-A & $3.2 \times 10^{5}$ & $3.7 \times 10^{5}$&   $8.1 \times 10^{5}$ & $6.1 \times 10^{5}$\\
 {\tt Gap-Pc}-Hier-C  & $1.9 \times 10^{5}$ & $1.9 \times 10^{5}$ & $7.1 \times 10^{5}$ & $5.8 \times 10^{5}$ \\
 {\tt Gap-Pc}-Hier-S  & $7.7 \times 10^{5}$ & $3.1 \times 10^{5}$&  $1.4 \times 10^{6}$ & $1.3 \times 10^{6}$\\
 {\tt Gap-Pc}-K-means-R &$4.9 \times 10^{5}$ &$8.0 \times 10^{5}$ & $1.8 \times 10^{6}$& $1.8 \times 10^{6}$\\
 {\tt Gap-Pc}-K-means-A & $3.8 \times 10^{5}$ & $7.0 \times 10^{5}$&   $1.3 \times 10^{6}$ & $1.4 \times 10^{6}$\\
 {\tt Gap-Pc}-K-means-C  &$4.1 \times 10^{5}$ & $5.8 \times 10^{5}$&   $1.3 \times 10^{6}$ & $1.2 \times 10^{6}$ \\
 {\tt Gap-Pc}-K-means-S &$6.5 \times 10^{5}$ & $7.6 \times 10^{5}$&   $2.4 \times 10^{6}$& $2.0 \times 10^{6}$ \\

 {\tt Gap-Pr}-Hier-A  & $2.5 \times 10^{5}$ & $1.6 \times 10^{5}$&  $3.3 \times 10^{5}$   & $3.8 \times 10^{5}$  \\
 {\tt Gap-Pr}-Hier-C   & $1.3 \times 10^{5}$ & $1.4 \times 10^{5}$&  $3.2 \times 10^{5}$   & $3.7 \times 10^{5}$  \\
 {\tt Gap-Pr}-Hier-S  & $6.8 \times 10^{5}$ & $1.9 \times 10^{5}$&  $1.1 \times 10^{6}$   & $1.4 \times 10^{6}$  \\
 {\tt Gap-Pr}-K-means-R & $8.6 \times 10^{5}$ & $5.4 \times 10^{5}$&   $1.5 \times 10^{6}$& $9.4 \times 10^{5}$  \\
 {\tt Gap-Pr}-K-means-A  & $7.4 \times 10^{5}$ & $5.0 \times 10^{5}$&   $8.7 \times 10^{5}$& $1.0 \times 10^{6}$  \\
 {\tt Gap-Pr}-K-means-C  & $6.7  \times 10^{5}$ &$4.6  \times 10^{5}$&   $8.6 \times 10^{5}$& $1.0 \times 10^{6}$   \\
 {\tt Gap-Pr}-K-means-S  & $1.2  \times 10^{6}$ &$5.4  \times 10^{5}$&   $1.8 \times 10^{6}$& $2.3 \times 10^{6}$ \\
 \hline
\end{tabular}
\caption{A summary of the timing results for {\tt Gap} on all algorithms, with use of three null models.  For
{\tt Gap-Pc}, on PBM, the experiments were stopped due to their high
computational demand.}\label{table:Gap-time}
\end{center}
\end{footnotesize}
\end{table}

The results for {\tt Gap} are somewhat disappointing, as Table~\ref{table:Gap}
shows. However, a few comments are in
order, the first one regarding the null models. Tibshirani et al.
find experimentally that, on simulated data,  {\tt Gap-Pc} is the
clear winner over {\tt Gap-Ps} (they did not consider {\tt Gap-Pr}).
The results reported here show that, as the dataset size increases, {\tt Gap-Pc}
incurs into a severe time performance degradation, due to the
repeated data transformation step. Moreover, on the smaller
datasets, no null model seems to have the edge. Some of the results
are also somewhat puzzling. In particular, although the datasets
have cluster structure, many algorithms return an estimate of
$k^*=1$, i.e., no cluster structure in the data. An analogous
situation was reported by Monti et al.. In their study, {\tt Gap-Ps}
returned $k^*=1$ on two artificial datasets. Fortunately, an
analysis of the corresponding {\tt Gap} curve showed that indeed the
first maximum was at $k^*=1$ but a local maximum was also present
at the correct number of classes, in each dataset. An analogous analysis of the relevant {\tt Gap} curves is also performed here to find that, in analogy with Monti et al., most
curves show a local maximum at or very close to the number of
classes in each dataset, following the maximum at $k^*=1$. An
example curve is given in Fig.~\ref{fig:Gap}. From the above, one can conclude
that inspection of the {\tt Gap} curves and {\em domain knowledge}
can greatly help in disambiguating the case $k^*=1$. It is worth pointing out that experiments conducted by Dudoit and Fridlyand and, independently by Yan and Ye \cite{Yan07}, show that {\tt Gap} tends to overestimate the correct number of clusters, although this does not seem to be the case for the datasets and algorithms used in this dissertation. The above
considerations seem to suggest that the automatic rule for the
prediction of $k^*$ based on {\tt Gap} is rather weak.

\begin{figure}[ht] \centering \epsfig{file=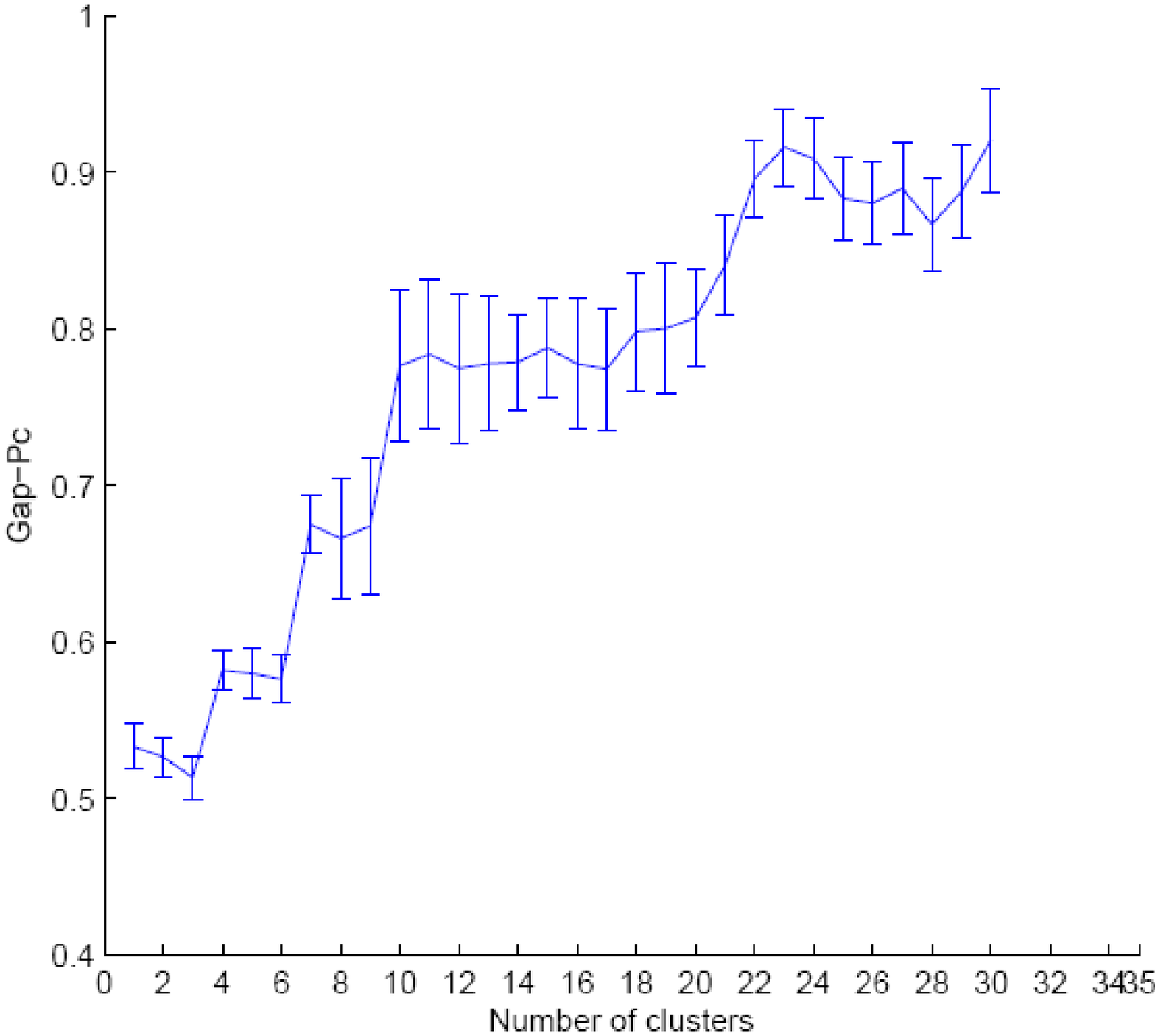,scale=0.6}
\caption{The {\tt Gap-Pc} curve on the Leukemia dataset, with use of the
Hier-S algorithm. At each point, error bars indicate the variation
of the curve across simulations. The curve shows a first maximum at
$k=1$, yielding a prediction of $k^*=1$, the next maximum is at
$k=4$, which is very close to the number of classes $k^*=3$.}\label{fig:Gap}
\end{figure}

\subsection{Clest}\label{subsec:resultclest}

For CNS Rat and Yeast and each clustering algorithm, {\tt
Clest} is computed for a number of cluster values in the range $[2,30]$ while,
for Leukemia,  NCI60 and Lymphoma,  the ranges  $[2,10]$, $[2,15]$
and $[2, 15]$ are used, respectively,  due to the small size of the
datasets. Moreover, although experiments have been started with PBM,
no substantial progress was made after a week of execution and, for
each clustering algorithm, the corresponding experiment was
terminated. Following the same experimental setup of Dudoit and
Fridlyand, for each cluster value $k$, 20 resampling
steps and 20 iterations are performed. In each step, 66\% of the rows of the data
matrix are extracted, uniformly and at random,  to create a learning
set, to be given to the clustering algorithm to be clustered in $k$
groups. As one of its input parameters, {\tt Clest} requires the use
of an  external index $E$ to establish the level of agreement
between two partitions of a dataset.  Here each of the following are used:
the {\tt FM} (the {\tt FM-Index}), {\tt Adj} (the {\tt Adjusted Rand}
Index) and  {\tt F} (the {\tt F-Index}) (see Section~\ref{sec:ext}).

The precision results are summarized in Table~\ref{table:Clest-summary}, while the timing results are reported in Table~\ref{table:Clest-summary-time}.
The  Leukemia, NCI60 and Lymphoma datasets were excluded since the
experiments  were performed on a smaller interval of cluster values
with respect to CNS Rat. This latter interval is the standard one used in this dissertation to make consistent comparisons across measures and
algorithms.

\begin{table}
\begin{footnotesize}
\begin{center}
\begin{tabular}{|l|ccccc|}\hline
& \multicolumn{5}{c|}{Precision}
\\\cline{2-6}
 & CNS Rat & Leukemia & NCI60 & Lymphoma & Yeast \\
{\tt Clest-FM}-Hier-A & 10 & 6 &  10  & 13  & 24  \\
{\tt Clest-FM}-Hier-C  & 10 &\ding{175} & \ding{180}  & 15  & 8  \\
{\tt Clest-FM}-Hier-S & 20 & 10 & 15  & 15  & 1 \\
{\tt Clest-FM}-K-means-R & 8  &\ding{175} &  \ding{189}& \ding{173}  & \ding{175} \\
{\tt Clest-FM}-K-means-A&  18  & 7 & 12  & 15  & 13 \\
{\tt Clest-FM}-K-means-C & 12 & 5 & 12  & 11  &\ding{175}  \\
{\tt Clest-FM}-K-means-S & 24 & 8 & 13 & 15  & 1 \\

{\tt Clest-Adj}-Hier-A & 13 & \ding{184} &3  & \ding{173}   & 11 \\
{\tt Clest-Adj}-Hier-C & 9 &  \ding{175}  & 2  &  \ding{173}  &\ding{175} \\
{\tt Clest-Adj}-Hier-S & 4 & 7 &\ding{180}  & 7  & 26 \\
{\tt Clest-Adj}-K-means-R & \ding{176}   &\ding{175} & 3 & \ding{173}   & 2  \\
{\tt Clest-Adj}-K-means-A & 12 & \ding{184} &3  & \ding{173}   & \ding{186}  \\
{\tt Clest-Adj}-K-means-C & 9 &\ding{173} & 2  & \ding{173}   &\ding{175}   \\
{\tt Clest-Adj}-K-means-S & 20 & 6 & 13  & 6  & 10  \\

{\tt Clest-F}-Hier-A &  \ding{178} & 7 & 10 & 15 & 27  \\
{\tt Clest-F}-Hier-C & 9 & \ding{184} &13  & \ding{184} & \ding{186}  \\
{\tt Clest-F}-Hier-S & 28 & 10 & 15  & 15  & 1  \\
{\tt Clest-F}-K-means-R &\ding{187}& \ding{184} &15  &\ding{173}  & \ding{175}  \\
{\tt Clest-F}-K-means-A & 8 & 6 & 10 & 14  & 11  \\
{\tt Clest-F}-K-means-C & 9 &5 & 12 & \ding{184}& \ding{175} \\
{\tt Clest-F}-K-means-S & 21 & 10 & 15  & 15  & 1  \\

\hline
{\bf Gold solution}  & {\bf 6} & {\bf 3} & {\bf 8} & {\bf 3} & {\bf 5} \\
\hline
\end{tabular}
\end{center}
\end{footnotesize}
\caption{
A summary of the precision results for {\tt Clest} on all algorithms and the
first four datasets, with use of three external indices. For
PBM, the experiments were terminated  due to their high
computational demand (weeks to complete).
}\label{table:Clest-summary}
\end{table}

\begin{table}
\begin{footnotesize}
\begin{center}
\begin{tabular}{|l|c|}\hline
& Timing\\
\cline{2-2}
 &  CNS Rat\\
{\tt Clest-FM}-Hier-A & $1.1 \times 10^{6}$ \\
{\tt Clest-FM}-Hier-C & $1.1 \times 10^{6}$ \\
{\tt Clest-FM}-Hier-S & $1.1 \times 10^{6}$\\
{\tt Clest-FM}-K-means-R &  $1.2 \times 10^{6}$ \\
{\tt Clest-FM}-K-means-A &  $1.4 \times 10^{6}$ \\
{\tt Clest-FM}-K-means-C &  $1.5 \times 10^{6}$ \\
{\tt Clest-FM}-K-means-S & $1.8 \times 10^{6}$ \\
{\tt Clest-Adj}-Hier-A & $1.1 \times 10^{6}$ \\
{\tt Clest-Adj}-Hier-C & $1.1 \times 10^{6}$ \\
{\tt Clest-Adj}-Hier-S & $1.1 \times 10^{6}$ \\
{\tt Clest-Adj}-K-means-R & $1.1 \times 10^{6}$ \\
{\tt Clest-Adj}-K-means-A & $1.4 \times 10^{6}$ \\
{\tt Clest-Adj}-K-means-C & $1.4 \times 10^{6}$ \\
{\tt Clest-Adj}-K-means-S & $1.8 \times 10^{6}$ \\
{\tt Clest-F}-Hier-A &  $1.1 \times 10^{6}$ \\
{\tt Clest-F}-Hier-C & $1.1 \times 10^{6}$ \\
{\tt Clest-F}-Hier-S & $1.1 \times 10^{6}$ \\
{\tt Clest-F}-K-means-R & $1.2 \times 10^{6}$ \\
{\tt Clest-F}-K-means-A & $1.4 \times 10^{6}$ \\
{\tt Clest-F}-K-means-C & $1.5 \times 10^{6}$ \\
{\tt Clest-F}-K-means-S & $1.8 \times 10^{6}$ \\
\hline
\end{tabular}
\end{center}
\end{footnotesize}
\caption{
A summary of the timing results for {\tt Clest} on all algorithms, with use of three external indices.  For
PBM, the experiments were terminated  due to their high
computational demand (weeks to complete). Therefore, the resulting
column is omitted from the table. For the Leukemia, NCI60 and
Lymphoma datasets, the timing experiments are not reported because
incomparable with those of CNS Rat and of the other measures. The
corresponding  columns are eliminated from the table.
}\label{table:Clest-summary-time}
\end{table}

The results show that {\tt Clest} has severe time demand limitations
on large datasets. It also seems to achieve a better performance,
across algorithms with {\tt Adj} and {\tt F}. Moreover, it is
clearly algorithm-dependent, with K-means-R being the best performer
with both  {\tt FM}  and {\tt F}. Those results  are reported in
Tables~\ref{table:summary-table}  {and~\ref{table:summary-table-time} for comparison with the performance of the other measures.

\subsection{ME}
For each of the first five datasets and each clustering algorithm,
{\tt ME} is computed for a number of cluster values in the range
$[2,30]$. Following the same experimental setup of Ben-Hur et al.,
for each cluster value $k$, 100 iterations are performed. In each step, two datasets to be given to the algorithm to be clustered
in $k$ groups are computed. Each dataset is created by extracting uniformly and at random  80\% of the rows. The prediction of $k^*$ is based on the plot of the corresponding histograms, as illustrated in Chapter~\ref{chap:Stability}. As for the external indices that are used, they are the same three
used for {\tt Clest}. The histograms obtained from such an
experimentation are reported at the following supplementary material web site~\cite{Bench}
in Figs. S20-S124. As for PBM, the computations were stopped because
of their computational demand. A summary of the results is given in
Tables~\ref{table:ME} and~\ref{table:ME-time}. Indeed, the performance of {\tt ME} was rather disappointing,
with the exception of Leukemia and Lymphoma, across algorithms and
external indices.

\begin{table}

\begin{footnotesize}\begin{center}
\begin{tabular}{|l|cccccc|}\hline
& \multicolumn{6}{c|}{Precision}
\\\cline{2-7}
  & CNS Rat & Leukemia & NCI60 & Lymphoma & Yeast & PBM \\
 {\tt ME-FM}-Hier-A & 4 & \ding{173} & 2 & \ding{173} & 1 & - \\
 {\tt ME-FM}-Hier-C & 2 & \ding{173} & 2 & \ding{173} & 1 & - \\
 {\tt ME-FM}-Hier-S & 8 & \ding{173} & 2 & \ding{173} & - & - \\
 {\tt ME-FM}-K-means-R & 2 & \ding{173} & 2 & \ding{173} & 3 & - \\
 {\tt ME-FM}-K-means-A & 2 & \ding{173} & 4 & \ding{173} & 2 & - \\
 {\tt ME-FM}-K-means-C & 2 & \ding{173} & 2 & \ding{173} & 3 & - \\
 {\tt ME-FM}-K-means-S & 2 & \ding{173} & 3 & \ding{173} & \ding{175} & - \\

 {\tt ME-Adj}-Hier-A & 3 & 1 & 4 & 1 & 1 & - \\
 {\tt ME-Adj}-Hier-C & 1 & 1 & 2 & \ding{173} & 1 & - \\
 {\tt ME-Adj}-Hier-S & 1 & 1 & 1 & 1 & 1 & - \\
 {\tt ME-Adj}-K-means-R & 1 & 1 & 1 & 2 & 1 & - \\
 {\tt ME-Adj}-K-means-A & 1 & \ding{173} & 1 & \ding{173} & 1 & - \\
 {\tt ME-Adj}-K-means-C & 1 & \ding{173} & 2 & \ding{173} & 1 & - \\
 {\tt ME-Adj}-K-means-S & 1 & 1 & 1 & 1 & 1 & - \\

 {\tt ME-F}-Hier-A & 4 & 1 & 3 & 1 & 1 & - \\
 {\tt ME-F}-Hier-C & 3 & 1 & 2 & \ding{173} & 1 & - \\
 {\tt ME-F}-Hier-S & \ding{178} & 1 & 2 & \ding{184} & - & - \\
 {\tt ME-F}-K-means-R & 1 & \ding{173} & 2 & \ding{173} & 2 & - \\
 {\tt ME-F}-K-means-A & 2 & \ding{184} & 4 & \ding{173} & 2 & - \\
 {\tt ME-F}-K-means-C & 2 & \ding{173} & 2 & \ding{173} & 2 & - \\
 {\tt ME-F}-K-means-S & 2 & 1 & 2 & \ding{175} & \ding{175} & - \\
\hline
 {\bf Gold solution}  & \textbf{6} & {\bf 3} & {\bf 8} & {\bf 3} & {\bf 5} & {\bf 18} \\
 \hline
\end{tabular}\end{center}
\end{footnotesize}

\caption{A summary of the precision results for {\tt ME} on all algorithms
and  {\tt Benchmark 1} datasets, with use of three external indices.  For
PBM, the experiments were stopped due to their high computational
demand (weeks to complete).}\label{table:ME}
\end{table}

\begin{table}

\begin{footnotesize}\begin{center}
\begin{tabular}{|l|cccc|}\hline
& \multicolumn{4}{c|}{Timing}
\\\cline{2-5}
& CNS Rat & Leukemia & NCI60 & Limphoma\\
 {\tt ME-FM}-Hier-A  & $2.8 \times 10^{5}$ & $2.3 \times 10^{5}$ & $7.6 \times 10^{5}$ & $6.5 \times 10^{5}$ \\
 {\tt ME-FM}-Hier-C  & $2.9 \times 10^{5}$  & $2.3 \times 10^{5}$ & $7.6 \times 10^{5}$  & $6.5 \times 10^{5}$ \\
 {\tt ME-FM}-Hier-S  & $2.9 \times 10^{5}$ & $2.3 \times 10^{5}$ & $7.7 \times 10^{5}$  & $6.5 \times 10^{5}$ \\
 {\tt ME-FM}-K-means-R  & $3.6 \times 10^{5}$ & $3.9 \times 10^{5}$ & $1.3 \times 10^{6}$  & $1.1 \times 10^{6}$ \\
 {\tt ME-FM}-K-means-A  & $4.6 \times 10^{5}$ & $3.6 \times 10^{5}$ & $1.1 \times 10^{6}$ & $9.8 \times 10^{5}$\\
 {\tt ME-FM}-K-means-C  & $4.4 \times 10^{5}$ & $3.7 \times 10^{5}$ & $1.1 \times 10^{6}$ & $9.8 \times 10^{5}$ \\
 {\tt ME-FM}-K-means-S  & $5.3 \times 10^{5}$ & $3.7 \times 10^{5}$ & $1.2 \times 10^{6}$ & $1.0 \times 10^{6}$ \\

 {\tt ME-Adj}-Hier-A  & $2.7 \times 10^{5}$ & $2.3 \times 10^{5}$ & $7.6 \times 10^{5}$ & $6.4 \times 10^{5}$ \\
 {\tt ME-Adj}-Hier-C  & $2.7 \times 10^{5}$ & $2.3 \times 10^{5}$ & $7.6 \times 10^{5}$ & $6.5 \times 10^{5}$ \\
 {\tt ME-Adj}-Hier-S  & $2.7 \times 10^{5}$ & $2.3 \times 10^{5}$ & $7.5 \times 10^{5}$ & $6.5 \times 10^{5}$ \\
 {\tt ME-Adj}-K-means-R & $3.4 \times 10^{5}$ & $3.8 \times 10^{5}$ & $1.3 \times 10^{6}$  & $1.1 \times 10^{6}$  \\
 {\tt ME-Adj}-K-means-A  & $4.4 \times 10^{5}$ & $3.5 \times 10^{5}$ & $1.1 \times 10^{6}$  & $9.8 \times 10^{5}$ \\
 {\tt ME-Adj}-K-means-C  & $4.2 \times 10^{5}$ & $4.2 \times 10^{5}$ & $1.1 \times 10^{6}$  & $9.9 \times 10^{5}$ \\
 {\tt ME-Adj}-K-means-S  & $5.1 \times 10^{5}$ & $3.7 \times 10^{5}$ & $1.2 \times 10^{6}$  & $1.0 \times 10^{6}$ \\

 {\tt ME-F}-Hier-A  & $2.8 \times 10^{5}$ & $2.1 \times 10^{5}$ & $7.3 \times 10^{5}$ & $6.4 \times 10^{5}$ \\
 {\tt ME-F}-Hier-C  & $2.8 \times 10^{5}$ & $2.2 \times 10^{5}$ & $7.4 \times 10^{5}$ & $6.4 \times 10^{5}$ \\
 {\tt ME-F}-Hier-S  & $2.8 \times 10^{5}$ & $2.1 \times 10^{5}$ & $7.4 \times 10^{5}$ & $6.4 \times 10^{5}$ \\
 {\tt ME-F}-K-means-R  & $3.6\times 10^{5}$ & $3.8 \times 10^{5}$  & $1.3 \times 10^{6}$ & $1.0 \times 10^{6}$ \\
 {\tt ME-F}-K-means-A  & $4.5 \times 10^{5}$ & $3.5 \times 10^{5}$ & $1.1 \times 10^{6}$ & $9.5 \times 10^{5}$ \\
 {\tt ME-F}-K-means-C  & $4.3 \times 10^{5}$ & $3.5 \times 10^{5}$ & $1.1 \times 10^{6}$ & $9.6 \times 10^{5}$ \\
 {\tt ME-F}-K-means-S  & $5.2 \times 10^{5}$ & $3.5 \times 10^{5}$ & $1.1 \times 10^{6}$ & $1.0 \times 10^{6}$ \\
\hline

\end{tabular}
\end{center}
\end{footnotesize}
\caption{A summary of the timing results for {\tt ME} on all algorithms, with use of three external indices.  For
PBM, the experiments were stopped due to their high computational
demand (weeks to complete).}\label{table:ME-time}
\end{table}

\subsection{Consensus}

For each of the first five datasets and each clustering algorithm,
{\tt Consensus} is computed for a number of cluster values in the
range $[2,30]$. Following the same experimental setup of Monti et
al., for each cluster value $k$, 500 resampling steps are  performed. In
each step, 80\% of the rows of the matrix are extracted uniformly
and at random to create a new dataset, to be given to the clustering
algorithm to be clustered in $k$ groups. The prediction of $k^*$ is
based on the plot of two curves, $\Delta(k)$ and $\Delta'(k)$, as a
function of the number $k$ of clusters. Both curves are defined in Chapter~\ref{chap:Stability}.
As suggested by Monti et al.,  the first curve
is suitable for hierarchical algorithms while the second suits
non-hierarchical ones. The experiment for PBM  were aborted since {\tt
Consensus} was very slow (execution on each algorithm was terminated
after a week). Contrary to Monti et al. indication, the $\Delta(k)$ curve is computed for all algorithms on the first five datasets,
for reasons that will be self-evident shortly. The corresponding
plots are available at the  following supplementary material web site~\cite{Bench} (Figures
section) as Figs. S125-S134. Moreover, the $\Delta'(k)$
curve is also computed for the K-means algorithms, on the same datasets. Recall from Chapter~\ref{chap:Stability} the recommendation to use the $\Delta'$ curve instead of the $\Delta$ curve for non-hierarchical algorithms as suggested by Monti et al.. Briefly, the reason is the following: $A(k)$ is a value that is expected to behaves like a non-decreasing function of $k$, for hierarchical algorithms. Therefore,  $\Delta(k)$ would be expected to be positive or, when negative, not too far from zero. Such a monotonicity of $A(k)$ is not expected for non-hierarchical algorithms. Therefore, another definition of $\Delta$ is needed to ensure a behavior of this function analogous to the hierarchical algorithms. However, from the experiments reported in Table~\ref{table:consensus}, for the K-means algorithms, $A(k)$ displays nearly the same monotonicity properties of the hierarchical algorithms. The end result is that  $\Delta$ can be used for both types of algorithms.  Consequently, since the $\Delta'$
curves are nearly identical to the $\Delta(k)$ ones, they are
omitted. In order to predict the number of clusters in the datasets, for all curves, the rule reported and explained in the Section~\ref{subsec:Consensus} is used: take as $k^*$ the abscissa corresponding to the
smallest non-negative value where the curve starts to stabilize;
that is, no big variation in the curve takes place from that point
on. An analysis on the $\Delta(k)$ curves is performed and the precision
results are summarized in Table~\ref{table:consensus} and the corresponding
timing results in Table~\ref{table:consensus-time}.

As for the precision of {\tt Consensus}, all  algorithms perform
well, except for Hier-S.

In conclusion, {\tt Consensus} seems to be limited by time demand
that makes  it not applicable to large datasets. However, on small
and medium sized datasets, it is remarkably precise across
algorithms. In fact, except for Hier-S, the performance of {\tt
Consensus} is among the best and reported in  Tables~\ref{table:summary-table} and~\ref{table:summary-table-time},  for comparison
with the performance of the other measures.

\begin{table}\begin{center}
\begin{footnotesize}

\begin{tabular}{|l|ccccc|}\hline
& \multicolumn{5}{c|}{Precision}
\\\cline{2-6}
 & CNS Rat & Leukemia & NCI60 & Lymphoma & Yeast  \\

Hier-A & \ding{178}& \ding{184} &\ding{189}& \ding{184} & \ding{186}\\
Hier-C & \ding{187}&\ding{175} & \ding{189}& 5 &\ding{177}\\
Hier-S & 2 & 8 & 10 & \ding{184} &10 \\

K-means-R & \ding{187}  &\ding{175} & \ding{178}& \ding{204}&\ding{177}  \\
K-means-A & \ding{178} & \ding{184} &\ding{189}& \ding{204}&\ding{177}  \\
K-means-C & \ding{187}& \ding{184} &\ding{189}&\ding{175} &\ding{177} \\
K-means-S & \ding{178}&\ding{175} & 10 &\ding{173} &\ding{177}  \\

\hline
{\bf Gold solution}  & {\bf 6} & {\bf 3} & {\bf 8} & {\bf 3} & {\bf 5} \\
 \hline
\end{tabular}
\end{footnotesize}
\caption{
 A summary of the precision results for {\tt Consensus} on all
algorithms and the first five datasets in {\tt Benchmark 1}. For
PBM, the experiments were terminated  due to their high
computational demand and the corresponding column has been removed
from the table.
 }\label{table:consensus}
\end{center}
\end{table}

\begin{table}\begin{center}
\begin{footnotesize}

\begin{tabular}{|l|cccc|}\hline
& \multicolumn{4}{c|}{Timing}
\\\cline{2-5}
  & CNS Rat & Leukemia & NCI60 & Lymphoma \\

Hier-A & $9.2 \times 10^{5}$ & $7.9 \times 10^{5}$ & $2.0 \times 10^{6}$ & $1.9 \times 10^{6}$\\
Hier-C & $8.7 \times 10^{5}$ & $6.9 \times 10^{5}$ & $2.0 \times 10^{6}$ & $2.0 \times 10^{6}$\\
Hier-S & $9.4 \times 10^{5}$ & $8.0 \times 10^{5}$& $2.0 \times 10^{6}$ & $1.7 \times 10^{6}$\\

K-means-R & $1.0 \times 10^{6}$ & $1.3 \times 10^{6}$ & $3.4 \times 10^{6}$ & $3.0 \times 10^{6}$ \\
K-means-A & $1.3 \times 10^{6}$ & $1.6 \times 10^{6}$ & $3.0 \times 10^{6}$ & $2.6 \times 10^{6}$ \\
K-means-C & $1.3 \times 10^{6}$ &  $1.8 \times 10^{6}$ & $2.9 \times 10^{6}$ & $2.6 \times 10^{6}$\\
K-means-S & $1.5 \times 10^{6}$ & $1.8 \times 10^{6}$ &  $3.2 \times 10^{6}$ &  $2.8 \times 10^{6}$ \\

\hline
\end{tabular}
\end{footnotesize}
\caption{
 A summary of the timing results for {\tt Consensus} on all
algorithms and the first five datasets in {\tt Benchmark 1}. For
PBM, the experiments were terminated  due to their high
computational demand and the corresponding column has been removed
from the table.
 }\label{table:consensus-time}
\end{center}
\end{table}

\subsection{FOM}
For each algorithm,  and
each dataset,  the same methodology outlined for
{\tt WCSS} is followed. The relevant plots are in Figs. S135-S136 at the following
supplementary material web site~\cite{Bench} (Figures section). The values
resulting from the application of this methodology to the relevant
plots are reported in Table~\ref{table:FOM-Index} and~\ref{table:FOM-Index-time} together with timing results for the
relevant datasets. From those results, it is possible to see as {\tt FOM} is  algorithm-dependent and gives no useful indication on
large datasets.  The best performing settings are reported in Tables~\ref{table:summary-table},  and~\ref{table:summary-table-time} for
comparison with the performance of the other measures.

\begin{table}\begin{center}
\begin{footnotesize}
\begin{tabular}{|l|cccccc|}\hline
& \multicolumn{6}{c|}{Precision}
\\ \cline{2-7}
 & CNS Rat & Leukemia & NCI60 & Lymphoma & Yeast & PBM \\
Hier-A & \ding{178}&\ding{184} & \ding{178} & 6 &\ding{177} &  - \\
Hier-C & 10 & \ding{175} &  \ding{178} & 7 & \ding{186} & -  \\
Hier-S & 3 & 7 &  \ding{178} &  9  & - &  - \\

K-means-R & \ding{178} & \ding{184} & 6 & 9 & \ding{175} &  - \\
K-means-A & \ding{178} & \ding{184} & 6 & 6 & \ding{175}  &  - \\
K-means-C & \ding{178} &  8  &  \ding{189}&  \ding{175}  & \ding{175} &  -  \\
K-means-S &\ding{187 }& \ding{184} & \ding{189} & 8 &\ding{175} &  -  \\
\hline
{\bf Gold solution}  & {\bf 6} & {\bf 3} & {\bf 8} & {\bf 3} & {\bf 5} & {\bf 18} \\
 \hline
\end{tabular}

\end{footnotesize}\end{center}
\caption{
A summary of the precision results for {\tt FOM} on all algorithms
and on {\tt Benchmark 1} datasets. Cells with a dash indicate that  {\tt
FOM} did not give any useful indication.
}\label{table:FOM-Index}
\end{table}

\begin{table}\begin{center}
\begin{footnotesize}
\begin{tabular}{|l|cccc|}\hline
& \multicolumn{4}{c|}{Timing}
\\ \cline{2-5}
 & CNS Rat & Leukemia & NCI60 & Lymphoma\\
Hier-A & $1.6  \times 10^{3}$ & $7.5  \times 10^{3}$ & $5.1 \times 10^{4}$&   $1.8 \times 10^{4}$\\
Hier-C & $1.6  \times 10^{3}$ & $7.7  \times 10^{3}$& $4.5 \times 10^{4}$& $1.8 \times 10^{4}$\\
Hier-S & $1.6  \times 10^{3}$ & $7.4  \times 10^{3}$& $4.9 \times 10^{5}$&  $1.7 \times 10^{4}$\\

K-means-R  & $2.9  \times 10^{4}$ & $1.9  \times 10^{5}$& $1.3 \times 10^{6}$  &  $6.7 \times 10^{5}$\\
K-means-A  & $2.2  \times 10^{4}$ & $9.3  \times 10^{4}$& $5.5 \times 10^{5}$&  $2.7 \times 10^{5}$\\
K-means-C  & $1.9 \times 10^{4}$ & $9.4 \times 10^{4}$& $5.5 \times 10^{5}$ & $2.6 \times 10^{5}$ \\
K-means-S & $2.9 \times 10^{4}$ & $1.0 \times 10^{5}$ & $7.1 \times 10^{5}$ & $3.6 \times 10^{5}$ \\
\hline
\end{tabular}

\end{footnotesize}\end{center}
\caption{
A summary of the timing results for {\tt FOM} on all algorithms
and on {\tt Benchmark 1} datasets.
}\label{table:FOM-Index-time}
\end{table}

\section{Relative Merits of Each Measure}\label{time_merits}

The discussion here refers to Tables~\ref{table:summary-table} and~\ref{table:summary-table-time}. It is evident that the
K-means algorithms have superior performance with respect to the
hierarchical ones, although Hier-A has an impressive and unmatched
performance with {\tt Consensus}.

However, {\tt Consensus} and {\tt FOM} stand out as being  the most
stable across algorithms. In particular, {\tt Consensus} has a
remarkable stability performance across algorithms and datasets.

For large datasets such as  PBM, the experiments show that all the
measures are severely limited due  either to speed  ({\tt Clest},
{\tt Consensus}, {\tt Gap-Pc}) or to precision as well as speed
(the others). Therefore, this fact stresses even more the need for
good data filtering and dimensionality reduction techniques since
they may help reduce such datasets to sizes manageable by the
measures studied in this chapter.

It is also obvious that, when one takes computer time into account,
there is a hierarchy of measures, with {\tt WCSS} being the fastest
and {\tt Consensus} the slowest. It is worth pointing out that from Table~\ref{table:summary-table-time}  that there is
a natural division of methods in two groups: slow ({\tt Clest}, {\tt
Consensus}, {\tt Gap}) and fast (the other measures). Since there
are at least two orders of magnitude of difference in time
performance between the two groups, it seems reasonable to use one
of the fast methods to limit the search
interval for $k^*$. One can then use {\tt Consensus} in the narrowed
interval. Although it may seem paradoxical, despite its precision
performance, {\tt FOM} does not seem to be competitive in this
scenario. Indeed, it is only marginally better than the best
performing setting of {\tt WCSS}  but at least an order of
magnitude slower in time.

When one does not account for time, {\tt Consensus} seems to be the
clear winner since it offers good precision performance across
algorithms at virtually the same price in terms of time performance.

It is also important  pointing  out that the three
instances of the Stability Measure paradigm have quite diverging performances. Such a fact gives evidence that care must be exercised in taking full advantage of such a powerful paradigm.
Indeed,  only {\tt Consensus} seems to take full advantage of the repeated data generation. A possible reason for this is in the different implementation of the Stability Statistic paradigm to collect the statistic. In fact, {\tt Consensus} builds a matrix, the consensus matrix,  that contains very punctual and global information about the cluster structure of the dataset, while the other two measures try to infer that structure by first splitting the dataset in two subsets and then by using a synoptic function (an external index) to assess the similarity between the partitions. That is, those latter two measures use a coarse assessment of consistency. Moreover, {\tt ME} uses the same algorithm for both of the datasets generated. Probably that induces a big dependency of the measure on the clustering algorithm.

Considering the results of Tables~\ref{table:summary-table} and~\ref{table:summary-table-time},
a promising avenue of
research is to design fast  approximation  algorithms for the
computation of the slowest measures, in particular {\tt Consensus}.
Finally, it is worth pointing out that {\tt Gap}, {\tt Clest}, {\tt ME} and {\tt
Consensus} have various parameters that a user needs to specify.
Those choices may affect both time performance and precision. However, no parameter tuning is available in the Literature.

The next chapter addresses those issues. Indeed, a first study of the best parameter setting for {\tt Consensus} is provided. Moreover, an approximation of several measures is presented,  with particular focus on {\tt Gap} and  {\tt Consensus}. Moreover,   a general scheme  speeding up  the Stability Statistic paradigm is also provided.

\begin{table}
\begin{center}
\begin{footnotesize}
\begin{tabular}{|l|ccccc|}\hline
& \multicolumn{5}{c|}{Precision}
\\ \cline{2-6}
  & CNS Rat & Leukemia & NCI60 & Lymphoma & Yeast  \\
 {\tt WCSS}-K-means-C & \ding{176}  &
\ding{184} &\ding{189}& 8 &\ding{175} \\
{\tt FOM}-Hier-A & \ding{178}&\ding{184} & \ding{178}& 6 & \ding{177} \\
 {\tt FOM}-K-means-C & \ding{178}&  8  &
\ding{189}&\ding{175}  &\ding{175} \\
 {\tt FOM}-K-means-S &\ding{187}&\ding{184}& \ding{189} & 8 &\ding{175} \\

 {\tt Clest-F}-K-means-R &\ding{187}& \ding{184}& 15  &\ding{173}  & \ding{175}  \\

 {\tt Clest-FM}-K-means-R & 8 &
 \ding{175} &  \ding{189}& \ding{173}  &
 \ding{175} \\

 {\tt Consensus}-Hier-A & \ding{178}& \ding{184} &\ding{189}& \ding{184} & \ding{186} \\
 {\tt Consensus}-Hier-C & \ding{187}&\ding{175} & \ding{189}& 5 & \ding{177} \\

 {\tt Consensus}-K-means-R & \ding{187}  &\ding{175} & \ding{178}& \ding{184} &\ding{177} \\
 {\tt Consensus}-K-means-A & \ding{178} & \ding{184} &\ding{189}& \ding{184} &\ding{177} \\
 {\tt Consensus}-K-means-C & \ding{187}& \ding{184} &\ding{189}&\ding{175} & \ding{177} \\
 {\tt Consensus}-K-means-S & \ding{178}&\ding{175} &10 &\ding{173} & \ding{177} \\
\hline
 {\bf Gold solution}  & {\bf 6} & {\bf 3} & {\bf 8} & {\bf 3} & {\bf 5} \\
\hline
\end{tabular}
\end{footnotesize}
\end{center}
\caption{
A summary of the precision for the best performances obtained by each measure. The PBM
dataset has been excluded because no measure gave useful information
about its cluster structure.
}\label{table:summary-table}
\end{table}

\begin{table}
\begin{center}
\begin{footnotesize}
\begin{tabular}{|l|cccc|}\hline
 & \multicolumn{4}{c|}{Timing}
\\ \cline{2-5}
  & CNS Rat & Leukemia & NCI60 & Lymphoma\\
 {\tt WCSS}-K-means-C &
$1.7 \times 10^{3}$ & $1.3 \times 10^{3}$&$5.0 \times 10^{3}$ & $4.0
\times 10^{3}$
\\
{\tt FOM}-Hier-A & $1.6  \times 10^{3}$ & $7.5  \times 10^{3}$ & $5.1 \times 10^{4}$&   $1.8 \times 10^{4}$\\
 {\tt FOM}-K-means-C &
$1.9 \times 10^{4}$ & $9.4 \times 10^{4}$& $5.5 \times 10^{5}$& $2.6
\times 10^{5}$
\\
 {\tt FOM}-K-means-S &  $2.9 \times
10^{4}$ & $1.0 \times 10^{5}$ & $7.1 \times 10^{5}$& $3.6 \times
10^{5}$
\\

 {\tt Clest-F}-K-means-R &  $1.2 \times 10^{6}$ & - & - & -  \\

 {\tt Clest-FM}-K-means-R & $1.2 \times 10^{6}$ & -& - & - \\

 {\tt Consensus}-Hier-A & $9.2 \times 10^{5}$ & $7.9 \times 10^{5}$ & $2.0 \times 10^{6}$ & $1.9 \times 10^{6}$\\
 {\tt Consensus}-Hier-C & $8.7 \times 10^{5}$ & $6.9 \times 10^{5}$ & $2.0 \times 10^{6}$ & $2.0 \times 10^{6}$\\

 {\tt Consensus}-K-means-R & $1.0 \times 10^{6}$ & $1.3 \times 10^{6}$ & $3.4 \times 10^{6}$ & $3.0 \times 10^{6}$ \\
 {\tt Consensus}-K-means-A & $1.3 \times 10^{6}$ & $1.6 \times 10^{6}$ & $3.0 \times 10^{6}$ & $2.6 \times 10^{6}$ \\
 {\tt Consensus}-K-means-C & $1.3 \times 10^{6}$ &  $1.8 \times 10^{6}$ & $2.9 \times 10^{6}$ & $2.6 \times 10^{6}$\\
 {\tt Consensus}-K-means-S & $1.5 \times 10^{6}$ & $1.8 \times 10^{6}$ &  $3.2 \times 10^{6}$ &  $2.8 \times 10^{6}$ \\
\hline
\end{tabular}
\end{footnotesize}
\end{center}
\caption{
A summary of the timing for the best performances obtained by each measure. The PBM
dataset has been excluded because no measure gave useful information
about its cluster structure.
}\label{table:summary-table-time}
\end{table} 

%% file: Chapter7.tex
\chapter{Speedups of Internal Validation Measures Based on Approximations}
\label{chap:7}
One open question that was made explicit by the study of the previous chapter is the  design of a data-driven  internal validation measure that is both precise and fast, and capable of granting scalability with dataset size. Such a lack of scalability for the most precise internal validation measures is one of the main computational bottlenecks in the process of cluster evaluation for microarray data analysis. Its elimination is far from trivial~\cite{klie10} and even
partial progress on  this problem is perceived as important. In the research area embodying the design and analysis of algorithms, when a problem is computationally difficult, a usual approach to its solution is to design fast heuristics and/or provably good approximation algorithms, in order to obtain solution  that are \vir{close} to {the ones that would be produced by the exact algorithms.

In this chapter,  the algorithmic approach just outlined is investigated in a systematic way in the realm of  internal validation measures, with the goal of narrowing the time performance gap,  identified in the previous chapter, between the most precise and the fastest measures. In particular, several algorithmic approximations and two general approximations schemes are presented.

\section{An Approximation of WCSS}

The approximation of {\tt WCSS} proposed here  is based
on the idea of reducing the number of executions of a clustering
algorithm $C$ for the computation of ${\tt WCSS}(k)$, for each $k$
in a given interval $[1, k_{max}]$. In fact, given an integer $R>0$, which is referred to as refresh step,
the approximate algorithm to compute {\tt WCSS}  uses algorithm $C$
to obtain a clustering solution with $k$ clusters, only for values
of $k$ multiples of $R$. For
all other $k$'s, a clustering solution is obtained by merging two
clusters in a chosen clustering solution already available. The
procedure in Fig.~\ref{algo:WCSS-R} gives the high level details. It takes as input $R,
C, D$ and $k_{max}$. Algorithm $C$ must be able to take as input a
clustering solution with $k$ clusters and refine it to give as
output a clustering solution with the same number of clusters.

\begin{figure}
\[
\setlength{\fboxsep}{12pt}
\setlength{\mylength}{\linewidth}
\addtolength{\mylength}{-2\fboxsep}
\addtolength{\mylength}{-2\fboxrule}
\ovalbox{
\parbox{\mylength}{
\setlength{\abovedisplayskip}{0pt}
\setlength{\belowdisplayskip}{0pt}

\begin{pseudocode}{WCSS-R}{R, C, D, k_{max}}
1. \mbox{ Let $P_{k_{max}}$ be the partition of $D$ into $k_{max}$ cluster with the use of $C$}\\
2.\mbox{ Compute ${\rm WCSS}(k_{max})$ using $P_{k_{max}}$}\\
\mbox{ }\FOR k\GETS k_{max}-1 \DOWNTO 1 \DO \\
\mbox{ }\mbox{ }\BEGIN
\mbox{ }\mbox{ }\mbox{ }\mbox{ }\mbox{ }\mbox{ }\COMMENT{Merge}\\
3.\mbox{ }\mbox{ }\mbox{ }\mbox{Merge
the two clusters in $P_{k+1}$  with minimum Euclidean  distance}\\ \mbox{ }\mbox{ }\mbox{ }\mbox{ }\mbox{
between centroids to obtain a  temporary clustering solution}\\ \mbox{ }\mbox{ }\mbox{ }\mbox{ }\mbox{
$P'_{k}$ with $k$ clusters}\\
\mbox{ }\mbox{ }\mbox{ }\mbox{ }\mbox{ }\mbox{ }\COMMENT{Refresh}\\
4.\mbox{ }\mbox{ }\mbox{ }\mbox{ }\IF  (R=0) \mbox{ or }(  k \hbox{ mod } R>0) \THEN
P_k \GETS P'_k
\mbox{ }\mbox{ }\mbox{ }\ELSE \mbox{Compute new $P_k$ based on $P'_k$.
That is, $P'_k$ is given as input to $C$, as an }\\ \mbox{ }\mbox{ }\mbox{ }\mbox{ }\mbox{initial partition of
$D$ in $k$ clusters, and $P_k$ is  the result of that computation.}\\
5.\mbox{ }\mbox{ }\mbox{ }\mbox{ }\mbox{ Compute ${\tt WCSS}(k)$ using $P_k$} \\
\mbox{ }\mbox{ }\END

\end{pseudocode}
}
}
\]
\caption{The \codice{WCSS-R} procedure}\label{algo:WCSS-R}
\end{figure}

Technically, the main idea in the approximation
scheme is to interleave the execution of a partitional clustering
algorithm $C$ with a merge step typical of agglomerative clustering.
The gain in speed is realized by having a fast merge step, based on
$k+1$ clusters, to obtain $k$ clusters instead of a new full fledged
computation, starting from scratch, of the algorithm $C$ to obtain
the same number of clusters. The approximation scheme would work
also for hierarchical algorithms, provided that they comply with the
requirement that, given as input a dataset, they will return a
partition into $k$ groups. However, in this circumstance, the
approximation scheme would be a nearly exact replica of the
hierarchical algorithm.  In conclusion, a general
approximation scheme is proposed, where the gain is realized when the merge step
is faster than a complete computation of a clustering algorithm $C$.
In this thesis,  experiments have been conducted with K-means-R on {\tt Benchmark 1} datasets,  and with values of the refresh
step $R=0,2,5$, i.e., the partitional clustering algorithm is used
only once, every two and five steps, respectively. The corresponding results are summarized in Tables~\ref{table:WCSSapp} and~\ref{table:WCSSapp-time}, together with the results of {\tt WCSS} already reported in Tables~\ref{table:WCSS} and~\ref{table:WCSS-time}. As is self-evident from the results in the former tables, the approximation has a better predicting power
than the original {\tt WCSS} curve  (obtained via all other
clustering algorithms one has experimented with). In fact, the approximation  is
among the best performers. Moreover, depending on the
dataset, it is from a few times to an order of magnitude faster
than the K-means algorithms (see Table~\ref{table:WCSSapp-time}).

\begin{table}
\begin{footnotesize}
\begin{center}
\begin{tabular}{|l|cccccc|}\hline
& \multicolumn{6}{c|}{Precision}
\\\cline{2-7}
  & CNS Rat & Leukemia & NCI60 & Lymphoma & Yeast & PBM \\
   Hier-A & 10 & \ding{184}& 3 & 6 & \ding{186}&  - \\
 Hier-C & 10 & \ding{184}& \ding{178}& 8 & 9 &  - \\
 Hier-S & 8 & 10  & \ding{178}& 9 & - &- \\
 R-R0 & \ding{176}  & \ding{175} & \ding{180} & \ding{184} &\ding{175} & -\\
 R-R2 & \ding{178}& 5 & 15 &  \ding{175}  &\ding{175} &  -\\
 R-R5 & \ding{187}& 5 & \ding{180} & 5 &\ding{175} & -\\
 K-means-R & 4 & \ding{184}& 3 & 8 &\ding{175}  &  - \\
 K-means-A & 4 & \ding{184} & \ding{178}& 6 & \ding{186}&  - \\
 K-means-C & \ding{176} & \ding{184} & \ding{189}& 8 &\ding{175} & - \\
K-means-S & 3 & \ding{175} & \ding{178} & 8 & 24 & - \\
\hline
{\bf Gold solution}  & {\bf 6} & {\bf 3} & {\bf 8} & {\bf 3} & {\bf 5} & {\bf 18} \\
\hline
\end{tabular}
\end{center}
\end{footnotesize}
\caption{A summary of the precision results for {\tt WCSS} of Table~\ref{table:WCSS} with the addition of its approximation. Cells with a dash indicate that  {\tt WCSS}  did not give any useful indication.}\label{table:WCSSapp}
\end{table}

\begin{table}
\begin{footnotesize}
\begin{center}
\begin{tabular}{|l|cccc|}\hline
& \multicolumn{4}{c|}{Timing}
\\\cline{2-5}
  &  CNS Rat & Leukemia & NCI60 & Limphoma\\
 Hier-A & $1.1  \times 10^{3}$ & $4.0  \times 10^{2}$&$2.1 \times 10^{3}$ & $1.9 \times 10^{3}$ \\
 Hier-C &  $7.0 \times 10^{2}$ &$ 4.0 \times 10^{2}$ & $1.7 \times 10^{3}$ & $1.4 \times 10^{3}$\\
 Hier-S & $2.6  \times 10^{3}$ & $6.0 \times 10^{2}$ & $3.2 \times 10^{3}$& $3.8 \times 10^{3}$\\
 R-R0 & $1.2 \times 10^{3}$ & $8.0 \times 10^{2}$ & $4.1 \times 10^{3}$ & $3.0 \times 10^{3}$\\
 R-R2 & $1.3 \times 10^{3}$ & $8.0 \times 10^{2}$&$5.3 \times 10^{3}$ & $3.2 \times 10^{3}$\\
 R-R5 &  $1.2 \times 10^{3}$ & $8.0 \times 10^{2}$& $4.6 \times 10^{3}$ &$3.2 \times 10^{3}$\\
 K-means-R  &  $2.4 \times 10^{3}$ & $2.0 \times 10^{3}$ & $8.4 \times 10^{3}$ &  $8.4 \times 10^{3}$ \\
 K-means-A  &$2.3 \times 10^{3}$ & $1.3 \times 10^{3}$& $5.4 \times 10^{3}$  & $5.8 \times 10^{3}$\\
 K-means-C & $1.7 \times 10^{3}$ & $1.3 \times 10^{3}$&$5.0 \times 10^{3}$ & $4.0 \times 10^{3}$ \\
 K-means-S & $2.6 \times 10^{3}$ & $1.6 \times 10^{3}$ & $7.3 \times 10^{3}$ & $7.4 \times 10^{3}$\\
\hline
\end{tabular}
\end{center}
\end{footnotesize}
\caption{A summary of the timing results for {\tt WCSS} of Table~\ref{table:WCSS-time} with the addition of its approximation.}\label{table:WCSSapp-time}
\end{table}

\section{A Geometric Approximation of Gap Statistics: G-Gap}

The geometric interpretation of {\tt Gap} given in Chapter~\ref{chap:ValidationMeasuers} and the behavior of the
{\tt WCSS} curve across null models suggests a fast
approximation, which it is referred to as {\tt G-Gap}. The intuition, based on experimental
observations, is that one can skip the entire simulation phase of
{\tt Gap}, without compromising too much the accuracy of the
prediction of $k^*$. Indeed, based on the {\tt WCSS} curve, the plot
of the $\log {\tt WCSS}$ curve one expects, for a given clustering
algorithm and null model, is a straight line with a slope somewhat
analogous to that of the $\log {\tt WCSS}$ curve and dominating it (see Fig.~\ref{fig:GAP-exp}).
Therefore, one can simply identify the \vir{knee} in the {\tt WCSS} by
translating the end-points of the $\log {\tt WCSS}$ curve on the
original dataset by a given amount $a$, to obtain two points $g_s$
and $g_e$. Those two points are then joined by a straight line,
which is used to replace the null model curve to compute the segment
lengths used to predict $k^*$, i.e, the first maximum among those
segment lengths as $k$ increases.  An example is provided in Fig.~\ref{fig:G-Gap}
with the {\tt WCSS} curve of Fig.~\ref{fig:WCSS-example}. The prediction is $k^*=3$,
which is very close to the correct $k^*=2$. It is worth pointing out that the
use of the {\tt WCSS} curve in the figure is to make clear the
behavior of the segment lengths, which would be unnoticeable with
the $\log {\tt WCSS}$ curve, although the result would be the same.

\begin{figure}[ht] \centering \epsfig{file=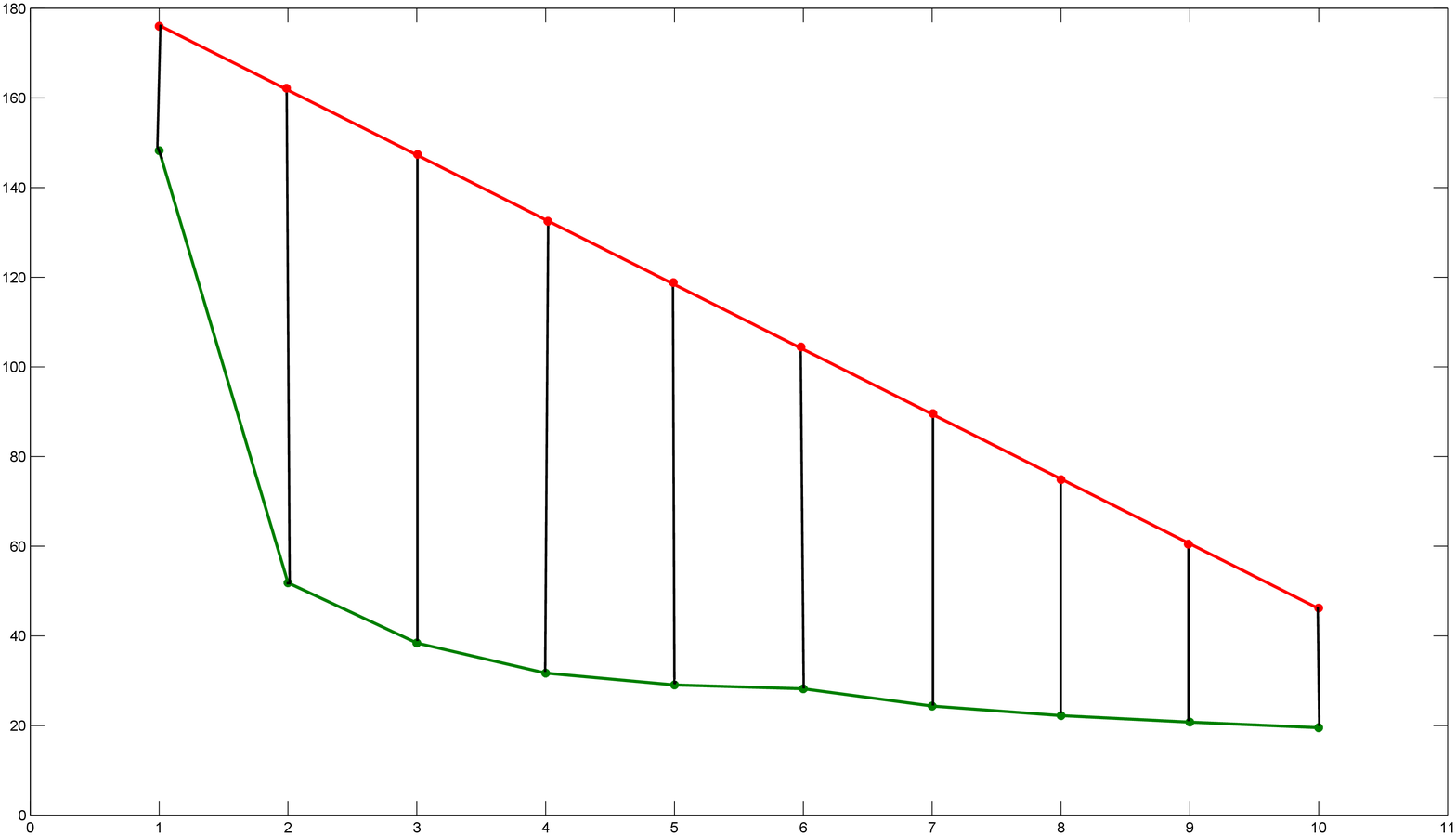,scale=0.4}
\caption{ The G-Gap Heuristic. The curve in green is a {\tt WCSS} curve obtained on the dataset of Fig.~\ref{fig:KmeansExample}(a) with the use of the K-means algorithm. The line in red is obtained by projecting upward the end points of the {\tt WCSS} curve by a units and then joining them. It is a heuristic approximation of {\tt WCSS} for a null model. The vertical lines have the same role as in Gap and the rule to identify $k^*$ is the same, yielding a value  $k^*$ = 3, a value very close to the correct number of classes (two) in the dataset.}\label{fig:G-Gap}
\end{figure}

As for {\tt G-Gap}, the geometric approximation of {\tt Gap}, each algorithm and each dataset, the
corresponding {\tt WCSS} curve and its approximations has been computed in the
interval $[1,30]$ on {\tt Benchmark 1} datasets. The rule described above has been applied to get the value of $k^*$. The corresponding results are
summarized in Tables~\ref{table:Gapapp} and~\ref{table:Gapapp-time}  with the addition of the results of {\tt Gap} reported in Tables~\ref{table:Gap} and~\ref{table:Gap-time}. As it is evident from Tables~\ref{table:Gapapp}, the overall
performance of {\tt G-Gap} is clearly superior to {\tt Gap},
irrespective of the null model. Moreover, depending on the dataset,
it is from two to three orders of magnitude faster (see Table ~\ref{table:Gapapp-time}).

\begin{table}
\begin{footnotesize}
\begin{center}
\begin{tabular}{|l|cccccc|}\hline
& \multicolumn{6}{c|}{Precision}
\\\cline{2-7}
 & CNS Rat & Leukemia & NCI60 & Lymphoma & Yeast & PBM \\
 {\tt G-Gap}-Hier-A & \ding{178} & \ding{184} & 1 & \ding{173} & 3 & 1 \\
 {\tt G-Gap}-Hier-C & \ding{178} & \ding{184} & 2 & \ding{173} & 7 & 4 \\
 {\tt G-Gap}-Hier-S & \ding{176} & \ding{184} &1  & 1 & 1 & 5 \\
 {\tt G-Gap}-R-R0 & 2 & 7 & 2 & \ding{175} & \ding{186} & 4 \\
 {\tt G-Gap}-R-R2 & 3 & \ding{173} & 2 & \ding{173} & \ding{186}& 6 \\
 {\tt G-Gap}-R-R5 & \ding{176}   & \ding{175}   & 2 &  \ding{173} & \ding{175} & 6 \\
 {\tt G-Gap}-K-means-R & \ding{178} & \ding{184} & 4 & \ding{175} & \ding{177} & 5 \\
 {\tt G-Gap}-K-means-A & 4 & \ding{184} & 1 & \ding{173} & \ding{177} &  4 \\
 {\tt G-Gap}-K-means-C & \ding{176} & \ding{184} & 2 & 8 & \ding{177} & 5 \\
 {\tt G-Gap}-K-means-S & 3  & \ding{184} & 1  & 1 & 1 & 1 \\

 {\tt Gap-Ps}-Hier-A & 1 & \ding{175} &  1 & 6 & 3 & 1 \\
 {\tt Gap-Ps}-Hier-C & 1 or 2 & \ding{175} &  2 & 1 or 25& 7 & 15 \\
 {\tt Gap-Ps}-Hier-S & 1 & 1&  1 & 1 & 1 & 1 \\
 {\tt Gap-Ps}-K-means-R & \ding{187}  or \ding{178}  & \ding{175}  or 5 &  3& 8 & 9 &7 \\
 {\tt Gap-Ps}-K-means-A & \ding{178}  & \ding{184} &  1 & 8 & 7 & 9 \\
 {\tt Gap-Ps}-K-means-C & \ding{178} & \ding{175} &  2& 1 or 25  & 12 & 6 \\
 {\tt Gap-Ps}-K-means-S & 9 & \ding{184} &  1 & 1 & 7 & 8 \\

 {\tt Gap-Pc}-Hier-A & 1 & \ding{184} or \ding{175} &  1 & 1 & 1 or 2 or 3  & - \\
 {\tt Gap-Pc}-Hier-C & 1 &  \ding{175} &  1 & 1 &  3  & - \\
 {\tt Gap-Pc}-Hier-S & 1 & 1 &  1& 1  &  1  & - \\
 {\tt Gap-Pc}-K-means-R & 2 & 1 &  1& 1  &   \ding{175}  & -\\
 {\tt Gap-Pc}-K-means-A & 2 &  \ding{175}  &  1 & 1& 3 & - \\
 {\tt Gap-Pc}-K-means-C & 2 & 1 &  1 & 1  & \ding{175} & -\\
 {\tt Gap-Pc}-K-means-S & 3 & 1 &  1 & 1  & 1 & -\\

 {\tt Gap-Pr}-Hier-A & 3 & \ding{175} & 1 & 6 & 3 & 1 \\
 {\tt Gap-Pr}-Hier-C & \ding{178} & \ding{175} & 1 & 1 or 25 & 16 & 1 \\
 {\tt Gap-Pr}-Hier-S & 1 or \ding{187} & 1&  2 & 1 & 1 & 2 \\
 {\tt Gap-Pr}-K-means-R & \ding{187} & \ding{175} & 5 & 8 & 8 & 8 \\
 {\tt Gap-Pr}-K-means-A & 8 & \ding{175} & 1 & 8 & 13 & 4 \\
 {\tt Gap-Pr}-K-means-C & \ding{176} & 6 & 1 & 1 or 25 & 8 & 1 \\
 {\tt Gap-Pr}-K-means-S & \ding{178} & \ding{184} & 2 & 1 & 11 & 1 \\

 \hline
 {\bf Gold solution}  & \textbf{6} & {\bf 3} & {\bf 8} & {\bf 3} & {\bf 5} & {\bf 18}  \\
 \hline
\end{tabular}
\caption{A summary of the precision results for {\tt Gap} of Table~\ref{table:Gap} with the addition of its approximations.  For
{\tt Gap-Pc}, on PBM, the experiments were stopped due to their high
computational demand.}\label{table:Gapapp}
\end{center}
\end{footnotesize}
\end{table}

\begin{table}
\begin{footnotesize}
\begin{center}
\begin{tabular}{|l|cccc|}\hline
 & \multicolumn{4}{c|}{Timing}
\\\cline{2-5}
 & CNS Rat & Leukemia & NCI60 & Limphoma\\
 {\tt G-Gap}-Hier-A & $1.1 \times 10^{3}$ & $4.0 \times 10^{2}$ & $2.0 \times 10^{3}$ & $1.9 \times 10^{3}$ \\
 {\tt G-Gap}-Hier-C & $7.0 \times 10^{2}$ &$ 4.0 \times 10^{2}$ & $1.7 \times 10^{3}$ & $1.4 \times 10^{3}$\\
 {\tt G-Gap}-Hier-S & $2.6 \times 10^{3}$ & $6.0 \times 10^{2}$ & $3.2 \times 10^{3}$ & $3.8 \times 10^{3}$\\
 {\tt G-Gap}-R-R0 & $1.2 \times 10^{3}$ & $8.0 \times 10^{2}$ & $4.0 \times 10^{3}$ & $3.0 \times 10^{3}$\\
 {\tt G-Gap}-R-R2 & $1.3 \times 10^{3}$ & $8.0 \times 10^{2}$ & $5.2 \times 10^{3}$ & $3.2 \times 10^{3}$ \\
 {\tt G-Gap}-R-R5 & $1.2 \times 10^{3}$ & $8.0 \times 10^{2}$ & $4.5 \times 10^{3}$ &$3.2 \times 10^{3}$\\
 {\tt G-Gap}-K-means-R & $2.4 \times 10^{3}$ & $2.0 \times 10^{3}$ & $8.3 \times 10^{3}$ &  $8.4 \times 10^{3}$ \\
 {\tt G-Gap}-K-means-A & $2.3 \times 10^{3}$ & $1.3 \times 10^{3}$ & $5.3 \times 10^{3}$  & $5.8 \times 10^{3}$\\
 {\tt G-Gap}-K-means-C & $1.7 \times 10^{3}$ & $1.3 \times 10^{3}$ & $5.0 \times 10^{3}$ & $4.0 \times 10^{3}$\\
 {\tt G-Gap}-K-means-S & $2.6 \times 10^{3}$ & $1.6 \times 10^{3}$ & $7.3 \times 10^{3}$ & $7.4 \times 10^{3}$\\

 {\tt Gap-Ps}-Hier-A & $2.7 \times 10^{5}$ & $1.4 \times 10^{5}$ &  $6.1 \times 10^{5}$ & $6.4 \times 10^{5}$\\
 {\tt Gap-Ps}-Hier-C  & $2.3 \times 10^{5}$ & $1.1 \times 10^{4}$&  $3.4 \times 10^{5}$  & $3.2 \times 10^{5}$\\
 {\tt Gap-Ps}-Hier-S  & $6.1 \times 10^{5}$ & $1.9 \times 10^{5}$&  $1.1 \times 10^{6}$ &$1.4 \times 10^{6}$ \\
 {\tt Gap-Ps}-K-means-R  & $8.4  \times 10^{5}$ & $5.0 \times 10^{5}$ &  $1.1 \times 10^{6}$ &$1.0 \times 10^{6}$\\
 {\tt Gap-Ps}-K-means-A & $6.1  \times 10^{5}$ & $4.7 \times 10^{5}$&   $1.1 \times 10^{6}$ & $1.0 \times 10^{6}$\\
 {\tt Gap-Ps}-K-means-C & $6.0  \times 10^{5}$ &$6.1  \times 10^{5}$&   $8.8 \times 10^{5}$& $7.6 \times 10^{5}$\\
 {\tt Gap-Ps}-K-means-S & $9.1  \times 10^{5}$ &$6.5  \times 10^{5}$&   $2.1 \times 10^{6}$ & $1.8 \times 10^{6}$\\

 {\tt Gap-Pc}-Hier-A & $3.2 \times 10^{5}$ & $3.7 \times 10^{5}$&   $8.1 \times 10^{5}$ & $6.1 \times 10^{5}$\\
 {\tt Gap-Pc}-Hier-C  & $1.9 \times 10^{5}$ & $1.9 \times 10^{5}$ & $7.1 \times 10^{5}$ & $5.8 \times 10^{5}$ \\
 {\tt Gap-Pc}-Hier-S  & $7.7 \times 10^{5}$ & $3.1 \times 10^{5}$&  $1.4 \times 10^{6}$ & $1.3 \times 10^{6}$\\
 {\tt Gap-Pc}-K-means-R &$4.9 \times 10^{5}$ &$8.0 \times 10^{5}$ & $1.8 \times 10^{6}$& $1.8 \times 10^{6}$\\
 {\tt Gap-Pc}-K-means-A & $3.8 \times 10^{5}$ & $7.0 \times 10^{5}$&   $1.3 \times 10^{6}$ & $1.4 \times 10^{6}$\\
 {\tt Gap-Pc}-K-means-C  &$4.1 \times 10^{5}$ & $5.8 \times 10^{5}$&   $1.3 \times 10^{6}$ & $1.2 \times 10^{6}$ \\
 {\tt Gap-Pc}-K-means-S &$6.5 \times 10^{5}$ & $7.6 \times 10^{5}$&   $2.4 \times 10^{6}$& $2.0 \times 10^{6}$ \\

 {\tt Gap-Pr}-Hier-A  & $2.5 \times 10^{5}$ & $1.6 \times 10^{5}$&  $3.3 \times 10^{5}$   & $3.8 \times 10^{5}$  \\
 {\tt Gap-Pr}-Hier-C   & $1.3 \times 10^{5}$ & $1.4 \times 10^{5}$&  $3.2 \times 10^{5}$   & $3.7 \times 10^{5}$  \\
 {\tt Gap-Pr}-Hier-S  & $6.8 \times 10^{5}$ & $1.9 \times 10^{5}$&  $1.1 \times 10^{6}$   & $1.4 \times 10^{6}$  \\
 {\tt Gap-Pr}-K-means-R & $8.6 \times 10^{5}$ & $5.4 \times 10^{5}$&   $1.5 \times 10^{6}$& $9.4 \times 10^{5}$  \\
 {\tt Gap-Pr}-K-means-A  & $7.4 \times 10^{5}$ & $5.0 \times 10^{5}$&   $8.7 \times 10^{5}$& $1.0 \times 10^{6}$  \\
 {\tt Gap-Pr}-K-means-C  & $6.7  \times 10^{5}$ &$4.6  \times 10^{5}$&   $8.6 \times 10^{5}$& $1.0 \times 10^{6}$   \\
 {\tt Gap-Pr}-K-means-S  & $1.2  \times 10^{6}$ &$5.4  \times 10^{5}$&   $1.8 \times 10^{6}$& $2.3 \times 10^{6}$ \\

 \hline
\end{tabular}
\caption{A summary of the timing results  for {\tt Gap} of Table~\ref{table:Gap-time} with the addition of its approximations.  For
{\tt Gap-Pc}, on PBM, the experiments were stopped due to their high
computational demand.}\label{table:Gapapp-time}
\end{center}
\end{footnotesize}
\end{table}

\section{An Approximation of FOM}
Recalling from Chapter~\ref{chap:ValidationMeasuers} that both  {\tt WCSS} and {\tt FOM} use
the same criteria that is used in order to infer $k^*$.  Such an analogy between {\tt FOM} and {\tt
WCSS} immediately suggest to extend some of the knowledge available
about {\tt WCSS} to {\tt FOM}, as follows:

\begin{itemize}

\item [$\bullet$] The approximation of  {\tt FOM}  is based on exactly the same ideas and
schemes presented for the approximation of {\tt WCSS}. Indeed, {\tt
FOM}$(e,k)$ in equation (\ref{math:fome}) can be approximated in
exactly the same way as {\tt WCSS}$(k)$. Then, one uses equation
(\ref{math:foma}) to approximate {\tt FOM}. We denote those
approximations as {\tt FOM-R}.

\item [$\bullet$] The {\tt G-Gap} idea can be extended verbatim to {\tt FOM}
to make it automatic and to obtain {\tt G-FOM}.

\item [$\bullet$] The {\tt KL} technique can be extended to {\tt
FOM}, although the extension is subtle. Indeed, a verbatim extension
of it would yield poor results (experiments not shown). Rather,
consider formula (\ref{eqn:Diff}), with {\tt WCSS}$(k)$ substituted
by {\tt FOM}$(k)$. As $k$ increases towards $k^*$, $DIFF(k)$
increases to decrease sharply and then assume nearly constant values
as it moves away from $k^*$. Fig.~\ref{fig:fomchap7} provides a small example of this
behavior.  So, one can  take as $k^*$ the abscissa corresponding to
the maximum of $DIFF(k)$ in the interval $[3, k_{max}]$. This method is referred to as {\tt DIFF-FOM}.

\end{itemize}

\begin{figure}[ht] \centering \epsfig{file=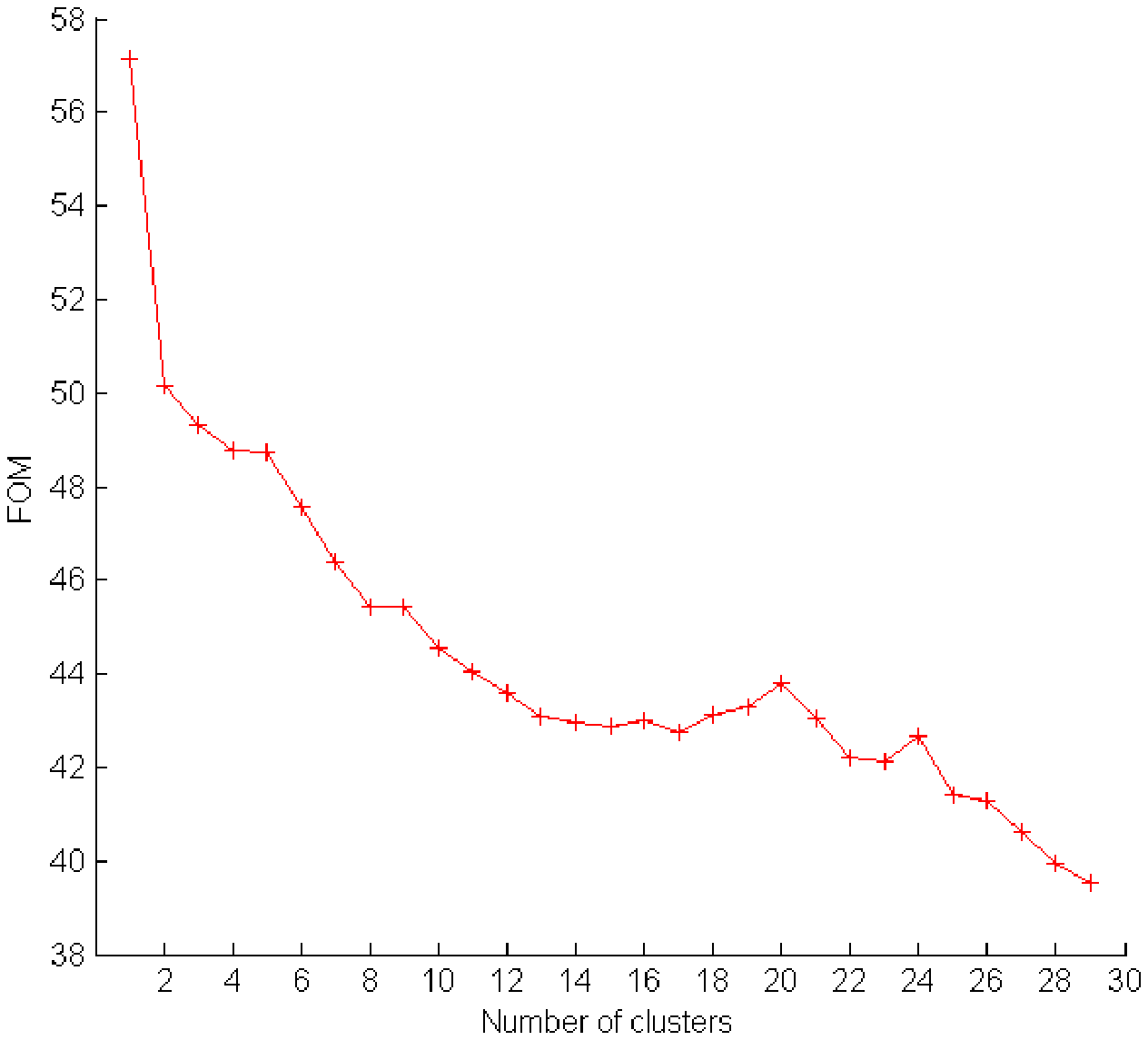,scale=0.8}
\caption{The {\tt FOM} curve computed on the Leukemia dataset with K-means-R. As for {\tt WCSS}, the \vir{knee} in the plot indicates the correct number of clusters in the dataset: $k^* = 3$.}\label{fig:fomchap7}
\end{figure}

For each algorithm, each of the {\tt FOM} approximations (denoted
{\tt FOM-R-R0}, {\tt FOM-R-R2}, {\tt FOM-R-R5}, respectively) and
each dataset in {\tt Benchmark 1}, the same methodology outlined for
{\tt WCSS} and its approximation has been followed. The relevant plots are in Figs. S135-S136 at the following
supplementary material web site~\cite{Bench} (Figures section). The values
resulting from the application of this methodology to the relevant
plots are reported in Table~\ref{table:FOMapp}, while the timing results for the relevant datasets are reported in Table~\ref{table:FOMapp-time} with the addition of the results for {\tt FOM} reported in Tables~\ref{table:FOM-Index} and~\ref{table:FOM-Index-time}, respectively.
Using the same experimental setting, {\tt G-FOM} and {\tt DIFF-FOM}, the
extensions of {\tt FOM} proposed here, are computed in order to predict $k^*$. The results are in Tables~\ref{table:Gfom}-\ref{table:Gfom-time}
and ~\ref{table:DiffFom}-\ref{table:DiffFom-time}, respectively. As those results show, {\tt G-FOM} does not
perform as well as the other two. Moreover, both {\tt FOM} and {\tt
DIFF-FOM} are algorithm-dependent and give no useful indication on
large datasets. As for the approximations of {\tt FOM}, i.e., {\tt
FOM-R-R0}, {\tt FOM-R-R2}, {\tt FOM-R-R5}, they compare very well
with the K-means algorithms in terms of precision and they are an
order of magnitude faster.

\begin{table}\begin{center}
\begin{footnotesize}
\begin{tabular}{|l|cccccc|}\hline
& \multicolumn{6}{c|}{Precision}
\\ \cline{2-7}
 & CNS Rat & Leukemia & NCI60 & Lymphoma & Yeast & PBM \\
Hier-A & \ding{178}&\ding{184} & \ding{178} & 6 &\ding{177} &  - \\
Hier-C & 10 & \ding{175} &  \ding{178} & 7 & \ding{186} & -  \\
Hier-S & 3 & 7 &  \ding{178} &  9  & - &  - \\

R-R0 & 10 & 5 & \ding{178} & \ding{175} & 7 &  -\\
R-R2 & 8 & 5 & \ding{189} &  5  & \ding{186} & - \\
R-R5 & \ding{187}&   \ding{184} & \ding{178} & 5 & \ding{186} & -\\

K-means-R & \ding{178} & \ding{184} & 6 & 9 & \ding{175} &  - \\
K-means-A & \ding{178} & \ding{184} & 6 & 6 & \ding{175}  &  - \\
K-means-C & \ding{178} &  8  &  \ding{189}&  \ding{175}  & \ding{175} &  -  \\
K-means-S &\ding{187 }& \ding{184} & \ding{189} & 8 &\ding{175} &  -  \\
\hline
{\bf Gold solution}  & {\bf 6} & {\bf 3} & {\bf 8} & {\bf 3} & {\bf 5} & {\bf 18} \\
 \hline
\end{tabular}

\end{footnotesize}\end{center}
\caption{
A summary of the precision results for {\tt FOM} of Table~\ref{table:FOM-Index} with the addition of its approximations. Cells with a dash indicate that  {\tt
FOM} did not give any useful indication.
}\label{table:FOMapp}
\end{table}

\begin{table}\begin{center}
\begin{footnotesize}
\begin{tabular}{|l|cccc|}\hline
& \multicolumn{4}{c|}{Timing}
\\ \cline{2-5}
 & CNS Rat & Leukemia & NCI60 & Lymphoma\\
Hier-A & $1.6  \times 10^{3}$ & $7.5  \times 10^{3}$ & $5.1 \times 10^{4}$&   $1.8 \times 10^{4}$\\
Hier-C & $1.6  \times 10^{3}$ & $7.7  \times 10^{3}$& $4.5 \times 10^{4}$& $1.8 \times 10^{4}$\\
Hier-S & $1.6  \times 10^{3}$ & $7.4  \times 10^{3}$& $4.9 \times 10^{5}$&  $1.7 \times 10^{4}$\\

R-R0 & $2.6 \times 10^{3}$ & $3.1 \times 10^{4}$ & $1.7 \times 10^{5}$ & $5.3 \times 10^{4}$ \\
R-R2 & $3.4 \times 10^{3}$ & $3.8 \times 10^{4}$ & $2.2 \times 10^{5}$ & $7.2 \times 10^{4}$ \\
R-R5 & $3.9 \times 10^{3}$ & $3.7 \times 10^{4}$ & $2.1 \times 10^{5}$ & $7.6 \times 10^{4}$ \\

K-means-R  & $2.9  \times 10^{4}$ & $1.9  \times 10^{5}$& $1.3 \times 10^{6}$  &  $6.7 \times 10^{5}$\\
K-means-A  & $2.2  \times 10^{4}$ & $9.3  \times 10^{4}$& $5.5 \times 10^{5}$&  $2.7 \times 10^{5}$\\
K-means-C  & $1.9 \times 10^{4}$ & $9.4 \times 10^{4}$& $5.5 \times 10^{5}$ & $2.6 \times 10^{5}$ \\
K-means-S & $2.9 \times 10^{4}$ & $1.0 \times 10^{5}$ & $7.1 \times 10^{5}$ & $3.6 \times 10^{5}$ \\
\hline
\end{tabular}

\end{footnotesize}\end{center}
\caption{
A summary of the timing results for {\tt FOM} of Table~\ref{table:FOM-Index-time} with the addition of its approximations.
}\label{table:FOMapp-time}
\end{table}

\begin{table}
\begin{footnotesize}
\begin{center}
\begin{tabular}{|l|cccccc|}\hline
& \multicolumn{6}{c|}{Precision}
\\ \cline{2-7}
 & CNS Rat & Leukemia & NCI60 & Lymphoma & Yeast & PBM \\
Hier-A & 3  & \ding{184} & \ding{178} & \ding{173} & 8 & 2 \\
Hier-C & 10 & \ding{175} &  2  & \ding{175} & 8 & 2 \\
Hier-S & \ding{178} & \ding{173} & 2 & \ding{173} & 2 & 2\\
R-R0 & 2 & 7 & 2 & 5 & 7 & 4 \\
R-R2 & \ding{178}&  5  & 2 & 5  & 8 & 4 \\
R-R5 & 4 & \ding{175} & 2 & 6 & \ding{177} & 4 \\
K-means-R & \ding{178}& 5 & 6 & 8 &\ding{177} &  7 \\
K-means-A & 2  & \ding{184} & \ding{178}&\ding{173}  &\ding{177}  &  6 \\
K-means-C & 2 & \ding{175} & 2 & \ding{175} & 7 & 6 \\
K-means-S & 3  & 5 & 2 &  \ding{173} &\ding{177} &  8 \\

\hline
{\bf Gold solution}  & {\bf 6} & {\bf 3} & {\bf 8} & {\bf 3} & {\bf 5} & {\bf 18}\\
\hline
\end{tabular}
\end{center}
\end{footnotesize}
\caption{
A summary of the results for {\tt G-FOM} on all algorithms
and on all datasets. The columns under the label precision indicate
the number of clusters predicted by {\tt G-FOM}, while the remaining
four indicate the timing in milliseconds for the execution of the
corresponding experiment. Cells with a dash indicate that {\tt
G-FOM} did not give any useful indication.
}\label{table:Gfom}
\end{table}

\begin{table}
\begin{footnotesize}
\begin{center}
\begin{tabular}{|l|cccc|}\hline
&  \multicolumn{4}{c|}{Timing}
\\ \cline{2-5}
 &  CNS Rat & Leukemia & NCI60 & Lymphoma\\

Hier-A &  $1.6  \times 10^{3}$ & $7.5  \times 10^{3}$ & $5.1 \times 10^{4}$&   $1.8 \times 10^{4}$\\
Hier-C &  $1.6  \times 10^{3}$ & $7.7  \times 10^{3}$& $4.5 \times 10^{4}$& $1.8 \times 10^{4}$\\
Hier-S & $1.6  \times 10^{3}$ & $7.4  \times 10^{3}$& $4.9 \times 10^{5}$&  $1.7 \times 10^{4}$\\

R-R0 & $2.6 \times 10^{3}$ & $3.1 \times 10^{4}$ & $1.7 \times 10^{5}$ & $5.3 \times 10^{4}$\\
R-R2 & $3.4 \times 10^{3}$ & $3.8 \times 10^{4}$ & $2.2 \times 10^{5}$ & $7.2 \times 10^{4}$\\
R-R5 & $3.9 \times 10^{3}$ & $3.7 \times 10^{4}$ & $2.1 \times 10^{5}$ & $7.6 \times 10^{4}$\\
K-means-R &  $2.9  \times 10^{4}$ & $1.9  \times 10^{5}$& $1.3 \times 10^{6}$  &  $6.7 \times 10^{5}$\\
K-means-A&  $2.2  \times 10^{4}$ & $9.3  \times 10^{4}$& $5.5 \times 10^{5}$&   $2.7 \times 10^{5}$\\
K-means-C & $1.9 \times 10^{4}$ & $9.4 \times 10^{4}$& $5.5 \times 10^{5}$&   $2.6 \times 10^{5}$\\
K-means-S & $2.9 \times 10^{4}$ & $1.0 \times 10^{5}$ & $7.1 \times 10^{5}$&   $3.6 \times 10^{5}$\\
\hline

\end{tabular}
\end{center}
\end{footnotesize}
\caption{
A summary of the precision results for {\tt G-FOM} on all algorithms
and on all datasets. Cells with a dash indicate that {\tt
G-FOM} did not give any useful indication.
}\label{table:Gfom-time}
\end{table}

\begin{table}
\begin{footnotesize}
\begin{center}
\begin{tabular}{|l|cccccc|}\hline
& \multicolumn{6}{c|}{Precision}
\\ \cline{2-7}
 & CNS Rat & Leukemia & NCI60 & Lymphoma & Yeast & PBM \\
K-means-R & 4 & \ding{184}&4 & \ding{184} &3 &  4 \\
K-means-A & \ding{178} & \ding{184} &3 & 6 & 3  &  8 \\
K-means-C & \ding{198} & \ding{184} & \ding{178} & \ding{175}  & 3 & 5 \\
K-means-S & \ding{178} & \ding{184} & 12 & 8 & 3 &  10 \\
R-R0 & 10 & \ding{175} & 17 & \ding{175} & 3 & 3 \\
R-R5 & 4 & \ding{184} & 11 & \ding{184} & 3 & 4 \\
R-R2 & \ding{178} & \ding{184} & 17 & \ding{184} & 3 &  7 \\

Hier-A & \ding{178} & \ding{184} & 3 & 6 & 3 & 25 \\
Hier-C & 9 & \ding{184} & \ding{178} & 7 & 3 & 7  \\
Hier-S & 20 & 7 & 22 &  9 & 7 & 20 \\
\hline
{\bf Gold solution}  & {\bf 6} & {\bf 3} & {\bf 8} & {\bf 3} & {\bf 5} & {\bf 18}\\
\hline
\end{tabular}
\end{center}
\end{footnotesize}
\caption{
A summary of the precision results for {\tt DIFF-FOM} on all
algorithms and on all datasets.
Cells with a dash indicate that {\tt DIFF-FOM} did not give any
useful indication.
}\label{table:DiffFom}
\end{table}

\begin{table}
\begin{footnotesize}
\begin{center}
\begin{tabular}{|l|cccc|}\hline
&  \multicolumn{4}{c|}{Timing}
\\ \cline{2-5}
 &  CNS Rat & Leukemia & NCI60 & Lymphoma\\

Hier-A & $1.6  \times 10^{3}$ & $7.5  \times 10^{3}$ & $5.1 \times 10^{4}$ &   $1.8 \times 10^{4}$\\
Hier-C & $1.6  \times 10^{3}$ & $7.7  \times 10^{3}$& $4.5 \times 10^{4}$ & $1.8 \times 10^{4}$\\
Hier-S & $1.6  \times 10^{3}$ & $7.4  \times 10^{3}$& $4.9 \times 10^{5}$ &  $1.7 \times 10^{4}$\\

R-R0 & $2.6 \times 10^{3}$ & $3.1 \times 10^{4}$ & $1.7 \times 10^{5}$ & $5.3 \times 10^{4}$\\
R-R2 & $3.4 \times 10^{3}$ & $3.8 \times 10^{4}$& $2.2 \times 10^{5}$&   $7.2 \times 10^{4}$\\
R-R5 & $3.9 \times 10^{3}$ & $3.7 \times 10^{4}$&$2.1 \times 10^{5}$&   $7.6 \times 10^{4}$\\

K-means-R & $2.9 \times 10^{4}$ & $1.9  \times 10^{5}$ & $1.3 \times 10^{6}$ & $6.7 \times 10^{5}$\\
K-means-A & $2.2 \times 10^{4}$ & $9.3  \times 10^{4}$ & $5.5 \times 10^{5}$ & $2.7 \times 10^{5}$\\
K-means-C & $1.9 \times 10^{4}$ & $9.4 \times 10^{4}$ & $5.5 \times 10^{5}$ & $2.6 \times 10^{5}$ \\
K-means-S & $2.9 \times 10^{4}$ & $1.0 \times 10^{5}$ & $7.1 \times 10^{5}$&   $3.6 \times 10^{5}$\\
\hline
\end{tabular}
\end{center}
\end{footnotesize}
\caption{
A summary of the timing results for {\tt DIFF-FOM} on all
algorithms and on all datasets. Cells with a dash indicate that  {\tt DIFF-FOM} did not give any
useful indication.
}\label{table:DiffFom-time}
\end{table}

\section{An Exhaustive Study of the Consensus Parameters}

It is helpful for the discussion to highlight, here, some key facts about
{\tt Consensus}, summarizing the detailed description of the procedure presented in Chapter~\ref{chap:Stability}.
For a given number of clusters, {\tt Consensus} computes a certain number of clustering solutions (resampling step), each from a sample of the original data (subsampling). The performance of {\tt Consensus} depends
on two parameters: the number of resampling  steps $H$ and the percentage of subsampling $p$, where  $p$ states how large the sample must be. From each clustering solution, a corresponding
connectivity matrix is computed: each entry in that matrix  indicates whether a  pair of elements is  in the same cluster or not. For the given number of clusters, the consensus matrix is a normalized sum of the corresponding $H$ connectivity matrices. Intuitively, the consensus matrix indicates the level of agreement of clustering solutions that have been obtained via independent sampling of the dataset.

Monti et al., in their seminal paper, set $H=500$ and $p=80\%$,
without any experimental or theoretical justification.  For this reason and based also on an open problem mentioned in Chapter~\ref{chap:6} and in~\cite{giancarlo08}, here several experiments with different parameter settings of $H$ and $p$ are performed,
in order to find the \vir{best} precision-time  trade-off, when {\tt Consensus} is regarded both as an internal validation measure and as a procedure that computes a similarity/distance matrix.
In particular, using the hierarchical algorithms and K-means, experiments with $H=500, 250, 100$ and $p= 80\%, 66\%$ have been performed, respectively,  reporting the precision values and times. The choice of the value of $p$ is justified by the results reported in~\cite{BenHur02,CLEST}.
Intuitively, a value of $p$ smaller then 66\% would fail to capture the entire cluster structure present in the data.
The results in this section will show, it is worthy to anticipate, that a simple reduction in terms of $H$ and $p$ is not enough to grant a good precision-time trade-off.
Such a finding, together with the state of the art outlined in Chapter~\ref{chap:6}, motivates a strong interest in the design of alternative methods, such as fast heuristics that are discussed in depth in Section~\ref{sec:FC}.
As for datasets in {\tt Benchmark 2}, only the experiments with $H=250$ and $p=80\%$ are performed, reporting precision results for all  and timing only for the microarray data, since the timing results for the artificial datasets are redundant. The choice for this parameter setting for {\tt Consensus} is justified when  the results of the experiments on the datasets in {\tt Benchmark 1} are discussed. Moreover, the study  of a proper parameter setting for {\tt Consensus} is limited only to the {\tt Benchmark 1} datasets for pragmatic reasons: that choice allows to complete the (rather high) required number of experiments in a reasonable amount of time.

Due to its high computational demand (see Chapter~\ref{chap:5}), experiments only with $H=250$ and $p=80\%$ for NMF have been performed. They are reported in the relevant Table together with the results of the other algorithms, but they are discussed separately. The choice for this parameter setting for {\tt Consensus} when used in conjunction with NMF is justified when the results of the experiments are discussed.

\subsection{Consensus as an Internal Validation Measure}\label{subsec:ConsensusSpeed}
The experiments summarized here refer to the {\tt Benchmark 1} datasets.
Separate tables for each experimental setup are reported: they are Tables~\ref{table:H500p80}-\ref{table:H100p66-time}.
For each dataset and each  clustering algorithm, {\tt Consensus} for a number of cluster values in the range $[2,30]$ is computed, while, for Leukemia, the range  $[2,25]$ is used when $p=66\%$, due to the small size of the dataset.
The prediction value, $k^*$, is based on the plot of the $\Delta(k)$ curve (defined in Chapter~\ref{chap:Stability}) as indicated in Chapter~\ref{chap:6} and in~\cite{giancarlo08}. The corresponding plots are available at the following supplementary material web site~\cite{SpeedUpWeb}, in the Figures section, as Figs. S1-S10 for $p=80\%$ and $H=500$, and Figs. S11-S24 for $p=80\%$ and $H=250$, Figs. S25-S34 for $p=80\%$ and $H=100$. For $p=66\%$ the relevant figures are Figs. S35-S44, S45-S54 and S55-S64 for $H=500$, $H=250$ and $H=100$, respectively.


For $p=80\%$, the precision results reported in Tables~\ref{table:H500p80}-\ref{table:H100p80-time}  show
there is very little  difference between the results obtained for
$H=500$ and $H=250$. That is in contrast with the results for
$H=100$,  where many prediction values are very far from the gold
solution for the corresponding dataset, e.g., the  Lymphoma dataset.
Such a finding seems to indicate that, in order to find a consensus matrix which captures well the inherent structure of the dataset, one needs a sensible number of connectivity matrices.
The results for a subsampling value of $p=66\%$ confirms that the number of connectivity matrices one needs to compute is more relevant than the percentage of the data matrix actually used to compute them. Indeed,  although it is obvious that a reduction in the number of resampling
steps results in a saving in terms of execution time, it is less
obvious that for subsampling values $p=66\%$ and $p=80\%$, there is
no  substantial difference in the results, both in terms of
precision and of time.

In regard to NMF, only the parameter setting $H=250$ and $p=80\%$ for this experiments is used, since it seems to be the most promising (as determined by the use of  the other algorithms).
The results are reported in Table~\ref{table:H250p80}.  Even so, the inefficiencies of {\tt Consensus} compound with those of NMF; that is, the relatively large number of connectivity matrices  needed by {\tt Consensus} and  the well-known slow convergence of NMF for the computation of a clustering solution, since connectivity matrices are obtained from clustering solutions. The end-result is a slow-down of one order of magnitude with respect to {\tt Consensus} used in conjunction with other clustering algorithms. As a consequence, NMF and {\tt Consensus} can be used together on a conventional PC only for relatively small datasets. In fact, the experiments for Yeast and PBM, the two largest datasets with which one has experimented, were stopped after four days.

It is also worth to point out that, although the parameter setting $H=250$ and $p=80\%$ grants a faster execution of {\tt Consensus} with respect to the original setting by Monti et al., the experiments on the PBM dataset were stopped after four days on all algorithms. That is, the largest of the datasets used here is still \vir{out of reach} of {\tt Consensus} even with a tuning of the parameters aimed at reducing its computational demand.

In conclusion, this experiments show that an effective parameter setting for {\tt Consensus} is $H=250$ and $p=80\%$: in fact Table~\ref{table:H250p80} displays the best trade-off between
precision and time. Moreover,  the experiments also show that inefficiencies of the {\tt Consensus} methodology are due to the large number of  connectivity matrices that are required to compute a reliable consensus matrix, rather than to the size of the sample taken from the data matrix that is then used to compute them. This is particularly important since a slow clustering algorithm, e.g.,  NMF, used in conjunction with the methodology makes it worthless on conventional computers.

\begin{table}
\begin{center}
\begin{footnotesize}
\begin{tabular}{|l|cccccc|}\hline
& \multicolumn{6}{c|}{Precision}
\\ \cline{2-7}
& CNS Rat & Leukemia & NCI60 & Lymphoma & Yeast  & PBM \\
Hier-A & \whitestone[7] & \blackstone[3] & \blackstone[8] & \blackstone[3] &  \blackstone[5] & - \\
Hier-C & \blackstone[6]  & \whitestone[4] & \blackstone[8] & 5 & \whitestone[6] & -  \\
Hier-S & 2 & 8 & 10  & \blackstone[3]  & 10 & - \\
K-means-R & \blackstone[6] & \whitestone[4] & \whitestone[7] & \blackstone[3] & \whitestone[6] & - \\
K-means-A & \whitestone[7]  & \blackstone[3] & \blackstone[8] & \blackstone[3] & \whitestone[6]  & - \\
K-means-C & \blackstone[6]  & \blackstone[3] & \blackstone[8] & \whitestone[4] & \whitestone[6] & - \\
K-means-S & \whitestone[7] & \whitestone[4] & 10 & \whitestone[2] & \whitestone[6] & -  \\
\hline
{\bf Gold solution}  & {\bf 6} & {\bf 3} & {\bf 8} & {\bf 3} & {\bf 5} & {\bf 18} \\
\hline
\end{tabular}
\end{footnotesize}
\end{center}
\caption{A summary of the precision results for {\tt Consensus} with $H=500$  and
$p=80\%$, on all algorithms, except NMF,  and  for the {\tt Benchmark 1} datasets.  Cells with a dash indicate that the experiments were terminated  due to their high
computational demand.}\label{table:H500p80}
\end{table}

\begin{table}
\begin{footnotesize}
\begin{center}
\begin{tabular}{|l|cccccc|}\hline
&  \multicolumn{6}{c|}{Timing}
\\ \cline{2-7}
&  CNS Rat & Leukemia & NCI60 & Lymphoma &  Yeast & PBM\\
Hier-A & $9.2 \times 10^{5}$ & $7.9 \times 10^{5}$ & $2.0 \times 10^{6}$ & $1.9 \times 10^{6}$ & $7.7 \times 10^{7}$ &-\\
Hier-C & $8.7 \times 10^{5}$ & $6.9 \times 10^{5}$ & $2.0 \times 10^{6}$ & $2.0 \times 10^{6}$ & $8.2 \times 10^{7}$ &-\\
Hier-S & $9.4 \times 10^{5}$ & $8.0 \times 10^{5}$& $2.0 \times 10^{6}$ & $1.7 \times 10^{6}$ &  $8.2 \times 10^{7}$ & -\\
K-means-R & $1.0 \times 10^{6}$ & $1.3 \times 10^{6}$ & $3.4 \times 10^{6}$ & $3.0 \times 10^{6}$ & $5.5 \times 10^{7}$ & - \\
K-means-A & $1.3 \times 10^{6}$ & $1.6 \times 10^{6}$ & $3.0 \times 10^{6}$ & $2.6 \times 10^{6}$ & $1.1 \times 10^{8}$ & -\\
K-means-C & $1.3 \times 10^{6}$ &  $1.8 \times 10^{6}$ & $2.9 \times 10^{6}$ & $2.6 \times 10^{6}$ & $9.3 \times 10^{7}$ & -\\
K-means-S & $1.5 \times 10^{6}$ & $1.8 \times 10^{6}$ &  $3.2 \times 10^{6}$ &  $2.8 \times 10^{6}$ & $9.7 \times 10^{7}$ &-\\
\hline
\end{tabular}
\end{center}
\end{footnotesize}
\caption{A summary of the timing results for {\tt Consensus} with $H=500$  and
$p=80\%$, on all algorithms, except NMF,  and  for the {\tt Benchmark 1} datasets. Cells with a dash indicate that the experiments were terminated  due to their high
computational demand.
}\label{table:H500p80-time}
\end{table}

\begin{table}
\begin{footnotesize}
\begin{center}
\begin{tabular}{|l|cccccc|}\hline
& \multicolumn{6}{c|}{Precision}
\\ \cline{2-7}
& CNS Rat & Leukemia & NCI60 & Lymphoma & Yeast  & PBM \\

 Hier-A & \whitestone[7] &   \blackstone[3] & \blackstone[8] & \blackstone[3] &  \blackstone[5] & - \\
 Hier-C & \blackstone[6]&    \whitestone[4] & \blackstone[8] & 5 & \whitestone[6] & - \\
 Hier-S & 2 & \blackstone[3] &  10 &  \whitestone[2] & 10 & - \\
 K-means-R &  \blackstone[6]& \whitestone[4] & \whitestone[7] & \whitestone[4]& \whitestone[6] & -\\
 K-means-A & \whitestone[7] &  \blackstone[3] &  \blackstone[8] &  \blackstone[3] & \whitestone[6] & - \\
 K-means-C &  \blackstone[6]&  \blackstone[3] & \blackstone[8]  &  \whitestone[4] & \whitestone[6]& - \\
 K-means-S & \whitestone[7] & 5 & \whitestone[9]  & \whitestone[2] & \whitestone[6] & - \\

 NMF-R & \blackstone[6] & \whitestone[4]  & \whitestone[7] & \whitestone[4] & - & - \\
 NMF-A & \whitestone[7] & \blackstone[3] & 2 & \blackstone[3] & - & - \\
 NMF-C & \whitestone[5] & \whitestone[4] & \whitestone[7] & \whitestone[4] & - & - \\
 NMF-S & 2 & 8 & \whitestone[9] &  \whitestone[2] & - & - \\

\hline {\bf Gold solution}  & {\bf 6} & {\bf 3} & {\bf 8} & {\bf 3} & {\bf 5} & {\bf 18} \\
\hline
\end{tabular}
\end{center}
\end{footnotesize}
\caption{
A summary of the precision results for {\tt Consensus} with $H=250$  and
$p=80\%$, on all algorithms and for the {\tt Benchmark 1} datasets.
Cells with a dash indicate that the experiments were terminated  due to their high
computational demand.
}\label{table:H250p80}
\end{table}

\begin{table}
\begin{footnotesize}
\begin{center}
\begin{tabular}{|l|cccccc|}\hline
&  \multicolumn{6}{c|}{Timing}
\\ \cline{2-7}
&  CNS Rat & Leukemia & NCI60 & Lymphoma &  Yeast & PBM\\
Hier-A & $8.9 \times 10^{5}$ & $6.0 \times 10^{5}$ & $1.4 \times 10^{6}$ & $1.3 \times 10^{6}$ & $5.0 \times 10^{7}$ & -\\
Hier-C & $8.1 \times 10^{5}$ & $6.0 \times 10^{5}$ & $1.3 \times 10^{6}$ & $1.3 \times 10^{6}$ & $4.8 \times 10^{7}$ & -\\
Hier-S & $4.3 \times 10^{5}$ & $6.2 \times 10^{5}$ & $1.0 \times 10^{5}$ & $1.2 \times 10^{6}$ & $4.8 \times 10^{7}$ & -\\

K-means-R & $5.6 \times 10^{5}$ & $3.7 \times 10^{5}$ & $1.2 \times 10^{6}$ & $1.1 \times 10^{6}$ & $2.7 \times 10^{7}$ & -\\
K-means-A & $1.0 \times 10^{6}$ & $6.4 \times 10^{5}$ & $1.8 \times 10^{6}$ & $1.7 \times 10^{6}$ & $5.6 \times 10^{7}$ & -\\
K-means-C & $9.8 \times 10^{5}$ & $6.5 \times 10^{5}$ & $1.7 \times 10^{6}$ & $1.3 \times 10^{6}$ & $5.3 \times 10^{7}$ & -\\
K-means-S & $1.2 \times 10^{6}$ & $4.7 \times 10^{5}$ & $1.2 \times 10^{6}$ & $1.2 \times 10^{6}$ & $5.7 \times 10^{7}$ & -\\

 NMF-R & $1.1 \times 10^{8}$ & $1.3 \times 10^{7}$ & $6.4 \times 10^{7}$ & $7.7 \times 10^{7}$ & - & -\\
 NMF-A & $3.0 \times 10^{7}$ & $3.0 \times 10^{7}$ & $1.3 \times 10^{7}$ & $1.6 \times 10^{7}$ & - & -\\
 NMF-C & $3.0 \times 10^{7}$ & $4.4 \times 10^{6}$ & $1.3 \times 10^{7}$ & $1.7 \times 10^{7}$ & - & -\\
 NMF-S & $3.6 \times 10^{7}$ & $4.7 \times 10^{6}$ & $1.3 \times 10^{7}$ & $1.6 \times 10^{7}$ & - & -\\
\hline
\end{tabular}
\end{center}
\end{footnotesize}
\caption{
A summary of the timing results for {\tt Consensus} with $H=250$  and
$p=80\%$, on all algorithms and for the {\tt Benchmark 1} datasets.
Cells with a dash indicate that the experiments were terminated  due to their high
computational demand.
}\label{table:H250p80-time}
\end{table}

\begin{table}
\begin{footnotesize}
\begin{center}
\begin{tabular}{|l|cccccc|}\hline
& \multicolumn{6}{c|}{Precision}
\\ \cline{2-7}
  & CNS Rat & Leukemia & NCI60 & Lymphoma & Yeast  & PBM \\
 Hier-A & \blackstone[6]& \blackstone[3] & \whitestone[7] & \blackstone[3] & \whitestone[6] & - \\
 Hier-C & \whitestone[7] & \whitestone[4] & \blackstone[8] & 6 & \blackstone[5] & - \\
 Hier-S & 2 & 9 & \whitestone[9] & 10 & 2 & - \\
 K-means-R & \blackstone[6]& \blackstone[3] - \whitestone[4] & \whitestone[7] & 6 & \whitestone[6] & - \\
 K-means-A & \whitestone[7] & \blackstone[3] & \whitestone[7] & 6 & \whitestone[6] & - \\
 K-means-C & \whitestone[5] & \whitestone[4] & \blackstone[8] & 6 & \whitestone[6] & -\\
 K-means-S &  8  & 8 & \whitestone[9] & 8 & \whitestone[6] & - \\

\hline {\bf Gold solution}  & {\bf 6} & {\bf 3} & {\bf 8} & {\bf 3} & {\bf 5} & {\bf 18} \\
\hline
\end{tabular}
\end{center}
\end{footnotesize}
\caption{
A summary of the precision results for {\tt Consensus} with $H=100$  and
$p=80\%$, on all algorithms, except NMF,  and for the {\tt Benchmark 1} datasets.  Cells with a dash indicate that the experiments were terminated  due to their high
computational demand.
}\label{table:H100p80}
\end{table}

\begin{table}
\begin{footnotesize}
\begin{center}
\begin{tabular}{|l|cccccc|}\hline
&  \multicolumn{6}{c|}{Timing}
\\ \cline{2-7}
&  CNS Rat & Leukemia & NCI60 & Lymphoma &  Yeast & PBM\\
 Hier-A & $3.5 \times 10^{5}$ & $2.6 \times 10^{5}$ & $4.2 \times 10^{5}$ & $3.8 \times 10^{5}$ & $2.9 \times 10^{7}$ & -\\
 Hier-C & $3.3 \times 10^{5}$ & $2.7 \times 10^{5}$ & $3.9 \times 10^{5}$ & $3.9 \times 10^{5}$ & $2.9 \times 10^{7}$ & -\\
 Hier-S & $3.5 \times 10^{5}$ & $2.8 \times 10^{5}$ & $3.9 \times 10^{5}$ & $4.2 \times 10^{5}$ & $2.8 \times 10^{7}$ & -\\

K-means-R & $3.6 \times 10^{5}$ & $3.4 \times 10^{5}$ & $7.4 \times 10^{5}$ & $6.2 \times 10^{5}$ & $2.3 \times 10^{7}$ & -\\
K-means-A & $4.2 \times 10^{5}$ & $1.9 \times 10^{5}$ & $7.2 \times 10^{5}$ & $5.2 \times 10^{5}$ & $7.5 \times 10^{6}$ & -\\
K-means-C & $4.4 \times 10^{5}$ & $2.6 \times 10^{5}$ & $6.7 \times 10^{5}$ & $5.5 \times 10^{5}$ & $3.3 \times 10^{7}$ & -\\
K-means-S & $4.8 \times 10^{5}$ & $2.8 \times 10^{5}$ & $6.2 \times 10^{5}$ & $5.8 \times 10^{5}$ & $3.7 \times 10^{7}$ & -\\

\hline
\end{tabular}
\end{center}
\end{footnotesize}
\caption{
A summary of the timing results for {\tt Consensus} with $H=100$  and
$p=80\%$, on all algorithms, except NMF,  and for the {\tt Benchmark 1} datasets.  Cells with a dash indicate that the experiments were terminated  due to their high
computational demand.
}\label{table:H100p80-time}
\end{table}

\begin{table}
\begin{footnotesize}
\begin{center}
\begin{tabular}{|l|cccccc|}\hline
& \multicolumn{6}{c|}{Precision}
\\ \cline{2-7}
 & CNS Rat & Leukemia & NCI60 & Lymphoma & Yeast  & PBM \\
 Hier-A & \whitestone[7] & \blackstone[3] & \blackstone[8] & \blackstone[3] & \whitestone[6] & -\\
 Hier-C & \whitestone[7] & \blackstone[3] & \blackstone[8] &  5  & \blackstone[5] & - \\
 Hier-S & 2 & 8 & \whitestone[9] & \whitestone[2] & 10 & -  \\
 K-means-R & \whitestone[7] & \whitestone[4] & \whitestone[7] & \blackstone[3] & \whitestone[6] & - \\
 K-means-A & \whitestone[7] & \blackstone[3] & \blackstone[8] & \blackstone[3] & \whitestone[6] & - \\
 K-means-C & \blackstone[6]& \blackstone[3] & \blackstone[8] & \whitestone[4] & \whitestone[6] & - \\
 K-means-S & \whitestone[7] &  5  & \whitestone[9]  & \whitestone[2] &  7 & - \\
\hline {\bf Gold solution}  & {\bf 6} & {\bf 3} & {\bf 8} & {\bf 3} & {\bf 5} & {\bf 18} \\
\hline
\end{tabular}
\end{center}
\end{footnotesize}
\caption{
A summary of the precision results for {\tt Consensus} with $H=500$  and
$p=66\%$, on all algorithms, except NMF,  and for the {\tt Benchmark 1} datasets.  Cells with a dash indicate that the experiments were terminated  due to their high
computational demand.
}\label{table:H500p66}
\end{table}

\begin{table}
\begin{footnotesize}
\begin{center}
\begin{tabular}{|l|cccccc|}\hline
&  \multicolumn{6}{c|}{Timing}
\\ \cline{2-7}
&  CNS Rat & Leukemia & NCI60 & Lymphoma &  Yeast & PBM\\
 Hier-A & $1.5 \times 10^{6}$ & - & $2.4 \times 10^{6}$ & $2.0 \times 10^{6}$ & $5.9 \times 10^{7}$ & - \\
 Hier-C & $1.5 \times 10^{6}$ & - & $2.3 \times 10^{6}$ & $1.6 \times 10^{6}$ & $5.9 \times 10^{7}$ & -  \\
 Hier-S & $1.6 \times 10^{6}$ & - & $1.5 \times 10^{6}$ & $1.6 \times 10^{6}$ & $5.8 \times 10^{7}$ & -  \\

 K-means-R & $1.5 \times 10^{6}$ & - & $3.4 \times 10^{6}$ & $2.7 \times 10^{6}$ & $4.7 \times 10^{7}$ & -  \\
 K-means-A & $1.8 \times 10^{6}$ & - & $3.4 \times 10^{6}$ & $2.0 \times 10^{6}$ & $7.9 \times 10^{7}$ & -  \\
 K-means-C & $1.0 \times 10^{6}$ & - & $2.4 \times 10^{6}$ & $2.0 \times 10^{6}$ & $7.5 \times 10^{7}$ & -  \\
 K-means-S & $1.3 \times 10^{6}$ & - & $2.5 \times 10^{6}$ & $2.8 \times 10^{6}$ & $4.8 \times 10^{7}$ & -  \\

\hline
\end{tabular}
\end{center}
\end{footnotesize}
\caption{
A summary of the timing results for {\tt Consensus} with $H=500$  and
$p=66\%$, on all algorithms, except NMF,  and for the {\tt Benchmark 1} datasets.  For the Leukemia dataset, the timing experiments are not reported because
incomparable with those of the remaining  datasets. Cells with a dash indicate that the experiments were terminated  due to their high
computational demand.
}\label{table:H500p66-time}
\end{table}

\begin{table}
\begin{footnotesize}
\begin{center}
\begin{tabular}{|l|cccccc|}\hline
& \multicolumn{6}{c|}{Precision}
\\ \cline{2-7}
 & CNS Rat & Leukemia & NCI60 & Lymphoma & Yeast  & PBM \\
 Hier-A & \whitestone[7] & \blackstone[3] & \blackstone[8] & \blackstone[3] & \blackstone[5] & - \\
 Hier-C & \whitestone[7] & \blackstone[3] & \blackstone[8] & \blackstone[3] & \whitestone[6] & - \\
 Hier-S & 2 & 8 & \whitestone[9] & \whitestone[2] & 10 & - \\
 K-means-R & \whitestone[7] & \whitestone[4] & \blackstone[8] & 6 & \whitestone[6] & - \\
 K-means-A & \whitestone[7] & \blackstone[3] & \blackstone[8] &  5  &  \whitestone[6] & - \\
 K-means-C & \blackstone[6]& \blackstone[3] & \whitestone[9] &  5  &  \whitestone[6] & - \\
 K-means-S & \whitestone[7] & 8 &  10  & \whitestone[2] & \whitestone[6]  & - \\
\hline {\bf Gold solution}  & {\bf 6} & {\bf 3} & {\bf 8} & {\bf 3} & {\bf 5} & {\bf 18} \\
\hline
\end{tabular}
\end{center}
\end{footnotesize}
\caption{
A summary of the precision results for {\tt Consensus} with $H=250$  and
$p=66\%$, on all algorithms, except NMF,  and for the {\tt Benchmark 1} datasets.  Cells with a dash indicate that the experiments were terminated  due to their high
computational demand.
}\label{table:H250p66}
\end{table}

\begin{table}
\begin{footnotesize}
\begin{center}
\begin{tabular}{|l|cccccc|}\hline
&  \multicolumn{6}{c|}{Timing}
\\ \cline{2-7}
 &  CNS Rat & Leukemia & NCI60 & Lymphoma &  Yeast & PBM\\
  Hier-A & $3.8 \times 10^{5}$ & - & $8.0 \times 10^{5}$ & $1.1 \times 10^{6}$ & $3.7 \times 10^{7}$ & - \\
 Hier-C & $3.8 \times 10^{5}$ & - & $8.2 \times 10^{5}$ & $1.1 \times 10^{6}$ & $3.7 \times 10^{7}$ & -\\
 Hier-S & $7.9 \times 10^{5}$ & - & $8.2 \times 10^{5}$ & $7.0 \times 10^{5}$ & $3.7 \times 10^{7}$ & -\\

 K-means-R & $5.2 \times 10^{5}$ & - & $1.4 \times 10^{6}$ & $1.4 \times 10^{6}$ & $3.1 \times 10^{7}$ & -\\
 K-means-A & $5.7 \times 10^{5}$ & - & $1.2 \times 10^{6}$ & $1.2 \times 10^{6}$ & $4.8 \times 10^{7}$ & - \\
 K-means-C & $5.3 \times 10^{5}$ & - & $1.2 \times 10^{6}$ & $1.3 \times 10^{6}$ & $4.4 \times 10^{7}$ & - \\
 K-means-S & $6.2 \times 10^{5}$ & - & $1.2 \times 10^{6}$ & $1.1 \times 10^{6}$ & $5.1 \times 10^{7}$ & -\\
\hline {\bf Gold solution}  & {\bf 6} & {\bf 3} & {\bf 8} & {\bf 3} & {\bf 5} & {\bf 18} \\
\hline
\end{tabular}
\end{center}
\end{footnotesize}
\caption{
A summary of the timing results for {\tt Consensus} with $H=250$  and
$p=66\%$, on all algorithms, except NMF,  and for the {\tt Benchmark 1} datasets. For the Leukemia dataset, the timing experiments are not reported because
incomparable with those of the remaining datasets. Cells with a dash indicate that the experiments were terminated  due to their high
computational demand.
}\label{table:H250p66-time}
\end{table}

\begin{table}
\begin{footnotesize}
\begin{center}
\begin{tabular}{|l|cccccc|}\hline
& \multicolumn{6}{c|}{Precision}
\\ \cline{2-7}
 & CNS Rat & Leukemia & NCI60 & Lymphoma & Yeast  & PBM \\

 Hier-A &  8  & \blackstone[3] &  \blackstone[8] &\blackstone[3] & \blackstone[5] & - \\
 Hier-C & \whitestone[7] & \blackstone[3] &  \blackstone[8] & \whitestone[4] & \whitestone[6] & - \\
 Hier-S & 2 & 8 & \whitestone[9] & 9 & 2 & - \\
 K-means-R & \whitestone[7] & \whitestone[4] & \blackstone[8] & 6 & \whitestone[6] & - \\
 K-means-A & \whitestone[7] & \blackstone[3] & \blackstone[8] &  5  & \whitestone[6] & -\\
 K-means-C & \whitestone[7] & \blackstone[3] & \blackstone[8] &  5  & \whitestone[6] & - \\
 K-means-S & \whitestone[7] & 8 &  10  & \whitestone[2] & \whitestone[6] & - \\

\hline
\end{tabular}
\end{center}
\end{footnotesize}
\caption{
A summary of the precision results for {\tt Consensus} with $H=100$  and
$p=66\%$, on all algorithms, except NMF,  and for the {\tt Benchmark 1} datasets.  Cells with a dash indicate that the experiments were terminated  due to their high
computational demand.
}\label{table:H100p66}
\end{table}

\begin{table}
\begin{footnotesize}
\begin{center}
\begin{tabular}{|l|cccccc|}\hline
&  \multicolumn{6}{c|}{Timing}
\\ \cline{2-7}
 &  CNS Rat & Leukemia & NCI60 & Lymphoma &  Yeast & PBM\\
 Hier-A & $5.9 \times 10^{4}$ & - & $9.6 \times 10^{4}$ & $4.3 \times 10^{4}$ & $1.0 \times 10^{6}$ & -\\
 Hier-C & $4.2 \times 10^{4}$ & - & $4.9 \times 10^{4}$ & $4.9 \times 10^{5}$ & $1.0 \times 10^{6}$ & -\\
 Hier-S & $4.8 \times 10^{4}$ & - & $5.1 \times 10^{4}$ & $5.3 \times 10^{5}$ & $1.6 \times 10^{6}$ & -\\

 K-means-R & $5.2 \times 10^{5}$ & - & $1.0 \times 10^{6}$ & $9.6 \times 10^{5}$  & $1.2 \times 10^{7}$ & -\\
 K-means-A & $5.6 \times 10^{5}$ & - & $8.7 \times 10^{5}$ & $6.8 \times 10^{5}$  & $2.5 \times 10^{6}$ & -\\
 K-means-C & $5.5 \times 10^{5}$ & - & $5.0 \times 10^{5}$ & $7.8 \times 10^{5}$  & $1.0 \times 10^{7}$ & -\\
 K-means-S & $3.5 \times 10^{5}$ & - & $5.6 \times 10^{5}$ & $8.0 \times 10^{5}$  & $1.5 \times 10^{7}$ & -\\

\hline
\end{tabular}
\end{center}
\end{footnotesize}
\caption{
A summary of the timing results for {\tt Consensus} with $H=100$  and
$p=66\%$, on all algorithms, except NMF,  and for the {\tt Benchmark 1} datasets.  For the Leukemia dataset, the timing experiments are not reported because
incomparable with those of the remaining datasets. Cells with a dash indicate that the experiments were terminated  due to their high
computational demand.
}\label{table:H100p66-time}
\end{table}

\subsection{Consensus and Similarity Matrices}\label{sec:ConsensusSimilarity}
One concentrates on two experiments that, together, assess the ability of {\tt Consensus} to produce a similarity/distance matrix that actually improves  the performance  of clustering algorithms. In particular, following Monti et al., this thesis concentrates on hierarchical algorithms. As in the previous section, only the {\tt Benchmark 1} datasets are used.

The first experiment is as follows: for each dataset and each hierarchical algorithm considered here, one takes the consensus matrix corresponding to the number of clusters $k^*$  predicted by {\tt Consensus}. That matrix is transformed into a distance matrix, which is then used by the clustering algorithm to produce $k^*$ clusters. The agreement of that clustering solution with the gold solution of the given dataset is measured via the {\tt Adjusted Rand} Index (defined in Section~\ref{sec:AdjRand}). In view of the results reported in the previous section, only the cases $H=500, 250$ and $p=80\%$ are discussed here. The corresponding results are reported in Tables~\ref{Table:simmConsensuns500-80} and~\ref{Table:simmConsensuns250-80}, respectively. The relevant values of $k^*$ are taken from Tables~\ref{table:H500p80} and~\ref{table:H250p80}, respectively. The interested reader will find, at the following supplementary material web site~\cite{SpeedUpWeb}, all the complete tables, in the Tables section, as Tables TS1-TS6. The second experiment follows the same lines as the first, but the clustering algorithm uses a Euclidean distance matrix. The results are reported in Table~\ref{table:siimEuclidean}. In this case, the relevant values of $k^*$ are taken from Table~\ref{table:H500p80}.

Tables~\ref{Table:simmConsensuns500-80} and~\ref{Table:simmConsensuns250-80} confirm the indication about the proper {\tt Consensus} parameter  setting identified in the previous section. Moreover, it is worth pointing out that there is no substantial difference between the results reported in Tables~\ref{Table:simmConsensuns500-80} and~\ref{table:siimEuclidean}. Combining those results with the analogous ones obtained by Monti et al., one has an indication that the consensus matrix is at least as good as an Euclidean distance matrix, when used as input to hierarchical clustering algorithms.

\begin{table}
\begin{footnotesize}
\begin{center}
\begin{tabular}{|l|ccccc|}\hline
 & CNS Rat & Leukemia & NCI60 & Lymphoma & Yeast  \\
\cline{2-6}
 Hier-A & 0.190350 & 0.919174 & 0.498265 & 0.430841 & 0.523873\\
 Hier-C & 0.176446 & 0.676293 & 0.414214 & 0.483664 & 0.492815 \\
 Hier-S & 0.000134 & 0.507680 & 0.161798 & -0.01777 & 0.002036 \\
\hline
\end{tabular}
\end{center}
\end{footnotesize}
\caption{For each dataset and each hierarchical algorithm considered here, the consensus matrix corresponding to the number of clusters $k^*$  predicted by {\tt Consensus} in Table~\ref{table:H500p80} is taken. That matrix is transformed into a distance matrix, which is then used by the clustering algorithm to produce $k^*$ clusters. The agreement of that clustering solution with the gold solution of the given dataset is measured via the {\tt Adjusted Rand} Index.}\label{Table:simmConsensuns500-80}
\end{table}

\begin{table}
\begin{footnotesize}
\begin{center}
\begin{tabular}{|l|ccccc|}\hline
 & CNS Rat & Leukemia & NCI60 & Lymphoma & Yeast \\
\cline{2-6}
 Hier-A & 0.190350 & 0.919174  & 0.498265 & 0.430841  & 0.522578  \\
 Hier-C & 0.237957 & 0.676293  & 0.414214 & 0.483664  & 0.555979 \\
 Hier-S & 0.000134 & -0.040230 & 0.161798 & -0.009141 & 0.002036\\
\hline
\end{tabular}
\end{center}
\end{footnotesize}
\caption{For each dataset and each hierarchical algorithm considered here, the consensus matrix corresponding to the number of clusters $k^*$  predicted by {\tt Consensus} in Table~\ref{table:H250p80} is taken. That matrix is transformed into a distance matrix, which is then used by the clustering algorithm to produce $k^*$ clusters. The agreement of that clustering solution with the gold solution of the given dataset is measured via the {\tt Adjusted Rand} Index.}\label{Table:simmConsensuns250-80}
\end{table}

\begin{table}
\begin{footnotesize}
\begin{center}
\begin{tabular}{|l|ccccc|}\hline
 & CNS Rat & Leukemia & NCI60 & Lymphoma & Yeast  \\
\cline{2-6}
 Hier-A & 0.190350 & 0.910081 & 0.498265 & 0.430841 & 0.558884\\
 Hier-C & 0.135778 & 0.676293 & 0.414214 & 0.483664 & 0.413154\\
 Hier-S & 0.000134 & 0.507680 & 0.161798 & -0.01777 & 0.002036 \\
\hline
\end{tabular}
\end{center}
\end{footnotesize}
\caption{For each dataset and each hierarchical algorithm considered here, the Euclidean distance matrix and number of clusters $k^*$  predicted by {\tt Consensus} in Table~\ref{table:H500p80} is taken. That matrix is used by the clustering algorithm to produce $k^*$ clusters. The agreement of that clustering solution with the gold solution of the given dataset is measured via the {\tt Adjusted Rand} Index.}\label{table:siimEuclidean}
\end{table}

\section{An Approximation of Consensus: FC}\label{sec:FC}
In this section an approximation of {\tt Consensus} is provided. This speedup is referred to as {\tt FC} (Fast Consensus).
Intuitively, a large number of clustering solutions, each obtained via a sample of the original dataset,  seem to be required in order to identify the correct number of clusters. However, there is no theoretic reason indicating that those clustering solutions must each be generated from a \emph{different} sample of the input dataset, as {\tt Consensus} does. Based on this observation, this thesis  proposes to perform, first, a sampling step to generate a data matrix  $D_{1}$, which is then used to generate all clustering solutions for $k$ in the range $[2, k_{max}]$. In terms of code, that implies a simple switch of the two iteration cycles of the {\tt Consensus} procedure (see Chapter~\ref{chap:Stability}).
Indeed, with reference to the stability measures discussed in Chapter~\ref{chap:Stability}, it is also worth noticing that each of the $k_{max}\times H$ clustering solutions needed is computed from a distinct dataset. As discussed here, this leads to inefficiencies, in particular in regard to agglomerative clustering algorithms, such as the hierarchical ones. Indeed, their ability to quickly compute a clustering solution with $k$ clusters from one with $k+1$, typical of these methods, cannot be used  within {\tt Consensus} because, for each $k$, the dataset changes. The same holds true for divisive methods.
In turn, that switch allows to obtain a speedup since costly computational duplications are avoided when the clustering algorithm $C_{1}$ is hierarchical. Indeed, once the switch is done,  it becomes  possible to interleave the computation of the measure with the level bottom-up construction of the hierarchical tree underlying the clustering algorithms. Specifically, only one dendogram construction is required rather than the repeated and partial construction of dendograms as in the {\tt Consensus} procedure. Therefore, one uses, in full, the main characteristic of agglomerative algorithms briefly discussed in the section regarding {\tt Consensus}.
{\tt FC} is formalized by the procedure given in Fig.~\ref{algo:FC}.  It is also worth pointing out that this switch is possible for several of the stability based measures detailed in Chapter~\ref{chap:Stability}. This general approximation paradigm is formalized by the procedure given in Fig.~\ref{algo:SpeedUpParadigm}, where the macro operations and inputs are the same used for the \codice{Stability\_Measure} procedure detailed in Chapter~\ref{chap:Stability}. The \vir{rule of thumb} that one uses to predict  $k^*$, via {\tt FC}, is the same as for {\tt Consensus}. An example is reported in  Fig.~\ref{fig:FCexample}(b). It is worth pointing  out that both the CDFs and $\Delta$ curve shapes for {\tt FC} closely track those of the respective curves for {\tt Consensus} Fig.~\ref{fig:FCexample}(a).

\begin{figure}[ht] \centering \epsfig{file=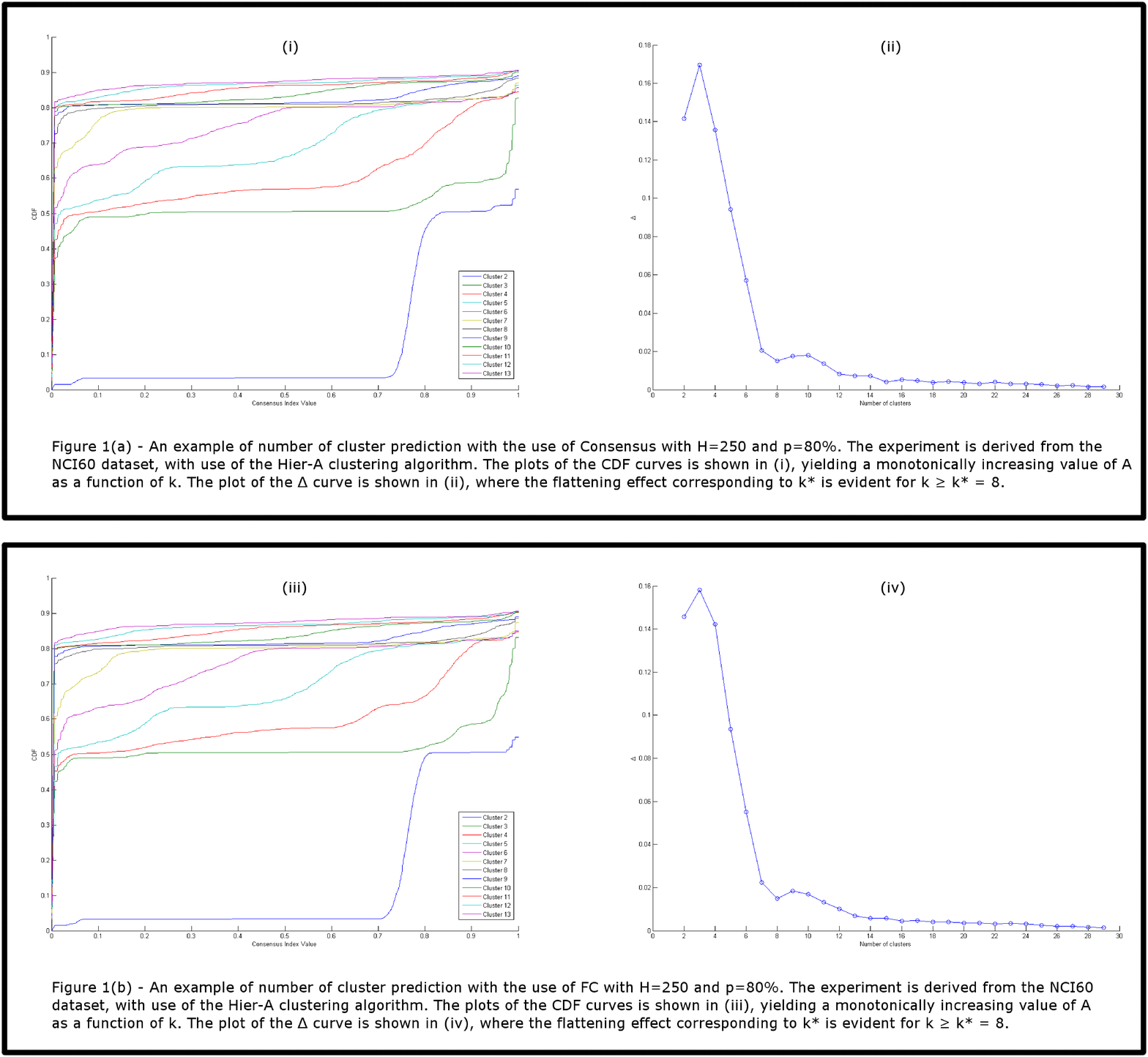,scale=0.16}
\caption{The experiment is derived from the NCI60 dataset, with the use of the
Hier-A clustering algorithm. (a) Figure for {\tt Consensus} with $H=250$ and $p=80\%$: the plot of the CDF curves is shown (i), yielding a monotonically increasing value of $A$ as a
function of $k$.  The plot of the $\Delta$ curve is  shown in (ii),
where the flattening effect corresponding to $k^*$  is evident for
$k\geq k^*=8$. (b) Figure for {\tt FC} with $H=250$ and $p=80\%$: the plots of the $CDF$ curves is
shown in (iii), yielding a monotonically increasing value of $A$ as a
function of $k$.  The plot of the $\Delta$ curve is  shown in (iv),
where the flattening effect corresponding to $k^*$  is evident for
$k\geq k^*=8$.}\label{fig:FCexample}
\end{figure}

\begin{figure}[ht] \centering \epsfig{file=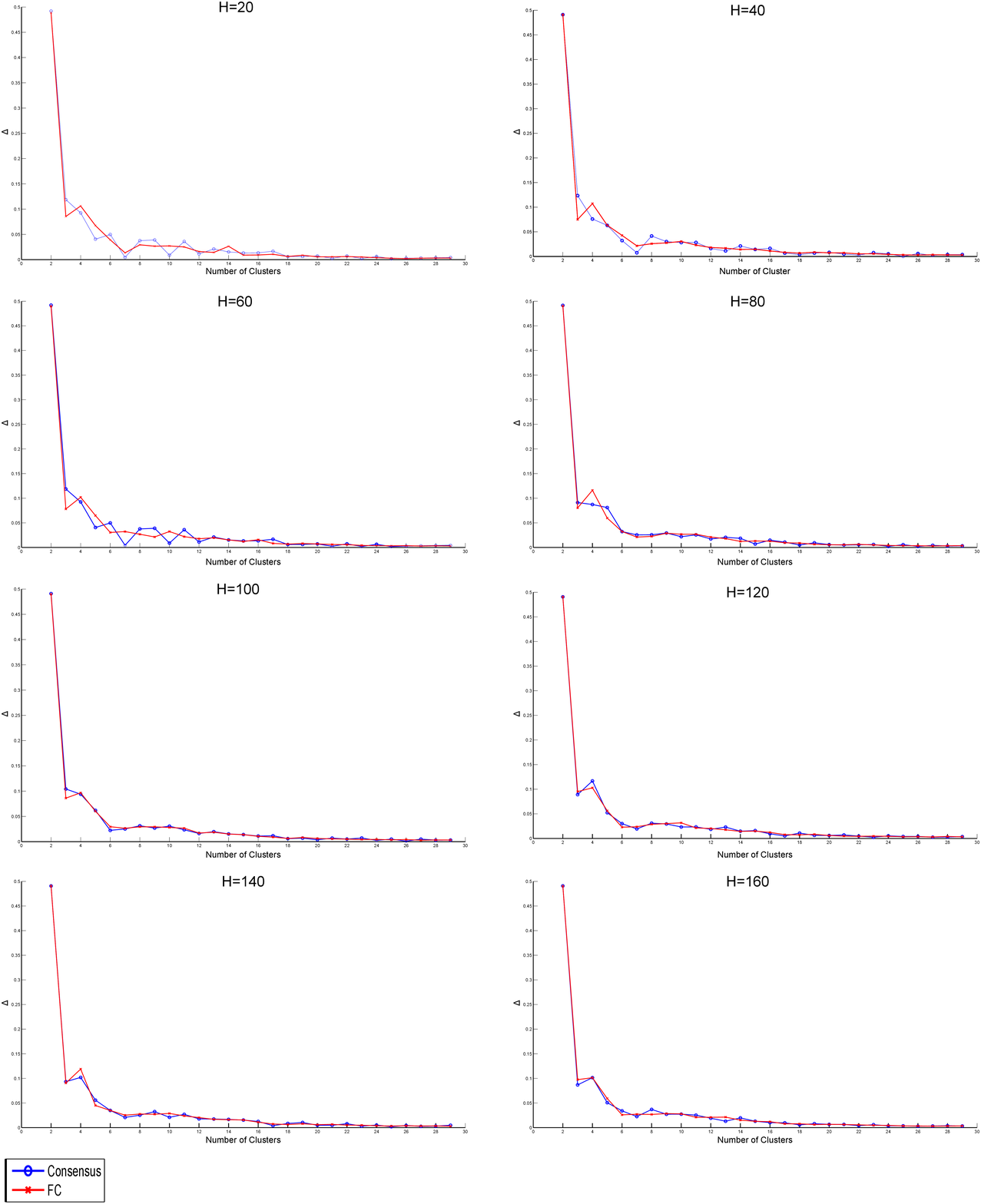,scale=0.18}
\caption{Plot of the $\Delta$ curves for {\tt Consensus} and {\tt FC} with p=80\% and H=20,40,60,80,100,120,140,160.
The experiment is derived from the Lymphoma dataset, with the use of the Hier-A clustering algorithm.}\label{fig:d1}
\end{figure}

\begin{figure}[ht] \centering \epsfig{file=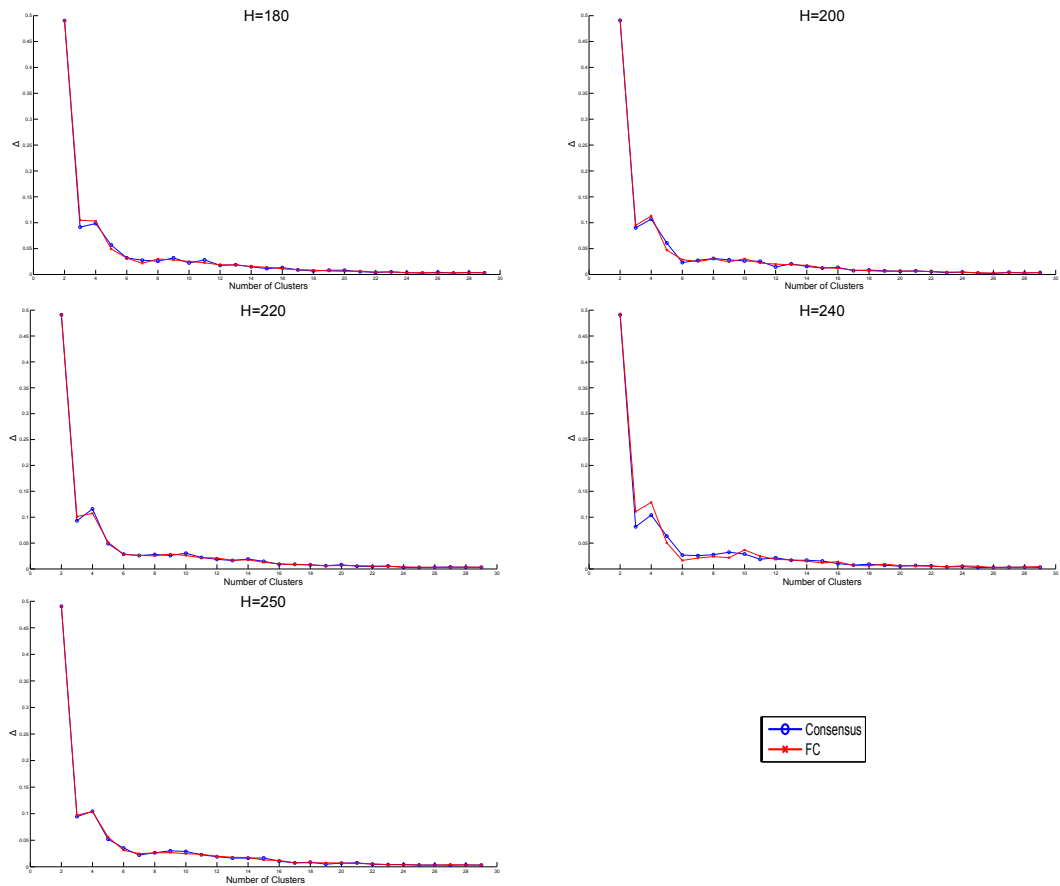,scale=0.18}
\caption{Plot of the $\Delta$ curves for {\tt Consensus} and {\tt FC} with p=80\% and H=180,200,220,250.
The experiment is derived from the Lymphoma dataset, with the use of the Hier-A clustering algorithm.}\label{fig:d2}
\end{figure}

\begin{figure}
\[
\setlength{\fboxsep}{12pt}
\setlength{\mylength}{\linewidth}
\addtolength{\mylength}{-2\fboxsep}
\addtolength{\mylength}{-2\fboxrule}
\ovalbox{
\parbox{\mylength}{
\setlength{\abovedisplayskip}{0pt}
\setlength{\belowdisplayskip}{0pt}

\begin{pseudocode}{FC}{H_{c}, <C_{1}>, D, k_{max}}
\FOR i\GETS 1 \TO H_{c} \DO \\
\BEGIN
1.\mbox{ }\mbox{ }\mbox{ }\mbox{ Generate (via a subsampling) a data matrix $D_i$}\\
\mbox{ }\mbox{ }\mbox{ }\mbox{ }\FOR k\GETS 2 \TO k_{max}   \DO \\
\mbox{ }\mbox{ }\mbox{ }\mbox{ }\BEGIN
2.\mbox{ }\mbox{ }\mbox{ }\mbox{ }\mbox{ }\mbox{ }\mbox{Let $P_{1}$ be the partition of $D_{i}$ into $k$ clusters with the use of }C_{1}\\
3.\mbox{ }\mbox{ }\mbox{ }\mbox{ }\mbox{ }\mbox{ Based on }P_{1}\mbox{ compute the connectivity matrix }M^{k}_{i}\\
\mbox{ }\mbox{ }\mbox{ }\mbox{ }\END\\
\END\\
\FOR k\GETS 2 \TO k_{max}   \DO \\
\BEGIN
4.\mbox{ }\mbox{ }\mbox{ Compute the consensus matrix }\mathcal{M}^{k}\\
\END\\
5.\mbox{ Based on the }k_{max}-1 \mbox{ consensus
matrices, return a prediction for }k^{*}
\end{pseudocode}
}
}
\]
\caption{The \codice{FC} procedure}\label{algo:FC}
\end{figure}


\begin{figure}
\[
\setlength{\fboxsep}{12pt}
\setlength{\mylength}{\linewidth}
\addtolength{\mylength}{-2\fboxsep}
\addtolength{\mylength}{-2\fboxrule}
\ovalbox{
\parbox{\mylength}{
\setlength{\abovedisplayskip}{0pt}
\setlength{\belowdisplayskip}{0pt}

\begin{pseudocode}{Fast\_Stability\_Measure}{k_{min}, k_{max}, D, H, \alpha, \beta, <C_{1},C_{2},\ldots,C_{t}>}
\WHILE H \DO\\
\BEGIN
1. \mbox{ }\mbox{ }<D_1, D_2,\ldots, D_l> \GETS <\CALL{DGP}{D_0, \beta}, \CALL{DGP}{D_0, \beta}, \ldots, \CALL{DGP}{D_0, \beta}>\\
\mbox{ }\mbox{ }\mbox{ }\FOR k\GETS k_{min} \TO k_{max} \DO \\
\mbox{ }\mbox{ }\mbox{ }\BEGIN
2. \mbox{ }\mbox{ }\mbox{ }\mbox{ }<D_{T,0},D_{T,1}, \ldots, D_{T,l}, D_{L,0},D_{L,1}, \ldots, D_{L,l}> \GETS \CALL{Split}{<D_0,D_1, \ldots, D_l>, \alpha}\\

3.\mbox{ }\mbox{ } \mbox{ }\mbox{ }<G> \GETS \CALL{Assign}{<D_{T,0}, D_{T,1},\ldots, D_{T,l}>,<C_{1}, C_{2},\ldots, C_{t}>}\\

4.\mbox{ }\mbox{ }\mbox{ }\mbox{ } <C_{i_{1}}, C_{i_{2}},\ldots, C_{i_{q}}> \GETS \CALL{Train}{<G>}\\

5.\mbox{ }\mbox{ } \mbox{ }\mbox{ }<\hat{G}> \GETS \CALL{Assign}{<D_{L,0}, D_{L,1},\ldots, D_{L,l}>,<C_{1}, C_{2},\ldots, C_{t}>}\\

6. \mbox{ }\mbox{ }\mbox{ }\mbox{ }<P_{1}, P_{2}, \ldots, P_{z}> \GETS \CALL{Cluster}{\hat{G},k}\\

7.~\mbox{ }\mbox{ }\mbox{ }\mbox{ }u \GETS \CALL {Collect\_Statistic}{<P_{1}, P_{2}, \ldots, P_{z}>}\\

8.~\mbox{ }\mbox{ }\mbox{ }\mbox{ }S^{k} \GETS S^{k} \bigcup \{u\}\\

\mbox{ }\mbox{ }\mbox{ }\END\\
\END\\
\mbox{ }\FOR k\GETS k_{min} \TO k_{max} \DO \\
\mbox{ }\mbox{ }\BEGIN
9.\mbox{ }\mbox{ }\mbox{ }\mbox{ }R^{k} \GETS \CALL{Synopsis}{S^{k}}\\
\mbox{ }\mbox{ }\END\\

10. \mbox{ }k^{*} \GETS

\CALL{Significance\_Analysis}{R^{k_{min}},\ldots,R^{k_{max}}}\\

\RETURN{k^{*}}
\end{pseudocode}
}
}
\]
\caption{The \codice{Fast\_Stability\_Measure} procedure.}\label{algo:SpeedUpParadigm}
\end{figure}

\subsection{FC and Its Parameters}
In this section, the results of the experiments obtained with {\tt FC} are reported and discussed. In analogy with {\tt Consensus}, its  precision and time performances depend on $H$ and $p$.  In order to compare the two measures along the parameters of interest, one uses, for {\tt FC}, the same experimental setup detailed in the previous section for {\tt Consensus}. Moreover, based on the results of  the previous section, the discussion here is based only on the experiments for {\tt FC} with $H=250$ and $p=80\%$.

It is worthy to anticipate that the results in this section  will show that {\tt FC} is a very good approximation of {\tt Consensus} both as an internal validation measure and as a preprocessor for clustering algorithms. Remarkably, it is at least one order of magnitude faster in time when used in conjunction with hierarchical clustering algorithms or with partitional algorithms with a hierarchical initialization.

As pointed out in Chapter~\ref{chap:5} in order to perform a better comparison between {\tt Consensus} and its approximation, both {\tt Benchmark 1} and {\tt Benchmark 2} datasets are taken in account (see Section~\ref{sec:dataset} for details) in this section.

\subsubsection{FC as an Internal Validation Measure}

Tables~\ref{table:FC-H250p80} and~\ref{table:FC-H250p80-time}  report the results regarding {\tt FC} as an internal validation measure for the {\tt Benchmark 1} datasets. For this discussion, they are compared with the {\tt Consensus} results reported in Tables~\ref{table:H250p80} and~\ref{table:H250p80-time}. The interested reader will find, at the following supplementary material web site~\cite{SpeedUpWeb}, all the complete tables as Tables TI7-TI12 for {\tt FC}  and the corresponding figures as Figs. S65-S139 in the Tables and Figures section, respectively. The time for the PBM dataset with $p=66\%$ in the corresponding table is not reported, since it does not provide any relevant information.

Note that, in terms of precision,  {\tt FC} and {\tt Consensus} provide  nearly identical predictions on the CNS Rat and Yeast datasets, while their predictions  are quite close  on the Leukemia dataset. Moreover, in terms of time, note that {\tt FC} is faster then {\tt Consensus} by at least one order of magnitude on all hierarchical algorithms and K-means-A, K-means-C and K-means-S. In particular, {\tt FC} is able to complete execution on the PBM dataset, as opposed to {\tt Consensus}, with all of the mentioned algorithms. It is also worthy of notice that K-means-C also provides, for that dataset, a reasonable estimate of the number of clusters present in it. Another point of interest is the performance of {\tt FC} with K-means-R since  the algorithm engineering used in its implementation grants good results on the largest datasets used with that clustering algorithm.

It is somewhat unfortunate, however, that those quite substantial speedups  have only minor effects when one uses {\tt NMF} as a clustering algorithm, which is a clear indication that the time taken by {\tt NMF} to converge to a clustering solution accounts for most of the time performance of {\tt FC} in that setting, in analogy with {\tt Consensus}.

As for {\tt Benchmark 2} datasets, both {\tt Consensus} and {\tt FC} are computed for a number of cluster values in the range $[2,30]$. The prediction value, $k^*$, is based on the plot of the $\Delta(k)$ curve (defined in Chapter~\ref{chap:Stability}) as indicated in~\cite{giancarlo08}. The corresponding plots are available at the following supplementary material web site~\cite{SpeedUpWeb}, in the Figures section, as Figs. M1-M12 and M13-M24 for  {\tt Consensus} and {\tt FC}, respectively.
The corresponding results are reported in Table~\ref{table:Consensus-H250p80Sim} and Table~\ref{table:FCSim} for the simulated datasets, while the corresponding results for the microarray datasets  are in Tables~\ref{table:Consensus-H250p80Monti}-\ref{table:Consensus-H250p80Monti-time} and Tables~\ref{table:FC-Monti}-\ref{table:FC-Monti-time} for {\tt Consensus} and {\tt FC}, respectively.

By comparing the results in the mentioned tables, it is of great interest to notice that,  on the datasets in {\tt Benchmark 2},   there is no difference whatsoever in the predictions between {\tt Consensus} and {\tt FC}. Even more remarkably, by analyzing  the $\Delta$ curves from which the predictions are made (see Methods section), one discovers that the ones produced by {\tt Consensus} and {\tt FC} are nearly identical (see again Figs. M1-M24 at the following supplementary material web site~\cite{SpeedUpWeb}). However, on the microarray datasets on {\tt Benchmark 2}, {\tt FC} is at least one order of magnitude faster than {\tt Consensus}, with exactly the same algorithms indicated  for the {\tt Benchmark 1} datasets. {\tt NMF} results to be problematic also on the datasets on {\tt Benchmark 2}.

It is of some interest to point out that, as detailed in the previous section, {\tt FC} builds the same number of connectivity matrices as {\tt Consensus}. However, it uses only $H$  \vir{new} matrices, each sampled from the input dataset, rather than $H\times k$ \vir{new} matrices as {\tt Consensus} does. Adding this observation to the ones of the preceding subsection, one understands that the number of connectivity matrices computed by  {\tt FC} is key to its precision performance, again in analogy with {\tt Consensus}. The novelty, by far non-obvious,  is that  those matrices can be computed by taking a relatively small number of samples from the input matrix.
Moreover, Figs.~\ref{fig:d1} and~\ref{fig:d2} provides the $\Delta$ curve both for {\tt Consensus} and {\tt FC} for $p=80\%$ and different values of $H$, in order to show how the behavior of the two curves is practically identically for a $H>100$. From these figures it is possible to see how {\tt FC} preserves the same outstanding properties of {\tt Consensus} and that for a reasonable value of $H$ the precision of the measures is the same (see Fig.~\ref{fig:d2} for $H=250$).

\begin{table}
\begin{footnotesize}
\begin{center}
\begin{tabular}{|l|cccccc|}\hline
& \multicolumn{6}{c|}{Precision}
\\ \cline{2-7}
 & CNS Rat & Leukemia & NCI60 & Lymphoma & Yeast  & PBM \\

Hier-A  &  \whitestone[7]  &  \blackstone[3] &  \blackstone[8]  & \blackstone[3]  &  \blackstone[5] & 2 \\
 Hier-C &  \blackstone[6]&  \whitestone[4] &  \blackstone[8] & 5  & \whitestone[6]  & 14 - \whitestone[17] \\
 Hier-S & 2 & 8 &  \blackstone[8]  &  \whitestone[2] & 10  & 2 \\

 K-means-R & \blackstone[6]&\whitestone[4] & \whitestone[7] &  \whitestone[4]& \whitestone[6] &  16  \\
 K-means-A &\whitestone[7] &  \blackstone[3]&  \blackstone[8] &  \blackstone[3] & \whitestone[6] & 12 \\
 K-means-C & \blackstone[6]&  \whitestone[4]  & \blackstone[8]  &  \whitestone[4] & \whitestone[6]& 12 \\
 K-means-S & \blackstone[6] & 7 & \whitestone[9]  & \whitestone[2] & \whitestone[6]& 2 \\

NMF-R & \blackstone[6] & \whitestone[4] & \whitestone[7] & \whitestone[4] & - & - \\
NMF-A & \whitestone[7] & \blackstone[3] & \whitestone[7] & \blackstone[3] & - & - \\
NMF-C & \blackstone[6] & \blackstone[3] & \blackstone[8] & \whitestone[4] & - & - \\
NMF-S & 2 & 8 & \whitestone[9] & \whitestone[2] & - & - \\
\hline {\bf Gold solution}  & {\bf 6} & {\bf 3} & {\bf 8} & {\bf 3} & {\bf 5} & {\bf 18} \\
\hline
\end{tabular}
\end{center}
\end{footnotesize}
\caption{A summary of the precision results for {\tt FC} with $H=250$  and
$p=80\%$, on all algorithms,  and for the {\tt Benchmark 1} datasets.  Cells with a dash indicate that the experiments were terminated  due to their high
computational demand.
}\label{table:FC-H250p80}
\end{table}

\begin{table}
\begin{footnotesize}
\begin{center}
\begin{tabular}{|l|cccccc|}\hline
&  \multicolumn{6}{c|}{Timing}
\\ \cline{2-7}
 &  CNS Rat & Leukemia & NCI60 & Lymphoma &  Yeast & PBM\\

Hier-A  & $5.9 \times 10^{4}$ & $2.7 \times 10^{4}$ & $7.0 \times 10^{4}$ & $6.8 \times 10^{4}$ & $1.5 \times 10^{6}$ & $4.2 \times 10^{7}$ \\
Hier-C & $5.9 \times 10^{4}$ &  $2.7 \times 10^{4}$ & $6.5 \times 10^{4}$ & $6.7 \times 10^{4}$ & $1.4 \times 10^{6} $  &$3.4 \times 10^{7}$\\
Hier-S & $8.1 \times 10^{4}$ &  $2.7 \times 10^{4}$ & $5.8 \times 10^{4}$ & $6.2 \times 10^{4}$ & $2.2 \times 10^{6}$   & $4.4 \times 10^{7}$\\

K-means-R & $3.7 \times 10^{5}$ &  $3.7 \times 10^{5}$ & $1.2 \times 10^{6}$ & $1.1 \times 10^{6}$  & $1.6 \times 10^{7}$  & $1.6 \times 10^{8}$ \\
K-means-A & $3.1 \times 10^{5}$ &  $2.0 \times 10^{5}$ & $9.3 \times 10^{5}$ & $9.0 \times 10^{5}$  & $1.8 \times 10^{7}$ & $2.1 \times 10^{8}$ \\
K-means-C & $2.5 \times 10^{5}$ & $6.0 \times 10^{5}$ & $6.5 \times 10^{5}$ & $9.4 \times 10^{5}$  & $1.4 \times 10^{7}$  & $2.0 \times 10^{8}$\\
K-means-S & $3.7 \times 10^{5}$ &  $5.8 \times 10^{5}$ & $6.9 \times 10^{5}$ &  $9.4 \times 10^{5}$ & $1.9 \times 10^{7}$  & $2.4 \times 10^{8}$\\

NMF-R & $1.1 \times 10^{8}$ & $1.3 \times 10^{7}$ & $6.3 \times 10^{7}$ & $7.5 \times 10^{7}$ & - & -\\
NMF-A & $3.0 \times 10^{7}$ & $4.0 \times 10^{6}$ & $1.2 \times 10^{7}$ & $1.6 \times 10^{7}$ & - & -\\
NMF-C & $2.9 \times 10^{7}$ & $4.0 \times 10^{6}$ & $1.2 \times 10^{7}$ & $1.6 \times 10^{7}$ & - & -\\
NMF-S & $3.5 \times 10^{7}$ & $4.0 \times 10^{6}$ & $1.2 \times 10^{7}$ & $1.5 \times 10^{7}$ & - & -\\
\hline
\end{tabular}
\end{center}
\end{footnotesize}
\caption{A summary of the timing results for {\tt FC} with $H=250$  and
$p=80\%$, on all algorithms, and for the {\tt Benchmark 1} datasets.  Cells with a dash indicate that the experiments were terminated  due to their high
computational demand.
}\label{table:FC-H250p80-time}
\end{table}

\begin{table}
\begin{footnotesize}
\begin{center}
\begin{tabular}{|l|ccc|}\hline
& \multicolumn{3}{c|}{Precision}
\\  \cline{2-4}
 &  Novartis & St.Jude & Normal \\

Hier-A  &  \whitestone[5] -  6  & \blackstone[6] &  10  \\
Hier-C &   \blackstone[4] - \whitestone[5] & \whitestone[5] - \blackstone[6] &  10  \\
Hier-S &  \whitestone[5]  &  2 &  10  \\

K-means-R &  \whitestone[5]  &  \blackstone[6] &  10   \\
K-means-A &  \whitestone[5] -  6   &  \blackstone[6] &  8    \\
K-means-C &  \blackstone[4] - \whitestone[5]  &  \whitestone[5] - \blackstone[6] &  10   \\
K-means-S &  \whitestone[5]  & \blackstone[6] &  10  \\

NMF-R &  -  &  - &  -   \\
NMF-A &  -  &  - &  -   \\
NMF-C &  -  &  - &  -  \\
NMF-S &  -  &  - &  -   \\

 \hline {\bf Gold solution}  & {\bf 4} & {\bf 6} & {\bf 13} \\
\hline
\end{tabular}
\end{center}
\end{footnotesize}
\caption{A summary of the precision results for {\tt Consensus} with $H=250$  and
$p=80\%$, on all algorithms and for the {\tt Benchmark 2} datasets. Cells with a dash indicate that the experiments were terminated  due to their high
computational demand.
}\label{table:Consensus-H250p80Monti}
\end{table}

\begin{table}
\begin{footnotesize}
\begin{center}
\begin{tabular}{|l|ccc|}\hline
& \multicolumn{3}{c|}{Timing}
\\ \cline{2-4}
 &  Novartis & St.Jude & Normal \\

Hier-A  & $1.0 \times 10^{7}$  & $3.7 \times 10^{7}$ & $9.5 \times 10^{6}$  \\
Hier-C & $1.0 \times 10^{7}$  &  $3.7 \times 10^{7}$ & $9.2 \times 10^{6}$ \\
Hier-S & $9.8 \times 10^{6}$ & $3.7 \times 10^{7}$ & $9.4 \times 10^{6}$ \\

K-means-R & $1.8 \times 10^{7}$  &  $1.5 \times 10^{7}$ & $6.3 \times 10^{6}$ \\
K-means-A & $1.4 \times 10^{7}$ & $6.8 \times 10^{7}$ & $1.1 \times 10^{7}$  \\
K-means-C & $1.5 \times 10^{7}$ & $6.8 \times 10^{7}$ & $1.0 \times 10^{7}$ \\
K-means-S & $1.6 \times 10^{7}$ & $6.8 \times 10^{7}$ & $1.1 \times 10^{7}$ \\

NMF-R &  -  &  - &  -   \\
NMF-A &  -  &  - &  -   \\
NMF-C &  -  &  - &  -  \\
NMF-S &  -  &  - &  -  \\

\hline
\end{tabular}
\end{center}
\end{footnotesize}
\caption{A summary of the timing results for {\tt Consensus} with $H=250$  and
$p=80\%$, on all algorithms and for the {\tt Benchmark 2} datasets. Cells with a dash indicate that the experiments were terminated  due to their high
computational demand.
}\label{table:Consensus-H250p80Monti-time}
\end{table}

\begin{table}
\begin{footnotesize}
\begin{center}
\begin{tabular}{|l|ccc|}\hline
& \multicolumn{3}{c|}{Precision}
\\  \cline{2-4}
 &  Gaussian3 & Gaussian5 & Simulated6 \\

Hier-A & \blackstone[3] & \blackstone[5]   & \whitestone[5] \\
Hier-C & \blackstone[3] & \blackstone[5]  & \whitestone[5]   \\
Hier-S & \whitestone[2] & 2  & \whitestone[7]  \\

K-means-R & \blackstone[3]  & \blackstone[5]  & \blackstone[6] \\
K-means-A & \blackstone[3]  & \blackstone[5]  & \whitestone[5] \\
K-means-C & \blackstone[3]  & \blackstone[5]  & \whitestone[5]  \\
K-means-S & \whitestone[2] & \blackstone[5]  & \blackstone[6] \\

 \hline {\bf Gold solution}  & {\bf 3} & {\bf 5} & {\bf 6} \\
\hline
\end{tabular}
\end{center}
\end{footnotesize}
\caption{A summary of the precision results for {\tt Consensus} with $H=250$  and
$p=80\%$, on all algorithms, except NMF, and for the simulated datasets in {\tt Benchmark 2}.
}\label{table:Consensus-H250p80Sim}
\end{table}

\begin{table}
\begin{footnotesize}
\begin{center}
\begin{tabular}{|l|ccc|}\hline
& \multicolumn{3}{c|}{Precision}
\\ \cline{2-4}
 &  Novartis & St.Jude & Normal \\

Hier-A  &  \whitestone[5] -  6  &  \blackstone[6] &  10  \\
 Hier-C &   \blackstone[4] - \whitestone[5] &  \whitestone[5] - \blackstone[6] &  10 \\
 Hier-S &  \whitestone[5]  &  2 & 10 \\

 K-means-R &  \whitestone[5]  & \blackstone[6] &  10  \\
 K-means-A &  \whitestone[5] -  6 & \blackstone[6] & 8  \\
 K-means-C &  \blackstone[4] - \whitestone[5]  &  \whitestone[5] - \blackstone[6] &  10  \\
 K-means-S &  \whitestone[5]  & \blackstone[6]  &  10  \\

NMF-R &  -  &  - &  -   \\
NMF-A &  -  &  - &  -  \\
NMF-C &  -  &  - &  - \\
NMF-S &  -  &  - &  -  \\

 \hline {\bf Gold solution}  & {\bf 4} & {\bf 6} & {\bf 13}\\
\hline
\end{tabular}
\end{center}
\end{footnotesize}
\caption{A summary of the precision results for {\tt FC} with $H=250$  and
$p=80\%$, on all algorithms and for the microarray datasets in {\tt Benchmark 2}. Cells with a dash indicate that the experiments were terminated  due to their high
computational demand.}\label{table:FC-Monti}
\end{table}

\begin{table}
\begin{footnotesize}
\begin{center}
\begin{tabular}{|l|ccc|}\hline
& \multicolumn{3}{c|}{Timing}
\\ \cline{2-4}
 &  Novartis & St.Jude & Normal \\
 Hier-A  & $4.0 \times 10^{5}$  &  $1.6 \times 10^{6}$ & $3.4 \times 10^{5}$ \\
 Hier-C &  $3.9 \times 10^{5}$  &  $1.4 \times 10^{6}$ & $3.3 \times 10^{5}$ \\
 Hier-S &   $4.4 \times 10^{5}$  &  $1.5 \times 10^{6}$ & $3.4 \times 10^{5}$ \\
 K-means-R &  $1.4 \times 10^{7}$  &  $5.9\times 10^{6}$ & $2.0\times 10^{6}$ \\
 K-means-A &  $5.5 \times 10^{6}$  &  $3.2 \times 10^{7}$ & $5.4 \times 10^{6}$  \\
 K-means-C &  $6.5 \times 10^{6}$ & $3.2 \times 10^{7}$  & $2.1\times 10^{6}$ \\
 K-means-S &  $7.8 \times 10^{6}$ & $4.9 \times 10^{7}$ & $2.1\times 10^{6}$  \\
NMF-R &  -  &  - &  -  \\
NMF-A &  -  &  - &  -   \\
NMF-C &  -  &  - &  -   \\
NMF-S &  -  &  - &  -   \\
\hline
\end{tabular}
\end{center}
\end{footnotesize}
\caption{A summary of the precision results for {\tt FC} with $H=250$  and
$p=80\%$, on all algorithms and for the microarray datasets in {\tt Benchmark 2}. Cells with a dash indicate that the experiments were terminated  due to their high
computational demand.}\label{table:FC-Monti-time}
\end{table}

\begin{table}

\begin{footnotesize}
\begin{center}
\begin{tabular}{|l|ccc|}\hline
& \multicolumn{3}{c|}{Precision}
\\  \cline{2-4}
 &  Gaussian3 & Gaussian5 & Simulated6 \\
Hier-A & \blackstone[3] & \blackstone[5]   & \whitestone[5] \\
Hier-C & \blackstone[3] & \blackstone[5]  & \whitestone[5]   \\
Hier-S & \whitestone[2] & 2  & \whitestone[7]  \\
K-means-R & \blackstone[3]  & \blackstone[5]  & \blackstone[6] \\
K-means-A & \blackstone[3]  & \blackstone[5]  & \whitestone[5] \\
K-means-C & \blackstone[3]  & \blackstone[5]  & \whitestone[5]  \\
K-means-S & \whitestone[2] & \blackstone[5]  & \blackstone[6] \\
 \hline {\bf Gold solution}  & {\bf 3} & {\bf 5} & {\bf 6} \\
\hline
\end{tabular}
\end{center}
\end{footnotesize}
\caption{A summary of the precision results for {\tt FC} with $H=250$  and
$p=80\%$, on all algorithms, except NMF, and for the simulated datasets in {\tt Benchmark 2}.
}\label{table:FCSim}
\end{table}

\subsubsection{FC and Similarity Matrices}

The same experiments described for the evaluation of {\tt Consensus} for the computation of a similarity matrix  have been performed here for {\tt FC}. For the presentation of the results, the same organization of Section~\ref{sec:ConsensusSimilarity} is followed here, i.e., the results for the {\tt Benchmark 1} datasets are presented and discussed first.
Indeed, the experiments reported in Table~\ref{table:FC-simGiancarlo} are the same as the ones reported in Table~\ref{Table:simmConsensuns250-80} for {\tt Consensus}. Again, there is no difference between the two tables. Therefore, also in this case, {\tt FC} is a good approximation of {\tt Consensus}.
For completeness, we report that, for {\tt FC} on the {\tt Benchmark 1} datasets, the interested reader will find, at the following supplementary material web site~\cite{SpeedUpWeb}, all the complete tables, in the Tables section, as Tables TS7-TS12,  for each experimental setup.

As for {\tt Benchmark 2} datasets, the results of the  experiments, for the microarrays datasets,  are reported in Tables~\ref{Table:simmConsensuns250-80-Monti} and~\ref{table:FC-simMonti} for {\tt Consensus} and {\tt FC}, respectively. Tables~\ref{table:ConsensussimSimulated} and~\ref{table:FCsimSimulated} report the results for the simulated datasets for  {\tt Consensus} and {\tt FC}, respectively.
Also for the  {\tt Benchmark 2} datasets, there is no difference between the two methods.

In analogy with {\tt Consensus} and the {\tt Benchmark 1} datasets, the clustering results obtained with the use of the similarity matrices computed by {\tt Consensus} and {\tt FC} are compared, for the {\tt Benchmark 2} datasets, against the clustering results obtained with the use of Euclidean distance. The relevant values of $k^*$ are taken from Table~\ref{table:Consensus-H250p80Monti} for {\tt Consensus}  and Table~\ref{table:FC-Monti} for {\tt FC}. The results are reported in Tables~\ref{table:FC-simMonti}-\ref{table:ConsensusEuclideanMonti} and Tables~\ref{table:ConsensussimSimulated}-\ref{table:simEuclideanSimulated}. They confirm that the consensus matrix is at least as good as an Euclidean distance matrix, when used as input to hierarchical clustering algorithms, even when computed by {\tt FC}.

\begin{table}
\begin{footnotesize}
\begin{center}
\begin{tabular}{|l|ccc|}\hline
 & Novartis & St.Jude & Normal\\
\cline{2-4}
 Hier-A &  0.641611 & 0.173717  & 0.572747 \\
 Hier-C &  0.515570 & 0.438039  & 0.521355 \\
 Hier-S &  0.320264 & -7.88788$e^{-4}$ & 0.502043 \\
\hline
\end{tabular}
\end{center}
\end{footnotesize}
\caption{For each dataset and each hierarchical algorithm considered here,  the consensus matrix corresponding to the number of clusters $k^*$  predicted by {\tt Consensus} in Table~\ref{table:Consensus-H250p80Monti} is taken. That matrix is transformed into a distance matrix, which is then used by the clustering algorithm to produce $k^*$ clusters. The agreement of that clustering solution with the gold solution of the given dataset is measured via the {\tt Adjusted Rand} Index.}\label{Table:simmConsensuns250-80-Monti}
\end{table}

\clearpage

\begin{table}
\begin{footnotesize}
\begin{center}
\begin{tabular}{|l|cccccc|}\hline
  & CNS Rat & Leukemia & NCI60 & Lymphoma & Yeast  & PBM   \\
\cline{2-7}
 Hier-A & 0.190350 & 0.919174 & 0.498265 & 0.430841 & 0.528360 & 0.000261\\
 Hier-C & 0.176446 & 0.676293 & 0.414214 & 0.483664 & 0.589179 & 0.672256 \\
 Hier-S & 0.000134 & 0.507680 & 0.171124 & -0.009141 & 0.00203 & 0.000261 \\
\hline
\end{tabular}
\end{center}
\end{footnotesize}
\caption{For each dataset and each hierarchical algorithm considered here, the consensus matrix corresponding to the number of clusters $k^*$  predicted by {\tt FC} in Table~\ref{table:FC-H250p80} is taken. That matrix is transformed into a distance matrix, which is then used by the clustering algorithm to produce $k^*$ clusters. The agreement of that clustering solution with the gold solution of the given dataset is measured via the {\tt Adjusted Rand} Index.}\label{table:FC-simGiancarlo}
\end{table}

\begin{table}
\begin{footnotesize}
\begin{center}
\begin{tabular}{|l|ccc|}\hline
 & Novartis & St.Jude & Normal\\
\cline{2-4}
 Hier-A & 0.641611 & 0.173717 & 0.572747\\
 Hier-C & 0.515570 & 0.435431 & 0.537849  \\
 Hier-S & 0.320264 & -7.88788$e^{-4}$ & 0.502043 \\
\hline
\end{tabular}
\end{center}
\end{footnotesize}
\caption{For each dataset and each hierarchical algorithm considered here,  the consensus matrix corresponding to the number of clusters $k^*$  predicted by {\tt FC} in Table~\ref{table:FC-Monti} is taken. That matrix is transformed into a distance matrix, which is then used by the clustering algorithm to produce $k^*$ clusters. The agreement of that clustering solution with the gold solution of the given dataset is measured via the {\tt Adjusted Rand} Index.}\label{table:FC-simMonti}
\end{table}

\begin{table}
\begin{footnotesize}
\begin{center}
\begin{tabular}{|l|ccc|}
\hline & Novartis & St.Jude & Normal\\
\cline{2-4}
 Hier-A & $0.544647$ & $0.16992$ & $0.57274$ \\
 Hier-C & $0.51557$ & $0.42637$ & $0.52135$\\
 Hier-S & $0.32026$ & $-7.88786e^{-4}$ & $0.50204$ \\
\hline
\end{tabular}
\end{center}
\end{footnotesize}
\caption{For each dataset and each hierarchical algorithm considered here, the Euclidean distance matrix and number of clusters $k^*$  predicted by {\tt Consensus} in Table~\ref{table:Consensus-H250p80Monti} is taken. That matrix is used by the clustering algorithm to produce $k^*$ clusters. The agreement of that clustering solution with the gold solution of the given dataset is measured via the {\tt Adjusted Rand} Index.}\label{table:ConsensusEuclideanMonti}
\end{table}

\begin{table}
\begin{footnotesize}
\begin{center}
\begin{tabular}{|l|ccc|}
\hline  & Gaussian3 & Gaussian5 & Simulated6   \\
\cline{2-4}
 Hier-A & $1.0$  & $0.85942$  & $-0.59990$ \\
 Hier-C & $1.0$ & $0.82638$  & $-0.59990$  \\
 Hier-S & $0.0$ & $0.0$ & $-0.19586$  \\
 \hline
\end{tabular}
\end{center}
\end{footnotesize}
\caption{For each dataset and each hierarchical algorithm considered here, the consensus matrix corresponding to the number of clusters $k^*$  predicted by {\tt Consensus} in Table~\ref{table:Consensus-H250p80Sim} is taken. That matrix is transformed into a distance matrix, which is then used by the clustering algorithm to produce $k^*$ clusters. The agreement of that clustering solution with the gold solution of the given dataset is measured via the {\tt Adjusted Rand} Index.}\label{table:ConsensussimSimulated}
\end{table}

\begin{table}
\begin{footnotesize}
\begin{center}
\begin{tabular}{|l|ccc|}
\hline  & Gaussian3 & Gaussian5 & Simulated6   \\
\cline{2-4}
 Hier-A & $1.0$ & $0.859429$ & $-0.59990$ \\
 Hier-C & $1.0$ & $0.830103$  & $-0.59990$ \\
 Hier-S & $0.0$ & $0.0$  & $-0.19586$ \\
\hline
\end{tabular}
\end{center}
\end{footnotesize}
\caption{For each dataset and each hierarchical algorithm considered here, the consensus matrix corresponding to the number of clusters $k^*$  predicted by {\tt FC} in Table~\ref{table:FCSim} is taken. That matrix is transformed into a distance matrix, which is then used by the clustering algorithm to produce $k^*$ clusters. The agreement of that clustering solution with the gold solution of the given dataset is measured via the {\tt Adjusted Rand} Index.}\label{table:FCsimSimulated}
\end{table}

\begin{table}
\begin{center}
\begin{footnotesize}
\begin{tabular}{|l|ccc|}
\hline  & Gaussian3 & Gaussian5 & Simulated6\\
\cline{2-4}
 Hier-A & 1.0 & 0.82729 & -0.59990 \\
 Hier-C & 1.0 & 0.65218 & -0.59990 \\
 Hier-S & 0.0 & 0.0 & -0.195867 \\
\hline
\end{tabular}
\end{footnotesize}
\end{center}
\caption{For each dataset and each hierarchical algorithm considered here, the Euclidean distance matrix and number of clusters $k^*$  predicted by {\tt Consensus} in Table~\ref{table:Consensus-H250p80Sim} is taken. That matrix is used by the clustering algorithm to produce $k^*$ clusters.  The agreement of that clustering solution with the gold solution of the given dataset is measured via the {\tt Adjusted Rand} Index.}\label{table:simEuclideanSimulated}
\end{table}

\subsection{Comparison of FC with other Internal Validation Measures}
It is also of interest to compare {\tt FC} with other validation measures that are available in the Literature. One takes, as reference, the benchmarking results reported in Chapter~\ref{chap:6},   since both the datasets and the experimental setup are identical to the ones used  here.
It is worth  pointing out that this benchmark show that there is a natural hierarchy, in terms of time, for the measures taken in account. Moreover, the faster the measure, the less accurate it is. From that study and for completeness, taking in account Tables~\ref{table:summary-table} and~\ref{table:summary-table-time} of Section~\ref{time_merits} one reports in Tables~\ref{table:fsummary-table} and~\ref{table:fsummary-table-time} the best performing measures, with the addition of {\tt FC} and the other \vir{best} approximations proposed in this chapter. From that table, one extract and report,  in Tables~\ref{table:fastsummary-table} and~\ref{table:fastsummary-table-time}, the fastest and best performing measures - again, with the addition of {\tt FC}.  As is self-evident from that latter table, {\tt FC} with Hier-A is within a one order of magnitude difference in speed with respect to the fastest measures, i.e., {\tt WCSS} and {\tt G-Gap}. Quite remarkably, it grants a better precision in terms of its ability to identify the underlying structure in each of the benchmark datasets. It is also of relevance to point out that {\tt FC} with Hier-A has a time performance comparable to that of {\tt FOM}, but again it  has a better precision performance. Notice that, none of the three just-mentioned measures depends on any parameter setting, implying that no speedup will result from a tuning of the algorithms.

\begin{table}
\begin{center}
\begin{footnotesize}
\begin{tabular}{|l|ccccc|} \hline
& \multicolumn{5}{c|}{Precision}
\\
\cline{2-6}
  & CNS Rat & Leukemia & NCI60 & Lymphoma & Yeast\\

 {\tt WCSS}-K-means-C & \whitestone[5] & \blackstone[3] &\blackstone[8]& 8 &\whitestone[4] \\
{\tt WCSS}-R-R0 & \whitestone[5] & \whitestone[4] &\blackstone[8] & \blackstone[3] &\whitestone[4]  \\
 {\tt G-Gap}-K-means-R & \whitestone[7] & \blackstone[3] & 4 & \whitestone[4] & \whitestone[6] \\
 {\tt G-Gap}-R-R5 & \whitestone[5] & \whitestone[4] & 2 & \whitestone[2] & \whitestone[4] \\
 {\tt FOM}-K-means-C & \whitestone[7]& 8 & \blackstone[8]&\whitestone[4] & \whitestone[4]  \\
 {\tt FOM}-K-means-S & \blackstone[6] &\blackstone[3] & \blackstone[8] & 8 &\whitestone[4]  \\
 {\tt FOM}-R-R5 & \blackstone[6] & \blackstone[3] & \whitestone[7] & 5 & \blackstone[5] \\
 {\tt FOM}-Hier-A & \whitestone[7] & \blackstone[3] & \whitestone[7] & 6 & \whitestone[6] \\
 {\tt DIFF-FOM}-K-means-C & \whitestone[7] & \blackstone[3] &\whitestone[7]& \whitestone[4] & 3
\\

{\tt FC}-Hier-A & \whitestone[7] & \blackstone[3] & \blackstone[8] & \blackstone[3] & \blackstone[5] \\
{\tt FC}-Hier-C & \blackstone[6]& \whitestone[4] & \blackstone[8] & 5 & \whitestone[6] \\

{\tt FC}-K-means-R & \blackstone[6] & \whitestone[4] & \whitestone[7] & \whitestone[4] & \whitestone[6] \\
{\tt FC}-K-means-A & \whitestone[7] & \blackstone[3] & \blackstone[8] & \blackstone[3] & \whitestone[6] \\
{\tt FC}-K-means-C & \blackstone[6]& \blackstone[3] & \blackstone[8] & \whitestone[4] & \whitestone[6] \\
{\tt FC}-K-means-S & \whitestone[7] & \whitestone[4] & 10 & \whitestone[2] & \whitestone[6] \\

 {\tt Clest-F}-K-means-R & \blackstone[6]& \blackstone[3]& 15  &\whitestone[2]  & \whitestone[4]  \\
 {\tt Clest-FM}-K-means-R & 8 & \whitestone[4] & \blackstone[8]& \whitestone[2] & \whitestone[4] \\
  {\tt Consensus}-Hier-A & \whitestone[7] & \blackstone[3] & \blackstone[8]& \blackstone[3] & \blackstone[5]\\
 {\tt Consensus}-Hier-C & \blackstone[6] &\whitestone[4] & \blackstone[8]& 5 & \whitestone[6] \\
 {\tt Consensus}-K-means-R & \blackstone[6] &\whitestone[4] & \whitestone[7]& \blackstone[3] & \whitestone[6] \\
 {\tt Consensus}-K-means-A & \whitestone[7] & \blackstone[3] & \blackstone[8] & \blackstone[3] & \whitestone[6] \\
 {\tt Consensus}-K-means-C & \blackstone[6]& \blackstone[3] &\blackstone[8]&\whitestone[4] & \whitestone[6] \\
 {\tt Consensus}-K-means-S & \whitestone[7] & \whitestone[4] & 10 & \whitestone[2] & \whitestone[6]\\
 \hline
 {\bf Gold solution}  & {\bf 6} & {\bf 3} & {\bf 8} & {\bf 3} & {\bf 5}\\
\hline
\end{tabular}
\end{footnotesize}
\end{center}
\caption{A summary of precision results of the best performing measures taken into  account in Chapter~\ref{chap:6}, with the addition of {\tt WCSS}-R, {\tt G-Gap}, {\tt FOM}-R,  {\tt DIFF-FOM} and {\tt FC} with $H=500$ and $p=80\%$.
}\label{table:fsummary-table}
\end{table}

\begin{table}
\begin{center}
\begin{footnotesize}
\begin{tabular}{|l|cccc|}\hline
& \multicolumn{4}{c|}{Timing}
\\ \cline{2-5}
  & CNS Rat & Leukemia & NCI60 & Lymphoma\\
 {\tt WCSS}-K-means-C & $1.7 \times 10^{3}$ & $1.3 \times 10^{3}$ & $5.0 \times 10^{3}$ & $4.0
\times 10^{3}$\\
{\tt WCSS}-R-R0 & $1.2 \times 10^{3}$ & $8.0 \times 10^{2}$ & $4.1 \times 10^{3}$ & $3.0
\times 10^{3}$ \\
 {\tt G-Gap}-K-means-R &  $2.4 \times 10^{3}$ & $2.0 \times 10^{3}$ & $8.3 \times 10^{4}$ & $8.4 \times 10^{3}$ \\
 {\tt G-Gap}-R-R5 & $1.2 \times 10^{3}$ & $8.0 \times 10^{2}$ & $4.5 \times 10^{4}$ & $3.2 \times 10^{3}$ \\
 {\tt FOM}-K-means-C & $1.9 \times 10^{4}$ & $9.4 \times 10^{4}$& $5.5 \times 10^{5}$& $2.6 \times 10^{5}$ \\
 {\tt FOM}-K-means-S & $2.9 \times 10^{4}$ & $1.0 \times 10^{5}$ & $7.1 \times 10^{5}$& $3.6 \times 10^{5}$ \\
 {\tt FOM}-R-R5 & $3.9 \times 10^{3}$ & $3.7 \times 10^{4}$ & $2.1 \times 10^{5}$ & $7.6 \times 10^{4}$ \\
 {\tt FOM}-Hier-A & $1.6  \times 10^{3}$ & $7.5 \times 10^{3}$ & $5.1 \times 10^{4}$ & $1.8 \times 10^{4}$\\
 {\tt DIFF-FOM}-K-means-C & $1.9 \times 10^{4}$ & $9.4 \times 10^{4}$ & $5.5 \times 10^{5}$ & $2.6 \times 10^{5}$
\\

{\tt FC}-Hier-A & $4.7 \times 10^{4}$ & $3.5 \times 10^{4}$  & $5.2 \times 10^{4}$ & $1.3 \times 10^{5}$\\
{\tt FC}-Hier-C & $4.4 \times 10^{4}$ & $2.7 \times 10^{4}$ & $1.3 \times 10^{5}$ & $1.3 \times 10^{5}$\\

{\tt FC}-K-means-R & $7.2 \times 10^{5}$ & $7.7 \times 10^{5}$ & $2.5 \times 10^{6}$ & $2.3 \times 10^{6}$\\
{\tt FC}-K-means-A & $5.6 \times 10^{5}$ & $4.2 \times 10^{5}$ & $1.1 \times 10^{6}$ & $1.5 \times 10^{6}$\\
{\tt FC}-K-means-C & $5.4 \times 10^{5}$ & $4.8 \times 10^{5}$ & $1.1 \times 10^{6}$ & $1.0 \times 10^{6}$\\
{\tt FC}-K-means-S & $7.7 \times 10^{5}$ & $4.7 \times 10^{5}$ & $1.3 \times 10^{6}$ & $1.1 \times 10^{6}$\\

 {\tt Clest-F}-K-means-R & $1.2 \times 10^{6}$ & - & - & -  \\
 {\tt Clest-FM}-K-means-R & $1.2 \times 10^{6}$ & -& - & - \\
  {\tt Consensus}-Hier-A & $9.2 \times 10^{5}$ & $7.9 \times 10^{5}$ & $2.0 \times 10^{6}$ & $1.9 \times 10^{6}$\\
 {\tt Consensus}-Hier-C & $8.7 \times 10^{5}$ & $6.9 \times 10^{5}$ & $2.0 \times 10^{6}$ & $2.0 \times 10^{6}$\\
 {\tt Consensus}-K-means-R & $1.0 \times 10^{6}$ & $1.3 \times 10^{6}$ & $3.4 \times 10^{6}$ & $3.0 \times 10^{6}$ \\
 {\tt Consensus}-K-means-A & $1.3 \times 10^{6}$ & $1.6 \times 10^{6}$ & $3.0 \times 10^{6}$ & $2.6 \times 10^{6}$ \\
 {\tt Consensus}-K-means-C & $1.3 \times 10^{6}$ &  $1.8 \times 10^{6}$ & $2.9 \times 10^{6}$ & $2.6 \times 10^{6}$\\
 {\tt Consensus}-K-means-S & $1.5 \times 10^{6}$ & $1.8 \times 10^{6}$ &  $3.2 \times 10^{6}$ &  $2.8 \times 10^{6}$ \\
 \hline
\end{tabular}
\end{footnotesize}
\end{center}
\caption{A summary of the timing results best performing measures taken into  account in Chapter~\ref{chap:6}, with the addition of {\tt WCSS}-R, {\tt G-Gap}, {\tt FOM}-R,  {\tt DIFF-FOM} and {\tt FC}  with $H=500$ and $p=80\%$. Cell with a dash indicates that the experiment was performed on a smaller interval of cluster values with respect to CNS Rat and so the time performance are
not comparable.
}\label{table:fsummary-table-time}
\end{table}

\begin{table}
\begin{center}
\begin{footnotesize}
\begin{tabular}{|l|ccccc|}\hline
& \multicolumn{5}{c|}{Precision}
\\ \cline{2-6}
  & CNS Rat & Leukemia & NCI60 & Lymphoma & Yeast  \\

 {\tt WCSS}-K-means-C & \whitestone[5] & \blackstone[3] &\blackstone[8]& 8 &\whitestone[4] \\
{\tt WCSS}-R-R0 & \whitestone[5] & \whitestone[4] &\blackstone[8] & \blackstone[3] &\whitestone[4]  \\
 {\tt G-Gap}-K-means-R & \whitestone[7] & \blackstone[3] & 4 & \whitestone[4] & \whitestone[6]  \\
 {\tt G-Gap}-R-R5 & \whitestone[5] & \whitestone[4] & 2 & \whitestone[2] & \whitestone[4]  \\
 {\tt FOM}-K-means-C & \whitestone[7] & 8 & \blackstone[8]&\whitestone[4] & \whitestone[4] \\
 {\tt FOM}-K-means-S & \blackstone[6] &\blackstone[3] & \blackstone[8] & 8 &\whitestone[4]  \\
 {\tt FOM}-R-R5 & \blackstone[6] & \blackstone[3] & \whitestone[7] & 5 & \blackstone[5]  \\
 {\tt FOM}-Hier-A & \whitestone[7] & \blackstone[3] & \whitestone[7] & 6 & \whitestone[6] \\
 {\tt DIFF-FOM}-K-means-C & \whitestone[7] & \blackstone[3] &\whitestone[7]& \whitestone[4] & 3
\\
{\tt FC}-Hier-A & \whitestone[7] & \blackstone[3] &  \blackstone[8]  & \blackstone[3] & \blackstone[5] \\
{\tt FC}-Hier-C & \blackstone[6] & \whitestone[4] &  \blackstone[8] & 5 & \whitestone[6]  \\
 \hline
 {\bf Gold solution}  & {\bf 6} & {\bf 3} & {\bf 8} & {\bf 3} & {\bf 5}  \\
\hline
\end{tabular}
\end{footnotesize}
\end{center}
\caption{
A summary of the precision results of best performing measures taken from the benchmarking of  Chapter~\ref{chap:6}, with the addition of {\tt WCSS}-R, {\tt G-Gap}, {\tt FOM}-R,  {\tt DIFF-FOM} and {\tt FC},  with $H=250$ and $p=80\%$.
}\label{table:fastsummary-table}
\end{table}

\begin{table}
\begin{center}
\begin{footnotesize}
\begin{tabular}{|l|cccc|}\hline
& \multicolumn{4}{c|}{Timing}
\\ \cline{2-5}
  & CNS Rat & Leukemia & NCI60 & Lymphoma\\

 {\tt WCSS}-K-means-C & $1.7 \times 10^{3}$ & $1.3 \times 10^{3}$ & $5.0 \times 10^{3}$ & $4.0
\times 10^{3}$\\
{\tt WCSS}-R-R0 & $1.2 \times 10^{3}$ & $8.0 \times 10^{2}$ & $4.1 \times 10^{3}$ & $3.0
\times 10^{3}$ \\
 {\tt G-Gap}-K-means-R & $2.4 \times 10^{3}$ & $2.0 \times 10^{3}$ & $8.3 \times 10^{4}$ & $8.4 \times 10^{3}$ \\
 {\tt G-Gap}-R-R5 & $1.2 \times 10^{3}$ & $8.0 \times 10^{2}$ & $4.5 \times 10^{4}$ & $3.2 \times 10^{3}$ \\
 {\tt FOM}-K-means-C & $1.9 \times 10^{4}$ & $9.4 \times 10^{4}$& $5.5 \times 10^{5}$& $2.6 \times 10^{5}$ \\
 {\tt FOM}-K-means-S & $2.9 \times 10^{4}$ & $1.0 \times 10^{5}$ & $7.1 \times 10^{5}$& $3.6 \times 10^{5}$ \\
 {\tt FOM}-R-R5 & $3.9 \times 10^{3}$ & $3.7 \times 10^{4}$ & $2.1 \times 10^{5}$ & $7.6 \times 10^{4}$ \\
 {\tt FOM}-Hier-A & $1.6  \times 10^{3}$ & $7.5 \times 10^{3}$ & $5.1 \times 10^{4}$ & $1.8 \times 10^{4}$\\
 {\tt DIFF-FOM}-K-means-C & $1.9 \times 10^{4}$ & $9.4 \times 10^{4}$ & $5.5 \times 10^{5}$ & $2.6 \times 10^{5}$
\\
{\tt FC}-Hier-A & $5.9 \times 10^{4}$ & $2.7 \times 10^{4}$ & $7.0 \times 10^{4}$ & $6.8 \times 10^{4}$ \\
{\tt FC}-Hier-C & $5.9 \times 10^{4}$ &  $2.7 \times 10^{4}$ & $6.5 \times 10^{4}$ & $6.7 \times 10^{4}$ \\

\hline
\end{tabular}
\end{footnotesize}
\end{center}
\caption{
A summary of the time results of best performing measures taken from the benchmarking of Chapter~\ref{chap:6}, with the addition of {\tt WCSS}-R, {\tt G-Gap}, {\tt FOM}-R,  {\tt DIFF-FOM} and {\tt FC},  with $H=250$ and $p=80\%$.
}\label{table:fastsummary-table-time}
\end{table}

The results outlined above  are  particularly significant since (i) {\tt FOM} is one of the most established and highly-referenced measures specifically designed for microarray data; (ii)  in purely algorithmic terms,  {\tt WCSS} and {\tt G-Gap}, are so simple as to represent a \vir{lower bound} in terms of the time performance that is achievable by any data-driven internal validation measure.   In conclusion, the experiments reported here show that {\tt FC} is quite close in time performance to three of the fastest data-driven validation measures available in the Literature, while also granting better precision results. In view of the fact that the former measures are considered reference points in this area, the speedup of {\tt Consensus} proposed here seems to be a non-trivial step forward in the area of data-driven internal validation measures. 

%% file: Conclusions.tex
\chapter{Conclusions and Future Directions}\label{chap:8}

In this thesis, an extensive study of internal validation measures is proposed, with attention to the analysis of microarray data.
In particular, this dissertation has contributed to the area as follows:

\paragraph{A Paradigm for Stability Measures.}  A new general paradigm of stability internal validation measures is proposed. It is also shown that each of the known stability based measures is an instance of such a novel paradigm. Surprisingly, also {\tt Gap} falls within the new paradigm. Moreover, from this general algorithmic paradigm, it is simple to design new stability internal measure combining the building blocks of the detailed measures.

\paragraph{Benchmarking of Internal Validation Measures.} A benchmarking of internal validation measures, taking into account both the precision and time, is proposed.
This study provides further insights into the relative merits of each of the measures considered, from which more accurate and useful guidelines for their use can be inferred. In particular, when computer time is taken into account,
there is a hierarchy of measures, with {\tt WCSS} being the fastest and {\tt Consensus} the slowest.
Overall, {\tt Consensus} results to be the method of choice.
It is also to be stressed that no measure performed well on large datasets.

\paragraph{Fast Approximations.} Based on the above benchmarking,  the idea of extensions and
approximations of internal validation measures has been systematically investigated. The resulting new measures turn out to be competitive, both in time and precision. In particular, {\tt G-Gap} and {\tt FC} an approximation of {\tt Gap} and {\tt Consensus}, respectively, are proposed.  As it is evident from the results obtained, the overall performance of the approximations is clearly superior to  the \vir{original} measures. Moreover, depending on the dataset, they are at least one orders of magnitude faster.
In terms of the existing Literature on data-driven internal validation measures, {\tt FC} is only one order of magnitude away from the fastest measures, yet granting a superior performance in terms of precision. Although {\tt FC} does not close the gap between the time performance of the fastest internal validation measures and the most precise, it is a substantial step forward towards that goal.

\paragraph{Benchmarking of NMF as a clustering algorithm.}  A benchmarking of NMF as a clustering algorithm on microarray data is proposed. Unfortunately, in view of the steep computational price one must pay, the use of NMF as a clustering algorithm does not seem to be justified. Indeed, NMF is at least two orders of magnitude slower than a classical clustering algorithm and with a worse precision.

\paragraph{Future Directions.} This thesis suggests several interesting directions of investigation. Some of them are mentioned  next:

\begin{itemize}

\item Techniques that would enhance the performance of NMF. In particular,  a relevant issue is to compute a solution for $k$ clusters starting from one with $k\pm1$ clusters. That is, an incremental/decremental version of NMF.  Such a version  could yield  a substantial speedup when NMF is used as a clustering algorithm in conjunction with {\tt Consensus} and  {\tt FC}.

\item The intrinsic and relative study of stability validation measure  generated from the stability paradigm, mixing the building blocks available today.

\item The design of fast approximations of other stability internal validation measures.

\item An internal validation measure that closes the gap between the time performance of the fastest internal validation measures and the most precise.

\item A comparison among the best data driven validation measures discussed here and Bayesian method that solve the same problem, i.e.,~\cite{BayesianCluster}.
\end{itemize}

%% file: Thesis.bbl
\begin{thebibliography}{100}

\bibitem{BroadInstitute}
Broad institute.
\newblock
  \url{http://www.broadinstitute.org/cgi-bin/cancer/publications/pub_paper.cgi%
?mode=view&paper_id=89}.

\bibitem{NCI60}
{NCI} 60 {C}ancer {M}icroarray {P}roject.
\newblock \url{http://genome-www.stanford.edu/NCI60}.

\bibitem{Bench}
Suppelementary material web site benchmarking.
\newblock
  \url{http://www.math.unipa.it/~raffaele/suppMaterial/benchmarking/benchmarki%
ng/Index.html}.

\bibitem{SpeedUpWeb}
Supplementary material web site speedup.
\newblock \url{http://www.math.unipa.it/~utro/suppMaterial/speedUp/}.

\bibitem{Achlioptas}
D.~Achlioptas.
\newblock Database-friendly random projections: {J}ohnson-{L}indenstrauss with
  binary coins.
\newblock {\em J. Comput. Syst. Sci.}, 66:671--687, 2003.

\bibitem{AilonChazelle}
N.~Ailon and B.~Chazelle.
\newblock Approximate nearest neighbors and the fast {J}ohnson-{L}indenstrauss
  transform.
\newblock In {\em STOC '06: Proceedings of the 38th annual ACM Symposium on
  Theory of computing}, pages 557--563. ACM, 2006.

\bibitem{Langville06initializations}
A.Langville, C.~Meyer, and R.~Albright.
\newblock Initializations for the {Nonnegative Matrix Factorization}, 2006.

\bibitem{Aliz00}
A.A. Alizadeh, M.B. Eisen, R.E. Davis, C.~Ma, I.S. Lossos, A.~Rosenwald, J.C
  Boldrick, H.~Sabet, T.~Tran, X.~Yu, J.I Powell, L.~Yang, G.E. Marti,
  T.~Moore, J.~Jr Hudson, L.~Lu, D.B. Lewis, R.~Tibshirani, G.~Sherlock, W.C.
  Chan, T.C. Greiner, D.D. Weisenburger, J.O. Armitage, R.~Warnke, R.~Levy,
  W.~Wilson, M.R. Grever, J.C Byrd, D.~Botstein, P.O. Brown, and L.M. Staudt.
\newblock Distinct types of diffuse large b-cell lymphoma identified by gene
  expression profiling.
\newblock {\em Nature}, 403:503--511, 2000.

\bibitem{Alon08}
N.~Alon.
\newblock Problems and results in extremal combinatorics - ii.
\newblock {\em Discrete Mathematics}, 308:4460--4472, 2008.

\bibitem{AlonU1999}
U.~Alon, N.~Barkai, D.A. Notterman, K.~Gish, S.~Ybarra, D.~Mack, and A.J.
  Levine.
\newblock Broad patterns of gene expression revealed by clustering analysis of
  tumor and normal colon tissues probed by oligonucleotide arrays.
\newblock {\em Proceedings of the National Academy of Sciences of the United
  States of America}, 96:6745--6750, 1999.

\bibitem{knuth}
R.B. Altman.
\newblock {Professor Donald Knuth on Bioinformatics}.
\newblock
  \url{http://www-helix.stanford.edu/people/altman/bioinformatics.html}.

\bibitem{AndoniIndyk08}
A.~Andoni and P.~Indyk.
\newblock Near-optimal hashing algorithms for approximate nearest neighbor in
  high dimensions.
\newblock {\em Commun. ACM}, 51:117--122, 2008.

\bibitem{Badea05clusteringand}
L.~Badea.
\newblock Clustering and metaclustering with {Nonnegative Matrix
  Decompositions}.
\newblock In {\em 16th European Conference on Machine Learning}. Springer,
  2005.

\bibitem{BellSejnowski}
A.J. Bell and T.J. Sejnowski.
\newblock The "independent components" of natural scenes are edge filters.
\newblock {\em Vision research}, 37:3327--3338, 1997.

\bibitem{Ben-David2006}
S.~Ben-David, U.~von Luxburg, and D.~P\'{a}l.
\newblock {A sober look at clustering stability}.
\newblock {\em Lecture Notes in Computer Science}, 4005:5, 2006.

\bibitem{BenHur02}
A.~Ben-Hur, A.~Elisseeff, and I.~Guyon.
\newblock A stability based method for discovering structure in clustering
  data.
\newblock In {\em Seventh Pacific Symposium on Biocomputing}, pages 6--17.
  ISCB, 2002.

\bibitem{noise1}
J.~Benesty, D.~Morgan, and M.~Sondhi.
\newblock A better understanding and an improved solution to the problems of
  stereophonic acoustic echo cancellation.
\newblock In {\em ICASSP '97: Proceedings of the 1997 IEEE International
  Conference on Acoustics, Speech, and Signal Processing (ICASSP '97) -Volume
  1}, page 303. IEEE Computer Society, 1997.

\bibitem{BerryBrowne}
M.W. Berry and M.~Browne.
\newblock {\em Understanding search engines: mathematical modeling and text
  retrieval}.
\newblock Society for Industrial and Applied Mathematics, 1999.

\bibitem{Berry06algorithmsand}
M.W. Berry, M.~Browne, A.N. Langville, V.P. Pauca, and R.J. Plemmons.
\newblock Algorithms and applications for approximate {Nonnegative Matrix
  Factorization}.
\newblock In {\em Computational Statistics and Data Analysis}, pages 155--173.
  Elsevier, 2006.

\bibitem{BertoniV06}
A.~Bertoni and G.~Valentini.
\newblock Randomized maps for assessing the reliability of patients clusters in
  {DNA} microarray data analyses.
\newblock {\em Artificial Intelligence in Medicine}, 37:85--109, 2006.

\bibitem{BertoniV07}
A.~Bertoni and G.~Valentini.
\newblock Model order selection for bio-molecular data clustering.
\newblock {\em BMC Bioinformatics}, 8, 2007.

\bibitem{BhattacharyaKP09}
A.~Bhattacharya, P.~Kar, and M.~Pal.
\newblock On low distortion embeddings of statistical distance measures into
  low dimensional spaces.
\newblock In {\em DEXA}, pages 164--172, 2009.

\bibitem{BinghamMannila}
E.~Bingham and H.~Mannila.
\newblock Random projection in dimensionality reduction: applications to image
  and text data.
\newblock In {\em KDD '01: Proceedings of the seventh ACM SIGKDD international
  conference on Knowledge discovery and data mining}, pages 245--250. ACM,
  2001.

\bibitem{bittner}
M.~Bittner, P.~Meltzer, Y.~Chen, Y.~Jiang, E.~Seftor, M.~Hendrix, M.~Radmacher,
  R.~Simon, Z.~Yakhini, A.~Ben-Dor, N.~Sampas, E.~Dougherty, E.~Wang,
  F.~Marincola F, C.~Gooden, J.~Lueders, A.~Glatfelter, P.~Pollock, J.~Carpten,
  E.~Gillanders, D.~Leja, K.~Dietrich, C.~Beaudry, M.~Berens, D.~Alberts, and
  V.~Sondak.
\newblock Molecular classification of cutaneous malignant melanoma by gene
  expression profiling.
\newblock {\em Nature}, 406:536--540, 2000.

\bibitem{Bock1985}
H.H. Bock.
\newblock {On some significance tests in cluster analysis}.
\newblock {\em Journal of Classification}, 2:77--108, 1985.

\bibitem{Bondy}
J.A. Bondy and U.S.R. Murty.
\newblock {\em Graph Theory With Applications}.
\newblock Elsevier Science Ltd, 1976.

\bibitem{Clusterin_high}
A.~Borodin, R.~Ostrovsky, and Y.~Rabani.
\newblock Subquadratic approximation algorithms for clustering problems in high
  dimensional space.
\newblock {\em Machine Learning}, 56:153--167, 2004.

\bibitem{BoutsidisGallopoulos}
C.~Boutsidis and E.~Gallopoulos.
\newblock {SVD} based initialization: A head start for nonnegative matrix
  factorization.
\newblock {\em Pattern Recognition}, 41:1350--1362, 2008.

\bibitem{Breckenridge89}
J.N. Breckenridge.
\newblock Replicating cluster analysis: Method, consistency, and validity.
\newblock {\em Multivariate Behavioral Research}, 24(2):147--161, 1989.

\bibitem{BrunetNMF}
J.-P. Brunet, P.~Tamayo, T.R. Golub, and J.P. Mesirov.
\newblock {Metagenes and molecular pattern discovery using matrix
  factorization}.
\newblock {\em Proceedings of the National Academy of Sciences of the United
  States of America}, 101:4164--4169, 2004.

\bibitem{CarmonaPascual}
P.~Carmona-Saez, R.D. Pascual-Marqui, F.~Tirado, J.M. Carazo, and
  A.~Pascual-Montano.
\newblock Biclustering of gene expression data by non-smooth {Non-negative
  Matrix Factorization}.
\newblock {\em BMC Bioinformatics}, 7:78, 2006.

\bibitem{lonardibook}
J.Y. Chen and S.~Lonardi.
\newblock {\em Biological Data Mining}.
\newblock Chapman \& Hall, 2009.

\bibitem{Chen05}
Z.~Chen and A.~Cichocki.
\newblock {Nonnegative Matrix Factorization} with temporal smoothness and/or
  spatial decorrelation constraints.
\newblock Technical report, Laboratory for Advanced Brain Signal Processing,
  RIKEN, 2005.

\bibitem{CichockiAmari}
A.~Cichocki and S.~Amari.
\newblock {\em Adaptive Blind Signal and Image Processing: Learning Algorithms
  and Applications}.
\newblock John Wiley \& Sons, Inc., 2002.

\bibitem{NMFLab}
A.~Cichocki and R.~Zdunek.
\newblock {NMFLAB MATLAB} toolbox for {Non-negative Matrix Factorization}.

\bibitem{Cichocki06}
A.~Cichocki, R.~Zdunek, and S.-I. Amari.
\newblock Csisz\'{a}r's divergences for {Non-negative Matrix Factorization}:
  Family of new algorithms.
\newblock In {\em LNCS}, pages 32--39. Springer, 2006.

\bibitem{CormodeDIM03}
G.~Cormode, M.~Datar, P.~Indyk, and S.~Muthukrishnan.
\newblock Comparing data streams using {H}amming norms (how to zero in).
\newblock {\em IEEE Trans. Knowl. Data Eng.}, 15:529--540, 2003.

\bibitem{Costanzo01012001}
M.C. Costanzo, M.E. Crawford, J.E. Hirschman, J.E. Kranz, P.~Olsen, L.S.
  Robertson, M.S. Skrzypek, B.R. Braun, K.L. Hopkins, P.~Kondu, C.~Lengieza,
  J.E. Lew-Smith, M.~Tillberg, and J.I. Garrels.
\newblock {YPDTM, PombePDTM and WormPDTM: model organism volumes of the
  BioKnowledgeTM Library, an integrated resource for protein information}.
\newblock {\em Nucl. Acids Res.}, 29:75--79, 2001.

\bibitem{DasguptaG03}
S.~Dasgupta and A.~Gupta.
\newblock An elementary proof of a theorem of {J}ohnson and {L}indenstrauss.
\newblock {\em Random Struct. Algorithms}, 22:60--65, 2003.

\bibitem{Ba03}
S.~Datta and S.~Datta.
\newblock Comparisons and validation of statistical clustering techniques for
  microarray gene expression data.
\newblock {\em Bioinformatics}, 19:459--466, 2003.

\bibitem{Devarajan}
K.~Devarajan.
\newblock {Nonnegative Matrix Factorization: An Analytical and Interpretive
  Tool in Computational Biology}.
\newblock {\em PLoS Comput. Biol.}, 4:e1000029, 2008.

\bibitem{Nature_Clustering}
P.~D'haeseleer.
\newblock How does gene expression cluster work?
\newblock {\em Nature Biotechnology}, 23:1499--1501, 2006.

\bibitem{Dhillon}
I.S. Dhillon and D.S. Modha.
\newblock Concept decompositions for large sparse text data using clustering.
\newblock {\em Machine Learning}, 42:143--175, 2001.

\bibitem{Dhillon05}
I.S. Dhillon and S.~Sra.
\newblock Generalized {Nonnegative Matrix} approximations with {Bregman}
  divergences.
\newblock In {\em Neural Information Proceedings Systems}, pages 283--290,
  2005.

\bibitem{Genclust}
V.~{Di Ges\'{u}}, R.~Giancarlo, G.~{Lo Bosco}, A.~Raimondi, and D.~Scaturro.
\newblock Genclust: A genetic algorithm for clustering gene expression data.
\newblock {\em BMC Bioinformatics}, 6:289, 2005.

\bibitem{Donoho03whendoes}
D.~Donoho and V.~Stodden.
\newblock When does {Non-Negative Matrix Factorization} give correct
  decomposition into parts?
\newblock In {\em Seventeenth Annual Conferfence on Neural Information
  Processing Systems}, 2003.

\bibitem{Shmulevichbook}
E.R. Dougherty, I.~Shmulevich, L.~Chen, and Z.J. Wang.
\newblock {\em {Genomic Signal Processing and Statistics}}, volume~2.
\newblock {Hindawi Publishing Corporation}, 2005.

\bibitem{CLEST}
S.~Dudoit and J.~Fridlyand.
\newblock A prediction-based resampling method for estimating the number of
  clusters in a dataset.
\newblock {\em Genome Biology}, 3, 2002.

\bibitem{Dudoit2003}
S.~Dudoit and J.~Fridlyand.
\newblock {Bagging to improve the accuracy of a clustering procedure}.
\newblock {\em Bioinformatics}, 19(9):1090--1099, 2003.

\bibitem{ET93}
B.~Efron and R.J. Tibshirani.
\newblock {\em An Introduction to the Bootstrap}.
\newblock Chapman \& Hall, London, 1993.

\bibitem{Eisen98}
M.B. Eisen, P.T. Spellman, P.O. Brown, and D.~Botstein.
\newblock Cluster analysis and display of genome-wide expression patterns.
\newblock {\em Proceedings of The National Academy of Science USA},
  95:14863--14868, 1998.

\bibitem{Ever}
B.~Everitt.
\newblock {\em Cluster Analysis}.
\newblock Edward Arnold, London, 1993.

\bibitem{Fogel2007}
P.~Fogel, S.S. Young, D.M. Hawkins, and N.~Ledirac.
\newblock Inferential, robust {Non-negative Matrix Factorization} analysis of
  microarray data.
\newblock {\em Bioinformatics}, 23:44--49, 2007.

\bibitem{FMMeasure}
E.B. Fowlkes and C.L. Mallows.
\newblock A method for comparing two hierarchical clusterings.
\newblock {\em Journal of the American Statistical Association}, 78:553--584,
  1983.

\bibitem{Kmeanscoresets}
G.~Frahling and C.~Sohler.
\newblock A fast {K-means} implementation using coresets.
\newblock In {\em Proceedings of the Twenty-Second Annual Symposium on
  Computational Geometry}, pages 135--143, New York, NY, USA, 2006. ACM.

\bibitem{YuanGao11012005}
Y.~Gao and G.~Church.
\newblock {Improving molecular cancer class discovery through sparse
  Non-negative Matrix Factorization}.
\newblock {\em Bioinformatics}, 21:3970--3975, 2005.

\bibitem{GareyJohnson}
M.R. Garey and D.S. Johnson.
\newblock {\em Computers and Intractability; A Guide to the Theory of
  NP-Completeness}.
\newblock W. H. Freeman \& Co., New York, NY, USA, 1990.

\bibitem{Gaujoux2010}
R.~Gaujoux and C.~Seoighe.
\newblock {A flexible R package for Nonnegative Matrix Factorization.}
\newblock {\em BMC Bioinformatics}, 11:367, 2010.

\bibitem{Bioconductor}
R.C. Gentleman, V.J. Carey, D.~M. Bates, B.~Bolstad, M.~Dettling, S.~Dudoit,
  B.~Ellis, L.~Gautier, Y.~Ge, J.~Gentry, K.~Hornik, T.~Hothorn, W.~Huber,
  S.~Iacus, R.~Irizarry, F.~Leisch, C.~Li, M.~Maechler, A.J. Rossini,
  G.~Sawitzki, C.~Smith, G.~Smyth, L.~Tierney, J.Y. Yang, and J.~Zhang.
\newblock Bioconductor: open software development for computational biology and
  bioinformatics.
\newblock {\em Genome biology}, 5:R80+, 2004.

\bibitem{Giancarlo2010}
R.~Giancarlo, G.~{Lo Bosco}, and L.~Pinello.
\newblock {Distance Functions, Clustering Algorithms and Microarray Data
  Analysis}.
\newblock In {\em Lecture Notes in Computer Science}, volume 6073, 2010.

\bibitem{giancarlocibb}
R.~Giancarlo, G.~{Lo Bosco}, P.~Pinello, and F.~Utro.
\newblock {The Three Steps of Clustering in the Post-Genomic Era}.
\newblock {\em Lecture Notes in Bioinformatics}, To Appear.

\bibitem{giancarlo08}
R.~Giancarlo, D.~Scaturro, and F.~Utro.
\newblock Computational cluster validation for microarray data analysis:
  experimental assessment of {C}lest, {C}onsensus {C}lustering, {F}igure of
  {M}erit, {G}ap {S}tatistics and {M}odel {E}xplorer.
\newblock {\em BMC Bioinformatics}, 9:462, 2008.

\bibitem{FC}
R.~Giancarlo and F.~Utro.
\newblock Speeding up the {Consensus Clustering methodology for microarray data
  analysis}.
\newblock Submitted, 2010.

\bibitem{MatrixComputation}
G.H. Golub and C.F. {Van Loan}.
\newblock {\em Matrix computations (3rd ed.)}.
\newblock Johns Hopkins University Press, 1996.

\bibitem{golub99}
T.R. Golub, D.K. Slonim, P.~Tamayo, C.~Huard, M.~Gaasenbeeck, J.P. Mesirov,
  H.~Coller, M.L. Loh, J.R. Downing, M.A. Caligiuri, C.D. Bloomfield, and E.S.
  Lander.
\newblock Molecular classification of cancer: Class discovery and class
  prediction by gene expression monitoring.
\newblock {\em Science}, 286(531):531--537, 5439 1999.

\bibitem{GonzalesZhang2005a}
E.F. Gonzales and Y.~Zhang.
\newblock Accelerating the {Lee-Seung} algorithm for {Non-negative Matrix
  Factorization}.
\newblock Technical report, Dept. Comput. Appl. Math., Rice University,
  Houston, TX, 2005.

\bibitem{Gordon1}
A.D. Gordon.
\newblock Clustering algorithms and cluster validation.
\newblock In P.~Dirschedl and R.~Ostermann, editors, {\em Computational
  Statistics}, pages 503--518. Physica-Verlag, Heidelberg, Germany, 1994.

\bibitem{Gordon}
A.D. Gordon.
\newblock Null models in cluster validation.
\newblock In {\em From Data to Knowledge: Theoretical and Practical Aspects of
  Classification}, pages 32--44. Springer Verlag, 1996.

\bibitem{GuillametBressan}
D.~Guillamet, M.~Bressan, and J.~Vitr\'{i}.
\newblock A weighted {Non-Negative Matrix Factorization} for local
  representations.
\newblock {\em Computer Vision and Pattern Recognition, IEEE Computer Society
  Conference on}, 1:942, 2001.

\bibitem{GuillametICPR}
D.~Guillamet, B.~Schiele, and J.~Vitri\'{a}.
\newblock Analyzing {Non-Negative Matrix Factorization} for image
  classification.
\newblock {\em Pattern Recognition, International Conference on}, 2:20116,
  2002.

\bibitem{GuillametSchiele}
D.~Guillamet, J.~Vitri\`{a}, and B.~Schiele.
\newblock Introducing a weighted {Non-negative Matrix Factorization} for image
  classification.
\newblock {\em Pattern Recognition Letters}, 24:2447--2454, 2003.

\bibitem{Guillametface}
D.~Guillamet and M.~Vitri\'{a}.
\newblock Classifying faces with {N}onnegative {M}atrix {F}actorization.
\newblock In {\em Proceedings of the 5th Catalan Conference for Artificial
  Intelligence}, 2002.

\bibitem{Handl05}
J.~Handl, J.~Knowles, and D.B. Kell.
\newblock Computational cluster validation in post-genomic data analysis.
\newblock {\em Bioinformatics}, 21(15):3201--3212, 2005.

\bibitem{Hansen93}
M.H. Hansen, W.N. Hurwitz, and W.G. Madow.
\newblock {\em Sample Survey Methods and Theory Methods and Applications},
  volume~1.
\newblock Wiley, 1993.

\bibitem{Hansen}
P.~Hansen and P.~Jaumard.
\newblock Cluster analysis and mathematical programming.
\newblock {\em Mathematical Programming}, 79:191--215, 1997.

\bibitem{Harper}
C.W.~Jr Harper.
\newblock {Groupings by locality in community ecology and paleoecology: tests
  of significance}.
\newblock {\em Lethaia}, 11:251--257, 1978.

\bibitem{Hartigan}
J.A. Hartigan.
\newblock {\em Clustering Algorithms}.
\newblock John Wiley and Sons, 1975.

\bibitem{Hartuv00}
E.~Hartuv, A.~Schmitt, J.~Lange, S.~Meier-Ewert, H.~Lehrach, and R.~Shamir.
\newblock An algorithm for clustering of c{DNA}s for gene expression analysis
  using short oligonucleotide fingerprints.
\newblock {\em Genomics}, 66:249--256, 2000.

\bibitem{Tibshrbook}
T.~Hastie, R.~Tibshirani, and J.~Friedman.
\newblock {\em The Elements of Statistical Learning}.
\newblock Springer, 2003.

\bibitem{HegerHolm}
A.~Heger and L.~Holm.
\newblock Sensitive pattern discovery with ``fuzzy'' alignments of distantly
  related proteins.
\newblock {\em Bioinformatics}, 19:i130--i137, 2003.

\bibitem{HoodGals}
L.~Hood and D.~Galas.
\newblock The digital code of {DNA}.
\newblock {\em Nature}, 421:444--448, 2003.

\bibitem{Hoyer}
P.O. Hoyer.
\newblock Nonnegative sparse coding.
\newblock In {\em IEEE Workshop on Neural Net- works for Signal Processing},
  2002.

\bibitem{Hoyer04}
P.O. Hoyer and P.~Dayan.
\newblock {Non-negative Matrix Factorization} with sparseness constraints.
\newblock {\em Journal of Machine Learning Research}, 5:1457--1469, 2004.

\bibitem{AdjRand0}
L.~Hubert and P.~Arabie.
\newblock Comparing partitions.
\newblock {\em Journal of Classification}, 2:193--218, 1985.

\bibitem{IndykMotwani}
P.~Indyk and R.~Motwani.
\newblock Approximate nearest neighbors: towards removing the curse of
  dimensionality.
\newblock In {\em STOC '98: Proceedings of the 30th annual ACM symposium on
  Theory of computing}, pages 604--613. ACM, 1998.

\bibitem{Isakoff06122005}
M.S. Isakoff, C.G. Sansam, P.~Tamayo, A.~Subramanian, J.A. Evans, C.M.
  Fillmore, X.~Wang, J.A. Biegel, S.L. Pomeroy, J.P. Mesirov, and C.W.M.
  Roberts.
\newblock {Inactivation of the Snf5 tumor suppressor stimulates cell cycle
  progression and cooperates with p53 loss in oncogenic transformation}.
\newblock {\em Proceedings of the National Academy of Sciences of the United
  States of America}, 102:17745--17750, 2005.

\bibitem{JainDubes}
A.K. Jain and R.C. Dubes.
\newblock {\em Algorithms for Clustering Data}.
\newblock Prentice-Hall, Englewood Cliffs, 1988.

\bibitem{jain}
A.K. Jain, M.N. Murty, and P.J. Flynn.
\newblock Data clustering: a review.
\newblock {\em ACM Computing Surveys}, 31(3):264--323, 1999.

\bibitem{JL}
W.B. Johnson and J.~Lindenstrauss.
\newblock Extensions of {L}ipschitz mappings into a {H}ilbert space.
\newblock {\em Contemp. Math.}, 26:189--206, 1984.

\bibitem{JohnsonN09}
W.B. Johnson and A.~Naor.
\newblock The {J}ohnson-{L}indenstrauss lemma almost characterizes {H}ilbert
  space, but not quite.
\newblock In {\em SODA}, pages 885--891, 2009.

\bibitem{JungLee}
I.~Jung, J.~Lee, S.-Y.Lee, and D.~Kim.
\newblock Application of {Nonnegative Matrix Factorization} to improve
  profile-profile alignment features for fold recognition and remote homolog
  detection.
\newblock {\em BMC Bioinformatics}, 9:298, 2008.

\bibitem{Tib06}
A.V. Kapp and R.~Tibshirani.
\newblock Are clusters found in one dataset present in another dataset ?
\newblock {\em Biostatistics}, 8:9--31, 2007.

\bibitem{KR90}
L.~Kaufman and P.J. Rousseeuw.
\newblock {\em Finding Groups in Data: An Introduction to Cluster Analysis}.
\newblock Wiley, New York, 1990.

\bibitem{Kelm2007}
B.M. Kelm, B.H. Menze, C.M. Zechmann, K.T. Baudendistel, and F.A. Hamprecht.
\newblock Automated estimation of tumor probability in prostate magnetic
  resonance spectroscopic imaging: Pattern recognition vs quantification.
\newblock {\em Magnetic Resonance in Medicine}, 57:150--159, 2007.

\bibitem{Kerr00bootstrappingcluster}
M.K. Kerr and G.A. Churchill.
\newblock Bootstrapping cluster analysis: Assessing the reliability of
  conclusions from microarray experiments.
\newblock {\em PNAS}, 98:8961--8965, 2000.

\bibitem{Kikuchi2003}
S.~Kikuchi, D.~Tominaga, M.~Arita, K.~Takahashi, and M.~Tomita.
\newblock {Dynamic modeling of genetic networks using genetic algorithm and
  S-system}.
\newblock {\em Bioinformatics}, 19:643--650, 2003.

\bibitem{KimPark}
H.~Kim and H.~Park.
\newblock Sparse {Non-negative Matrix Factorizations} via alternating
  non-negativity-constrained least squares for microarray data analysis.
\newblock {\em Bioinformatics}, 23:1495--1502, 2007.

\bibitem{Kim01072003}
P.M. Kim and B.~Tidor.
\newblock {Subsystem Identification Through Dimensionality Reduction of
  Large-Scale Gene Expression Data}.
\newblock {\em Genome Research}, 13:1706--1718, 2003.

\bibitem{klie10}
S.~Klie, Z.~Nikoloski, and J.~Selbig.
\newblock Biological cluster evaluation for gene function prediction.
\newblock {\em Journal of Computational Biology}, 17:1--18, 2010.

\bibitem{Kraus10}
J.~Kraus and H.~Kestler.
\newblock A highly efficient multi-core algorithm for clustering extremely
  large datasets.
\newblock {\em BMC Bioinformatics}, 11, 2010.

\bibitem{KLMeasure}
W.~Krzanowski and Y.~Lai.
\newblock A criterion for determining the number of groups in a dataset using
  sum of squares clustering.
\newblock {\em Biometrics}, 44:23--34, 1985.

\bibitem{KullbackLeibler1951}
S.~Kullback and R.A. Leibler.
\newblock On information and sufficiency.
\newblock {\em The Annals of Mathematical Statistics}, 22:79--86, 1951.

\bibitem{LangvilleMeyerAlbright}
A.~Langville, C.~Meyer, R.~Albright, J.~Cox, and D.~Duling.
\newblock Algorithms, initializations, and convergence for the {Nonnegative
  matrix factorization}.
\newblock In {\em Twelfth Annual SIGKDD International Conference on Knowledge
  Discovery and Data Mining}, 2006.

\bibitem{LawrenceMatusik}
J.~Lawrence, A.~Ben-Artzi, C.~{De Coro}, W.~Matusik, H.~Pfister,
  R.~Ramamoorthi, and S.~Rusinkiewicz.
\newblock Inverse shade trees for non-parametric material representation and
  editing.
\newblock {\em ACM Transactions on Graphics}, 25:735--745, 2006.

\bibitem{LawrenceRusinkiewicz}
J.~Lawrence, S.~Rusinkiewicz, and R.~Ramamoorthi.
\newblock Efficient {BRDF} importance sampling using a factored representation.
\newblock In {\em SIGGRAPH '04: ACM SIGGRAPH 2004 Papers}, pages 496--505. ACM,
  2004.

\bibitem{LawsonHanson}
C.L. Lawson and R.J. Hanson.
\newblock {\em Solving Least Squares Problems (Classics in Applied
  Mathematics)}.
\newblock Society for Industrial Mathematics, new edition edition, 1987.

\bibitem{NMF}
D.D. Lee and H.S. Seung.
\newblock Learning the parts of objects by {Non-negative Matrix Factorization}.
\newblock {\em Nature}, 401:788--791, 1999.

\bibitem{lee00algorithms}
D.D. Lee and H.S. Seung.
\newblock Algorithms for {Non-negative Matrix Factorization}.
\newblock In {\em {NIPS}}, pages 556--562, 2000.

\bibitem{zscore}
M-Y. Leung, G.M. Marsch, and T.P. Speed.
\newblock Over and underrepresentation of short {DNA} words in {H}erphesvirus
  genomes.
\newblock {\em Journal of Computational Biology}, 3:345--360, 1996.

\bibitem{LevineDomany}
E.~Levine and E.~Domany.
\newblock Resampling method for unsupervised estimation of cluster validity.
\newblock {\em Neural Computation}, 13:2573--2593, 2001.

\bibitem{Lin07}
C.-J. Lin.
\newblock On the convergence of multiplicative update algorithms for
  {Non-negative Matrix Factorization}.
\newblock {\em IEEE Transactions on Neural Networks}, 18:1589--1596, 2007.

\bibitem{Lin07Neural}
C.-J. Lin.
\newblock Projected gradient methods for {Non-negative Matrix Factorization}.
\newblock {\em Neural Computation}, 19:2756--2779, 2007.

\bibitem{MacQueen}
J.B. MacQueen.
\newblock Some methods for classification and analysis of multivariate
  observations.
\newblock In {\em Proceedings of the fifth Berkeley Symposium on Mathematical
  Statistics and Probability}, volume~1, pages 281--297. University of
  California Press, 1967.

\bibitem{Marriot71}
F.H.C. Marriot.
\newblock Practical problems in a method of cluster analysis.
\newblock {\em Biometrics}, 27:501--514, 1971.

\bibitem{McAdams-1995}
H.H. Mcadams and L.~Shapiro.
\newblock Circuit simulation of genetic networks.
\newblock {\em Science}, 269(5224):650--6, 1995.

\bibitem{LK04}
G.J. McLachlan and N.~Khan.
\newblock On a resampling approach for tests on the number of clusters with
  mixture model-based clustering of tissue samples.
\newblock {\em J. Multivar. Anal.}, 90(1):90--105, 2004.

\bibitem{MCShane02}
L.M. McShane, M.D. Radmacher, B.~Freidlin, R.~Yu, M.-C. Li, and R.~Simon.
\newblock {Methods for assessing reproducibility of clustering patterns
  observed in analyses of microarray data}.
\newblock {\em Bioinformatics}, 18:1462--1469, 2002.

\bibitem{Mehta2004}
T.~Mehta, M.~Tanik, and D.B. Allison.
\newblock Towards sound epistemological foundations of statistical methods for
  high-dimensional biology.
\newblock {\em Nature genetics}, 36:943--947, 2004.

\bibitem{MC85}
G.W. Milligan and M.C. Cooper.
\newblock An examination of procedures for determining the number of clusters
  in a data set.
\newblock {\em Psychometrika}, 50:159--179, 1985.

\bibitem{AdjRand2}
G.W. Milligan and M.C. Cooper.
\newblock A study of the comparability of external criteria for hierarchical
  cluster analysis.
\newblock {\em Multivariate Behavioral Research}, 21:441--458, 1986.

\bibitem{Mirkin}
B.~Mirkin.
\newblock {\em Mathematical Classification and Clustering}.
\newblock Kluwer Academic Publisher, 1996.

\bibitem{Monti03}
S.~Monti, P.~Tamayo, J.~Mesirov, and T.~Golub.
\newblock Consensus clustering: A resampling-based method for class discovery
  and visualization of gene expression microarray data.
\newblock {\em Machine Learning}, 52:91--118, 2003.

\bibitem{NielsenBalslevHansen}
F.~{\AA}. Nielsen, D.~Balslev, and L.~K. Hansen.
\newblock Mining the posterior cingulate: Segregation between memory and pain
  components.
\newblock {\em NeuroImage}, 27:520--532, 2005.

\bibitem{Okun2006}
O.~Okun and H.~Priisalu.
\newblock Fast {Nonnegative Matrix Factorization} and its application for
  protein fold recognition.
\newblock {\em EURASIP J. Appl. Signal Process}, 2006:62--62, 2006.

\bibitem{Paatero1997}
P.~Paatero.
\newblock Least squares formulation of robust {Non-negative Factor Analysis}.
\newblock {\em Chemometrics and Intelligent Laboratory Systems}, 37, 1997.

\bibitem{PaateroPARAFAC}
P.~Paatero.
\newblock A weighted {Non-negative Least Squares} algorithm for three-way
  `{PARAFAC}' factor analysis.
\newblock {\em Chemometrics and Intelligent Laboratory Systems}, 38:223--242,
  1997.

\bibitem{Paatero99}
P.~Paatero.
\newblock The multilinear engine: A table-driven, least squares program for
  solving multilinear problems, including the n-way parallel factor analysis
  model.
\newblock {\em Journal of Computational and Graphical Statistics}, 8:854--888,
  1999.

\bibitem{PaateroTapper}
P.~Paatero and U.~Tapper.
\newblock {Positive Matrix Factorization}: A non-negative factor model with
  optimal utilization of error estimates of data values.
\newblock {\em Environmetrics}, 5:111--126, 1994.

\bibitem{PapadimitriouS82}
C.H. Papadimitriou and K.~Steiglitz.
\newblock {\em Combinatorial Optimization: Algorithms and Complexity}.
\newblock Prentice-Hall, 1982.

\bibitem{PascualMontano2006}
A.~Pascual-Montano, J.M. Carazo, K.~Kochi, D.~Lehmann, and R.D. Pascual-Marqui.
\newblock {Non-smooth Non-Negative Matrix Factorization (nsNMF)}.
\newblock {\em IEEE Transactions on Pattern Analysis and Machine Intelligence},
  28:403--415, 2006.

\bibitem{bioNMF}
A.~Pascual-Montano, P.~Carmona-Saez, M.~Chagoyen, F.~Tirado, J.M. Carazo, and
  R.D. Pascual-Marqui.
\newblock bionmf: a versatile tool for {Non-negative Matrix Factorization} in
  biology.
\newblock {\em BMC Bioinformatics}, 7:366, 2006.

\bibitem{TiradoMontano}
A.~Pascual-Montano, F.~Tirado, P.~Carmona-Saez, J.M. Carazo, and R.D.
  Pascual-Marqui.
\newblock Two-way clustering of gene expression profiles by {Sparse Matrix
  Factorization}.
\newblock In {\em CSBW '05: Proceedings of the 2005 IEEE Computational Systems
  Bioinformatics Conference - Workshops}, pages 103--104. IEEE Computer
  Society, 2005.

\bibitem{Pauca200629}
V.P. Pauca, J.~Piper, and R.J. Plemmons.
\newblock {Nonnegative Matrix Factorization} for spectral data analysis.
\newblock {\em Linear Algebra and its Applications}, 416:29--47, 2006.

\bibitem{PaucaSBP04}
V.P. Pauca, F.~Shahnaz, M.W. Berry, and R.J. Plemmons.
\newblock Text mining using {Non-Negative Matrix Factorizations}.
\newblock In {\em SDM}, 2004.

\bibitem{Perou1999}
C.M. Perou, S.S. Jeffrey, M.~{van de Rijn}, C.A. Rees, M.B. Eisen, D.T. Ross,
  A.~Pergamenschikov, C.F. Williams, S.X. Zhu, J.C.F. Lee, D.~Lashkari,
  D.~Shalon, P.O. Brown, and D.~Botstein.
\newblock {Distinctive gene expression patterns in human mammary epithelial
  cells and breast cancers}.
\newblock {\em Proceedings of the National Academy of Sciences of the United
  States of America}, 96:9212--9217, 1999.

\bibitem{Perrin2003}
B.-E. Perrin, L.~Ralaivola, A.~Mazurie, S.~Bottani, J.~Mallet, and
  F.~{d'Alch\'{e}}$-$Buc.
\newblock {Gene networks inference using dynamic Bayesian networks}.
\newblock {\em Bioinformatics}, 19:ii138--ii148, 2003.

\bibitem{Piper04}
J.~Piper, V.P. Pauca, R.J. Plemmons, and M.~Giffin.
\newblock Object characterization from spectral data using {Nonnegative
  Factorization and Information Theory}.
\newblock In {\em Proceedings Amos Technical Conf.}, 2004.

\bibitem{Pollack99}
J.R. Pollack, C.M. Perou, A.A. Alizadeh, M.B. Eisen, A.~Pergamenschikov amd
  C.F.~Williams, S.S. Jeffrey, D.~Botstein, and P.O. Brown.
\newblock {Genome-wide analysis of DNA copy-number changes using cDNA
  microarrays}.
\newblock {\em Nature Genetics}, 23:41--46, 1999.

\bibitem{Priness07}
I.~Priness, O.~Maimon, and I.~Ben-Gal.
\newblock Evaluation of gene-expression clustering via {Mutual Information}
  distance measure.
\newblock {\em BMC Bioinformatics}, 8:111, 2007.

\bibitem{Ramaswamy2001}
S.~Ramaswamy, P.~Tamayo, R.~Rifkin, S.~Mukherjeen, C.-H. Yeang, M.~Angelo,
  C.~Ladd, M.~Reich, E.~Latulippe, J.P. Mesirov, T.~Poggio, W.~Gerald, M.~Loda,
  E.S. Lander, and T.R. Golub.
\newblock {Multiclass cancer diagnosis using tumor gene expression signatures}.
\newblock {\em Proceedings of the National Academy of Sciences of the United
  States of America}, 98:15149--15154, 2001.

\bibitem{BayesianCluster}
M.F. Ramoni, P.~Sebastiani, and I.S. Kohane.
\newblock {Cluster analysis of gene expression dynamics}.
\newblock {\em Proceedings of the National Academy of Sciences of the United
  States of America}, 99:9121--9126, 2002.

\bibitem{rand_ref}
W.M. Rand.
\newblock Objective criteria for the evaluation of clustering methods.
\newblock {\em Journal of the American Statistical Association}, 66:846--850,
  1971.

\bibitem{Raviv96bootstrappingwith}
Y.~Raviv and N.~Intrator.
\newblock Bootstrapping with noise: An effective regularization technique.
\newblock {\em Connection Science}, 8:355--372, 1996.

\bibitem{Rice}
J.A. Rice.
\newblock {\em Mathematical Statistics and Data Analysis}.
\newblock Wadsworth, 1996.

\bibitem{CVRijs}
C.~Van Rijsbergen.
\newblock {\em Information Retrieval, second edition.}
\newblock Butterworths, London, 1979.

\bibitem{Ross00}
D.T. Ross, U.~Scherf, M.B. Eisen, C.M. Perou, P.~Spellman, V.~Iyer, S.S.
  Jeffrey, M.~{van de Rijn}, M.~Walthama, A.~Pergamenschikov, J.C.F. Lee,
  D.~Lashkari, D.~Shalon, T.G. Myers, J.N. Weistein, D.~Botstein, and P.O.
  Brown.
\newblock Systematic variation in gene expression patterns in human cancer cell
  lines.
\newblock {\em Nature Genetics}, 24:227--235, 2000.

\bibitem{Roth02aresampling}
V.~Roth, T.~Lange, M.~Braun, and J.~Buhmann.
\newblock A resampling approach to cluster validation.
\newblock In {\em Proceedings 15th Symposium in Computational Statistics},
  pages 123--128, 2002.

\bibitem{Sachs02}
K.~Sachs, D.~Gifford, T.~Jaakkola, P.~Sorger, and D.A. Lauffenburger.
\newblock Bayesian network approach to cell signaling pathway modeling.
\newblock {\em Science's STKE : signal transduction knowledge environment},
  2002, 2002.

\bibitem{Sarle1983}
W.S. Sarle.
\newblock {Cubic clustering criterion}.
\newblock Technical report, SAS, 1983.

\bibitem{Optimal_Hier}
S.~Seal, S.~Comarina, and S.~Aluru.
\newblock An optimal hierarchical clustering algorithm for gene expression
  data.
\newblock {\em Information Processing Letters}, 93:143--147, 2004.

\bibitem{ShahnazBPP06}
F.~Shahnaz, M.W. Berry, V.P. Pauca, and R.J. Plemmons.
\newblock Document clustering using {Nonnegative Matrix Factorization}.
\newblock {\em Information Process. Manage.}, 42:373--386, 2006.

\bibitem{Roded4}
R.~Shamir and R.~Sharan.
\newblock Algorithmic approaches to clustering gene expression data.
\newblock In T.~Jiang, T.~Smith, Y.~Xu, and M.~Q. Zhang, editors, {\em Current
  Topics in Computational Biology}, pages 120--161. {MIT} Press, Cambridge,
  Ma., 2003.

\bibitem{Shmulevich2002}
I.~Shmulevich, E.R. Dougherty, S.~Kim, and W.~Zhang.
\newblock {Probabilistic Boolean networks: a rule-based uncertainty model for
  gene regulatory networks}.
\newblock {\em Bioinformatics}, 18:261--274, 2002.

\bibitem{Silvescu01temporalboolean}
A.~Silvescu and V.~Honavar.
\newblock Temporal boolean network models of genetic networks and their
  inference from gene expression time series.
\newblock {\em Complex Systems}, 13:61--78, 2001.

\bibitem{Smolen2000}
P.~Smolen, D.A. Baxter, and J.H. Byrne.
\newblock Modeling transcriptional control in gene networks--methods, recent
  results, and future directions.
\newblock {\em Bulletin of mathematical biology}, 62:247--292, 2000.

\bibitem{SmolkinGhosh}
M.~Smolkin and D.~Ghosh.
\newblock Cluster stability scores for microarray data in cancer studies.
\newblock {\em BMC Bioinformatics}, 4, 2003.

\bibitem{Snyman}
J.A. Snyman.
\newblock Practical mathematical optimization: An introduction to basic
  optimization theory and classical and new gradient-based algorithms.
\newblock {\em Structural and Multidisciplinary Optimization}, 31:249, 2006.

\bibitem{Speed03}
T.P. Speed.
\newblock {\em Statistical analysis of gene expression microarray data}.
\newblock Chapman \& Hall/CRC, 2003.

\bibitem{Spell98}
P.T. Spellman, G.~Sherlock, M.Q. Zhang, V.~R. Iyer, K.~Anders, M.B. Eisen, P.O.
  Brown, D.~Botstein, and B.~Futcher.
\newblock Comprehensive identification of cell cycle regulated genes of the
  yeast {S}accharomyces {C}erevisiae by microarray hybridization.
\newblock {\em Mol. Biol. Cell}, 9:3273--3297, 1998.

\bibitem{Strauss}
R.E. Strauss.
\newblock {Statistical significance of species clusters in association
  analysis}.
\newblock {\em Ecology}, 63:634--639, 1978.

\bibitem{Su2002}
A.I. Su, M.P. Cooke, K.A. Ching, Y.~Hakak, J.R. Walker, T.~Wiltshire, A.P.
  Orth, R.G. Vega, L.M. Sapinoso, A.~Moqrich, A.~Patapoutian, G.M. Hampton,
  P.G. Schultz, and J.B. Hogenesch.
\newblock {Large-scale analysis of the human and mouse transcriptomes}.
\newblock {\em Proceedings of the National Academy of Sciences of the United
  States of America}, 99:4465--4470, 2002.

\bibitem{SwagatamSambarta}
D.~Swagatam, D.~Sambarta, B.~Arijit, A.~Ajith, and K.~Amit.
\newblock On stability of the chemotactic dynamics in bacterial-foraging
  optimization algorithm.
\newblock {\em Trans. Sys. Man Cyber. Part A}, 39:670--679, 2009.

\bibitem{Tamayo2007}
P.~Tamayo, D.~Scanfeld, B.L. Ebert, M.A. Gillette, C.W.M. Roberts, and J.P.
  Mesirov.
\newblock Metagene projection for cross-platform, cross-species
  characterization of global transcriptional states.
\newblock {\em Proceedings of the National Academy of Sciences of the United
  States of America}, 104:5959--5964, 2007.

\bibitem{Tibshirani}
R.~Tibshirani, G.~Walther, and T.~Hastie.
\newblock Estimating the number of clusters in a dataset via the gap
  statistics.
\newblock {\em Journal Royal Statistical Society B.}, 2:411--423, 2001.

\bibitem{TrefethenBau}
L.N. Trefethen and D.~Bau.
\newblock {\em Numerical Linear Algebra}.
\newblock SIAM: Society for Industrial and Applied Mathematics, 1997.

\bibitem{Tropp}
J.A. Tropp.
\newblock Literature survey: {Non-Negative Matrix Factorization}.
\newblock Available at:
  {http://citeseerx.ist.psu.edu/viewdoc/summary?doi=10.1.1.84.9645}.

\bibitem{mosclust}
G.~Valentini.
\newblock {Mosclust: a software library for discovering significant structures
  in bio-molecular data}.
\newblock {\em Bioinformatics}, 23:387--389, 2007.

\bibitem{Vassiliou1989}
A.~Vassiliou, L.~Ignatiades, and M.~Karydis.
\newblock {Clustering of transect phytoplankton collections with a quick
  randomization algorithm}.
\newblock {\em Journal of experimental marine biology and ecology},
  130:135--145, 1989.

\bibitem{WangKO06}
G.~Wang, A.V. Kossenkov, and M.F. Ochs.
\newblock {LS-NMF}: A modified non-negative matrix factorization algorithm
  utilizing uncertainty estimates.
\newblock {\em BMC Bioinformatics}, 7:175, 2006.

\bibitem{WangZL05}
K.~Wang, N.~Zheng, and W.~Liu.
\newblock Natural image matting with {Non-negative Matrix Factorization}.
\newblock In {\em ICIP (2)}, pages 1186--1189, 2005.

\bibitem{WangJiaHuTurk}
Y.~Wang, Y.~Jia, C.~Hu, and M.~Turk.
\newblock Fisher {Non-negative Matrix Factorization} for learning local
  features.
\newblock In {\em Asian Conference on Computer Vision}, 2004.

\bibitem{CNSRat}
X.~Wen, S.~Fuhrman, G.S. Michaels, G.S. Carr, D.B. Smith, J.L. Barker, and
  R.~Somogyi.
\newblock Large scale temporal gene expression mapping of central nervous
  system development.
\newblock {\em Proceedings of The National Academy of Science USA},
  95:334--339, 1998.

\bibitem{Wild2003}
S.~Wild.
\newblock {\em {Seeding Non-negative Matrix Factorizations with the spherical
  K-Means clustering}}.
\newblock PhD thesis, University of Colorado, 2003.

\bibitem{Witten00}
I.H. Witten.
\newblock {\em Data Mining: Practical Machine Learning Tools and Techniques
  with Java Implementations}.
\newblock Academic Press, San Diego, CA,, 2000.

\bibitem{wolfinger01}
R.D. Wolfinger, G.~Gibson, E.D. Wolfinger, L.~Bennet, H.~Hamadeh, C.A. Bushel,
  and R.S. Paules.
\newblock Assessing gene significance from c{DNA} microarray expression data
  via mixed models.
\newblock {\em Journal of Computational Biology}, pages 625--637, 2001.

\bibitem{XuLiu}
W.~Xu, X.~Liu, and Y.~Gong.
\newblock Document clustering based on {Non-negative Matrix Factorization}.
\newblock In {\em SIGIR '03: Proceedings of the 26th annual international ACM
  SIGIR conference on Research and development in informaion retrieval}, pages
  267--273. ACM, 2003.

\bibitem{Yamagishi}
J.~Yamagishi, H.~Kawai, and T.~Kobayashi.
\newblock Phone duration modeling using gradient tree boosting.
\newblock {\em Speech Commun.}, 50:405--415, 2008.

\bibitem{Yan07}
M.~Yan and K.~Ye.
\newblock Determining the number of clusters with the weighted {Gap
  Statistics}.
\newblock {\em Biometrics}, 63:1031--1037, 2007.

\bibitem{Yeoh2002}
E.-J. Yeoh, M.E. Ross, S.A. Shurtleff, W.K. Williams, D.~Patel, R.~Mahfouz,
  F.G. Behm, S.C. Raimondi, M.V. Relling, A.~Patel, C.~Cheng, D.~Campana,
  D.~Wilkins, X.~Zhou, J.~Li, H.~Liu, C.-H. Pui, W.E. Evans, C.~Naeve, L.~Wong,
  and J.R. Downing.
\newblock Classification, subtype discovery, and prediction of outcome in
  pediatric acute lymphoblastic leukemia by gene expression profiling.
\newblock {\em Cancer Cell}, 1:133--143, 2002.

\bibitem{KaYeeDiss}
K.Y. Yeung.
\newblock {\em Cluster Analysis of Gene Expression Data}.
\newblock PhD thesis, University of Washington, 2001.

\bibitem{KaYeeFOM}
K.Y. Yeung, D.R. Haynor, and W.L. Ruzzo.
\newblock Validating clustering for gene expression data.
\newblock {\em Bioinformatics}, 17:309--318, 2001.

\end{thebibliography}
